\documentclass{article}
\usepackage[letterpaper, left=1in, right=1in, top=1in,bottom=1in]{geometry}

\usepackage{graphicx}       
\usepackage{amsmath}        
\usepackage[T1]{fontenc}        
\usepackage{url}              
\usepackage{amsfonts}       
\usepackage{nicefrac}       
\usepackage{microtype}      
\usepackage{lipsum}         
\usepackage{fancyhdr}       
\usepackage{graphicx}       
\graphicspath{{media/}}     
\usepackage{subfigure}      
\usepackage{enumitem}       
\usepackage{booktabs}       
\usepackage{mathrsfs}       
\usepackage{tikz}           
\usepackage{cancel}         
\usepackage{amssymb}        
\usepackage{mathtools}   
\usepackage{amsthm}         
\usepackage{algorithm}
\usepackage{algpseudocode}
\usepackage{natbib}

\usepackage{auxiliary}      

\usepackage{color}        

\usepackage[normalem]{ulem} 

\usepackage{hyperref}       
\usepackage[capitalise]{cleveref}   

\title{Optimal Exploration of New Products under Assortment Decisions}
\author{
Jackie Baek\thanks{Stern School of Business, New York University, \texttt{baek@stern.nyu.edu}} \and 
Atanas Dinev\thanks{Massachusetts Institute of Technology, \texttt{adinev@mit.edu}} \and Thodoris Lykouris\thanks{Massachusetts Institute of Technology, \texttt{lykouris@mit.edu}}
}
\date{}

\begin{document}

\maketitle

\begin{abstract}
We study online learning for new products on a platform that makes capacity-constrained assortment decisions on which products to offer. For a newly listed product, its quality is initially unknown, and quality information propagates through social learning: when a customer purchases a new product and leaves a review, its quality is revealed to both the platform and future customers. Since reviews require purchases, the platform must feature new products in the assortment (``explore'') to generate reviews to learn about new products. Such exploration is costly because customer demand for new products is lower than for incumbent products. We characterize the optimal assortments for exploration to minimize regret, addressing two questions. 
\emph{(1)} Should the platform offer a new product alone or alongside incumbent products?
The former maximizes the purchase probability of the new product but yields lower short-term revenue. Despite the lower purchase probability, we show it is always optimal to pair the new product with the top incumbent products.
\emph{(2)} With multiple new products, should the platform explore them simultaneously or one at a time? We show that the optimal number of new products to explore simultaneously has a simple threshold structure: it increases with the ``potential'' of the new products and, surprisingly, does not depend on their individual purchase probabilities. 
We also show that two canonical bandit algorithms, UCB and Thompson Sampling, both fail in this setting for opposite reasons: UCB over-explores while Thompson Sampling under-explores. Our results provide structural insights on how platforms should learn about new products through assortment decisions.
\end{abstract}

\section{Introduction}\label{sec:intro}

Online marketplaces such as Amazon, Airbnb, and Etsy continuously expand their offerings by adding new listings each day, yet every new item arrives with no transaction history. Customers rely heavily on reviews and ratings to inform their purchase decisions, making a product's review history a key driver of demand \citep{chevalier2006effect,zhu2010impact}. Without any reviews, neither customers nor the platform can assess a new product's quality, keeping potentially high-quality products hidden beneath incumbent best-sellers. We study a setting where quality information propagates through \textit{social learning}: early purchases generate reviews that reveal quality to future customers and the platform \citep{crapis2017monopoly, besbes2018information,ifrach2019bayesian,chen2021reviews,acemoglu2022learning,guo2024leveraging,baek2024social}.

We study this problem through the lens of \textit{capacitated assortment} decisions. Because online marketplaces host millions of products, only a small subset can be shown to a customer at any time. For example, upon a keyword search for ``sweater,'' the platform must choose which among thousands of results appear on the first page, where nearly all customer attention is concentrated. Since social learning requires purchases, and purchases require exposure, the platform's assortment decisions directly govern when and how quickly quality information is revealed. We refer to the act of deliberately featuring a new product in the assortment as \textit{exploration}. Exploration is costly: because customers prefer established products with known appeal over new products of uncertain quality \citep{muthukrishnan2009ambiguity}, featuring new product displaces known products that would have generated reliable revenue. Yet exploration is necessary to learn: a new product may be of higher quality than existing products, but learning this requires purchases, which cannot happen without exposure. This paper studies how a platform should design exploration policies that balance these competing objectives.

We consider a platform that makes sequential assortment decisions over two types of products: \textit{incumbents} (existing products with an established track record) and \textit{entrants} (newly listed products). We model customer choice using the multinomial logit (MNL) framework, a standard in revenue management, where each product carries an \textit{attraction parameter} reflecting its quality and appeal to customers. For incumbents, this parameter is known to both the platform and customers. For entrants, the parameter is unknown to both, but a prior distribution over it is common knowledge.  Information about an entrant's attraction parameter is gained only when a customer purchases the product and leaves a review (we assume all purchases lead to reviews).
Thus, to learn about the entrant, the platform must include it in the assortment, and it must be purchased by a customer. Since the assortment is capacity-constrained, including an entrant displaces an incumbent and hence decreases short-term revenue, creating the central trade-off of our problem.

A key feature of our model is that customer behavior is governed by the current state of reviews. The platform knows the choice model at any point in time; what is unknown (to both the platform and customers) is the \textit{true} attraction parameter of each entrant. 
Learning this parameter only happens when a customer purchases the entrant and leaves a review. Observed choices that do not involve an entrant purchase carry no new information. This distinguishes our setting from MNL bandits, where the choice model itself is unknown and is learned from every customer choice.

The central question this paper addresses is \textit{how} to explore efficiently. This is in contrast to the classical exploration-exploitation tradeoff in multi-armed bandits, which centers on \textit{whether} to explore. In our setting, exploration is necessary: the platform must include every entrant in the assortment to learn about it. The non-trivial question is how to structure exploration assortments given the combinatorial decision space. We address this through two concrete questions:

\begin{enumerate}
  \item \emph{Should the platform offer an entrant product alone, or offer it together with incumbents?}
  Offering an assortment with \textit{only} an entrant product maximizes its purchase probability and thus minimizes the time until a review arrives. However, every round spent showing only the entrant sacrifices the revenue that incumbents would have generated. Pairing the new product with strong incumbents raises per-round revenue but lowers the new product's purchase probability through substitution, extending the time until the new product is purchased.

  \item \emph{With multiple entrants, is it better to explore them simultaneously or one at a time?}
  Including many entrants speeds up the rate at which some entrant is purchased, but each individual entrant's purchase probability drops, and the per-round revenue may suffer from displacing too many incumbents.
\end{enumerate}

\subsection{Our contribution}

\paragraph{Model.} 
The platform decides on an assortment of at most $c$ items to offer at each time step, and customers choose among these items according to a multinomial logit (MNL) model. The attraction parameter $w_i$ for product $i$ is known for incumbent products and unknown for entrants.
For entrants, there is a known prior distribution $\mathcal{F}$ from which their attraction parameters are drawn. When a customer purchases an entrant, they leave a review for it\footnote{Our results easily extend if customers leave a review with some probability.}, which reveals its true attraction parameter $w_i$ to future customers. Prior to a purchase, the entrant has a cold-start attraction parameter $\tildew$. This paper focuses on the noiseless setting where a single purchase fully reveals the true attraction parameter (we discuss this assumption below). We study an infinite-horizon setting and measure performance by regret, defined as the revenue loss relative to the revenue from the optimal assortment when all attraction parameters are known. The full model is formalized in Section~\ref{sec:model}.

\paragraph{Our results.} 
Under this model, we characterize the optimal exploration policy that minimizes total regret, when all items have the same revenue. In particular, our results are summarized below.
\begin{enumerate}
  \item \textbf{Single entrant}: With one entrant and multiple incumbents (Section~\ref{sec:exploring_single_entrant}), the optimal assortment combines the entrant with the $c-1$ most attractive incumbents. This is optimal for any prior distribution $\mathcal{F}$ where there is a non-zero probability that the optimal assortment contains the entrant (irrespective of the incumbent or nominal attraction parameters). This result offers a clear, easy-to-implement guideline for platforms: include top incumbents in the assortment along with the entrant.
  \item \textbf{Multiple entrants}: With multiple entrants (Section~\ref{sec:optimal_exploration_multiple_entrants}), the optimal assortment displays $\ell$ of them together with $c-\ell$  of the top incumbents. We derive a closed-form expression for the optimal $\ell$, which increases when entrants are more likely to be valuable compared to  incumbents (based on the prior $\mathcal{F}$ and the weights $w_i$). That is, when entrants have high potential, then exploring with more entrants simultaneously is beneficial so that the platform can discover a high-value entrant   sooner. 

  \item \textbf{Canonical bandit algorithms fail}: 
  In Section~\ref{sec:comparison_to_simple_strategies}, we show that two widely used bandit algorithms, Upper Confidence Bounds (UCB) and Thompson Sampling  (TS), can incur regret that is arbitrarily larger than the optimal exploration policy. Interestingly, we show that UCB and TS fail due to opposite reasons: UCB tends to over-explore (include too many entrants), while TS tends to under-explore (include too few entrants). Although our model has resemblance to classical bandit models, this negative result shows that canonical tools for the classical setting fundamentally fail. We also show that simple heuristic policies such as ``explore as many entrants as possible'' or ``explore one entrant at a time'' also fail for similar reasons to UCB and TS respectively.
  
  \item \textbf{Extensions.} Finally, in Section~\ref{sec:extensions}, we provide extensions of the basic model that allow for heterogeneous incumbent rewards, heterogeneous entrant priors (albeit only for two entrants), and noisy observations upon purchase (albeit only for a single entrant). Extending those results to more general settings is an interesting open direction that seems to require new techniques (as we discuss in Section~\ref{sec:conclusion}).
\end{enumerate}

Together, these contributions provide structural insights that translate into prescriptive guidance for platforms. 
Our results show that optimal exploration strategies have clean characterizations which scale naturally from single to multiple entrants, and cannot be replaced by off-the-shelf bandit algorithms.

\paragraph{Main modeling assumptions.} We make two crucial simplifying modeling assumptions: \emph{undiscounted infinite horizon}, \emph{homogeneous entrant priors}, and \emph{noiseless learning} (a single review reveals the attraction parameter exactly). Under a finite horizon, it may be optimal not to explore at all due to a short horizon, which reintroduces the question of \emph{whether} to explore. 
Therefore, the infinite horizon assumption ensures that exploration is always necessary, allowing us to isolate the problem of designing optimal exploration policies. 
The  latter two assumptions are related and are made for simplicity, as the problem is interesting and non-trivial even with these assumptions.
Relaxing the setting to noisy learning, where each review reveals a noisy signal of the attraction parameter, would require using a finite horizon (as the regret of all policies  will be infinite under an infinite horizon), as well as handling heterogeneous entrant priors.
We make progress in relaxing this assumption in Sections \ref{sec:heterogeneous_priors} and \ref{sec:single_entrant_noisy_obs}, and we further discuss this in future directions in Section~\ref{sec:conclusion}.

\subsection{Related work} \label{sec:related_work}

\paragraph{Multi-armed bandits.}
Our problem is related to, but distinct from, the classical multi‑armed bandit literature \citep{lai1985asymptotically}. In canonical bandit problems, the fundamental question is deciding whether to explore or to exploit, and algorithms such as UCB~\citep{auer2002finite} or Thompson sampling~\citep{thompson1933likelihood} achieve this trade-off optimally in many bandit environments. By contrast, as described above, this work focuses on a different question of how to explore efficiently. Even when the platform decides to explore a new product, the optimal way to explore it remains non-obvious. Our work focuses on the design of optimal exploration strategies rather than the classical exploration–exploitation trade-off.

\paragraph{MNL-bandits.}
Our model resembles the line of work that studies assortment decisions under the MNL-bandit model, where customer choices follow an unknown multinomial logit model. Explore-then-exploit policies have been analyzed in this setting \citep{rusmevichientong2010dynamic,saure2013optimal}, and the same holds for UCB-based \citep{agrawal2019mnl} and TS-based~\citep{agrawal2017thompson} approaches, with lower bounds established in \cite{chen2018note}. Contextual extensions include UCB and TS policies \citep{chen2020dynamic,oh2019thompson,oh2021multinomial} with optimal regret bounds\citep{lee2024nearly,lee2025improved} and in the adversarial setting \citep{perivier2022dynamic}. \cite{dong2023pasta} and \cite{han2025learning} study learning optimal assortments in an offline setting. 

A key assumption in this line of work is that customers behave according to a \textit{fixed} choice model, but this choice model is \textit{unknown} to the platform. The analogous interpretation of MNL-bandits in our setting would be that the \textit{customers know the true parameters} of all products and act based on this knowledge, while the platform does not. In contrast, we study a setting where, when a new product appears, \textit{neither the platform nor the customers know its true parameter}, and all knowledge is common across the platform and customers. Therefore, unlike in MNL-bandits, the current choice model is always known to the platform, but this choice model can change over time as the quality of an entrant is revealed through customer reviews. This distinction separates our framework from the standard MNL-bandit formulation. 

\paragraph{Assortment optimization under MNL without learning.}
Assortment optimization is a foundational topic in operations and revenue management, where a seller selects a subset of products to maximize expected revenue under substitution effects modeled by discrete choice. The multinomial logit model, originating from \cite{luce1959individual} and \cite{mcfadden1972conditional}, has become the standard due to its tractability and interpretability. Under MNL, the unconstrained assortment problem admits a simple structure: the optimal assortment is revenue-ordered, containing all products whose prices exceed a threshold \citep{talluri2004revenue}.
This result allows polynomial-time algorithms and linear-programming formulations \citep{gallego2015general}.

Subsequent work generalized this classical setting to include cardinality and capacity constraints \citep{rusmevichientong2010dynamic,rusmevichientong2014assortment,sumida2021revenue,desir2022capacitated}. Our work considers cardinality constraints, motivated by a limited number of products that can be shown on a webpage. We assume that the revenue of all products are equal; therefore, if all parameters were known, the optimal assortment is to simply offer the products with the highest attraction parameters. We leave as future direction to consider the setting where revenues differ. 

There is also recent work that studies assortment decisions under visibility or fairness constraints, where there are constraints to ensure a minimum requirement for the offering or the market share of each product \citep{chen2022fair,lu2023simple,barre2024assortment,housni2024assortment,zhu2025unified}. One motivation of such constraints is to provide visibility to newer or less established products; our work shares this motivation but differs by explicitly modeling the learning dynamics. \cite{besbes2024fault} study repeated assortment recommendations where a platform optimizes measurable engagement rather than utilities, and they show this can systematically under-recommend niche items even when those items deliver substantial utility to a minority of users.

\paragraph{Endogenous behavior and social learning.}
Our setting is motivated by social learning in marketplaces where information diffuses through reviews. Prior works show how such dynamics can shape demand and platform decisions \citep{crapis2017monopoly, besbes2018information,ifrach2019bayesian,chen2021reviews,acemoglu2022learning,guo2024leveraging,baek2024social}. 
\cite{bai2025endogenous} develop a choice model with a reference effect driven by customer feedback; like ours, only purchases provide feedback, but they focus on optimization under fixed parameters rather than learning dynamics.
Endogenous dynamics have also been studied via reference-price effects in dynamic pricing \citep{popescu2007dynamic,chen2017efficient,den2022dynamic} and within-assortment reference effects \citep{bai2023assortment}.  There is also a line of work in mechanism design where customers do not know their own value and learn them over time \citep{weed2016online,feng2018learning,papadimitriou2016complexity,chawla2016simple,kandasamy2020vcg}.

\section{Model}\label{sec:model}
Consider a set $\mathcal{N}=\{1,\ldots,n\}$ of $n$ products, where the first $m \leq n$ products are \textit{entrants}, and the rest are \textit{incumbents}. Each product $i \in \mathcal{N}$ has an attraction parameter $w_i$, and we denote by 0 the outside option which has an attraction parameter $w_0 = 1$. The attraction parameter of each incumbent is known, while the attraction parameter of each entrant is initially unknown. For each entrant, their attraction parameter is independently and identically distributed (i.i.d.) from a known prior $\mathcal{F}$, i.e., $w_i \sim_{i.i.d} \mathcal{F}$ for $i =1 , \ldots, m$. Let the support of $\mathcal{F}$ be $[\underline{\theta}, \overline{\theta}]$. Every product has the same reward of 1.

The attraction parameter of an entrant is revealed after the product is purchased. As this information is gained over time, the set of products for which the attraction parameters are known expands. Let $\mathcal{I}_t$ be the set of products whose attraction parameters are known before round $t$. Initially, $\mathcal{I}_{1} = \{m+1, \ldots,n\}$. Letting $\mathcal{H}_t = \{(S_{s}, Y_{s}, w_{Y_{s}})\}_{s=1}^{t-1}$ be the history of all information observed up to time $t$, the following occurs at each time $t = 1, \dots, T$: 
\begin{enumerate}
    \item The platform offers a set of products $S_t$ of size at most $c$ based on a policy $\pi$ that maps the history $\mathcal{H}_t$ to a (possibly randomized) assortment $S_t$. (For simplicity, we assume at least $c$ incumbents.\footnote{This can be achieved by introducing dummy incumbents with attraction parameter equal to $0$. This assumption allows a revenue-maximizing assortment that is restricted to known products to be always well defined, which simplifies exposition.}). 
    
    \item The customer chooses a product from $S_t$ or the outside option according to the choice model described below. Let $Y_t \in S_t \cup \{0\}$ denote their choice.
    
    \item The platform obtains a reward of $R_t = 1$ if a product is chosen ($Y_t \in S_t$), and a reward of $R_t = 0$ otherwise ($Y_t = 0$). 
    
    \item If an unknown product is chosen ($Y_t \not \in \mathcal{I}_t \cup \{0\}$), 
    the attraction parameter $w_{Y_t}$ is observed and the set of known products $\mathcal{I}_t$ is updated, i.e., $\mathcal{I}_{t+1} = \mathcal{I}_t \cup \{Y_t\}$. Otherwise, the set of known products is not updated, i.e, $\mathcal{I}_{t+1} = \mathcal{I}_t$.
\end{enumerate}

Customers choose among the offered products according to a multinomial logit (MNL) model, where each product $i$ has attraction parameter $w_i(t)$ at time $t$.
For each known product $i \in \mathcal{I}_t$, $w_i(t) = w_i$, and for each unknown product $i \notin \mathcal{I}_t$, $w_i(t)$ is the same and equal to a common \emph{cold-start attraction parameter}
$\tildew$. We take $\tildew$ as a known primitive of the model; it can be thought of as a prior-based estimate of an entrant's attractiveness (e.g., the mean or a percentile of $\mathcal{F}$), though our results do not rely on any particular form. The weight of the outside option stays constant at $w_0(t) = w_0$. Specifically, if $S$ is offered at time $t$, the probability that item $i \in S$ is chosen is:
$$\prob[\text{$i \in S$ is chosen} \;|\; \text{$S$ is offered at round $t$}] = 
 \frac{w_{i}(t)}{\sum_{j \in S} w_{j}(t) + w_0(t)}.
$$
Correspondingly, the expected revenue of offering $S$ at round $t$ is 
\begin{equation}\label{notation:revenue_assortment}
    \textsc{Rev}_t(S) = \frac{\sum_{i \in S} w_i(t)}{\sum_{i \in S} w_i(t) + w_0(t)}.
\end{equation}
The quantities $(n, m, \{w_i\}_{i=m+1}^n, \mathcal{F}, c)$ define an instance of our problem.

\paragraph{Notation.} For a set of known products $\mathcal{I}$ and $k \leq |\mathcal{I}|$, let $w^{\star}_{(k)}(\mathcal{I})$ be the attraction parameter of the $k$-th most attractive product in $\mathcal{I}$, and let $W^{\star}_{(k)}(\mathcal{I}) = \sum_{i=1}^{k} w^{\star}_{(i)}(\mathcal{I})$ be the sum of the attraction parameters of the $k$ most attractive products in $\mathcal{I}$. Let $S^{\star}_{(k)}(\mathcal{I}) \in \argmax_{S \subseteq \mathcal{I},|S| \leq k} \sum_{i \in S} w_i$ be any set of $k$ most attractive products in $\mathcal{I}$. 

\paragraph{Objective.}
The total expected reward obtained by a policy $\pi$ for a horizon $T$ is $\expect\Big[\sum_{t=1}^{T} R_t \Big]$.
We compare this to a benchmark where the attraction parameters of all products are known, in which it is optimal to offer the $c$ products with the highest attraction parameters.  
Formally, we define the infinite-horizon regret (or regret as a shorthand) of a policy $\pi$ as 

\begin{equation*}\label{eq:regret_definition}
    \textsc{Reg}(\pi) = \liminf_{T \to \infty} \expect \Bigg[\sum_{t=1}^{T} \left( \textsc{OPT} -R_t \right) \Bigg] \quad \text{ where } \quad \textsc{OPT} = \frac{W^{\star}_{(c)}(\mathcal{N})}{W^{\star}_{(c)}(\mathcal{N})+w_0}.
\end{equation*}
 and  the expectation is over the realizations $w_i \sim_{i.i.d.} \mathcal{F}$ for $i = 1, \ldots, m$, the randomness in the customers' choices and the (potential) randomness in the platform's policy $\pi$. A policy $\pi$ is \textit{optimal} if it minimizes the infinite-horizon regret. Note that once the attraction parameters of all products are known, the per-round regret becomes zero. Therefore, the expected regret will be finite for any policy that ensures exploration of all unknown products and then offers the revenue-maximizing assortment. 

\section{Optimal exploration with a single entrant}\label{sec:exploring_single_entrant}
In this section, we consider a single entrant with multiple incumbents.
We consider instances where offering the myopically revenue-maximizing assortment is not the optimal policy. Specifically, we assume that the entrant is not already included in the myopic revenue-maximizing assortment, i.e., $\tildew <  w^{\star}_{(c)}(\mathcal{I}_1)$, and we assume that there is a positive probability that the entrant would be in the optimal assortment once its attraction parameter is known, i.e.,  $\overline{\theta} > w^{\star}_{(c)}(\mathcal{I}_1)$. Under these conditions, exploration (offering the entrant) is necessary but results in short-term revenue loss. We characterize the platform's optimal exploration policy.

Since it is necessary to offer the entrant until it is purchased and its attraction parameter is revealed, the main question becomes: which \emph{other} products should be offered?
On one extreme, the platform can offer the entrant on its own; this minimizes the number of rounds until a purchase of the entrant but also minimizes the per-round revenue. On the other extreme, the platform can offer the entrant with the $c-1$ most attractive incumbents; this maximizes the per-round revenue, but also maximizes the number of rounds until a purchase of the entrant. We show that the optimal policy is the latter. 

For a set $S \subseteq \mathcal{I}_1$ of incumbents, let $\pi(S)$ be the policy that offers the entrant along with $S$ until the entrant is purchased. We refer to the set of rounds up to and including the unknown entrant's purchase as the \emph{exploration period}. After the exploration period, $\pi(S)$ offers the $c$ products with the highest attraction parameters.
To ease notation, we omit the dependence on $\mathcal{I}_1$ and write
$S^\star_{(c-1)} = S^\star_{(c-1)}(\mathcal{I}_1)$, $w^{\star}_{(c)} = w^{\star}_{(c)}(\mathcal{I}_1)$.

\begin{theorem}\label{thm:explore_one_unknown_with_best}
    For any instance with one unknown entrant and $\overline{\theta} > w^{\star}_{(c)}$, 
    $\pi(S_{(c-1)}^\star)$ is an optimal policy.
    That is, it is optimal to explore the unknown entrant with the $c-1$ most attractive incumbents. 
\end{theorem}
To prove the theorem, the following lemma (proven in Appendix~\ref{appendix_sec:proof_lemma_there_exists_an_optimal_policy_which_offers_the_same_assortment_until_the_entrant_is_purchased}) shows that it suffices to consider policies which offer the entrant repeatedly with the same assortment during the exploration period and subsequently offer the revenue-maximizing assortment. Intuitively, if the entrant is not purchased, the set of known products remains the same and thus so does the consumer purchase behavior. As a result, it is optimal to offer the same assortment containing the entrant until its purchase. 

\begin{lemma}\label{lemma:there_exists_an_optimal_policy_which_offers_the_same_assortment_until_the_entrant_is_purchased}
     For any instance with one unknown entrant and $\overline{\theta} > w^{\star}_{(c)}$, there exists an optimal policy which offers the entrant with a fixed set of incumbents until the entrant is purchased and offers the $c$ products with the highest attraction parameters in subsequent rounds.
\end{lemma}

\begin{proof}[Proof of Theorem~\ref{thm:explore_one_unknown_with_best}]

Let $\alpha^{\star}= \expect_{w_1 \sim \mathcal{F}} [\textsc{OPT} ]$ be the expected ex-post optimum. 
For a policy $\pi(S)$, let $\tau(S)$ be the expected total number of rounds during the exploration period, and $o(S)$ be the expected number of rounds that the outside option is chosen during the exploration period.

During the exploration period of the policy $\pi(S)$, the regret is $\alpha^{\star}$ when the outside option is chosen, which happens $o(S)$ times in expectation.
When a product is purchased, the reward is 1 and hence the regret is $\alpha^{\star}-1$; this happens $\tau(S) - o(S)$ times in expectation.
Since $\pi(S)$ does not incur any regret after the exploration period, its total regret is
\begin{equation} \label{eq:exploration_regret}
\textsc{Reg}(\pi(S)) = (\alpha^{\star}-1) (\tau(S) - o(S)) + \alpha^{\star} o(S).
\end{equation}
Next, we show that $o(S)$ is independent of $S$. Using $\tau(S) = 1/\prob[\text{entrant is chosen}|S]$, we can write
\begin{equation*}
    o(S) = \prob[\text{outside option is chosen}|S] \tau(S) 
    = \frac{\prob[\text{outside option is chosen}|S]}{\prob[\text{entrant is chosen}|S]}= \frac{\frac{w_0}{\sum_{i \in S} w_i  + \tildew + w_0} }{\frac{\tildew}{\sum_{i \in S} w_i  + \tildew + w_0} }=
    \frac{w_0}{\tildew}.
\end{equation*}
We note the lack of dependence on $S$ follows from the independence of irrelevant alternatives property of the MNL choice model.
Combining this with \eqref{eq:exploration_regret} and the fact that $\alpha^{\star}< 1$ (as all products have a reward of $1$ and there exists an outside option), minimizing $\textsc{Reg}(\pi(S))$ is equivalent to maximizing $\tau(S)$. Using
$$\tau(S) = \frac{1}{\prob[\text{entrant is chosen}|S]} = \frac{\sum_{i \in S} w_i  + \tildew + w_0}{\tildew},$$
$\tau(S)$ is maximized when $\sum_{i \in S} w_i$ is maximized, which occurs when $S =  S^{\star}_{(c-1)}$. 
As a result, $\pi(S^{\star}_{(c-1)})$ yields lower regret than any other policy $\pi(S)$. Given that there exists an optimal policy which is $\pi(S)$ for some $S$ by Lemma~\ref{lemma:there_exists_an_optimal_policy_which_offers_the_same_assortment_until_the_entrant_is_purchased}, $\pi(S^{\star}_{(c-1)})$ is an optimal policy. 
\end{proof}

\section{Optimal exploration with multiple entrants}\label{sec:optimal_exploration_multiple_entrants}
In the previous section (where there is a single entrant), we focus on \emph{which incumbents} to include in the assortment, along with the entrant.
With multiple entrants, the question becomes \emph{how many entrants} to explore at once and which incumbents to include.
We provide an optimal algorithm for exploring with multiple entrants, described in Section~\ref{subsec:algorithm_main_theorem}.
The rest of this section proves the optimality of this algorithm.

\subsection{Optimal policy for exploring multiple entrants}\label{subsec:algorithm_main_theorem}

Recall that the algorithm in Section \ref{sec:exploring_single_entrant} includes the unknown entrant if it has a positive probability of being in the optimal assortment, and offers it the $c-1$ most attractive incumbents. To generalize this approach to multiple entrants, our algorithm, which we term \textsc{Exploration with Fictitious Assortments} or \textsc{EFA} for short (Algorithm~\ref{alg:optimal_explorer}) explores at least one unknown entrant if it has a positive probability of being in the optimal assortment. In that case, the algorithm explores $\ell$ unknown entrants,\footnote{Note that all unknown entrants have the same attraction parameter $\tildew$ prior to a purchase.} we denote their set by $U_{\ell} \subseteq \mathcal{N} \setminus \mathcal{I}_t$, together with the $c-\ell$ most attractive known products (line \ref{alg_line:offer_set} of Algorithm~\ref{alg:optimal_explorer}), where we find the optimal $\ell$ to minimize regret. 

Given a history of $\mathcal{H}_t =  \{(S_s, Y_s, w_{Y_s})\}_{s=1}^{t-1}$ of all information observed until round $t$, the expected ex-post optimum is\footnote{Note that for any algorithm which explores all unknown entrants eventually $(\mathcal{I}_t \to_{t \to \infty} \mathcal{N})$ and the expected ex-post optimum converges to the ex-post optimum, i.e., $\textsc{OPT}_t \to_{t \to \infty} \textsc{OPT}$. }
\begin{equation}\label{eq:a-posteriori_opt}
    \textsc{OPT}_t = \expect_{w_j \sim \mathcal{F}, j \not \in \mathcal{I}_t} \Big[ \textsc{OPT} \Big| \mathcal{H}_t\Big].
\end{equation}  
Our algorithm explores at least one unknown product if the expected ex-post optimum is higher than the revenue of the myopic assortment $S^{\star}_{(c)}(\mathcal{I}_t)$ (line~\ref{alg_line:worth_it_to_explore} of Algorithm~\ref{alg:optimal_explorer}). 

The crux of our algorithm is to optimize $\ell$, the number of unknown products to explore. To do so, a useful object is a collection of \textit{fictitious} assortments of known products denoted by $S_t(\ell) \subseteq \mathcal{I}_t$ for $\ell \in \{1, \ldots, c\}$. The fictitious assortment $S_t(\ell)$ contains the $c-\ell$ most attractive known products and $\ell$ replicas of the $(c-\ell+1)$-th most attractive known product. Note that $S_t(1)$ is the assortment of the $c$ most attractive known products (which corresponds to the myopic revenue-maximizing assortment), and $S_t(c)$ is the assortment with $c$ replicas of the most attractive known product.
We denote the revenue of $S_t(\ell)$ by 
\begin{equation}\label{eq:alpha_l_I_t}
    \alpha_{t}(\ell) = \textsc{Rev}_t(S_t(\ell)) =  \frac{W_{(c-\ell)}^{\star}(\mathcal{I}_t) + \ell  w_{(c-\ell+1)}^{\star}(\mathcal{I}_t)}{W_{(c-\ell)}^{\star}(\mathcal{I}_t) + \ell  w_{(c-\ell+1)}^{\star}(\mathcal{I}_t) + w_0}.
\end{equation} 
The revenue $\alpha_t(\ell)$ is non-decreasing as $\ell$ increases by the definition of the fictitious assortments. 
We illustrate the construction of the fictitious assortments in the following example. 

\begin{example}\label{ex:example_algorithm_step}
Consider an instance with $n = 9$ products, $m = 5$ unknown products, assortment capacity $c = 4$, and each unknown entrant's attraction parameter is $10$ with probability $0.1$ and $5$ otherwise. Suppose at round $t$ the known products are $\mathcal{I}_t = \{5, 6, 7, 8, 9\}$, and their attraction parameters are $w_5= 5, w_6 = 6, w_7 = 7, w_8 = 8, w_9 = 9$. Using \eqref{eq:a-posteriori_opt}, the expected ex-post optimum is $\textsc{OPT}_t = 0.969$. The fictitious assortments and their respective revenues are:
\begin{align*}
    S_t(1) &= \{9,8,7,6\}, \alpha_t(1) = 0.967 \\
    S_t(2) &= \{9,8, 7, 7\}, \alpha_t(2) = 0.968 \\
    S_t(3) &= \{9,8,8,8\}, \alpha_t(3) = 0.970 \\
    S_t(4) &= \{9,9,9,9\}, \alpha_t(4) = 0.972. 
\end{align*}
\end{example}

The number of unknown products to explore, $\ell_t$, is chosen by taking the largest number such that $\textsc{OPT}_t \geq \alpha_t(\ell_t)$ (line~\ref{alg_line:set_number_products} of Algorithm~\ref{alg:optimal_explorer}).
In Example~\ref{ex:example_algorithm_step}, $\ell_t = 2$ as $\alpha_t(2) < \textsc{OPT}_t < \alpha_t(3)$.
The algorithm then offers $\ell_t$ unknown products along with $c-\ell_t$ most attractive known products.

\begin{algorithm}[!htbp]
\caption{\textsc{Exploration with Fictitious Assortments} (\textsc{EFA})}\label{alg:optimal_explorer} 
\begin{algorithmic}[1]
\State \text{Initialize }$I_1 = \{m+1, \ldots, n\}$ and $m_1 = m$.
\State \textbf{for} $t = 0,1,2, \ldots$:
\State \hspace*{1em} Calculate $\textsc{Rev}_t(S^{\star}_{(c)}(\mathcal{I}_t))$, $\textsc{OPT}_t$, and $\alpha_{t}(\ell)$ for all $\ell = 1, \ldots, c$ using  \eqref{notation:revenue_assortment}, \eqref{eq:a-posteriori_opt}, and \eqref{eq:alpha_l_I_t}.  
\State \hspace*{1em} \textbf{if} $\textsc{OPT}_t > \textsc{Rev}_t(S^{\star}_{(c)}(\mathcal{I}_t)):$ \Comment{ exploration-worthy condition} \label{alg_line:worth_it_to_explore}
\State \hspace*{1em}\hspace*{1em}  Set $\ell_t =  \max\{ \ell \in \{1, \ldots, k_t\}: \textsc{OPT}_{t} \geq \alpha_{t}(\ell) \}$ where $k_t = \min(c, m_t)$. \label{alg_line:set_number_products}
\State \hspace*{1em}\hspace*{1em}  Offer set $S_t  =U_{\ell_t} \cup S^{\star}_{(c-\ell_t)}(\mathcal{I}_t)$;  observe choice $Y_t \in S_t \cup \{0\}$. \label{alg_line:offer_set}
\State \hspace*{1em}\hspace*{1em}  If $Y_t \in U_{\ell_t}$ then $\mathcal{I}_{t+1} = \mathcal{I}_t \cup \{Y_t\}$ and $m_{t+1} = m_t-1$; else $\mathcal{I}_{t+1} = \mathcal{I}_t $ and $m_{t+1} = m_t$. 
\State \hspace*{1em} \textbf{else}: Offer set $S_t = S^{\star}_{(c)}(\mathcal{I}_t)$; observe choice $Y_t \in S_t \cup \{0\}$; $\mathcal{I}_{t+1} = \mathcal{I}_t $ and $m_{t+1} = m_t$.  \label{alg_line:not_exploration_worthy_case}
\end{algorithmic}
\end{algorithm}

Our main result is that \textsc{EFA} (Algorithm~\ref{alg:optimal_explorer}) minimizes the expected infinite-horizon regret. 

\begin{theorem}\label{thm:optimality_of_algo}
    For any instance, \textsc{EFA} (Algorithm \ref{alg:optimal_explorer}) is optimal. 
\end{theorem}

\begin{remark}
    Calculating $\textsc{OPT}_t$ can be achieved by enumeration when the support of $\mathcal{F}$ is finite and by a Monte Carlo simulation if the support of $\mathcal{F}$ is infinite or continuous. 
\end{remark}

\subsection{Optimality of EFA (Proof of Theorem~\ref{thm:optimality_of_algo})}\label{subsec:proof_algorithm_optimality}
To prove optimality of \textsc{EFA}, we partition the time horizon into \textit{epochs}, where in each epoch, the set of unknown products stays the same. That is, a new epoch is triggered when an unknown product is purchased. We show that the regret can be decomposed by epoch, and we show that \textsc{EFA} chooses the assortment that minimizes the regret in each epoch.

Let $\mathcal{S}_t = \{S \subseteq \mathcal{N}, |S| \leq c, S \cap (\mathcal{N} \setminus \mathcal{I}_t) \neq \emptyset\}$ be the set of all assortments of size at most $c$ which contain at least one unknown entrant at round $t$. Let $\tau(S;\mathcal{H}_t)$ be the expected number of rounds and $r(S;\mathcal{H}_t)$ be the expected total reward in an epoch if the platform offers an assortment $S \in \mathcal{S}_t$ and the history at round $t$ is $\mathcal{H}_t$. A useful quantity in our analysis is the \emph{expected regret in an epoch}: $$\textsc{EpochReg}_t(S) = \textsc{OPT}_t \tau(S; \mathcal{H}_t) - r(S; \mathcal{H}_t).$$ 
The following lemma establishes that a policy which minimizes the expected regret in every epoch minimizes the infinite-horizon regret.

\begin{lemma}\label{lemma:epoch_optimal_policy_is_optimal}
A policy $\pi^{\star}$ is optimal if for any round $t$ and any history $\mathcal{H}_t$ it satisfies 
\begin{itemize}
    \item $\pi^{\star}(\mathcal{H}_t) \in \argmin\limits_{S \in \mathcal{S}_t} \textsc{EpochReg}_t(S)$ if $\textsc{OPT}_t > \textsc{Rev}_t(S^{\star}_{(c)}(\mathcal{I}_t))$, and 
    \item  $\pi^{\star}(\mathcal{H}_t) \in \argmax\limits_{S \subseteq \mathcal{I}_t, |S| \leq c} \textsc{Rev}_t(S)$ otherwise.
\end{itemize}
\end{lemma}

The proof of the lemma uses that all entrants' attraction parameters are drawn from the same prior and is provided in Appendix~\ref{appendix_sec:there_exists_optimal_policy_which_is_stationary}. 

Next, we fix the number $\ell$ of the unknown entrants to include in the assortment, and we characterize the set of known products to include as well.
Recall that in the single entrant setting from Section~\ref{sec:exploring_single_entrant}, it is optimal to explore the entrant with the $c-1$ most attractive products (Theorem~\ref{thm:explore_one_unknown_with_best}). The following lemma extends this result by showing it is optimal to explore the $\ell$ unknown entrants with the $c-\ell$ most attractive known products. Let $k_t = \min(c, m_t)$ be the minimum between the assortment capacity $c$ and the number of remaining unknown products $m_t$ (line \ref{alg_line:set_number_products}). 

\begin{lemma}\label{lemma:optimal_completion_given_ell_unknown_prods}
For any round $t$ and any $\ell \in \{1, \ldots, k_t\}$, the assortment $S$ that minimizes $\textsc{EpochReg}_t(S)$ subject to including exactly $\ell$ unknown products $U_{\ell}$ complements them with the set $S^{\star}_{(c-\ell)}(\mathcal{I}_t)$ of the $c-\ell$ most attractive known products, i.e., 
$S^{\star}_{(c-\ell)}(\mathcal{I}_t) \in \argmin\limits_{S' \subseteq \mathcal{I}_t, |S'| \leq c- \ell} \textsc{EpochReg}_t(U_{\ell} \cup S')$.
\end{lemma}

The proof follows a similar approach to that of Theorem~\ref{thm:explore_one_unknown_with_best} and is provided in Appendix~\ref{appendix_sec_proof_optimal_completion_given_ell_unknown_prods}. 
Our technical crux is characterizing how to select the optimal number of unknown products at round $t$. This is shown in the following lemma, whose proof we discuss in detail in Section~\ref{subsec:proof_lemma_characterization_time_diff_reward_diff}.

\begin{lemma}\label{lemma:num_unknown_ORFA_minimizes_epoch_reg}
For any round $t$, the number of unknown entrants $\ell_{t}$ selected by Algorithm~\ref{alg:optimal_explorer} minimizes $\textsc{EpochReg}_t( U_{\ell} \cup S^{\star}_{(c-\ell)}(\mathcal{I}_t))$ subject to $\ell \in \{1, \ldots, k_t\}$.
\end{lemma}

\begin{proof}[Proof of Theorem~\ref{thm:optimality_of_algo}]
By Lemma~\ref{lemma:epoch_optimal_policy_is_optimal}, a policy is optimal if it minimizes within-epoch regret when it is worth it to explore and maximizes revenue otherwise. EFA maximizes revenue when it is not worth it to explore (line \ref{alg_line:not_exploration_worthy_case}) and by Lemmas \ref{lemma:optimal_completion_given_ell_unknown_prods} and \ref{lemma:num_unknown_ORFA_minimizes_epoch_reg} it minimizes within-epoch regret when it is worth it to explore (line \ref{alg_line:worth_it_to_explore}). Thus, EFA is optimal. 
\end{proof}

\subsection{Optimal number of unknown entrants within an epoch (Lemma~\ref{lemma:num_unknown_ORFA_minimizes_epoch_reg})
}\label{subsec:proof_lemma_characterization_time_diff_reward_diff}
For a number of unknown entrants $\ell$, let $\tau^{\star}_t(\ell) = \tau(U_{\ell} \cup S^{\star}_{(c-\ell)}(\mathcal{I}_t); \mathcal{H}_t)$, $r^{\star}_t(\ell) = r(U_{\ell} \cup S^{\star}_{(c-\ell)}(\mathcal{I}_t); \mathcal{H}_t)$, and $\textsc{EpochReg}_t^{\star}(\ell) = \textsc{EpochReg}_t(U_{\ell} \cup S^{\star}_{c-\ell})$. We compare the regret of including $\ell$ unknown entrants to that of including $\ell+1$ unknown entrants. Including $\ell+1$ entrants yields a lower regret than including $\ell$ entrants if and only if the expected ex-post optimum is at least the ratio of the reward loss to the time gain:
\begin{align}\label{ineq:ell_plus_one_smaller_epoch_regret_thanell} \textsc{EpochReg}_t^{\star}(\ell) \geq \textsc{EpochReg}_t^{\star}(\ell+1) 
   &\iff \textsc{OPT}_t \tau^{\star}_t(\ell)-r^{\star}_t(\ell) \geq  \textsc{OPT}_t \tau^{\star}_t(\ell+1)-r^{\star}_t(\ell+1) \nonumber\\
    &\iff \textsc{OPT}_t  \geq \frac{ r^{\star}_t(\ell)-r^{\star}_t(\ell+1)}{ \tau^{\star}_t(\ell)-\tau^{\star}_t(\ell+1)}.
\end{align}
The following lemma (proven at the end of the section) shows that the ratio of the reward loss to the time gain for $\ell$ unknown entrants equals the revenue $\alpha_t(\ell+1)$ of the fictitious assortment $S_t(\ell+1)$. 

\begin{lemma}\label{lemma:expression_ratio_reward_loss_time_gain}
    For any round $t$ and any number of unknown entrants $\ell$, 
    $$\frac{ r^{\star}_t(\ell)-r^{\star}_t(\ell+1)}{ \tau^{\star}_t(\ell)-\tau^{\star}_t(\ell+1)} = \alpha_t(\ell+1) = \frac{W^{\star}_{(c-\ell-1)}(\mathcal{I}_t) + (\ell + 1) w^{\star}_{(c-\ell)}(\mathcal{I}_t)}{W^{\star}_{(c-\ell-1)}(\mathcal{I}_t) + (\ell + 1) w^{\star}_{(c-\ell)}(\mathcal{I}_t)+w_0}.$$
\end{lemma}

\begin{proof}[Proof of Lemma~\ref{lemma:num_unknown_ORFA_minimizes_epoch_reg}]
To prove the lemma it suffices to show that
 $$
\begin{aligned}
    \textsc{EpochReg}_t^{\star}(\ell) &\geq \textsc{EpochReg}_t^{\star}(\ell+1) \text{ for $\ell \in \{0, \ldots, \ell_{t}-1\}$, and }\\
   \textsc{EpochReg}_t^{\star}(\ell) &< \textsc{EpochReg}_t^{\star}(\ell+1) \text{ for $\ell \in \{\ell_{t}, \ldots, k_t-1\}$.} 
\end{aligned} (\Delta)
$$
Using \eqref{ineq:ell_plus_one_smaller_epoch_regret_thanell} $\textsc{EpochReg}_t^{\star}(\ell) \geq \textsc{EpochReg}_t^{\star}(\ell+1)$ is equivalent to $\textsc{OPT}_t \geq \frac{ r^{\star}_t(\ell)-r^{\star}_t(\ell+1)}{ \tau^{\star}_t(\ell)-\tau^{\star}_t(\ell+1)}$. Using Lemma~\ref{lemma:expression_ratio_reward_loss_time_gain} the right hand side of the former is $$
\frac{ r^{\star}_t(\ell)-r^{\star}_t(\ell+1)}{ \tau^{\star}_t(\ell)-\tau^{\star}_t(\ell+1)} = \alpha_{t}(\ell+1)=\frac{W^{\star}_{(c-\ell-1)}(\mathcal{I}_t) + (\ell + 1) w^{\star}_{(c-\ell)}(\mathcal{I}_t)}{W^{\star}_{(c-\ell-1)}(\mathcal{I}_t) + (\ell + 1) w^{\star}_{(c-\ell)}(\mathcal{I}_t) + w_0} .
$$
Thus, 
\begin{equation}\label{prop:cond_worth_it_to_add_unknown_prod}
    \textsc{EpochReg}_t^{\star}(\ell) \geq \textsc{EpochReg}_t^{\star}(\ell+1) \text{ if and only if }  \textsc{OPT}_t  \geq \alpha_{t}(\ell+1).
\end{equation}
Express $\alpha_t(\ell) = \frac{x_{t}(\ell)}{x_{t}(\ell) + w_0}$
where $x_{\ell} =  W^{\star}_{(c-\ell)}(\mathcal{I}_t) + \ell w^{\star}_{(c-\ell+1)}(\mathcal{I}_t)$. Observe that $0 \leq x_{\ell} \leq x_{\ell+1} $ as the difference $x_{\ell} - x_{\ell+1} $ is
\begin{equation*}
    W^{\star}_{(c-\ell)}(\mathcal{I}_t) + \ell w^{\star}_{(c-\ell+1)}(\mathcal{I}_t) - (W^{\star}_{(c-(\ell+1))}(\mathcal{I}_t) + (\ell+1) w^{\star}_{(c-(\ell+1)+1)}(\mathcal{I}_t)) = \ell (w^{\star}_{(c-\ell+1)}(\mathcal{I}_t)-w^{\star}_{(c-\ell)}(\mathcal{I}_t)) \leq 0.
\end{equation*}
Combining this with the fact that the function $x \to \frac{x}{x+w_0}$ is strictly increasing in $x$ (as $w_0 > 0$) yields
\begin{equation}\label{ineq:alpha_ell_non_decreasing}
    \alpha_t(\ell) \leq \alpha_{t}(\ell+1) \text{ for $\ell \in \{0,\ldots, k_t-1\}$}.
\end{equation}
Let $\ell \in \{0, \ldots, \ell_{t}-1\}$. Then $\ell + 1\leq \ell_{t}$. Combining this with \eqref{ineq:alpha_ell_non_decreasing} and the fact that $\alpha_t(\ell_{t}) \leq \textsc{OPT}_t$ (line \ref{alg_line:set_number_products}, Algorithm \ref{alg:optimal_explorer}) yields $\alpha_{t}(\ell+1) \leq \alpha_{t}(\ell_{t}) \leq \textsc{OPT}_t$. Thus, $\textsc{EpochReg}_t^{\star}(\ell) \geq \textsc{EpochReg}_t^{\star}(\ell+1)$ by \eqref{prop:cond_worth_it_to_add_unknown_prod}. Let $\ell \in \{\ell_{t}, \ldots, k_t-1\}$. Thus, $\ell +1 \geq \ell_{t} + 1$ and therefore $\alpha_t(\ell+1)> \textsc{OPT}_t$ by the definition of $\ell_{t}$ (line \ref{alg_line:set_number_products}, Algorithm \ref{alg:optimal_explorer}). By \eqref{prop:cond_worth_it_to_add_unknown_prod}, $\textsc{EpochReg}_t^{\star}(\ell) < \textsc{EpochReg}_t^{\star}(\ell+1)$, yielding $(\Delta)$. 
\end{proof}
    
\subsection{Technical crux: Characterizing the ratio of reward loss to time gain (Lemma~\ref{lemma:expression_ratio_reward_loss_time_gain})}

\begin{proof}[Proof of Lemma~\ref{lemma:expression_ratio_reward_loss_time_gain}]
For a number of unknown products $\ell'$, let $S_t^{\star}(\ell') = U_{\ell'}\cup S^{\star}_{(c-\ell')}(\mathcal{I}_t)$. An important quantity in our analysis is the number of rounds $N_i$ option $i \in S^{\star}_{(c-\ell')}(\mathcal{I}_t) \cup \{0\}$ is chosen in an epoch if $S^{\star}_t(\ell')$ is offered. The expected reward in an epoch if $S^{\star}_t(\ell')$ is offered is
$$r^{\star}_t(\ell') = \expect \Bigg[\sum_{i \in S^{\star}_{(c-\ell')}(\mathcal{I}_t)} N_i \Bigg] +1  = \sum_{i \in S^{\star}_{(c-\ell')}(\mathcal{I}_t)} \expect[N_i] + 1.$$
Given that the probability that an unknown entrant is purchased if $S_{t}^{\star}(\ell')$ is offered is $\frac{\ell' \tildew}{W^{\star}_{(c-\ell')}(\mathcal{I}_t) + \ell' \tildew + w_0}$, the expected number of rounds in an epoch is $\tau_{t}^{\star}(\ell') = \frac{W^{\star}_{(c-\ell')}(\mathcal{I}_t) + \ell' \tildew + w_0}{\ell' \tildew}$. Thus, we can write $\expect[N_i]$ as
\begin{equation*}
    \expect[N_i] = \prob[\text{$i$ is chosen}|S^{\star}_t(\ell')] \cdot \tau_{t}^{\star}(\ell') = \frac{w_i}{\cancel{W^{\star}_{(c-\ell')}(\mathcal{I}_t) + \ell'\tildew + w_0}} \cdot  \frac{\cancel{W^{\star}_{(c-\ell')}(\mathcal{I}_t) + \ell' \tildew + w_0}}{\ell' \tildew}= \frac{w_i}{\ell' \tildew}.
\end{equation*}
Then, we have that
\begin{equation*}\label{eq:formula_reward}
    r^{\star}_t(\ell') = \sum_{i \in S^{\star}_{(c-\ell')}(\mathcal{I}_t)} \expect[N_i] +1 = \sum_{i \in S^{\star}_{(c-\ell')}(\mathcal{I}_t)} \frac{w_i}{\ell' \tildew} +1 = \sum_{k=1}^{c-\ell'} \frac{w^{\star}_{(k)}(\mathcal{I}_t)}{\ell' \tildew} + 1.
\end{equation*}
Applying the above for $\ell' \in \{\ell, \ell+1\}$, taking a difference yields 
\begin{align}\label{eq:reward_loss}
    r^{\star}_t(\ell) - r^{\star}_t(\ell+1) &= \sum_{k=1}^{c-\ell} \frac{w^{\star}_{(k)}(\mathcal{I}_t)}{\ell \tildew} - \sum_{k=1}^{c-\ell-1} \frac{w^{\star}_{(k)}(\mathcal{I}_t)}{(\ell+1) \tildew} \nonumber \\
    &= \sum_{k=1}^{c-\ell-1} w_{(k)}^{\star}(\mathcal{I}_t)\Big(\frac{1}{\ell \tildew} - \frac{1}{(\ell+1) \tildew} \Big) + \frac{w^{\star}_{(c-\ell)}(\mathcal{I}_t)}{\ell \tildew} = \sum_{k=1}^{c-\ell-1} \frac{w_{(k)}^{\star}(\mathcal{I}_t)}{\ell(\ell+1) \tildew}  + \frac{w^{\star}_{(c-\ell)}(\mathcal{I}_t)}{\ell \tildew} \nonumber \\
    &= \frac{\sum_{k=1}^{c-\ell-1} w_{(k)}^{\star}(\mathcal{I}_t) + (\ell+1) w^{\star}_{(c-\ell)}(\mathcal{I}_t)}{\ell(\ell+1) \tildew} = \frac{W^{\star}_{(c-\ell-1)}(\mathcal{I}_t) + (\ell + 1) w^{\star}_{(c-\ell)}(\mathcal{I}_t)}{\tildew \ell (\ell + 1)}
\end{align}
where the third equality uses $\frac{1}{\ell \tildew} - \frac{1}{(\ell+1) \tildew} = \frac{1}{\ell(\ell+1) \tildew}$ and the fourth equality brings all fractions under a common denominator by multiplying the numerator and denominator of $\frac{w^{\star}_{(c-\ell)}(\mathcal{I}_t)}{\ell \tildew}$ by $\ell+1$.

Given that the reward from each product is $1$, the expected number of rounds in the epoch is $\tau^{\star}_t(\ell') = r^{\star}_t(\ell') + \expect[N_0] = r^{\star}_t(\ell')+ \frac{w_0}{\ell' \tildew}$. Applying this for $\ell' \in \{\ell, \ell+1\}$, taking a difference and using \eqref{eq:reward_loss} yields 
\begin{align}\label{eq:time_gain}
    \tau^{\star}_t(\ell) - \tau^{\star}_t(\ell + 1) &= r^{\star}_t(\ell) +\frac{w_0}{\ell \tildew} -\Big(r^{\star}_t(\ell+1) +\frac{w_0}{(\ell+1) \tildew} \Big) = r^{\star}_t(\ell)-r^{\star}_t(\ell+1) + w_0 \Big(\frac{1}{\ell \tildew} - \frac{1}{(\ell+1) \tildew} \Big) \nonumber\\
    &= r^{\star}_t(\ell)-r^{\star}_t(\ell+1) + \frac{w_0}{\tildew \ell(\ell+1)} = \frac{W^{\star}_{(c-\ell-1)}(\mathcal{I}_t) + (\ell + 1) w^{\star}_{(c-\ell)}(\mathcal{I}_t) + w_0}{\tildew \ell (\ell + 1)}.
\end{align}
Taking the ratio of \eqref{eq:reward_loss} and \eqref{eq:time_gain} the denominator is the same and cancels yielding the lemma. 
\end{proof}

\section{Comparison to canonical algorithms for bandit learning}\label{sec:comparison_to_simple_strategies}

We compare \textsc{EFA} to four natural baselines that capture widely used algorithmic approaches for bandit learning problems. We consider the following policies:
\begin{itemize}
 \item \textsc{ExploreAll}: For any round $t$, the algorithm $\textsc{ExploreAll}$ offers as many unknown entrants in the assortment as possible along with the most attractive known products if $\textsc{OPT}_t > \textsc{Rev}_t(S^{\star}_{(c)}(\mathcal{I}_t))$ and offers the $c$ most attractive known products otherwise.
    \item \textsc{ExploreOne}: For any round $t$, the algorithm $\textsc{ExploreOne}$ explores one unknown entrant along with the $c-1$ most attractive known products if $\textsc{OPT}_t > \textsc{Rev}_t(S^{\star}_{(c)}(\mathcal{I}_t))$ and offers the $c$ most attractive known products otherwise.
    \item Upper confidence bound (\textsc{UCB}): For a non-decreasing sequence $\boldsymbol{p} = (p_t)_{t=1}^\infty$, with $p_t \in [0,1]$ for all $t$, the algorithm $\textsc{UCB}(\boldsymbol{p})$ assigns an upper confidence bound $\textsc{u}_t(i)$ for every round $t$ and product $i$ such that $\textsc{u}_t(i)$ is the $p_t$-th quantile of the prior $\mathcal{F}$ if $i$ is an unknown entrant and $\textsc{u}_t(i) = w_i$ if $i$ is a known product at round $t$. For any round $t$, $\textsc{UCB}(\boldsymbol{p})$ offers the assortment of $c$ products with the highest upper confidence bounds. In particular, this definition includes the Bayes-UCB algorithm in \cite{kaufmann2012bayesian}.
    \item Thompson Sampling (\textsc{TS}): For any round $t$, the algorithm $\textsc{TS}$ samples attraction parameters from the prior for each unknown entrant and offers the revenue-maximizing assortment under the sampled attraction parameters for the unknown entrants and the true attraction parameters for the known products. 
\end{itemize}

\textsc{ExploreAll} and \textsc{ExploreOne} are natural and simple policies for our framework. \textsc{UCB} and \textsc{TS} are widely used algorithmic approaches in the multi-armed bandit framework \cite{lai1985asymptotically, agrawal2012analysis} as well as in the MNL bandit framework \cite{agrawal2017thompson,agrawal2019mnl}. We show that each of these policies can have an arbitrarily larger regret than \textsc{EFA}. Interestingly, the reason for failure is that \textsc{UCB} and \textsc{ExploreAll} tend to \textit{over}-explore (include too many entrants), while \textsc{TS} and \textsc{ExploreOne} \textit{under}-explore (include too few entrants). We analyze these two failure modes in Section~\ref{subsec:overexploration_approaches} and Section~\ref{subsec:underexploration_approaches} respectively.

\subsection{Failure due to overexploration}\label{subsec:overexploration_approaches}

We begin by analyzing the \textsc{ExploreAll} policy. Formally, \textsc{ExploreAll} explores as many unknown entrants as possible, along with the most attractive known products. Specifically, at round $t$, let $k_t = \min(c, m_t)$ be the minimum between the assortment capacity $c$ and the number of remaining unknown entrants $m_t$. \textsc{ExploreAll} offers $k_t$ unknown entrants along with the most attractive known products $S^{\star}_{(c - k_t)}(\mathcal{I}_t)$.  We show that \textsc{ExploreAll} can perform arbitrarily worse than \textsc{EFA}.

For a capacity $c \in \mathbb{N}$ and probability $q > 0$ we consider the following instance which we denote by $\mathcal{I}(c, q)$. There are $m = c$ unknown entrants and each unknown entrant has a prior $\mathcal{F} = \bern(q)$. There are $c$ incumbents where $\frac{c}{2}$ of them have an attraction parameter of $0.9$ and the remaining $\frac{c}{2}$ have an attraction parameter of $2q$. The customer's cold-start attraction parameter $\tilde{w}$ is the mean of the prior $q$.

We show that the regret of $\textsc{ExploreAll}$ becomes arbitrarily worse than the regret of $\textsc{EFA}$ as the prior probability $q$ goes to $0$.  Intuitively, $\textsc{ExploreAll}$ will explore too many unknown entrants relative to $\textsc{EFA}$ which will induce it to have a higher regret. In particular, under the instance $\mathcal{I}(c,q)$,  $\textsc{EFA}$ will always explore at most $\frac{c}{2}$ unknown entrants while $\textsc{ExploreAll}$ will always explore all unknown entrants. In the first epoch, given that $\textsc{ExploreAll}$ explores all of the unknown entrants it takes at least $\frac{1}{cq}$ rounds in expectation to obtain a purchase and the regret per round is at least a constant. Thus, the regret of $\textsc{ExploreAll}$ is at least $\frac{A}{q}$ for some constant $A$. In comparison, the regret of $\textsc{EFA}$ can be shown to be upper-bounded by a constant which does not depend on $q$. This intuition is formalized in the following proposition (proof in Appendix~\ref{appendix_subsec:prop_explore_all_is_arbitrarily_suboptimal_proof}). 

\begin{proposition}\label{prop:explore_all_is_arbitrarily_suboptimal}
   For any capacity $c \in \mathbb{N}$ and probability $q < \frac{0.9}{2+2(c+1)^2}$, $\textsc{Reg}(\textsc{ExploreAll}) \geq \frac{1}{q} \cdot \frac{9}{29c} - 1$ and $\textsc{Reg}(\textsc{EFA}) \leq 2c^2(c+1)$ on the instance $\mathcal{I}(c,q)$.
\end{proposition}

We next analyze the performance of the $\textsc{UCB}$ policy on the instance $\mathcal{I}(c,q)$. For any non-decreasing sequence $\mathbf{p} = (p_t)_{t=1}^{\infty}$, $\textsc{UCB}(\mathbf{p})$ first assigns only $0$ to each unknown entrant (and thus does not explore), and it then assigns only $1$ (and thus explores all); it thus incurs a large regret due to a reason similar to \textsc{ExploreAll}. This is formalized in the following proposition (proof in Appendix~\ref{appendix_subsec:ucb_is_arbitrarily_suboptimal_proof}). 

\begin{proposition}\label{prop:ucb_performance_bound_on_instance}
    For any non-decreasing sequence $\mathbf{p} = (p_t)_{t=1}^{\infty}$, with $p_t \in [0,1]$ for all $t$, capacity $c \in \mathbb{N}$, and probability $q  < \frac{0.9}{2+2(c+1)^2}$, $\textsc{Reg}(\textsc{UCB}(\mathbf{p})) \geq \frac{1}{q} \cdot \frac{9}{29c} - 1$ and $\textsc{Reg}(\textsc{EFA}) \leq 2c^2(c+1)$ on the instance $\mathcal{I}(c,q)$.
\end{proposition}

\subsection{Failure due to underexploration}\label{subsec:underexploration_approaches}
We first compare \textsc{EFA} to the $\textsc{TS}$ policy. Recall that, for any round $t$, $\textsc{TS}$ samples attraction parameters from the prior for each unknown entrant and offers the revenue-maximizing assortment under the sampled attraction parameters for the unknown entrants and the true attraction parameters for the known products.

We show that the regret of $\textsc{TS}$ is arbitrarily worse than that of $\textsc{EFA}$ on the instance $\mathcal{I}(2,q)$ when the prior probability $q$ goes to $0$. Intuitively, to obtain a purchase of an unknown product $\textsc{TS}$ needs to (i) obtain a posterior sample of $1$ (with probability $q$) and (ii) obtain a purchase of an unknown entrant given the offered assortment (with probability at most $q$). Thus, $\textsc{TS}$ induces a purchase with a probability of $q^2$ which implies that the first epoch contains at least $\frac{1}{q^2}$ rounds in expectation. Given that the regret per round is at least $Bq$ for some constant $B$, the regret of $\textsc{TS}$ will be at least $\frac{B}{q}$. This intuition is formalized in the following proposition (proof in Appendix~\ref{appendix_subsec:thompson_sampling_is_arbitrarily_suboptimal_proof}).

\begin{proposition}\label{prop:TS_is_arbitrarily_suboptimal}
    For any probability $q <0.045$, $\textsc{Reg}(\textsc{TS}) \geq \frac{1}{q} \cdot \frac{(1-q)^2}{9}$ and $\textsc{Reg}(\textsc{EFA}) \leq 24$ on the instance $\mathcal{I}(2,q)$.
\end{proposition}

We next compare our algorithm to the $\textsc{ExploreOne}$ policy. 
The following proposition shows that the regret of this policy can be a factor of $c$ larger compared to our algorithm (proof in Appendix~\ref{appendix_subsec:explore_one_at_a_time_c_factor_lower_bound_proof}). 

\begin{proposition}\label{prop:explore_one_at_a_time_c_factor}
    For any assortment capacity $c$ there exists an instance where $\frac{\textsc{Reg}(\textsc{ExploreOne})}{\textsc{Reg}(\textsc{EFA})} \geq \frac{c}{2}$.
\end{proposition}

\begin{remark}\label{remark:explore_one_at_a_time_c_factor_upper_bound}
    We also show that $\frac{\textsc{Reg}(\textsc{ExploreOne})}{\textsc{Reg}(\textsc{EFA})} \leq c$ for any instance where the attraction parameter of the $c$-th most attractive incumbent is larger than $\tildew$ (see Appendix~\ref{appendix_subsec:explore_one_at_a_time_c_factor_upper_bound_proof}). 
\end{remark}

\section{Extensions}\label{sec:extensions}
In this section we add several extensions to our model. First, we consider a model where the incumbents have different rewards, and we characterize the optimal policy. Second, we consider an extension to different priors and two entrants. Third, we consider an extension to a single entrant and noisy observations and characterize the optimal policy. 

\subsection{Heterogeneous rewards}\label{subsec:heterogeneous_rewards}
We extend our model so that each incumbent $i \in \{m+1, \ldots, n\}$ has a different reward $r_i$ and each entrant $j \in \{1, \ldots, m\}$ has the same reward $r_{j} = r$.\footnote{
Our algorithm treats the entrants identically and this allows us to decompose our regret analysis into epochs without reasoning about which entrant is explored within a particular epoch. Extending this to settings where the entrant rewards are ex-ante heterogeneous seems to this require new techniques and is an exciting open direction.
} 
All rewards are known to both the platform and the customers at the beginning. All other aspects of the model remain the same as in Section~\ref{sec:model}. As before, the attraction parameters $w_i$ are known for each incumbent $i\in \{m+1,\ldots, n\}$ and are drawn from a known distribution $\mathcal{F}$ for each entrant $i\in\{1,\ldots,m\}$. Given a set of known products $\mathcal{I}_t$ at round $t$, the expected ex-post optimum (including unknown entrants) and the revenue of known products are denoted respectively by
$$\textsc{OPT}_t = \expect\limits_{w_i \sim \mathcal{F}, i \not \in \mathcal{I}_t}\Bigg[  \max_{S \subseteq \mathcal{N}, |S| \leq c} \frac{\sum_{i \in S} r_i w_i}{\sum_{i \in S} w_i + w_0}\Bigg] \quad \text{ and } \quad \textsc{Rev}^{\star}_{t} = \max_{S \subseteq \mathcal{I}_t, |S| \leq c} \frac{\sum_{i \in S} r_i w_i}{\sum_{i \in S} w_i + w_0}.
$$
An important quantity in extending our results to heterogeneous rewards is the \emph{scaled interim regret}, $\textsc{SIR}_i(t) = (\textsc{OPT}_t-r_i) \frac{w_i}{w_0}$. This corresponds to a scaled version of the regret contribution of a known product $i$ at round $t$.\footnote{The exact interim regret of product $i$ if the platform offers $\ell$ unknown entrants is $\textsc{SIR}_i(t) = (\textsc{OPT}_t-r_i) \frac{w_i}{\ell \tildew}$ since product $i$ is purchased $\frac{w_i}{\ell \tildew}$ rounds in expectation, and each purchase contributes $\textsc{OPT}_t-r_i$ to the total regret.
} We note that the regret here is not comparing against the ex-post optimum but rather against the \emph{interim} expected optimum conditioning on the set of known products $\mathcal{I}_t$ and their corresponding attraction parameters. Our algorithm only considers the set of known products with negative scaled interim regret; we denote this set by $\mathcal{I}^{<0}_t = \{i \in \mathcal{I}_t: \textsc{SIR}_i(t) < 0\}$ and its cardinality by $n^{<0}_t$.

First, we characterize the optimal policy for a single entrant; we again consider instances where exploration is worthwhile, i.e., $\textsc{OPT}_1 > \textsc{Rev}_1^{\star}$. Recall that under homogeneous rewards, it is optimal to explore the entrant with the $c-1$ incumbents with the highest values $w_i$. To extend this to heterogeneous rewards, we offer the entrant only with incumbents with negative scaled interim regret and include as many of these incumbents as possible; this yields $\min(n^{<0}_1,c-1)$ incumbents. We prioritize the incumbents with the lowest scaled interim regret and refer to the corresponding set as $S_1=\argmin_{S\subseteq \mathcal{I}_1, |S|\leq c-1} \sum_{i\in S} \textsc{SIR}_i(1)$. Similar to Section~\ref{sec:exploring_single_entrant}, we define policy $\pi(S)$ as the one that offers the unknown entrant together with the set $S$ until the entrant is purchased and offers the most attractive known products afterwards. The following result is a generalization of Theorem~\ref{thm:explore_one_unknown_with_best} and its proof is provided in Appendix~\ref{appendix_subsec:optimal_policy_different_rewards_single_entrant}. 

\begin{theorem}\label{thm:optimal_policy_different_rewards_single_entrant}
    For any instance with $\textsc{OPT}_1 > \textsc{Rev}^{\star}_1$, $\pi(S_1)$ is an optimal policy. That is, it is optimal to explore the entrant with the $\min(n^{<0}_1, c-1)$ products with lowest scaled interim regret.
\end{theorem}
Next, we characterize the optimal policy in the presence of multiple entrants. 
We show that a natural extension to the EFA algorithm is optimal in this setting. The main difference with EFA is that we sort known products according to their scaled interim regret. Note that the negative of the scaled interim regret of product $i$ corresponds to the myopic benefit of product $i$ compared to the expected interim optimum. This plays a similar role to the attraction parameter in \textsc{EFA}. Hence, for $\ell \in \{ 1, \ldots, c\}$ an important quantity in our algorithm is the cumulative benefit of the $\ell$-th fictitious assortment
$$\beta_t(\ell) = -\sum_{i=1}^{c-\ell} \textsc{SIR}_{(i)}(t) - \ell \times \textsc{SIR}_{(c-\ell+1)}(t)$$ 
where $\textsc{SIR}_{(k)}(t)$ is the $k$-th lowest scaled interim regret. The cumulative benefit of the $\ell$-th fictitious assortment $\beta_t(\ell)$ is an analogue of the revenue of the $\ell$-th fictitious assortment $\alpha_t(\ell)$ in Section~\ref{subsec:algorithm_main_theorem}. For $k \in \{1, \ldots, c\}$, denote by $V^{\star}_{(k)}(\mathcal{I}_t)$ a set of $k$ known products with the lowest scaled interim regret values. Choose $S^{\star}_{t} \in \argmax_{S \subseteq \mathcal{I}_t, |S| \leq c} \frac{\sum_{i \in S} r_i w_i}{\sum_{i \in S} w_i + w_0}$ to be a revenue-maximizing set of known products. Algorithm~\ref{alg:optimal_explorer_different_rewards}, which we term \textsc{Heterogeneous Exploration with Fictitious Assortments (HEFA)}, is a generalization of \textsc{EFA}. The number of unknown entrants $\ell_t^{\star}$ is the maximum between $c-n_t^{<0}$ and the largest number $\ell_t$ such that $\textsc{OPT}_t \geq \beta_t(\ell_t)$ (line~\ref{alg_different_rewards_line:set_number_products}, Algorithm~\ref{alg:optimal_explorer_different_rewards}). Then, our algorithm offers $\ell_t^{\star}$ unknown entrant together with the $c-\ell_t^{\star}$ known products with the lowest scaled interim regrets (line~\ref{alg_different_rewards_line:offer_set}, Algorithm~\ref{alg:optimal_explorer_different_rewards}). In Appendix~\ref{appendix_subsec:beta_connection_to_rewards_one}
, we show that, when rewards are homogeneous, HEFA reduces to EFA. The following theorem (proof in Appendix~\ref{appendix_subsec:optimal_policy_multiple_entrants}) shows that \textsc{HEFA} minimizes the expected infinite-horizon regret.

\begin{algorithm}[h!]
\caption{\textsc{Heterogeneous Exploration with Fictitious Assortments} (\textsc{HEFA})}\label{alg:optimal_explorer_different_rewards} 
\begin{algorithmic}[1]
\State \text{Initialize }$I_1 = \{m+1, \ldots, n\}$ and $m_1 = m$.
\State \textbf{for} $t = 0,1,2, \ldots$:
\State \hspace*{1em} Calculate $\textsc{Rev}^{\star}_{t}$, $\textsc{OPT}_t$, and $\beta_{t}(\ell)$ for $\ell \in \{1, \ldots, c\}$.  
\State \hspace*{1em} \textbf{if} $\textsc{OPT}_t > \textsc{Rev}^{\star}_{t}:$ \Comment{ exploration-worthy condition} \label{alg_different_rewards_line:worth_it_to_explore}
\State \hspace*{1em}\hspace*{1em}  Set $\ell_t =  \max\{ \ell \in \{1, \ldots, k_t\}: \textsc{OPT}_{t} \geq \beta_{t}(\ell) \}$ where $k_t = \min(c, m_t)$; set $\ell_t^{\star} = \max(\ell_t, c-n_t^{<0})$.\label{alg_different_rewards_line:set_number_products}
\State \hspace*{1em}\hspace*{1em}  Offer set $S_t  =U_{\ell_t^{\star}} \cup V^{\star}_{(c-\ell_t^{\star})}(\mathcal{I}_t)$;  observe choice $Y_t \in S_t \cup \{0\}$. \label{alg_different_rewards_line:offer_set}
\State \hspace*{1em}\hspace*{1em}  If $Y_t \in U_{\ell_t}$ then $\mathcal{I}_{t+1} = \mathcal{I}_t \cup \{Y_t\}$ and $m_{t+1} = m_t-1$; else $\mathcal{I}_{t+1} = \mathcal{I}_t $ and $m_{t+1} = m_t$. 
\State \hspace*{1em} \textbf{else}: Offer set $S_t = S^{\star}_{t}$; observe choice $Y_t \in S_t \cup \{0\}$; $\mathcal{I}_{t+1} = \mathcal{I}_t $ and $m_{t+1} = m_t$.  \label{alg_different_rewards_line:not_exploration_worthy_case}
\end{algorithmic}
\end{algorithm}

\begin{theorem}\label{thm:optimal_policy_multiple_entrants} For any instance with heterogeneous rewards,    \textsc{HEFA} (Algorithm~\ref{alg:optimal_explorer_different_rewards}) is optimal.
\end{theorem}

\subsection{Heterogeneous entrant priors} 
\label{sec:heterogeneous_priors}
We extend our model from Section~\ref{sec:model} to entrants with heterogeneous priors for the belief of the attraction parameters. In the original model with homogeneous prior beliefs, it did not matter exactly \textit{which} entrants were explored first, since all entrants were identical.
In that case, the optimal algorithm EFA characterized \textit{how many} entrants to explore with at once, but it did not have to specify exactly which entrants to include.
When the prior beliefs are heterogeneous, the main question we ask is \textit{which} entrant should be explored first, and how does this depend on the prior beliefs? 

We answer this question for a simple class of instances $\mathcal{C}$ with $n = 4$ products, two entrants with prior beliefs $\mathcal{F}_1$ and $\mathcal{F}_2$, two incumbents with attraction parameters $w_3$ and $w_4$ with $w_3 \geq w_4$, and capacity $c = 2$. Entrant $i$ has a cold-start attraction parameter $\tildew_i$ for $i \in \{1,2\}$. For a subset of entrants $S \subseteq \{1,2\}$ let $\pi^{S}$ be the policy which explores $S$ with the most attractive $2-|S|$ incumbents until an entrant is purchased, subsequently explores the remaining entrant with the most attractive known product until the remaining entrant is purchased, and offers the revenue-maximizing assortment when both entrants have been purchased. We assume that $\prob_{w_1 \sim \mathcal{F}_1}[w_1 > w_3] > 0$ and $\prob_{w_2 \sim \mathcal{F}_2}[w_2 > w_3] > 0$. Under this assumption, it is always worth it to explore an unknown entrant. Then, the optimal policy is one of $\pi^{\{1\}}$, $\pi^{\{2\}}$, and $\pi^{\{1, 2\}}$.
The following theorem (proof in  Appendix~\ref{appendix_subsec:characterization_of_the_optimal_policy_with_two_entrants}) characterizes the optimal policy as a function of the parameters.

\begin{theorem}
\label{thm:characterization_of_the_optimal_policy_with_two_unknown_entrants_with_heterogeneous_priors}
    For any instance in the class $\mathcal{C}$, 
    \begin{itemize}
        \item $\textsc{Reg}(\pi^{\{1,2\}})-\textsc{Reg}(\pi^{\{1\}}) > 0$ if and only if 
    $$\expect_{\substack{w_1 \sim \mathcal{F}_1 \\ w_2 \sim \mathcal{F}_2}} \Bigg[\Big(\frac{\tildew_2}{\tildew_1 + \tildew_2} \Big(\frac{\max(w_2, w_3)}{\tildew_1}-\frac{\max(w_1, w_3)}{\tildew_2} \Big) - \frac{w_3}{\tildew_1}  \Big)(\textsc{OPT}-1) \Bigg] > \expect_{\substack{w_1 \sim \mathcal{F}_1 \\ w_2 \sim \mathcal{F}_2}} \Bigg[ \frac{1}{\tildew_1+ \tildew_2} \textsc{OPT} \Bigg].$$
    \item $\textsc{Reg}(\pi^{\{1,2\}}) - \textsc{Reg}(\pi^{\{2\}}) > 0$ if and only if 
$$\expect_{\substack{w_1 \sim \mathcal{F}_1 \\ w_2 \sim \mathcal{F}_2}} \Bigg[\Big( \frac{\tildew_1}{\tildew_1 + \tildew_2} \Big(\frac{\max(w_1, w_3)}{\tildew_2} -\frac{\max(w_2, w_3)}{\tildew_1}\Big)-\frac{w_3}{\tildew_2} \Big)(\textsc{OPT}-1) \Bigg] > \expect_{\substack{w_1 \sim \mathcal{F}_1 \\ w_2 \sim \mathcal{F}_2}} \Bigg[ \frac{1}{\tildew_1 + \tildew_2} \textsc{OPT} \Bigg], $$
\item $\textsc{Reg}(\pi^{\{2\}})-\textsc{Reg}(\pi^{\{1\}}) >0 $  if and only if 
$$\expect_{\substack{w_1 \sim \mathcal{F}_1 \\ w_2 \sim \mathcal{F}_2}} \Bigg[ \Bigg(\frac{w_3}{\tildew_2} + \frac{\max(w_2,w_3)}{\tildew_1}\Bigg) (\textsc{OPT}-1) \Bigg] > \expect_{\substack{w_1 \sim \mathcal{F}_1 \\ w_2 \sim \mathcal{F}_2}} \Bigg[ \Bigg( \frac{w_3}{\tildew_1} +\frac{\max(w_1,w_3)}{\tildew_2} \Bigg)(\textsc{OPT}-1) \Bigg].$$
 \end{itemize}
\end{theorem}

We provide two interpretations of this result under certain conditions.
The following proposition (proof in Appendix~\ref{appendix_subsec_stochastic_dominance}) shows that if $\mathcal{F}_1$ stochastically dominates $\mathcal{F}_2$, then $\pi^{\{2\}}$ is never optimal.
That is, loosely speaking, it is better to first explore with the product with a ``higher'' attraction parameter.

\begin{proposition}\label{corollary:stochastic_dominance}
    For any instance in the class $\mathcal{C}$ where $\mathcal{F}_1$ stochastically dominates $\mathcal{F}_2$, $\tildew_1 > \tildew_2 $, and $\prob_{w_1 \sim \mathcal{F}_1}[w_1 > w_3] > 0$, $\pi^{\{2\}}$ is not an optimal policy. 
\end{proposition}
    
Next, we consider a setting where $\mathcal{F}_1$ and $\mathcal{F}_2$ have the \textit{same mean} $\mu$ but have different tails.
We consider Bernoulli-like priors, where $w_1$ equals $\mu/p_1$ with probability $p_1$ and 0 otherwise, and $w_2$ equals $\mu/p_2$ with probability $p_2$ and 0 otherwise, where $p_1 > p_2$. 
We refer to $\mu/p_i$ as the \textit{upside} of item $i$. Entrant 1 has a lower upside than entrant 2, but a higher probability of realizing this upside. We assume that $\tilde{w}_i = \expect_{w_i \sim \mathcal{F}_i}[w_i]$ is the mean of entrant $i$'s prior for $i \in \{1,2\}$. We fix $p_1$ and evaluate the optimal policy as a function of $p_2 \in (0, p_1)$.

Surprisingly, we show that either $\pi^{\{1\}}$ or $\pi^{\{2\}}$ can be optimal. 
That is, it is not the case that always exploring with the entrant with a higher upside first is better, nor is it optimal to always explore the lower-upside entrant first.
We show that it is better to explore with entrant 1 first if $p_2$ is \textit{much} smaller than~$p_1$, but it is better to explore with entrant 2 first if $p_2$ is \textit{slightly} smaller than~$p_1$.
We formalize this in the following proposition (proof in Appendix~\ref{appendix_subsec:different_priors}), where we show that the optimal policy depends on whether $p_2$ is above or below a threshold $\theta$.

\begin{proposition}\label{prop:different_priors_bernoulli}
    For any $\mu, p_1,w_3, w_4$ with $p_1 \in (0, \frac{1}{6})$, $w_3< \frac{\mu}{p_1} < 2 w_3$, $w_4 = \frac{w_3}{2}$, $w_4 > \mu$, and $w_3 > 4$, there exists some $\theta( \mu , p_1,w_3) \in (0, p_1)$ 
    such that the optimal policy is $\pi^{\{1\}}$ if $p_2 < \theta(\mu , p_1,w_3)$ and the optimal policy $\pi^{\{2\}}$ if $p_2 > \theta( \mu , p_1,w_3)$.
\end{proposition}

The exact conditions of \cref{prop:different_priors_bernoulli} are chosen to ensure that $\pi^{\{1, 2\}}$ is never the optimal policy, so that this provides a clean characterization of which of $\pi^{\{1\}}$ and $\pi^{\{2\}}$ is optimal.

This result shows that which entrant to explore with first has a complex dependence on the priors.
This is in contrast to the traditional multi-armed bandit setting, where a common  approach is to simply order all arms by their UCBs, and explore those with the highest UCBs first.
\cref{prop:different_priors_bernoulli} shows that simply using UCBs does not work in our setting.
We leave as an open direction of whether there exists an index that is a function of the prior, where it is optimal to explore those with a higher index first.

\subsection{Single entrant and noisy observations.}\label{sec:single_entrant_noisy_obs}
We also extend our model from Section~\ref{sec:model} when there is a single entrant to noisy observations. For ease of exposition, we assume that a Beta prior for the entrant, i.e., $\mathcal{F} = \mathrm{Beta}(a,b)$. After the $i$-th purchase of the entrant the platform and the customers learn a review feedback for the entrant given by $u_i \sim \bern(w_1)$ for $i \in \{1, \ldots, k-1\}$. To ensure that the optimal policy has finite regret, we assume that the entrant's attraction parameter is fully realized after the $k$-th purchase, i.e., $u_i = w_1$ for $i \geq k$. The entrant's attraction parameter given $i \leq k-1$ purchases is $h(\mathcal{F}_i)$ where $\mathcal{F}_i$ is the entrant's posterior given the first $i$ review feedbacks and $h$ is an arbitrary mapping (e.g., the mean, the $X$-th percentile, etc.). When there are no purchases ($i = 0$), the entrant's attraction parameter $h(\mathcal{F})$ is the cold-start attraction parameter $\tildew$. We formalize below the parts of the model which differ from Section~\ref{sec:model}.
\begin{itemize}
    \item 
    The history $\mathcal{H}_t$ now keeps track of the historical reviews for the entrant. Letting $p_s$ be the number of purchases for the entrant before round $s$, the history at round $t$ is $\mathcal{H}_t = \{(S_s, Y_s, i_s)\}_{s=1}^{t-1}$ where $i_s = \perp$ if $Y_s$ is an incumbent ($Y_s \neq 1$) and $i_s = u_{p_s+1}$ if $Y_s$ is the entrant ($Y_s = 1$).
    \item 
    If the entrant is chosen in round $t$ ($Y_t = 1$), a new review for the entrant is obtained given by $u_{p_t+1} \sim \bern(w_1)$ if $p_t \leq k-2$ and $u_{p_t + 1} = w_1$ otherwise;  $\mathcal{H}_{t+1} = \mathcal{H}_t \cup (S_t, 1, u_{p_t+1})$ and $p_{t+1} = p_t + 1$ respectively. If an incumbent is chosen ($Y_t \neq 1$), $\mathcal{H}_{t+1} = \mathcal{H}_t \cup (S_t, Y_t, \perp)$ and $p_{t+1} = p_t$ respectively.
    \item The customer's posterior for the entrant $\mathcal{F}(\mathcal{H}_t)$ depends on the entrant reviews $(u_1, \ldots, u_{p_t})$ obtained by round $t$. The customer's attraction parameter for the entrant at round $t$ is $w_1(t) = h(\mathcal{F}(\mathcal{H}_t))$.
\end{itemize}
We assume that there is a positive probability that the entrant would be in the optimal assortment once its attraction parameter is known $(\overline{\theta} > w^{\star}_{(c)}(\mathcal{I}_1))$. The following theorem (proof in Appendix~\ref{appendix_subsec:proof_characterization_opt_policy_single_entrant_noisy_obs}) establishes that it is optimal to explore the entrant with the $c-1$ most attractive incumbents. This shows that the result from Theorem~\ref{thm:explore_one_unknown_with_best} extends to noisy observations. 

\begin{theorem}\label{theorem:single_entrant_non-full_realizability}
    For any instance with $\overline{\theta} > w^{\star}_{(c)}(\mathcal{I}_1)$, it is optimal to offer the entrant with the $c-1$ most attractive incumbents until the $k$-th review, and offer the $c$ most attractive products in subsequent rounds. 
\end{theorem}

\section{Conclusion}\label{sec:conclusion}
This paper studies how an online marketplace should learn the quality of new products when learning is driven solely by purchase behavior and customer reviews. In contrast to the classical multi-armed bandit literature which centers on the question of \textit{whether} to explore, we show that another primary lever for minimizing revenue loss in our setting is optimizing \textit{how exploration is structured}. We characterize the optimal exploration policy, EFA, which keeps the best-known incumbents in the assortment and uses a simple threshold rule to determine how many entrants to explore simultaneously. We further show that standard bandit policies such as UCB and Thompson sampling can be arbitrarily worse than EFA, because they tend to drive exploration intensity toward extremes rather than the intermediate level that is often optimal.

Our paper opens up several avenues of future work which we summarize here. First, in our model, a purchase of an entrant reveals its attraction parameter exactly. A natural extension is one where each purchase provides a noisy signal, so the platform and customers can use this to update their belief about entrants over time; Section~\ref{sec:single_entrant_noisy_obs} takes a first step towards this direction.  In that setting, the platform is never fully certain about an entrant’s quality, and the usual multi-armed bandit question of \emph{whether} an entrant is worth continued exploration reappears. Second, a prerequisite for analyzing the above setting is to allow heterogeneous entrant priors. Section \ref{sec:heterogeneous_priors} takes an initial step in this direction by studying heterogeneous priors in a simplified two-entrant setting, but a full characterization remains open. Third, we adopt an infinite-horizon formulation to remove horizon-driven stopping considerations. With a finite horizon, the platform must additionally decide whether exploration is worthwhile given the remaining time. Finally, our characterization relies on the \emph{independence of irrelevant alternatives} property that is satisfied by the MNL choice model. Understanding the structure of optimal exploration policies in choice models that do not satisfy this property is an interesting open direction.

\subsection*{Acknowledgements}
We thank Marouane Ibn Brahim for providing detailed feedback that helped us improve the presentation of our results.

\bibliographystyle{alpha}
\bibliography{References}

\appendix

\section{Optimality of stationary and deterministic policies (Lemma~\ref{lemma:there_exists_an_optimal_policy_which_offers_the_same_assortment_until_the_entrant_is_purchased}) }\label{appendix_sec:proof_lemma_there_exists_an_optimal_policy_which_offers_the_same_assortment_until_the_entrant_is_purchased}
Recall that $\pi(S)$ offers the entrant along with a set of at most $c-1$ incumbents $S$ until the entrant is purchased and offers the $c$ products with the highest attraction parameters in subsequent rounds. Lemma~\ref{lemma:there_exists_an_optimal_policy_which_offers_the_same_assortment_until_the_entrant_is_purchased} posits that there exists an optimal policy $\pi^{\star} = \pi(S)$ for some set of at most $c-1$ incumbents $S$.

We first derive a lower bound on the regret of any policy. For a policy $\pi$, let $Z^{\pi}$ be the first round at which the unknown entrant is purchased; note that $Z^{\pi}$ is a random variable which depends on the policy $\pi$. For a policy $\pi$ and an assortment $S$, let $N^{\pi}(S)$ be the number of rounds in which $S$ is offered during the exploration period of $\pi$; note that $N^{\pi}(S)$ is a random variable which depends on $\pi$ and $S$. Recall that $\alpha^{\star} = \expect_{w_1 \sim \mathcal{F}}[\textsc{OPT}]$ is
the expected ex-post optimum (defined in the proof of Theorem~\ref{thm:explore_one_unknown_with_best}). The next lemma lower bounds the regret of any policy $\pi$ in terms of the total regret incurred over all assortments during the exploration period. Recall that $\textsc{Rev}_1(S)$ is the instantaneous revenue of $S$ before the entrant is purchased. 

\begin{lemma}\label{lemma:lower_bound_regret_of_any_policy_total_regret_incurred_over_all_assortments_during_the_exploration_period}
 The regret of any policy $\pi$ is lower bounded by $\textsc{Reg}(\pi) \geq \sum\limits_{S \subseteq \mathcal{N}, |S| \leq c} (\alpha^{\star} - \textsc{Rev}_1(S)) \expect[N^{\pi}(S)]$.
\end{lemma}

Let $\tau(S')$ be the expected number of rounds and $r(S')$ the expected total reward during the exploration period if the entrant is offered repeatedly with a set $S'$ of at most $c-1$ incumbents. Let $S^{\star} \in \argmin\limits_{S' \subseteq \mathcal{I}_1, |S'| \leq c-1} \alpha^{\star} \tau(S') - r(S')$ be a set of incumbents which yields the lowest regret during the exploration period. The following lemma upper bounds the regret of $\pi(S^{\star})$. 

\begin{lemma}\label{lemma:the_regret_of_pi_S_star_is_bounded}
   The regret of the policy $\pi(S^\star)$ is upper bounded by $\textsc{Reg}(\pi(S^{\star}))\leq \frac{\sum_{i \in S^{\star}} w_i + \tildew + w_0}{\tildew}$.
\end{lemma}

\begin{proof}[Proof of Lemma~\ref{lemma:there_exists_an_optimal_policy_which_offers_the_same_assortment_until_the_entrant_is_purchased}.]
   For any policy $\pi$, it suffices to show that $\textsc{Reg}(\pi) \geq \textsc{Reg}(\pi(S^{\star}))$ for any policy $\pi$. 

\paragraph{Case 1: $\expect[Z^{\pi}] < +\infty$}. Let $\mathcal{E} = \{ S \subseteq \mathcal{N}:|S| \leq c, 1 \in S\}$ be the set of assortments containing the entrant. Let $q(S) = \frac{\tildew}{\sum_{i \in S \setminus \{1\}} w_i + \tildew + w_0}$ be the probability that the entrant is purchased for a set $S \in \mathcal{e}$. Then, the regret of $\pi$ is lower bounded as
\begin{align*}
    \textsc{Reg}(\pi) &\geq \sum\limits_{S \subseteq \mathcal{N}, |S| \leq c} (\alpha^{\star} - \textsc{Rev}_1(S)) \expect[N^{\pi}(S)] \geq \sum\limits_{S \in \mathcal{E}} (\alpha^{\star} - \textsc{Rev}_1(S)) \expect[N^{\pi}(S)] \nonumber \\
    &= \sum\limits_{S \in \mathcal{E}} \frac{\alpha^{\star} - \textsc{Rev}_1(S)}{q(S)} q(S)\expect[N^{\pi}(S)]  \geq \min\limits_{S \in \mathcal{E}} \Big\{ \frac{\alpha^{\star} - \textsc{Rev}_1(S)}{q(S)} \Big\} \Big(\sum\limits_{S \in \mathcal{E}} q(S)\expect[N^{\pi}(S)] \Big)\\
    &= \min\limits_{S \in \mathcal{E}} \Big\{ \frac{\alpha^{\star} - \textsc{Rev}_1(S)}{q(S)} \Big\} = \min\limits_{S' \subseteq \mathcal{I}_1, |S'| \leq c-1} \{\alpha^{\star} \tau(S') - r(S')\}  = \textsc{Reg}(\pi(S^{\star}))
\end{align*}
where the first inequality uses Lemma~\ref{lemma:lower_bound_regret_of_any_policy_total_regret_incurred_over_all_assortments_during_the_exploration_period} and the second inequality uses that $\alpha^{\star} - \textsc{Rev}_1(S) > 0$ for any $S \not \in \mathcal{E}$. The third equality uses that $\sum\limits_{S \in \mathcal{E}} q(S)\expect[N^{\pi}(S)] = 1$ as $\expect[Z^{\pi}] < \infty$ and $\sum\limits_{S \in \mathcal{E}} q(S)\expect[N^{\pi}(S)]$ is the expected number of times the entrant is purchased during the exploration period since any assortment $S$ is offered $\expect[N^{\pi}(S)]$ rounds in expectation and every time the entrant is purchased independently with probability $q(S)$. The fourth equality uses that $\min \limits_{S\in \mathcal{E}} \frac{\alpha^{\star} - \textsc{Rev}_1(S)}{q(S)} = \min\limits_{S' \subseteq \mathcal{I}_1, |S'| \leq c-1} \alpha^{\star} \tau(S') - r(S') = \textsc{Reg}(\pi(S^{\star}))$ as for any set $S = \{1\} \cup S'$ containing the entrant along with a set of at most $c-1$ incumbents $S'$, $\frac{1}{q(S)} = \tau(S')$ and $\textsc{Rev}_1(S) \cdot \frac{1}{q(S)} = r(S')$.

\paragraph{Case 2: $\expect[Z^{\pi}] = +\infty$.} Given that $Z^{\pi} = \sum\limits_{S \subseteq \mathcal{N}, |S| \leq c} N^{\pi}(S)$, there exists some $\tilde{S}$ with $\expect[N^{\pi}(\tilde{S})] = +\infty$. All assortments $S \in \mathcal{E}$ have $\expect[N^{\pi}(S)] < \infty$ since the entrant is purchased with probability $\frac{\tildew}{\sum_{i\in S \setminus \{1\}} w_i  + \tildew + w_0}$ every time $S$ is offered. As a result, $S \not \in \mathcal{E}$. Recall that that there is a positive probability that the entrant strictly increases the revenue of the revenue-maximizing assortment (as $\overline{\theta} > w^{\star}_{(c)}$). Hence, $\alpha^{\star}> \textsc{Rev}_1(\tilde{S})$ and thus $(\alpha^{\star}-\textsc{Rev}_1(\tilde{S})) \expect[N^{\pi}(\tilde{S})] = +\infty$. Combining this with Lemma~\ref{lemma:lower_bound_regret_of_any_policy_total_regret_incurred_over_all_assortments_during_the_exploration_period} yields $\textsc{Reg}(\pi) = +\infty$. On the other hand by Lemma~\ref{lemma:the_regret_of_pi_S_star_is_bounded},  $\textsc{Reg}(\pi(S^{\star})) \leq  \frac{\sum_{i \in S^{\star}} w_i + \tildew + w_0}{\tildew}$, which implies that $\textsc{Reg}(\pi) \geq \textsc{Reg}(\pi(S^{\star}))$ and completes the proof.
\end{proof}

\begin{proof}[Proof of Lemma~\ref{lemma:lower_bound_regret_of_any_policy_total_regret_incurred_over_all_assortments_during_the_exploration_period}]
The regret of any policy $\pi$ can be lower bounded by:
\begin{align*}
    \textsc{Reg}(\pi) &=  \expect\Big[\sum_{t=1}^{\infty} (\textsc{OPT} - \textsc{Rev}_t(S_t))\Big] \geq \expect\Big[\sum_{t=1}^{Z^{\pi}} (\textsc{OPT} - \textsc{Rev}_t(S_t))\Big]= \expect\Big[\sum_{t=1}^{Z^{\pi}} (\alpha^{\star} - \textsc{Rev}_1(S_t))\Big]\\
    &= \expect\Big[\sum\limits_{S \subseteq \mathcal{N}, |S| \leq c} (\alpha^{\star} - \textsc{Rev}_1(S)) N^{\pi}(S)\Big] = \sum\limits_{S \subseteq \mathcal{N}, |S| \leq c} (\alpha^{\star} - \textsc{Rev}_1(S)) \expect[N^{\pi}(S)],
\end{align*}
where the inequality uses that $\textsc{OPT} \geq \textsc{Rev}_t(S)$ if $t$ is after the exploration period, the second equality uses that $\expect[\sum_{t=1}^{Z^{\pi}} \alpha^{\star}] = \expect[\sum_{t=1}^{Z^{\pi}} \textsc{OPT}]$ (since $\alpha^{\star} = \expect_{w_1 \sim \mathcal{F}}[\textsc{OPT}]$ and $Z^{\pi}$ and $w_1$ are independent), and the third equality uses that $\textsc{Rev}_t(S_t) = \textsc{Rev}_1(S_t)$ for any round $t$ in the exploration period (as the attraction parameters of all products do not change during the exploration period). 
\end{proof}

\begin{proof}[Proof of Lemma~\ref{lemma:the_regret_of_pi_S_star_is_bounded}.]
Given that the entrant is purchased with probability $\frac{\tildew}{\sum_{i \in S^{\star}} w_i + \tildew + w_0}$, $\tau(S^{\star}) = \frac{\sum_{i \in S^{\star}} w_i + \tildew + w_0}{\tildew}$. Combining this with the fact that $\alpha^{\star} < 1$ and $r(S^{\star}) \geq 0$, 
$$\textsc{Reg}(\pi(S^{\star})) = \alpha^{\star} \tau(S^{\star}) - r(S^{\star}) \leq \tau(S^{\star}) = \frac{\sum_{i \in S^{\star}} w_i + \tildew + w_0}{\tildew}.$$
\end{proof}

\section{Omitted proofs from Section~\ref{sec:optimal_exploration_multiple_entrants}}

\subsection{Epoch-regret-minimizing policies are optimal (Lemma~\ref{lemma:epoch_optimal_policy_is_optimal})}\label{appendix_sec:there_exists_optimal_policy_which_is_stationary}
For a history $\mathcal{H} = \{(S_i, Y_i, w_{Y_i})\}_{i=1}^{t-1}$, let $M(\mathcal{H}) =\{Y_i: Y_i \in \{1, \ldots, m\}, i \in \{1, \ldots, t-1\} \}$ be the set of known entrants and $I(\mathcal{H}) = M(\mathcal{H})\cup  \{m+1, \ldots, n\}$ be the set of all known products. 

Let $\mathcal{E}(\mathcal{H}) = \{S \subseteq \mathcal{N}: |S| \leq c, S \cap (\mathcal{N} \setminus I(\mathcal{H})) \neq \emptyset \}$ be the set of all assortments which contain at least one unknown entrant given a history $\mathcal{H}$. Let $\tau(S; \mathcal{H})$ be the expected number of rounds and $r(S; \mathcal{H})$ be the expected reward until an unknown entrant is purchased if $S \in \mathcal{E}(\mathcal{H})$ is offered repeatedly starting from a history $\mathcal{H}$. Let $\textsc{OPT}(\mathcal{H}) = \expect[\textsc{OPT}|\mathcal{H}]$ be the expected ex-post optimum given a history $\mathcal{H}$. Let $\textsc{EpochReg}(S;\mathcal{H}) = \textsc{OPT}(\mathcal{H}) \cdot \tau(S;\mathcal{H}) -r(S; \mathcal{H})$ be the expected regret until an unknown entrant is purchased given a history $\mathcal{H}$ if $S \in \mathcal{E}(\mathcal{H})$ is offered repeatedly.

Let $w_i(\mathcal{H}) = w_{i}$ if $i \in I(\mathcal{H})$ and $w_i(\mathcal{H}) = \tildew$ otherwise. Let $\textsc{Rev}(S;\mathcal{H}) = \frac{\sum_{i \in S} w_i(\mathcal{H})}{\sum_{i \in S} w_i(\mathcal{H}) + w_0}$ be the expected revenue of assortment $S$ given history $\mathcal{H}$.  We say a history $\mathcal{H}$ is \textit{terminal} if an unknown entrant cannot increase the myopic optimum of known products in expectation, i.e., 
$\textsc{OPT}(\mathcal{H}) = \max\limits_{S \subseteq I(\mathcal{H}), |S| \leq c} \textsc{Rev}(S; \mathcal{H})$. 
The policy $\pi^{\star}$ of Lemma~\ref{lemma:epoch_optimal_policy_is_optimal} is such that 
\begin{enumerate}
    \item $\pi^{\star}(\mathcal{H}) \in \argmin\limits_{S \in \mathcal{E}(\mathcal{H})} \textsc{EpochReg}(S;\mathcal{H})$ if $\mathcal{H}$ is not terminal and 
    \item  $\pi^{\star}(\mathcal{H}) \in \argmax\limits_{S \subseteq I(\mathcal{H}), |S| \leq c} \textsc{Rev}(S; \mathcal{H})$ otherwise.
\end{enumerate}

For a policy $\pi$ and a history $\mathcal{H}$, let $\textsc{Rev}(\pi; \mathcal{H}) = \expect[\sum_{t=1}^{\infty} \textsc{OPT}-R_t|\mathcal{H}_1 = \mathcal{H}]$ be the regret of $\pi$ conditioned on starting from a history $\mathcal{H}$. To prove Lemma~\ref{lemma:epoch_optimal_policy_is_optimal}, the next lemma (proof in Appendix~\ref{appendix_subsec:proof_lemma_any_policy_has_nonegative_regret_given_terminal_history}) shows that any policy has a non-negative regret starting from a terminal history and that $\pi^{\star}$ has zero regret starting from a terminal history. 

\begin{lemma}\label{lemma:any_policy_has_a_nonnegative_regret_starting_from_a_terminal_history_and_pi_star_as_a_zero_regret_starting_from_a_terminal_history}
    For any terminal history $\mathcal{H}$ and any policy $\pi$, $\textsc{Reg}(\pi;\mathcal{H}) \geq 0$ and $\textsc{Reg}(\pi^{\star};\mathcal{H}) = 0$.
\end{lemma}

Let $J(\mathcal{H}) = \{w_i:i \in I(\mathcal{H})\}$ be the set of the known products' attraction parameters given a history $\mathcal{H}$. The following lemma (proof in Appendix~\ref{appendix_subsec:proof_lemma_optimal_regret_depends_only_on_known_products_attraction_parameters}) establishes that the regret of $\pi^{\star}$ given any history $\mathcal{H}$ depends on $\mathcal{H}$ only through the set of known products' attraction parameters $J(\mathcal{H})$.

\begin{lemma}\label{lemma:the_regret_of_pi_star_depends_on_H_only_through_J_H}
For any history $\mathcal{H}$, $\textsc{Reg}(\pi^{\star};\mathcal{H})$ depends on $\mathcal{H}$ only through the known products' attraction parameters $J(\mathcal{H})$. 
\end{lemma}
Hence, the regret of the optimal policy depends only on the attraction parameters of the known products in the initial history. We thus express the regret of $\pi^{\star}$ given history $\mathcal{H}$ as $\textsc{Reg}(\pi^{\star},J(\mathcal{H})) = \textsc{Reg}(\pi^{\star};\mathcal{H})$. The following lemma (proof in Appendix~\ref{appendix_subsec:proof_lemma_one_epoch_lookahead_decomposition_optimal_policy}) provides a one-epoch-lookahead decomposition of the regret of~$\pi^{\star}$.

\begin{lemma}\label{lemma:the_regret_of_pi_star_can_be_decomposed_into_the_first_epoch_regret_and_the_future_regret}
 For any non-terminal history $\mathcal{H}$, 
    $$\textsc{Reg}(\pi^{\star}, J(\mathcal{H}))= \min\limits_{S \in \mathcal{E}(\mathcal{H})}  \textsc{EpochReg}(S;\mathcal{H})+\expect_{w \sim \mathcal{F}} [ \textsc{Reg}(\pi^{\star},J(\mathcal{H}) \cup \{w\})].$$
\end{lemma}

Let $Z^{\pi, \mathcal{H}}$ be the first round at which an unknown entrant is purchased if a policy $\pi$ starts from a history $\mathcal{H}$. The following lemma (proof in Appendix~\ref{appendix_subsec:proof_lemma_epoch_regret_of_any_policy_is_at_least_minimum_epoch_regret_of_static_assortment_policies}) establishes that the expected regret of $\pi$ until $Z^{\pi, \mathcal{H}}$ is at least the minimum epoch regret achieved by repeatedly offering a fixed assortment. Let $$\textsc{EpochReg}^{\pi}(\mathcal{H}) = \expect\Bigg[\sum_{t=1}^{Z^{\pi, \mathcal{H}}} \textsc{OPT}(\mathcal{H}_t) - \textsc{Rev}(S_t, \mathcal{H}_t)|\mathcal{H}_1 = \mathcal{H}\Bigg]$$ be the expected regret of $\pi$ until the first purchase of an unknown entrant.

\begin{lemma}\label{lemma:epoch_regret_of_any_policy_is_at_least_the_minimum_epoch_regret_over_fixed_assortments}
    For any policy $\pi$ and non-terminal history $\mathcal{H}$, 
    $\textsc{EpochReg}^{\pi}(\mathcal{H}) \geq \min\limits_{S \in \mathcal{E}(\mathcal{H})}  \textsc{EpochReg}(S;\mathcal{H})$
    if $\expect[Z^{\pi, \mathcal{H}}] < \infty$, and $\textsc{EpochReg}^{\pi}(\mathcal{H}) = +\infty$ otherwise. 
\end{lemma}

The next lemma (proof in Appendix~\ref{appendix_subsec:regret_of_optimal_policy_is_bounded_below_and_above}) establishes that the regret of $\pi^{\star}$ for any history $\mathcal{H}$ is bounded from below and from above by constants independent of $\mathcal{H}$. 

\begin{lemma}\label{lemma:for_any_history_with_k_unknown_entrants_the_regret_of_pi_star_is_lower_bounded_and_upper_bounded_by_quantities_which_depend_on_k}
    For any history $\mathcal{H}$, $\textsc{Reg}(\pi^{\star};\mathcal{H}) \geq -m$ and  $\textsc{Reg}(\pi^{\star};\mathcal{H}) \leq \frac{m}{\tildew}$.
\end{lemma}

\begin{proof}[Proof of Lemma~\ref{lemma:epoch_optimal_policy_is_optimal}.]
For any policy $\pi$, we show that $\textsc{Reg}(\pi;\mathcal{H}) \geq \textsc{Reg}(\pi^{\star};\mathcal{H})$ for any history $\mathcal{H}$ by induction on the number of unknown entrants given $\mathcal{H}$. For the base case suppose that there are no unknown entrants given $\mathcal{H}$. Thus $\mathcal{H}$ is terminal and Lemma~\ref{lemma:any_policy_has_a_nonnegative_regret_starting_from_a_terminal_history_and_pi_star_as_a_zero_regret_starting_from_a_terminal_history} implies that $\textsc{Reg}(\pi;\mathcal{H}) \geq 0$ and $\textsc{Reg}(\pi^{\star};\mathcal{H}) = 0$ yielding the claim. Suppose the claim is true for any history $\mathcal{H}'$ with at most $k$ unknown entrants remaining and let $\mathcal{H}$ be a history with $k+1$ unknown entrants remaining. If $\mathcal{H}$ is terminal, then Lemma~\ref{lemma:any_policy_has_a_nonnegative_regret_starting_from_a_terminal_history_and_pi_star_as_a_zero_regret_starting_from_a_terminal_history} implies that $\textsc{Reg}(\pi;\mathcal{H}) \geq 0$ and $\textsc{Reg}(\pi^{\star};\mathcal{H}) = 0$ yielding the claim. Suppose $\mathcal{H}$ is not terminal. Let $Z = Z^{\pi, \mathcal{H}}$ as a shorthand. Recall that the random variables $\mathcal{H}_{Z}$, $ S_{Z}$,$ Y_{Z}$, and $ w_{Y_{Z}}$ are the history, the assortment offered, the unknown entrant chosen, and the attraction parameter of the unknown entrant at round $Z$ respectively. The regret of $\pi$ given $\mathcal{H}$ is lower bounded as
\begin{align}\label{eq:regret_pi_decomposition_epoch_and_after_epoch}
    \textsc{Reg}(\pi;\mathcal{H})& = \expect\Big[\sum_{t=1}^{Z} \textsc{OPT}(\mathcal{H}_t) - \textsc{Rev}(S_t, \mathcal{H}_t)|\mathcal{H}_1 = \mathcal{H}\Big] \nonumber\\
    &\hspace{0.5in}+\expect[\textsc{Reg}(\pi;\mathcal{H}_{Z} \cup (S_{Z}, Y_{Z}, w_{Y_{Z}}))|\mathcal{H}_1 = \mathcal{H}, Z< \infty] \prob[Z < \infty] \nonumber\\
    &\geq \expect\Big[\sum_{t=1}^{Z} \textsc{OPT}(\mathcal{H}_t) - \textsc{Rev}(S_t, \mathcal{H}_t)|\mathcal{H}_1 = \mathcal{H}\Big] \nonumber\\
    &\hspace{0.5in}+\expect [\textsc{Reg}(\pi^{\star};\mathcal{H}_{Z} \cup (S_{Z}, Y_{Z}, w_{Y_{Z}}))|\mathcal{H}_1 = \mathcal{H}, Z< \infty] \prob[Z < \infty]\nonumber\\
    &= \textsc{EpochReg}^{\pi}(\mathcal{H}) + \expect_{w \sim \mathcal{F}}[\textsc{Reg}(\pi^{\star},J(\mathcal{H}) \cup \{w\})]\prob[Z < \infty]
\end{align}
where the first inequality uses $\textsc{Reg}(\pi;\mathcal{H}_{Z} \cup (S_{Z}, Y_{Z}, w_{Y_{Z}})) \geq \textsc{Reg}(\pi^{\star};\mathcal{H}_{Z} \cup (S_{Z}, Y_{Z}, w_{Y_{Z}}))$ from the inductive hypothesis, and the second equality uses $\textsc{Reg}(\pi^{\star};\mathcal{H}_{Z} \cup (S_{Z}, Y_{Z}, w_{Y_{Z}})) = \textsc{Reg}(\pi^{\star}, J(\mathcal{H}) \cup \{w_{Y_Z}\})$ by Lemma~\ref{lemma:the_regret_of_pi_star_depends_on_H_only_through_J_H} and the fact that $w_{Y_Z} \sim \mathcal{F}$.

\paragraph{Case 1: $\expect[Z] = \infty$.} Lemma~\ref{lemma:epoch_regret_of_any_policy_is_at_least_the_minimum_epoch_regret_over_fixed_assortments} yields that $\textsc{EpochReg}^{\pi}(\mathcal{H}) = +\infty$ and Lemma~\ref{lemma:for_any_history_with_k_unknown_entrants_the_regret_of_pi_star_is_lower_bounded_and_upper_bounded_by_quantities_which_depend_on_k} implies that $\expect_{w \sim \mathcal{F}}[\textsc{Reg}(\pi^{\star},J(\mathcal{H}) \cup \{w\})]\prob[Z < \infty] \geq -m$. Combining this with \eqref{eq:regret_pi_decomposition_epoch_and_after_epoch} yields that $\textsc{Reg}(\pi;\mathcal{H})= + \infty$. Given that $\textsc{Reg}(\pi^{\star};\mathcal{H})< \infty$ by Lemma~\ref{lemma:for_any_history_with_k_unknown_entrants_the_regret_of_pi_star_is_lower_bounded_and_upper_bounded_by_quantities_which_depend_on_k}, $\textsc{Reg}(\pi;\mathcal{H}) \geq \textsc{Reg}(\pi^{\star};\mathcal{H})$. 

\paragraph{Case 2: $\expect[Z] < \infty$.} Thus, $\prob[Z < \infty] =1$. Combining this with \eqref{eq:regret_pi_decomposition_epoch_and_after_epoch}, the regret of $\pi$ can be lower bound as
\begin{align*}
    \textsc{Reg}(\pi;\mathcal{H}) &\geq \textsc{EpochReg}^{\pi}(\mathcal{H}) + \expect_{w \sim \mathcal{F}}[\textsc{Reg}(\pi^{\star},J(\mathcal{H}) \cup \{w\})]\prob[Z < \infty]\\
    &=  \textsc{EpochReg}^{\pi}(\mathcal{H}) + \expect_{w \sim \mathcal{F}}[\textsc{Reg}(\pi^{\star},J(\mathcal{H}) \cup \{w\})]\\
    &\geq \min\limits_{S \in \mathcal{E}(\mathcal{H})}  \textsc{EpochReg}(S;\mathcal{H})+\expect_{w \sim \mathcal{F}} [ \textsc{Reg}(\pi^{\star},J(\mathcal{H}) \cup \{w\})] = \textsc{Reg}(\pi^{\star};\mathcal{H})
\end{align*}
where the first inequality uses \eqref{eq:regret_pi_decomposition_epoch_and_after_epoch}, the first equality uses $\prob[Z < \infty] =1$, the second inequality uses $\textsc{EpochReg}^{\pi}(\mathcal{H}) \geq \min\limits_{S \in \mathcal{E}(\mathcal{H})}  \textsc{EpochReg}(S;\mathcal{H})$ by Lemma~\ref{lemma:epoch_regret_of_any_policy_is_at_least_the_minimum_epoch_regret_over_fixed_assortments}, and the last equality uses the one-epoch-lookahead decomposition of $\pi^{\star}$ by Lemma~\ref{lemma:the_regret_of_pi_star_can_be_decomposed_into_the_first_epoch_regret_and_the_future_regret}.
\end{proof}

\subsection{Optimal completion to a fixed number of unknown products (Lemma~\ref{lemma:optimal_completion_given_ell_unknown_prods})}\label{appendix_sec_proof_optimal_completion_given_ell_unknown_prods}
\begin{proof}[Proof of Lemma~\ref{lemma:optimal_completion_given_ell_unknown_prods}.]
   For a set $S' \subseteq \mathcal{I}_t$ of at most $c-\ell$ known products, let $o(U_{\ell} \cup S' ; \mathcal{H}_t)$ be the expected number of rounds in which the outside option is chosen in the current epoch if the platform offers the assortment $U_{\ell} \cup S'$. If the platform offers assortment $U_{\ell} \cup S'$, the regret is $\textsc{OPT}_t$, when the outside option is chosen which happens $o(U_{\ell} \cup S' ; \mathcal{H}_t)$ times in expectation. When a product is purchased, the reward is $1$ and hence the regret is $\textsc{OPT}_t -1$; this happens $\tau(U_{\ell} \cup S' ; \mathcal{H}_t)-o(U_{\ell} \cup S' ; \mathcal{H}_t)$ times in expectation. Thus, \begin{equation}\label{eq:epoch_regret_as_a_function_of_tau_and_o}
   \textsc{EpochReg}_t(U_{\ell} \cup S') = (\textsc{OPT}_t-1)(\tau(U_{\ell} \cup S'; \mathcal{H}_t)- o(U_{\ell} \cup S'; \mathcal{H}_t)) +\textsc{OPT}_t o(U_{\ell} \cup S'; \mathcal{H}_t).
   \end{equation}
Given that the probability that some entrant is chosen if $U_{\ell} \cup S'$ is offered is $\frac{\ell \tildew}{\sum_{i \in S'} w_i + \ell \tildew + w_0}$, $$\tau(U_{\ell} \cup S'; \mathcal{H}_t) = \frac{\sum_{i \in S'} w_i + \ell \tildew + w_0}{\ell \tildew}.$$ Next, we show that $o(U_{\ell} \cup S'; \mathcal{H}_t)$ does not depend on $S'$:
$$o(U_{\ell} \cup S'; \mathcal{H}_t) = \prob[\text{outside option is chosen}|U_{\ell} \cup S'] \tau(U_{\ell} \cup S'; \mathcal{H}_t) = \frac{\frac{w_0}{\sum_{i \in S'} w_i + \ell \tildew + w_0}}{\frac{\ell \tildew}{\sum_{i \in S'} w_i + \ell \tildew + w_0}} = \frac{w_0}{\ell \tildew}.$$
Combining this with \eqref{eq:epoch_regret_as_a_function_of_tau_and_o}, minimizing $ \textsc{EpochReg}_t(U_{\ell} \cup S')$ is equivalent to maximizing $\tau(U_{\ell} \cup S'; \mathcal{H}_t)$ (as $\textsc{OPT}_t < 1$). Using that $\tau(U_{\ell} \cup S'; \mathcal{H}_t) = \frac{\sum_{i \in S'} w_i + \ell \tildew + w_0}{\ell \tildew}$, $\tau(U_{\ell} \cup S'; \mathcal{H}_t)$ is maximized when $\sum_{i \in S'} w_i$ is maximized which occurs when $S' = S^{\star}_{(c-\ell)}(\mathcal{I}_t)$ yielding the lemma. \end{proof}

\subsection{Any policy has non-negative regret given a terminal history (Lemma~\ref{lemma:any_policy_has_a_nonnegative_regret_starting_from_a_terminal_history_and_pi_star_as_a_zero_regret_starting_from_a_terminal_history})}\label{appendix_subsec:proof_lemma_any_policy_has_nonegative_regret_given_terminal_history}

To prove Lemma~\ref{lemma:any_policy_has_a_nonnegative_regret_starting_from_a_terminal_history_and_pi_star_as_a_zero_regret_starting_from_a_terminal_history} we establish that if the initial history is terminal then any subsequent history is also terminal (Lemma~\ref{lemma:if_H_1_is_terminal_H_t_is_terminal_with_probability_1}) and that the expected-ex post optimum equals the revenue of the myopic revenue-maximizing assortment for any terminal history (Lemma~\ref{lemma:if_H_is_terminal_the_expected_opt_equals_the_maximum_instantaneous_revenue}). 

\begin{lemma}\label{lemma:if_H_1_is_terminal_H_t_is_terminal_with_probability_1}
    For any terminal initial history $\mathcal{H}_1$, round $t$, and policy $\pi$, the history $\mathcal{H}_t$ under $\pi$ is terminal with probability $1$. 
\end{lemma}

\begin{lemma}\label{lemma:if_H_is_terminal_the_expected_opt_equals_the_maximum_instantaneous_revenue}
    For any terminal history $\mathcal{H}$, $\textsc{OPT}(\mathcal{H}) = \max\limits_{S \subseteq \mathcal{N}, |S| \leq c}\textsc{Rev}(S;\mathcal{H})$
\end{lemma}

\begin{proof}[Proof of Lemma~\ref{lemma:any_policy_has_a_nonnegative_regret_starting_from_a_terminal_history_and_pi_star_as_a_zero_regret_starting_from_a_terminal_history}.]
Given that $\mathcal{H}_1 = \mathcal{H}$ is terminal, the expected regret of $\pi$ at any round $t$ is non-negative:
$$\expect[\textsc{OPT} -R_t] = \expect[\textsc{OPT}(\mathcal{H}_t) -\textsc{Rev}(S_t;\mathcal{H}_t)] \geq  \expect\Big[\textsc{OPT}(\mathcal{H}_t) -  \max\limits_{S \subseteq \mathcal{N}, |S| \leq c}\textsc{Rev}(S;\mathcal{H}_t) \Big] = 0$$
where the last equality uses that $\mathcal{H}_t$ is terminal with probability $1$ by  Lemma~\ref{lemma:if_H_1_is_terminal_H_t_is_terminal_with_probability_1} and the fact that $\textsc{OPT}(\mathcal{H}_t) -  \max\limits_{S \subseteq \mathcal{N}, |S| \leq c}\textsc{Rev}(S;\mathcal{H}_t) = 0$ by Lemma~\ref{lemma:if_H_is_terminal_the_expected_opt_equals_the_maximum_instantaneous_revenue}. Thus, $\textsc{Reg}(\pi;\mathcal{H}) = \liminf_{T \to \infty} \expect \Bigg[\sum_{t=1}^{T} \textsc{OPT} - R_t \Bigg] \geq 0.$ 

Given that $\textsc{Rev}(\pi^{\star}(\mathcal{H}_t);\mathcal{H}_t) = \max\limits_{S \subseteq I(\mathcal{H}_t), |S| \leq c} \textsc{Rev}(S;\mathcal{H}_t)$, the expected regret of $\pi^{\star}$ at any round $t$ is $0$:
$$\expect[\textsc{OPT} -R_t] = \expect[\textsc{OPT}(\mathcal{H}_t) -\textsc{Rev}(\pi^{\star}(\mathcal{H}_t);\mathcal{H}_t)] =  \expect\Big[\textsc{OPT}(\mathcal{H}_t) -  \max\limits_{S \subseteq I(\mathcal{H}_t), |S| \leq c}\textsc{Rev}(S;\mathcal{H}_t) \Big]=0,$$
which holds as $\mathcal{H}_t$ is terminal (Lemma~\ref{lemma:if_H_1_is_terminal_H_t_is_terminal_with_probability_1}). Thus, $\textsc{Reg}(\pi^{\star};\mathcal{H}) =\liminf_{T \to \infty} \expect \Bigg[\sum_{t=1}^{T} \textsc{OPT} -R_t \Bigg] =0.$
\end{proof}

\begin{proof}[Proof of Lemma~\ref{lemma:if_H_1_is_terminal_H_t_is_terminal_with_probability_1}.]
The maximum revenue achievable by known products at round $t$ is at least the maximum revenue from known products at round $1$:
\begin{equation}\label{eq:knwon_opt_given_H_t_is_at_least_known_opt_given_H_1}
    \max\limits_{S \subseteq I(\mathcal{H}_t), |S| \leq c} \textsc{Rev}(S; \mathcal{H}_t) \geq  \max\limits_{S \subseteq I(\mathcal{H}_1), |S| \leq c} \textsc{Rev}(S; \mathcal{H}_t) = \max\limits_{S \subseteq I(\mathcal{H}_1), |S| \leq c} \textsc{Rev}(S; \mathcal{H}_1) = \textsc{OPT}(\mathcal{H}_1)
\end{equation}
where the inequality uses the $I(\mathcal{H}_1) \subseteq I(\mathcal{H}_t)$ as any known product at round $1$ is also known at round $t$ and the equality uses that 
$\textsc{Rev}(S; \mathcal{H}_t) = \textsc{Rev}(S;\mathcal{H}_1)$ for any $S \subseteq I(\mathcal{H}_1)$ as any known product at round $1$ has the same attraction parameter at round $t$, and the second equality uses the fact that $\mathcal{H}_1$ is terminal. The expected ex-post optimum at round $t$ is lower bounded by
\begin{align}\label{ineq:expected_ex_post_opt_rount_t_is_at_least_opt_revenue_from_known_is_at_least_expected_ex_post_at_round_1}
    \textsc{OPT}(\mathcal{H}_t) & = \expect[\max_{S \subseteq \mathcal{N}, |S| \leq c}\textsc{Rev}(S; \mathcal{H}_t)]\geq \max\limits_{S \subseteq I(\mathcal{H}_t), |S| \leq c} \textsc{Rev}(S; \mathcal{H}_t) \nonumber\\
    &\geq \textsc{OPT}(\mathcal{H}_1) = \expect[\textsc{OPT}|\mathcal{H}_1] =  \expect[\expect[\textsc{OPT}|\mathcal{H}_t]|\mathcal{H}_1] =\expect[\textsc{OPT}(\mathcal{H}_t)|\mathcal{H}_1]
\end{align}
where the first inequality uses that $I(\mathcal{H}_t) \subseteq \mathcal{N}$, 
the second inequality uses \eqref{eq:knwon_opt_given_H_t_is_at_least_known_opt_given_H_1}. Therefore for all realizations of $\mathcal{H}_t$, $\textsc{OPT}(\mathcal{H}_t) = \expect[\textsc{OPT}(\mathcal{H}_t)|\mathcal{H}_1]$; thus all inequalities in \eqref{ineq:expected_ex_post_opt_rount_t_is_at_least_opt_revenue_from_known_is_at_least_expected_ex_post_at_round_1} are equalities. In particular, it holds that $\textsc{OPT}(\mathcal{H}_t)=  \max\limits_{S \subseteq I(\mathcal{H}_t), |S| \leq c} \textsc{Rev}(S; \mathcal{H}_t)$, i.e., $\mathcal{H}_t$ is terminal.
\end{proof}

\begin{proof}[Proof of Lemma~\ref{lemma:if_H_is_terminal_the_expected_opt_equals_the_maximum_instantaneous_revenue}.]
Given that the history $\mathcal{H}$ is terminal the expected ex-post optimum is equal to maximum revenue from known products which is less than the maximum revenue from all products, i.e., $$\textsc{OPT}(\mathcal{H}) = \max\limits_{S \subseteq I(\mathcal{H}), |S| \leq c}\textsc{Rev}(S;\mathcal{H}) \leq \max\limits_{S \subseteq \mathcal{N}, |S| \leq c}\textsc{Rev}(S;\mathcal{H}).$$ 
For the sake of contradiction suppose that $\textsc{OPT}(\mathcal{H}) <\max\limits_{S \subseteq \mathcal{N}, |S| \leq c}\textsc{Rev}(S;\mathcal{H})$. Thus, there exists an assortment~$S'$ with at least one unknown entrant such that $\textsc{Rev}(S';\mathcal{H}) > \textsc{OPT}(\mathcal{H}) = \max\limits_{S \subseteq I(\mathcal{H}), |S| \leq c}\textsc{Rev}(S;\mathcal{H})$. Let~$\mathcal{E}^{\geq h}$ be the event the unknown entrants in $S'$ have attraction parameters at least $\tildew$ given the history~$\mathcal{H}$ and let $\mathcal{E}^{< h}$ be the complement of~$\mathcal{E}^{\geq h}$ .  The law of total probability yields 
\begin{equation}\label{eq:law_of_total_probability_event_ell_unknown_higher_than_h}
    \textsc{OPT}(\mathcal{H}) =\expect[\textsc{OPT}|\mathcal{E}^{\geq h}]\prob[\mathcal{E}^{\geq h}] + \expect[\textsc{OPT}|\mathcal{E}^{< h}]\prob[\mathcal{E}^{< h}].
\end{equation}
We denote the number of unknown entrants by $\ell$, the set of unknown entrants by $A_{\ell}$ , and the sum of the attraction parameters of known products in $S'$ by $W_{\textsc{known}}$. Conditioned on $\mathcal{E}^{\geq h}$, the expectation in the first term is at least

\begin{align}\label{eq:lower_bound_expectation_all_at_least_h}
    \expect[\textsc{OPT}|\mathcal{E}^{\geq h}] &\geq \expect \Bigg[\frac{\sum_{i \in S'} w_i }{\sum_{i \in S'} w_i + 1} \Bigg|\mathcal{E}^{\geq h}\Bigg] = \expect \Bigg[\frac{\sum_{i \in A_{\ell}} w_i + W_{\textsc{known}}}{\sum_{i \in A_{\ell}} w_i + W_{\textsc{known}} + 1} \Bigg|\mathcal{E}^{\geq h}\Bigg] \geq \frac{\ell \tildew + W_{\textsc{known}}}{\ell \tildew + W_{\textsc{known}} + 1} \nonumber\\
    &= \textsc{Rev}(S';\mathcal{H}) > \textsc{OPT}(\mathcal{H})
\end{align}
where the second inequality uses that $w_i \geq \tildew$ for $i \in A_\ell$ conditioned on $\mathcal{E}^{\geq h}$ and that $\frac{x +W_{\textsc{known}}}{x +W_{\textsc{known}} + 1}$ is strictly increasing in $x$; the last inequality uses the choice of $S'$. The expectation in the second term is at least
\begin{equation}\label{eq:lower_bound_expectation_at_least_one_smaller_than_h}
\expect[\textsc{OPT}|\mathcal{E}^{< h}]  = \expect\Big[\max\limits_{S \subseteq \mathcal{N}, |S| \leq c}\textsc{Rev}(S;\mathcal{H})|\mathcal{E}^{< h} \Big] \geq \max\limits_{S \subseteq I(\mathcal{H}), |S| \leq c}\textsc{Rev}(S;\mathcal{H}) = \textsc{OPT}(\mathcal{H})
\end{equation}
where the inequality uses that $\mathcal{I}(\mathcal{H}) \subseteq \mathcal{N}$ and last equality uses that $\mathcal{H}$ is terminal. Finally, 
\begin{equation}\label{ineq:prob_strictly_pos}
    \prob[\mathcal{E}^{\geq h}] = \prod_{i \in A_{\ell}} \prob_{w_i \sim \mathcal{F}}[w_i \geq \tildew] > 0
\end{equation}
because all unknown entrants are independent and  $\tildew \in [\underline{\theta}, \overline{\theta}]$. Combining \eqref{eq:law_of_total_probability_event_ell_unknown_higher_than_h}, \eqref{eq:lower_bound_expectation_all_at_least_h}, \eqref{eq:lower_bound_expectation_at_least_one_smaller_than_h}, and \eqref{ineq:prob_strictly_pos}: 
$$\textsc{OPT}(\mathcal{H}) = \expect[\textsc{OPT}|\mathcal{E}^{\geq h}]\prob[\mathcal{E}^{\geq h}] + \expect[\textsc{OPT}|\mathcal{E}^{< h}]\prob[\mathcal{E}^{< h}] > \textsc{OPT}(\mathcal{H})\prob[\mathcal{E}^{\geq h}]  + \textsc{OPT}(\mathcal{H})\prob[\mathcal{E}^{< h}] = \textsc{OPT}(\mathcal{H})$$ that implies $\textsc{OPT}(\mathcal{H}) > \textsc{OPT}(\mathcal{H})$ and thus is a contradiction.
\end{proof}

\subsection{Optimal regret depends only on known products' parameters (Lemma~\ref{lemma:the_regret_of_pi_star_depends_on_H_only_through_J_H})}\label{appendix_subsec:proof_lemma_optimal_regret_depends_only_on_known_products_attraction_parameters}

To prove Lemma~\ref{lemma:the_regret_of_pi_star_depends_on_H_only_through_J_H}, the next lemma (proof in Appendix~\ref{appendix_subsec:proof_lemma_B8}) shows that the expected-ex post optimum, the myopic revenue from known products, and the minimum epoch regret given a history $\mathcal{H}$ depend on $\mathcal{H}$ only through the attraction parameters of the known products in $\mathcal{H}$.

\begin{lemma}\label{lemma:for_any_history_H_the_expected_ex_post_optimum_the_maximum_revenue_from_known_products_and_the_minimum_epoch_regret_depend_on_H_only_through_J(H)}
    For any history $\mathcal{H}$, the expected ex-post optimum $\textsc{OPT}(\mathcal{H})$, the myopic revenue from known products$\max\limits_{S \subseteq I(\mathcal{H}), |S| \leq c} \textsc{Rev}(S; \mathcal{H})$, and the minimum epoch regret $\min_{S \in \mathcal{E}(\mathcal{H})} \textsc{EpochReg}(S;\mathcal{H})$ depend on $\mathcal{H}$ only through $J(\mathcal{H})$. 
\end{lemma}

The next lemma (proof in Appendix~\ref{appendix_subsec:proof_lemma_B9}) provides a one-epoch-lookahead decomposition for the regret of $\pi^{\star}$. For a history $\mathcal{H}$, let $Z^{\mathcal{H}}_{\star}$ be the first round an unknown entrant is purchased if $\pi^{\star}$ starts from a history $\mathcal{H}$. Recall that the random variables $\mathcal{H}_{Z^{\mathcal{H}}_{\star}}$, $ S_{Z^{\mathcal{H}}_{\star}}$,$ Y_{Z^{\mathcal{H}}_{\star}}$, and $ w_{Y_{Z^{\mathcal{H}}_{\star}}}$ are the history, the assortment offered, the unknown entrant chosen, and the attraction parameter of the unknown entrant at round $Z^{\mathcal{H}}_{\star}$ respectively. 

\begin{lemma}\label{lemma:the_regret_of_pi_star_can_be_decomposed_into_the_first_epoch_regret_and_the_future_regret_general_form}
    For any non-terminal history $\mathcal{H}$, the regret of $\pi^{\star}$ can be decomposed as
     $$\textsc{Reg}(\pi^{\star};\mathcal{H}) = \min\limits_{S \in \mathcal{E}(\mathcal{H})} \textsc{EpochReg}(S;\mathcal{H}) \nonumber + \expect[\textsc{Reg}(\pi^{\star};\mathcal{H}_{Z^{\mathcal{H}}_{\star}} \cup (S_{Z^{\mathcal{H}}_{\star}}, Y_{Z^{\mathcal{H}}_{\star}}, w_{Y_{Z^{\mathcal{H}}_{\star}}}))|\mathcal{H}_1 = \mathcal{H}].$$
\end{lemma}

\begin{proof}[Proof of Lemma~\ref{lemma:the_regret_of_pi_star_depends_on_H_only_through_J_H}.]
To prove the lemma it suffices to show that 
\begin{equation}\label{preoperty:regret_of_pi_star_is_the_for_any_two_histories_with_the_same_set_of_known_products_attraction_parameters}
    \textsc{Reg}(\pi^{\star};\mathcal{H}^1) = \textsc{Reg}(\pi^{\star};\mathcal{H}^2) \text{ for any $\mathcal{H}^1$ and $\mathcal{H}^2$ with $J(\mathcal{H}^1) =J(\mathcal{H}^2)$}.
\end{equation}
We prove \eqref{preoperty:regret_of_pi_star_is_the_for_any_two_histories_with_the_same_set_of_known_products_attraction_parameters} by induction on the number of unknown entrants remaining given $\mathcal{H}^1$ and $\mathcal{H}^2$. 

\paragraph{Base case.} There are no unknown entrants given $\mathcal{H}^1$ and $\mathcal{H}^2$; thus both $\mathcal{H}^1$ and $\mathcal{H}^2$ are terminal. Therefore,  $\textsc{Reg}(\pi^{\star};\mathcal{H}^1) = \textsc{Reg}(\pi^{\star};\mathcal{H}^2) = 0$ by Lemma~\ref{lemma:any_policy_has_a_nonnegative_regret_starting_from_a_terminal_history_and_pi_star_as_a_zero_regret_starting_from_a_terminal_history}.

\paragraph{Induction step.}
Suppose \eqref{preoperty:regret_of_pi_star_is_the_for_any_two_histories_with_the_same_set_of_known_products_attraction_parameters} holds for any two histories $\tilde{\mathcal{H}}^1$ and $\tilde{\mathcal{H}}^2$ with $J(\tilde{\mathcal{H}}^1) = J(\tilde{\mathcal{H}}^2)$ and at most $k$ unknown entrants left given $\tilde{\mathcal{H}}^1$ and $\tilde{\mathcal{H}}^2$. Let $\mathcal{H}^1$ and $\mathcal{H}^2$ be such that $J(\mathcal{H}^1) = J(\mathcal{H}^2)$ and there are $k+1$ unknown entrants left given $\mathcal{H}^1$ and $\mathcal{H}^2$. We consider cases based on whether $\mathcal{H}^1$ is terminal. 
\begin{itemize}
    \item \textbf{$\mathcal{H}^1$ is terminal.} Then $\mathcal{H}^2$ is terminal as:
    $$\textsc{OPT}(\mathcal{H}^2) = \textsc{OPT}(\mathcal{H}^1) = \max\limits_{S \subseteq I(\mathcal{H}^1), |S| \leq c} \textsc{Rev}(S; \mathcal{H}^1) = \max\limits_{S \subseteq I(\mathcal{H}^2), |S| \leq c} \textsc{Rev}(S; \mathcal{H}^2) $$
    where the first and third equalities use $J(\mathcal{H}^1)= J(\mathcal{H}^2)$ and Lemma~\ref{lemma:for_any_history_H_the_expected_ex_post_optimum_the_maximum_revenue_from_known_products_and_the_minimum_epoch_regret_depend_on_H_only_through_J(H)}, and the second equality uses that $\mathcal{H}^1$ is terminal. Therefore, $\textsc{Reg}(\pi^{\star};\mathcal{H}^1) =\textsc{Reg}(\pi^{\star};\mathcal{H}^2) =0$ by Lemma~\ref{lemma:any_policy_has_a_nonnegative_regret_starting_from_a_terminal_history_and_pi_star_as_a_zero_regret_starting_from_a_terminal_history}.
    \item    
\textbf{$\mathcal{H}^1$ is not terminal} Thus, $\pi^{\star}$ offers at least one unknown entrant. Let $Z^1$ be the round at which the first unknown entrant is purchased under $\pi^{\star}$ given $\mathcal{H}^1$. By the induction hypothesis there exists a function $f$ such that $\textsc{Reg}(\pi^{\star};\mathcal{H}) = f(J(\mathcal{H}))$ for any history $\mathcal{H}$ with at most $k$ unknown entrants. The regret of $\pi^{\star}$ is decomposed as
\begin{align}\label{eq:one_epoch_lookahead_decomposition_pi_star_given_H1}
      \textsc{Reg}(\pi^{\star};\mathcal{H}^1) &= \min\limits_{S \in \mathcal{E}(\mathcal{H}^1)} \textsc{EpochReg}(S;\mathcal{H}^1) \nonumber \\
      &+ \expect\limits_{\mathcal{H}_{Z^1}, S_{Z^1}, Y_{Z^1}, w_{Y_{Z^1}}} [\textsc{Reg}(\pi^{\star};\mathcal{H}_{Z^1} \cup (S_{Z^1}, Y_{Z^1}, w_{Y_{Z^1}}))|\mathcal{H}_1 = \mathcal{H}^1] \nonumber \\
      &= \min\limits_{S \in \mathcal{E}(\mathcal{H}^1)} \textsc{EpochReg}(S;\mathcal{H}^1) + \expect\limits_{\mathcal{H}_{Z^1}, S_{Z^1}, Y_{Z^1}, w_{Y_{Z^1}}} [f(J(\mathcal{H}^1) \cup \{w_{Y_{Z^1}}\})|\mathcal{H}_1 = \mathcal{H}^1] \nonumber \\
      &= \min\limits_{S \in \mathcal{E}(\mathcal{H}^1)} \textsc{EpochReg}(S;\mathcal{H}^1) + \expect\limits_{w \sim \mathcal{F}} [f(J(\mathcal{H}^1) \cup \{w\})].
\end{align}
where the first equality uses Lemma~\ref{lemma:the_regret_of_pi_star_can_be_decomposed_into_the_first_epoch_regret_and_the_future_regret_general_form}, the second equality uses $\textsc{Reg}(\pi^{\star};\mathcal{H}_{Z^1} \cup (S_{Z^1}, Y_{Z^1}, w_{Y_{Z^1}})) =f(J(\mathcal{H}^1) \cup \{w_{Y_{Z^1}}\})$ as $\mathcal{H}_{Z^1} \cup (S_{Z^1}, Y_{Z^1}, w_{Y_{Z^1}})$ contains $k$ unknown entrants and $J(\mathcal{H}_{Z^1} \cup (S_{Z^1}, Y_{Z^1}, w_{Y_{Z^1}})) = J(\mathcal{H}^1) \cup \{w_{Y_{Z^1}}\}$, and the third equality uses that $w_{Y_{Z^1}}$ is independent of $\mathcal{H}_{Z^1}$, $S_{Z^1}$, and $Y_{Z^1}$ and that $w_{Y_{Z^1}} \sim \mathcal{F}$. Given that $\mathcal{H}^1$ is not terminal, $\mathcal{H}^2$ is not terminal. Due to an analogous argument for $\mathcal{H}^2$, 
\begin{equation}\label{eq:one_epoch_lookahead_decomposition_pi_star_given_H2}
     \textsc{Reg}(\pi^{\star};\mathcal{H}^2) = \min\limits_{S \in \mathcal{E}(\mathcal{H}^2)} \textsc{EpochReg}(S;\mathcal{H}^2) + \expect\limits_{w \sim \mathcal{F}} [f(J(\mathcal{H}^2) \cup \{w\})].
\end{equation}
We conclude that $\textsc{Reg}(\pi^{\star};\mathcal{H}^1) =  \textsc{Reg}(\pi^{\star};\mathcal{H}^2)$ because the first terms of the right-hand sides of \eqref{eq:one_epoch_lookahead_decomposition_pi_star_given_H1} and \eqref{eq:one_epoch_lookahead_decomposition_pi_star_given_H2} are equal (by Lemma~\ref{lemma:for_any_history_H_the_expected_ex_post_optimum_the_maximum_revenue_from_known_products_and_the_minimum_epoch_regret_depend_on_H_only_through_J(H)} as $J(\mathcal{H}^1) = J(\mathcal{H}^2)$) and the second terms are equal (as $J(\mathcal{H}^1) = J(\mathcal{H}^2)$).
\end{itemize}
\end{proof}

\subsection{Dependence of optimal values on known products' parameters (Lemma~\ref{lemma:for_any_history_H_the_expected_ex_post_optimum_the_maximum_revenue_from_known_products_and_the_minimum_epoch_regret_depend_on_H_only_through_J(H)}) }\label{appendix_subsec:proof_lemma_B8}

\begin{proof}[Proof of Lemma~\ref{lemma:for_any_history_H_the_expected_ex_post_optimum_the_maximum_revenue_from_known_products_and_the_minimum_epoch_regret_depend_on_H_only_through_J(H)}.]

Let $\mathcal{H}^1$ and $\mathcal{H}^2$ be two histories with $J(\mathcal{H}^1) = J(\mathcal{H}^2)$. To show the lemma it suffices to show that $\textsc{OPT}(\mathcal{H}^1) = \textsc{OPT}(\mathcal{H}^2)$, $\max\limits_{S \subseteq I(\mathcal{H}^1), |S| \leq c} \textsc{Rev}(S;\mathcal{H}^1) = \max\limits_{S \subseteq I(\mathcal{H}^2), |S| \leq c} \textsc{Rev}(S;\mathcal{H}^2)$, and $\min_{S \in \mathcal{E}(\mathcal{H}^1)} \textsc{EpochReg}(S;\mathcal{H}^1) =\min_{S \in \mathcal{E}(\mathcal{H}^2)} \textsc{EpochReg}(S;\mathcal{H}^2)$. A useful quantity in our analysis is a one-to-one mapping $a: \{1, \ldots, n\} \to \{1, \ldots, n\}$ which maps each known product under $\mathcal{H}^1$ to a corresponding known product under $\mathcal{H}^2$ with the same attraction parameter. Formally, $a$ is such that (i) $a$ maps each incumbent to itself, i.e, $a(i) = i$ for $i \in \{m+1, \ldots, n\}$, (ii) $a$ maps each known entrant in $\mathcal{H}^1$ to a known entrant in $\mathcal{H}^2$, so that $w_i(\mathcal{H}^1) = w_{a(i)}(\mathcal{H}^2)$ for each known entrant $i$ under $\mathcal{H}^1$. Let $a(S) = \{a(i): i \in S\}$ be the image of a set $S \subseteq \mathcal{N}$ under $a$. Let $U(\mathcal{H})$ be the set of unknown entrants given a history $\mathcal{H}$ and $U^i = U(\mathcal{H}^i)$ for $i \in \{1,2\}$ as a shorthand.

\paragraph{Establishing $\textsc{OPT}(\mathcal{H}^1) = \textsc{OPT}(\mathcal{H}^2)$.} For a set of realizations $\{v_j\}_{j \in U^2}$ for the attraction parameters of the unknown entrants under $\mathcal{H}^2$, we couple the event $\mathcal{E}^2(\{v_j\}_{j \in U^2}) = \{w_j = v_j \text{ for } j \in U^2\}$ that unknown entrant $j$ has attraction parameter $v_j$ given $\mathcal{H}^2$ to the event $\mathcal{E}^1(\{v_{a(i)}\}_{i \in U^1}) = \{w_{i} = v_{a(i)} \text{ for } i \in U^1\}$ that unknown entrant $i$ has the attraction parameter of its image under $a$, $v_{a(i)}$, given $\mathcal{H}^1$. Let $\textsc{OPT}^1(\{v_{a(i)}\}_{i \in U^1})$ be the optimum given $\mathcal{E}^1(\{v_{a(i)}\}_{i \in U^1})$ and $\mathcal{H}^1$ and $\textsc{OPT}^2(\{v_j\}_{j \in U^2})$ the optimum $\mathcal{E}^2(\{v_j\}_{j \in U^2})$ and $\mathcal{H}^2$. We show that $\textsc{OPT}^1(\{v_{a(i)}\}_{i \in U^1}) = \textsc{OPT}^2(\{v_j\}_{j \in U^2})$. For a set $S$ and a history $\mathcal{H}$, let $S^{\textsc{known}}(\mathcal{H})$ and  $S^{\textsc{unknown}}(\mathcal{H})$ be the subset of attraction parameters of the known and unknown products in $S$ given $\mathcal{H}$. For a history $\mathcal{H}$ and realization for the attraction parameters of the unknown entrants $\{v_i\}_{i \in U(\mathcal{H})}$, let $$\textsc{Rev}(S|\mathcal{H},\{w_i = v_i, i \in U(\mathcal{H})\}) = \frac{\sum_{i \in S^{\textsc{known}}(\mathcal{H})} w_i(\mathcal{H}) + \sum_{i \in S^{\textsc{unknown}}(\mathcal{H})} v_i}{\sum_{i \in S^{\textsc{known}}(\mathcal{H})} w_i(\mathcal{H}) + \sum_{i \in S^{\textsc{unknown}}(\mathcal{H})} v_i + 1}$$
be the expected revenue of $S$ given $\mathcal{H}$ and that $w_i = v_i$ for each unknown entrant $i$. The revenue of $S$ given history $\mathcal{H}^1$ and the event $\mathcal{E}^1(\{v_{a(i)}\}_{i \in U^1})$ equals the revenue of $a(S)$ given $\mathcal{E}^2(\{v_j\}_{j \in U^2})$
\begin{align}\label{eq:rev_S_given_H_1_and_v_ai_equals_rev_a_S_given_H_2_and_vi}
    \textsc{Rev}(S|\mathcal{H}^1,\{w_i = v_{a(i)}, i \in U^1\}) &= \frac{\sum_{i \in S^{\textsc{known}}(\mathcal{H}^1)} w_i(\mathcal{H}^1) + \sum_{i \in S^{\textsc{unknown}}(\mathcal{H}^1)} v_{a(i)}}{\sum_{i \in S^{\textsc{known}}(\mathcal{H}^1)} w_i(\mathcal{H}^1) + \sum_{i \in S^{\textsc{unknown}}(\mathcal{H}^1)} v_{a(i)} + 1} \nonumber \\
    &= \frac{\sum_{i \in S^{\textsc{known}}(\mathcal{H}^1)} w_{a(i)}(\mathcal{H}^2) + \sum_{i \in S^{\textsc{unknown}}(\mathcal{H}^1)} v_{a(i)}}{\sum_{i \in S^{\textsc{known}}(\mathcal{H}^1)} w_{a(i)}(\mathcal{H}^2)  + \sum_{i \in S^{\textsc{unknown}}(\mathcal{H}^1)} v_{a(i)} + 1} \nonumber\\
    &= \frac{\sum_{j \in (a(S))^{\textsc{known}}(\mathcal{H}^2)} w_{j}(\mathcal{H}^2) + \sum_{j \in (a(S))^{\textsc{unknown}}(\mathcal{H}^2)} v_{j}}{\sum_{j \in (a(S))^{\textsc{known}}(\mathcal{H}^2)} w_{j}(\mathcal{H}^2)  + \sum_{j \in (a(S))^{\textsc{unknown}}(\mathcal{H}^2)} v_{j} + 1} \nonumber\\
    &=  \textsc{Rev}(a(S)|\mathcal{H}^2,\{w_j = v_{j}, j \in U^2\})
\end{align}
where the second equality uses that $w_i(\mathcal{H}^1) = w_{a(i)}(\mathcal{H}^2)$ for known $i$ under $\mathcal{H}^1$ and the third equality uses that $a$ maps the unknown products in $S$ given $\mathcal{H}^1$ to the unknown products in $a(S)$ given $\mathcal{H}^2$. Thus, $\textsc{OPT}^1$ can be expressed as
\begin{align}\label{eq:optima_given_paired_events_are_equal}
    \textsc{OPT}^1(\{v_{a(i)}\}_{i \in U^1}) &= \max\limits_{S \subseteq \mathcal{N}, |S| \leq c} \textsc{Rev}(S|\mathcal{H}^1,\{w_i = v_{a(i)}, i \in U^1\}) \nonumber \\
    &= \max\limits_{S \subseteq \mathcal{N}, |S| \leq c}  \textsc{Rev}(a(S)|\mathcal{H}^2,\{w_j = v_{j}, j \in U^2\}) \nonumber\\
    &= \max\limits_{S \subseteq \mathcal{N}, |S| \leq c}  \textsc{Rev}(S|\mathcal{H}^2,\{w_j = v_{j}, j \in U^2\}) = \textsc{OPT}^2(\{v_j\}_{j \in U^2}).
\end{align}
where the second equality uses \eqref{eq:rev_S_given_H_1_and_v_ai_equals_rev_a_S_given_H_2_and_vi}, and the third equality uses that $a$ is a one-to-one mapping.

Letting $f(\cdot)$ be the density function of $\mathcal{F}$, the expected ex-post optimum given $\mathcal{H}^1$ is expressed as:
\begin{align*}
    \textsc{OPT}(\mathcal{H}^1) &= \int \textsc{OPT}^1(\{v_{a(i)}\}_{i \in U^1}) \Big(\prod_{i \in U^1} f(v_{a(i)}) d v_{a(i)} \Big)= \int \textsc{OPT}^2(\{v_j\}_{j \in U^2}) \Big(\prod_{i \in U^1} f(v_{a(i)}) d v_{a(i)} \Big)\\
    &= \int \textsc{OPT}^2(\{v_j\}_{j \in U^2}) \Big(\prod_{j \in U^2} f(v_{j}) d v_{j} \Big)= \textsc{OPT}(\mathcal{H}^2).
\end{align*}
where the first equality expresses the expected ex-post optimum over all realizations of the unknown entrants and uses that each the entrants' attraction parameters are independent, the second equality uses \eqref{eq:optima_given_paired_events_are_equal}, the the third equality uses the fact that $a$ maps the unknown entrants under $\mathcal{H}^1$ to the unknown entrants under $\mathcal{H}^2$. 

\paragraph{Establishing $\max\limits_{S \subseteq I(\mathcal{H}^1), |S| \leq c} \textsc{Rev}(S;\mathcal{H}^1) = \max\limits_{S \subseteq I(\mathcal{H}^2), |S| \leq c} \textsc{Rev}(S;\mathcal{H}^2)$.} For any set of known products $S \subseteq I(\mathcal{H}_1)$, using that $w_i(\mathcal{H}^1) = w_{a(i)}(\mathcal{H}^2)$ for all $i \in S$:
$$\textsc{Rev}(S;\mathcal{H}^1) = \frac{\sum_{i \in S} w_i(\mathcal{H}^1)}{\sum_{i \in S} w_i(\mathcal{H}^1) + 1} = \frac{\sum_{i \in S} w_{a(i)}(\mathcal{H}^2)}{\sum_{i \in S} w_{a(i)}(\mathcal{H}^2) + 1}  = \frac{\sum_{j \in a(S)} w_{j}(\mathcal{H}^2)}{\sum_{j \in a(S)} w_{j}(\mathcal{H}^2) + 1}  = \textsc{Rev}(a(S);\mathcal{H}^2).$$
Therefore, using that $a$ maps the known products in $\mathcal{H}^1$ to the known products in $\mathcal{H}^2$:
$$\max\limits_{S \subseteq I(\mathcal{H}^1), |S| \leq c} \textsc{Rev}(S; \mathcal{H}^1) = \max\limits_{S \subseteq I(\mathcal{H}^1), |S| \leq c} \textsc{Rev}(a(S); \mathcal{H}^2) =  \max\limits_{S \subseteq I(\mathcal{H}^2), |S| \leq c} \textsc{Rev}(S; \mathcal{H}^2).$$

\paragraph{Establishing $\min_{S \in \mathcal{E}(\mathcal{H}^1)} \textsc{EpochReg}(S;\mathcal{H}^1) =\min_{S \in \mathcal{E}(\mathcal{H}^2)} \textsc{EpochReg}(S;\mathcal{H}^2)$.} 
For any set containing at least one unknown entrant $S \in \mathcal{E}(\mathcal{H})$, $\textsc{EpochReg}(S;\mathcal{H}) = \textsc{OPT}(\mathcal{H}^1) \tau(S;\mathcal{H}) - r(S;\mathcal{H})$.
Moreover,
\begin{align*}
    \tau(S; \mathcal{H}^1) &= \frac{\sum_{i \in S}w_i(\mathcal{H}^1) + w_0}{\sum_{i \in S^{\textsc{unk}}(\mathcal{H}^1)}w_i(\mathcal{H}^1)} = \frac{\sum_{i \in S}w_{a(i)}(\mathcal{H}^2) + w_0}{\sum_{i \in S^{\textsc{unk}}(\mathcal{H}^1)}w_{a(i)}(\mathcal{H}^2)}= \frac{\sum_{j \in a(S)}w_{j}(\mathcal{H}^2) + w_0}{\sum_{j \in (a(S))^{\textsc{unk}}(\mathcal{H}^2)}w_{j}(\mathcal{H}^2)} = \tau(a(S);\mathcal{H}^2),
\end{align*}
\begin{align*}
    r(S; \mathcal{H}^1) &= \frac{\sum_{i \in S}w_i(\mathcal{H}^1)}{\sum_{i \in S^{\textsc{unk}}(\mathcal{H}^1)}w_i(\mathcal{H}^1)} = \frac{\sum_{i \in S}w_{a(i)}(\mathcal{H}^2) }{\sum_{i \in S^{\textsc{unk}}(\mathcal{H}^1)}w_{a(i)}(\mathcal{H}^2)}= \frac{\sum_{j \in a(S)}w_{j}(\mathcal{H}^2) }{\sum_{j \in (a(S))^{\textsc{unk}}(\mathcal{H}^2)}w_{j}(\mathcal{H}^2)} = r(a(S);\mathcal{H}^2).
\end{align*}
Given that $\textsc{OPT}(\mathcal{H}^1) = \textsc{OPT}(\mathcal{H}^2)$, this yields $\textsc{EpochReg}(S;\mathcal{H}^1) = \textsc{EpochReg}(a(S);\mathcal{H}^2)$. As a result, $$\min_{S \in \mathcal{E}(\mathcal{H}^1)} \textsc{EpochReg}(S;\mathcal{H}^1)  = \min_{S \in \mathcal{E}(\mathcal{H}^1)} \textsc{EpochReg}(a(S);\mathcal{H}^2) = \min_{S \in \mathcal{E}(\mathcal{H}^2)} \textsc{EpochReg}(S;\mathcal{H}^2)$$
where the last equality uses that $a$ is a one-to-one mapping from the set containing at least one unknown entrant given $\mathcal{H}^1$ to the sets containing at least unknown entrant given $\mathcal{H}^2$.
\end{proof}

\subsection{Expanded one-epoch-lookahead regret decomposition (Lemma~\ref{lemma:the_regret_of_pi_star_can_be_decomposed_into_the_first_epoch_regret_and_the_future_regret_general_form})}\label{appendix_subsec:proof_lemma_B9}

\begin{proof}[Proof of Lemma~\ref{lemma:the_regret_of_pi_star_can_be_decomposed_into_the_first_epoch_regret_and_the_future_regret_general_form}.]
Recall that $Z^{\mathcal{H}}_{\star}$ is the random variable of the first round at which an unknown entrant is purchased if $\pi^{\star}$ starts from $\mathcal{H}$. Given that $\pi^{\star}$ offers an unknown entrant for every round until $Z^{\mathcal{H}}_{\star}$ and an unknown entrant is purchased with probability at least $\frac{\tildew}{(c+1) w_{\text{max}}}$ where $w_{\text{max}} = \max\{w_i(\mathcal{H}): i \in \mathcal{N} \cup \{0\}\}$ at each round until $Z^{\mathcal{H}}_{\star}$, $\expect[Z^{\mathcal{H}}_{\star}] < \frac{(c+1) w_{\text{max}}}{\tildew}$. Thus, $\prob[Z^{\mathcal{H}}_{\star} < \infty] = 1$. Therefore, 
\begin{equation*}\label{eq:regret_pi_star_given_H_decomposition_until_Z_H_star}
    \textsc{Reg}(\pi^{\star};\mathcal{H}) = \underbrace{\expect \Bigg[\sum_{t=1}^{Z^{\mathcal{H}}_{\star}} \textsc{OPT}(\mathcal{H}_t) - \textsc{Rev}(S_t;\mathcal{H}_t) |\mathcal{H}_1 = \mathcal{H}\Bigg]}_{(\star)} + \expect[\textsc{Reg}(\pi^{\star};\mathcal{H}_{Z^{\mathcal{H}}_{\star}} \cup (S_{Z^{\mathcal{H}}_{\star}}, Y_{Z^{\mathcal{H}}_{\star}}, w_{Y_{Z^{\mathcal{H}}_{\star}}}))|\mathcal{H}_1 = \mathcal{H}].
\end{equation*}
To conclude the lemma we show below $(\star) = \min_{S \in S(\mathcal{H})} \textsc{EpochReg}(S;\mathcal{H})$. For all rounds $t \leq Z_{\star}^{\mathcal{H}}$  $\mathcal{H}_t$ and $\mathcal{H}$ contain the same set of known products and the same corresponding attraction parameters because no unknown entrant has been purchased. As a result, $\textsc{Rev}(S;\mathcal{H}_t) = \textsc{Rev}(S;\mathcal{H})$ for all assortments $S$ and rounds $t \leq Z^{\mathcal{H}}_{\star}$. Additionally, $\textsc{OPT}(\mathcal{H}_t) = \textsc{OPT}(\mathcal{H})$, $\tau(S;\mathcal{H}_t) = \tau(S;\mathcal{H})$, and $r(S;\mathcal{H}_t) = r(S;\mathcal{H})$ and thus $\textsc{EpochReg}(S;\mathcal{H}_t) = \textsc{EpochReg}(S;\mathcal{H})$ for all assortments $S$ and rounds $ t \leq Z^{\mathcal{H}}_{\star}$. 

At any round $ t \leq Z^{\mathcal{H}}_{\star}$, $\pi^{\star}$ offers an assortment in $\argmin_{S \in S(\mathcal{H})} \textsc{EpochReg}(S;\mathcal{H})$. Letting $N(S)$ be the number of times assortment $S$ is chosen until $Z^{\mathcal{H}}_{\star}$, $q(S;\mathcal{H}) = \frac{\sum_{i \in S \setminus I(\mathcal{H})} w_i(\mathcal{H})}{\sum_{i \in S} w_i(\mathcal{H}) + 1}$ be the the probability that an unknown entrant is purchased given $S$, and $\mathcal{A}(\mathcal{H}) = \argmin_{S \in S(\mathcal{H})} \textsc{EpochReg}(S;\mathcal{H})$ be the set of assortments minimizing the epoch regret, $(\star)$ equals 
\begin{align}\label{eq:expressing_regret_until_first_unknown_entrant_purchase}
    \expect \Bigg[\sum_{t=1}^{Z^{\mathcal{H}}_{\star}} \textsc{OPT}(\mathcal{H}_t) - \textsc{Rev}(S_t;\mathcal{H}_t) |\mathcal{H}_1 = \mathcal{H}\Bigg] 
    &=  \expect \Bigg[\sum_{t=1}^{Z^{\mathcal{H}}_{\star}} \textsc{OPT}(\mathcal{H}) - \textsc{Rev}(S_t;\mathcal{H}) |\mathcal{H}_1 = \mathcal{H}\Bigg] \nonumber\\
    &=\expect \Bigg[\sum_{S \in \mathcal{A}(\mathcal{H}) } (\textsc{OPT}(\mathcal{H}) - \textsc{Rev}(S;\mathcal{H})) N(S) |\mathcal{H}_1 = \mathcal{H}\Bigg] \nonumber \\
    &= \sum_{S \in \mathcal{A}(\mathcal{H}) } (\textsc{OPT}(\mathcal{H}) - \textsc{Rev}(S;\mathcal{H})) \expect[N(S) |\mathcal{H}_1 = \mathcal{H}] \nonumber\\
    &=\sum_{S \in \mathcal{A}(\mathcal{H}) } \frac{(\textsc{OPT}(\mathcal{H}) - \textsc{Rev}(S;\mathcal{H}))}{q(S;\mathcal{H})} q(S;\mathcal{H})\expect[N(S) |\mathcal{H}_1 = \mathcal{H}] \nonumber\\
    &= \min_{S \in S(\mathcal{H})} \textsc{EpochReg}(S;\mathcal{H}) \Big( \sum_{S \in \mathcal{A}(\mathcal{H}) } q(S;\mathcal{H})\expect[N(S) |\mathcal{H}_1 = \mathcal{H}] \Big) \nonumber \\
    &= \min_{S \in S(\mathcal{H})} \textsc{EpochReg}(S;\mathcal{H}) 
\end{align}
The first equality uses $\textsc{OPT}(\mathcal{H}_t) = \textsc{OPT}(\mathcal{H})$ and $\textsc{Rev}(S_t;\mathcal{H}_t) = \textsc{Rev}(S_t; \mathcal{H})$ for all $t \leq Z^{\mathcal{H}}_{\star}$. The second equality uses that $\pi^{\star}$ offers only assortments in $\mathcal{A}(\mathcal{H})$. The fifth equality uses that 
\begin{itemize}
    \item  $\textsc{EpochReg}(S;\mathcal{H}) = \frac{\textsc{OPT}(\mathcal{H}) - \textsc{Rev}(S;\mathcal{H})}{q(S;\mathcal{H})}$ as $\frac{1}{q(S;\mathcal{H})}$ is the expected time until an entrant is purchased and
    \item $\textsc{EpochReg}(S;\mathcal{H}) = \min_{S \in S(\mathcal{H})} \textsc{EpochReg}(S;\mathcal{H})$ for all $S \in \mathcal{A}(\mathcal{H})$.
\end{itemize}
 The sixth equality uses that $\sum_{S \in \mathcal{A}(\mathcal{H}) } q(S;\mathcal{H})\expect[N(S) |\mathcal{H}_1 = \mathcal{H}]  = 1$ as $\sum_{S \in \mathcal{A}(\mathcal{H}) } q(S;\mathcal{H})\expect[N(S) |\mathcal{H}_1 = \mathcal{H}]$ is the expected number of rounds in which an unknown entrant is purchased until $Z^{\mathcal{H}}_{\star}$ as assortment $S$ is offered $\expect[N(S) |\mathcal{H}_1 = \mathcal{H}]$ rounds in expectation and an unknown entrant is purchased independently with probability $q(S;\mathcal{H})$ every time $S$ is offered. As a result, $(\star) = \min_{S \in S(\mathcal{H})} \textsc{EpochReg}(S;\mathcal{H})$ concluding the proof.
\end{proof}

\subsection{One-epoch-lookahead decomposition of regret of optimal policy (Lemma~\ref{lemma:the_regret_of_pi_star_can_be_decomposed_into_the_first_epoch_regret_and_the_future_regret})}\label{appendix_subsec:proof_lemma_one_epoch_lookahead_decomposition_optimal_policy}
\begin{proof}[Proof of Lemma~\ref{lemma:the_regret_of_pi_star_can_be_decomposed_into_the_first_epoch_regret_and_the_future_regret}.]
The regret of $\pi^{\star}$ given a history $\mathcal{H}$ can be decomposed as
\begin{align*}
    \textsc{Reg}(\pi^{\star},J(\mathcal{H})) &= \textsc{Reg}(\pi^{\star};\mathcal{H}) \\
    &= \min\limits_{S \in \mathcal{E}(\mathcal{H})} \textsc{EpochReg}(S;\mathcal{H}) \nonumber + \expect[\textsc{Reg}(\pi^{\star};\mathcal{H}_{Z^{\mathcal{H}}_{\star}} \cup (S_{Z^{\mathcal{H}}_{\star}}, Y_{Z^{\mathcal{H}}_{\star}}, w_{Y_{Z^{\mathcal{H}}_{\star}}}))|\mathcal{H}_1 = \mathcal{H}]\\
    &=  \min\limits_{S \in \mathcal{E}(\mathcal{H})} \textsc{EpochReg}(S;\mathcal{H}) \nonumber + \expect_{w_{Y_{Z^{\mathcal{H}}_{\star}}} \sim \mathcal{F}}[\textsc{Reg}(\pi^{\star}, J(\mathcal{H}) \cup \{w_{Y_{Z^{\mathcal{H}}_{\star}}}\})|\mathcal{H}_1 = \mathcal{H}]\\
    &= \min\limits_{S \in \mathcal{E}(\mathcal{H})} \textsc{EpochReg}(S;\mathcal{H}) \nonumber + \expect_{w \sim \mathcal{F}}[\textsc{Reg}(\pi^{\star}, J(\mathcal{H}) \cup \{w\})|\mathcal{H}_1 = \mathcal{H}]
\end{align*}
where the first equality follows from Lemma~\ref{lemma:the_regret_of_pi_star_depends_on_H_only_through_J_H}; the second equality follows by Lemma~\ref{lemma:the_regret_of_pi_star_can_be_decomposed_into_the_first_epoch_regret_and_the_future_regret_general_form}; the third equality uses Lemma~\ref{lemma:the_regret_of_pi_star_depends_on_H_only_through_J_H}, the fact that $w_{Y_{Z^{\mathcal{H}}_{\star}}}$ is independent of $\mathcal{H}_{Z^{\mathcal{H}}_{\star}}$,$S_{Z^{\mathcal{H}}_{\star}}$, and $ Y_{Z^{\mathcal{H}}_{\star}}$, and $w_{Y_{Z^{\mathcal{H}}_{\star}}} \sim \mathcal{F}$. 
\end{proof}

\subsection{Static-assortment policies minimize the epoch regret (Lemma~\ref{lemma:epoch_regret_of_any_policy_is_at_least_the_minimum_epoch_regret_over_fixed_assortments})}\label{appendix_subsec:proof_lemma_epoch_regret_of_any_policy_is_at_least_minimum_epoch_regret_of_static_assortment_policies}

\begin{proof}[Proof of Lemma~\ref{lemma:epoch_regret_of_any_policy_is_at_least_the_minimum_epoch_regret_over_fixed_assortments}.]
Let $Z = Z^{\pi, \mathcal{H}}$ for convenience. Similar to our proof of Lemma~\ref{lemma:the_regret_of_pi_star_can_be_decomposed_into_the_first_epoch_regret_and_the_future_regret_general_form}, given that $\mathcal{H}_t$ and $\mathcal{H}$ contain the same set of known products and their attraction parameters for $t \leq Z$, $\textsc{OPT}(\mathcal{H}_t) = \textsc{OPT}(\mathcal{H})$ and $\textsc{Rev}(S;\mathcal{H}_t) = \textsc{Rev}(S;\mathcal{H})$ for $t \leq Z$. Letting $N(S)$ be the number of rounds in which assortment $S$ is offered until round $Z$, the regret of $\pi$ in the first epoch can be expressed as
\begin{align}\label{eq:regret_until_Z_expression_as_sum_over_assortments}
  \textsc{EpochReg}^{\pi}(\mathcal{H}) &=\expect\Bigg[\sum_{t=1}^{Z}\textsc{OPT}(\mathcal{H}_t) - \textsc{Rev}(S_t, \mathcal{H}_t)|\mathcal{H}_1 = \mathcal{H}\Bigg] \nonumber\\
  &=\expect\Bigg[\sum\limits_{S \subseteq \mathcal{N}, |S| \leq c} (\textsc{OPT}(\mathcal{H}) - \textsc{Rev}(S, \mathcal{H})) N(S) |\mathcal{H}_1 = \mathcal{H}\Bigg] \nonumber\\
    &= \sum\limits_{S \subseteq \mathcal{N}, |S| \leq c} (\textsc{OPT}(\mathcal{H}) - \textsc{Rev}(S, \mathcal{H})) \expect[N(S)^{\pi, \mathcal{H}}|\mathcal{H}_1 = \mathcal{H}].
\end{align}

\paragraph{Case 1: $\expect[Z] = +\infty$.} Given that $Z = \sum\limits_{S \subseteq \mathcal{N}, |S| \leq c} N(S)$, there exists some $\tilde{S}$ with \begin{equation}\label{eq:infinite_expected_number_of_visits_tilde_S}
    \expect[N(\tilde{S})|\mathcal{H}_1 = \mathcal{H}] = +\infty,
\end{equation} 
For any $S \in \mathcal{E}(\mathcal{H})$ it holds that $\expect[N(S)|\mathcal{H}_1 = \mathcal{H}\Big] \leq \frac{\sum_{i\in S} w_i(\mathcal{H}) + w_0}{\tildew}$ as an unknown entrant is purchased with probability at least $\frac{h(\mathcal{H})}{\sum_{i\in S} w_i(\mathcal{H}) + w_0}$ every time $S$ is offered. As a result $\tilde{S}$ contains only known products, i.e., $\tilde{S} \not \in \mathcal{E}(\mathcal{H})$. Given that $\mathcal{H}$ is non-terminal, $\textsc{OPT}(\mathcal{H}) - \textsc{Rev}(\tilde{S}, \mathcal{H}) > 0$, which together with \eqref{eq:regret_until_Z_expression_as_sum_over_assortments} and \eqref{eq:infinite_expected_number_of_visits_tilde_S} implies $ \textsc{EpochReg}^{\pi}(\mathcal{H}) = +\infty$. 

\paragraph{Case 2: $\expect[Z] < +\infty$.} Let $q(S;\mathcal{H}) = \frac{\sum_{i \in S \setminus I(\mathcal{H})} w_i(\mathcal{H})}{\sum_{i \in S} w_i(\mathcal{H}) + w_0}$ be the probability that an unknown entrant is purchased given $S$. Using \eqref{eq:regret_until_Z_expression_as_sum_over_assortments}, the regret of $\pi$ during the first epoch can be lower bounded as
\begin{align}\label{ineq:lower_bound_on_the_regret_of_the_first_epoch_of_any_policy}
    \textsc{EpochReg}^{\pi}(\mathcal{H}) &= \sum\limits_{S \subseteq \mathcal{N}, |S| \leq c} (\textsc{OPT}(\mathcal{H}) - \textsc{Rev}(S, \mathcal{H})) \expect[N(S)|\mathcal{H}_1 = \mathcal{H}] \nonumber\\
    &\geq \sum\limits_{S \in \mathcal{E}(\mathcal{H})} \frac{(\textsc{OPT}(\mathcal{H}) - \textsc{Rev}(S, \mathcal{H}))}{q(S;\mathcal{H})}q(S;\mathcal{H})\expect[N(S)|\mathcal{H}_1 = \mathcal{H}] \nonumber\\
    &\geq  \min \limits_{S \in \mathcal{E}(\mathcal{H})} \Bigg\{\frac{\textsc{OPT}(\mathcal{H}) - \textsc{Rev}(S; \mathcal{H})}{q(S;\mathcal{H})}\Bigg\} \Bigg[\sum\limits_{S \in \mathcal{E}(\mathcal{H})}q(S;\mathcal{H})\expect[N(S)|\mathcal{H}_1 = \mathcal{H}] \Bigg] \nonumber\\
    &= \min \limits_{S \in \mathcal{E}(\mathcal{H})} \Bigg\{\frac{\textsc{OPT}(\mathcal{H}) - \textsc{Rev}(S; \mathcal{H})}{q(S;\mathcal{H})} \Bigg\}=  \min \limits_{S \in \mathcal{E}(\mathcal{H})} \textsc{EpochReg}(S;\mathcal{H}).
\end{align}
The first inequality multiplies and divides each term by $q(S;\mathcal{H})$ and lower bounds the terms corresponding to $S \not \in \mathcal{E}(\mathcal{H})$ with $0$
because $\textsc{OPT}(\mathcal{H}) - \textsc{Rev}(S;\mathcal{H}) > 0$ for $S \not \in \mathcal{E}(\mathcal{H})$ as $\mathcal{H}$ is non-terminal. Regarding the second equality, note that $\sum\limits_{S \in \mathcal{E}(\mathcal{H})}q(S;\mathcal{H})\expect[N(S)|\mathcal{H}_1 = \mathcal{H}]$ is the expected number of times an entrant is purchased during the epoch since any assortment $S$ is offered $\expect[N^{\pi}(S)|\mathcal{H}_1 = \mathcal{H}]$ rounds in expectation and an entrant is purchased independently with probability $q(S;\mathcal{H})$ each time $S$ is offered. Given that an unknown entrant is always purchased (as $\expect[Z] < +\infty$), $\sum\limits_{S \in \mathcal{E}(\mathcal{H})}q(S;\mathcal{H})\expect[N(S)|\mathcal{H}_1 = \mathcal{H}] = 1$.
The last equality holds as $\textsc{EpochReg}(S;\mathcal{H}) = \tau(S;\mathcal{H}) \textsc{OPT}(\mathcal{H})-r(S;\mathcal{H})$; 
$\tau(S;\mathcal{H}) = \frac{1}{q(S;\mathcal{H})}$; $r(S;\mathcal{H}) = \frac{\textsc{Rev}(S;\mathcal{H})}{q(S;\mathcal{H})}$.\end{proof}

\subsection{Optimal policy regret is bounded independently of the history (Lemma~\ref{lemma:for_any_history_with_k_unknown_entrants_the_regret_of_pi_star_is_lower_bounded_and_upper_bounded_by_quantities_which_depend_on_k})}\label{appendix_subsec:regret_of_optimal_policy_is_bounded_below_and_above}

\begin{proof}[Proof of Lemma~\ref{lemma:for_any_history_with_k_unknown_entrants_the_regret_of_pi_star_is_lower_bounded_and_upper_bounded_by_quantities_which_depend_on_k}]
We show by induction on $k$ that for any history $\mathcal{H}$ with $k$ unknown entrants remaining, 
\begin{equation}\label{ineq:upper_and_lower_bound_claim_for_any_k}
    \textsc{Reg}(\pi^{\star};\mathcal{H}) \geq -k \quad \text{ and } \quad \textsc{Reg}(\pi^{\star};\mathcal{H}) \leq \frac{k}{\tildew}
\end{equation}

\paragraph{Base of induction.} If $k = 0$, $\mathcal{H}$ is terminal and Lemma~\ref{lemma:any_policy_has_a_nonnegative_regret_starting_from_a_terminal_history_and_pi_star_as_a_zero_regret_starting_from_a_terminal_history} implies that $\textsc{Reg}(\pi^{\star};\mathcal{H}) = 0$. 

\paragraph{Induction step. } Suppose \eqref{ineq:upper_and_lower_bound_claim_for_any_k} holds for $k$ and we will show it for $k+1$. If $\mathcal{H}$ is terminal, then by Lemma~\ref{lemma:the_regret_of_pi_star_depends_on_H_only_through_J_H}, $\textsc{Reg}(\pi^{\star};\mathcal{H}) = 0$ which satisfies the conditions. Suppose $\mathcal{H}$ is not terminal. Then by Lemma~\ref{lemma:the_regret_of_pi_star_can_be_decomposed_into_the_first_epoch_regret_and_the_future_regret},
\begin{equation*}
      \textsc{Reg}(\pi^{\star};\mathcal{H}) = \min\limits_{S \subseteq \mathcal{E}(\mathcal{H})} \textsc{EpochReg}(S;\mathcal{H}) + \underbrace{\expect[\textsc{Reg}(\pi^{\star};\mathcal{H}_{Z^{\mathcal{H}}_{\star}} \cup (S_{Z^{\mathcal{H}}_{\star}}, Y_{Z^{\mathcal{H}}_{\star}}, w_{Z^{\mathcal{H}}_{\star}}))|\mathcal{H}_1 = \mathcal{H}]}_{(\star)}.
\end{equation*}
By the induction hypothesis \eqref{ineq:upper_and_lower_bound_claim_for_any_k} , $(\star) \geq -k$ and  $(\star)\leq \frac{k}{\tildew}$. It remains to show that $  \textsc{Reg}(\pi^{\star};\mathcal{H}) \leq \frac{1}{\tildew}$ which we establish below in \eqref{ineq:epoch_reg_upper_bound} and $  \textsc{Reg}(\pi^{\star};\mathcal{H}) \geq -1$ which we establish below in \eqref{ineq:epoch_reg_lower_bound}.

For $S \in \mathcal{E}(\mathcal{H})$, let $S^{\textsc{k}}$ and $S^{\textsc{u}}$ be the set of known and unknown products in $S$ given $\mathcal{H}$. Then
$$\textsc{EpochReg}(S;\mathcal{H}) = \underbrace{\Bigg(\textsc{OPT}(\mathcal{H}) - \frac{\sum_{i \in S^{\textsc{k}}} w_i(\mathcal{H}) + \sum_{i \in S^{\textsc{u}}} w_i(\mathcal{H})}{\sum_{i \in S^{\textsc{k}}} w_i(\mathcal{H}) + \sum_{i \in S^{\textsc{u}}} w_i(\mathcal{H}) + 1} \Bigg)}_{\text{regret per round of $S$}} \underbrace{\frac{\sum_{i \in S^{\textsc{k}}} w_i(\mathcal{H}) + \sum_{i \in S^{\textsc{u}}} w_i(\mathcal{H}) + 1}{\sum_{i \in S^{\textsc{u}}} w_i(\mathcal{H})}}_{\text{expected number of rounds in an epoch}}.$$
Using that $\textsc{OPT}(\mathcal{H}) \leq 1$, the epoch regret is upper bounded as
\begin{align}\label{ineq:epoch_reg_upper_bound}
    \textsc{EpochReg}(S;\mathcal{H}) &\leq 
    \Bigg(1 - \frac{\sum_{i \in S^{\textsc{k}}} w_i(\mathcal{H}) + \sum_{i \in S^{\textsc{u}}} w_i(\mathcal{H})}{\sum_{i \in S^{\textsc{k}}} w_i(\mathcal{H}) + \sum_{i \in S^{\textsc{u}}} w_i(\mathcal{H}) + 1} \Bigg) \frac{\sum_{i \in S^{\textsc{k}}} w_i(\mathcal{H}) + \sum_{i \in S^{\textsc{u}}} w_i(\mathcal{H}) + 1}{\sum_{i \in S^{\textsc{u}}} w_i(\mathcal{H})} \nonumber \\
    &= \frac{1}{\cancel{\sum_{i \in S^{\textsc{k}}} w_i(\mathcal{H}) + \sum_{i \in S^{\textsc{u}}} w_i(\mathcal{H}) + 1}}\frac{\cancel{\sum_{i \in S^{\textsc{k}}} w_i(\mathcal{H}) + \sum_{i \in S^{\textsc{u}}} w_i(\mathcal{H}) + 1}}{\sum_{i \in S^{\textsc{u}}} w_i(\mathcal{H})} \leq \frac{1}{\tildew}.
\end{align}
where the last inequality uses that $w_i(\mathcal{H}) \geq \tildew$ for all $i \in S^{\textsc{u}}$. 

Using that $\textsc{OPT}(\mathcal{H}) \geq \frac{\sum_{i \in S^{\textsc{k}}} w_i(\mathcal{H}) }{\sum_{i \in S^{\textsc{k}}} w_i(\mathcal{H}) + 1}$, the epoch regret is lower bounded as
\begin{align}\label{ineq:epoch_reg_lower_bound}
    &\textsc{EpochReg}(S;\mathcal{H}) \nonumber\\
    &\geq \Bigg(\frac{\sum_{i \in S^{\textsc{k}}} w_i(\mathcal{H}) }{\sum_{i \in S^{\textsc{k}}} w_i(\mathcal{H}) + 1} - \frac{\sum_{i \in S^{\textsc{k}}} w_i(\mathcal{H}) + \sum_{i \in S^{\textsc{u}}} w_i(\mathcal{H})}{\sum_{i \in S^{\textsc{k}}} w_i(\mathcal{H}) + \sum_{i \in S^{\textsc{u}}} w_i(\mathcal{H}) + 1} \Bigg) \frac{\sum_{i \in S^{\textsc{k}}} w_i(\mathcal{H}) + \sum_{i \in S^{\textsc{u}}} w_i(\mathcal{H}) + 1}{\sum_{i \in S^{\textsc{u}}} w_i(\mathcal{H})} \nonumber \\
    &= -\frac{\sum_{i \in S^{\textsc{u}}} w_i(\mathcal{H})}{(\sum_{i \in S^{\textsc{k}}} w_i(\mathcal{H}) + 1)(\sum_{i \in S^{\textsc{k}}} w_i(\mathcal{H}) + \sum_{i \in S^{\textsc{u}}} w_i(\mathcal{H}) + 1)}\frac{\sum_{i \in S^{\textsc{k}}} w_i(\mathcal{H}) + \sum_{i \in S^{\textsc{u}}} w_i(\mathcal{H}) + 1}{\sum_{i \in S^{\textsc{u}}} w_i(\mathcal{H})} \nonumber\\
    &= -\frac{1}{\sum_{i \in S^{\textsc{k}}} w_i(\mathcal{H}) + 1} \geq -1.
\end{align}
The first equality uses $\frac{A}{A+1} -\frac{B}{B+1} = \frac{A-B}{(A+1)(B+1)}$ for $(A,B) = (\sum_{i \in S^{\textsc{k}}} w_i(\mathcal{H}), \sum_{i \in S^{\textsc{k}}} w_i(\mathcal{H}) + \sum_{i \in S^{\textsc{U}}} w_i(\mathcal{H}))$. 
\end{proof}

\section{Omitted proofs from Section~\ref{sec:comparison_to_simple_strategies}}
\subsection{Exploring all products can be arbitrarily suboptimal (Proposition~\ref{prop:explore_all_is_arbitrarily_suboptimal})}\label{appendix_subsec:prop_explore_all_is_arbitrarily_suboptimal_proof}

The following lemma lower bounds the regret of $\textsc{ExploreAll}$. 
\begin{lemma}\label{lemma:lower_bound_regret_ucb}
    For any $c \in \mathbb{N}$ and $q \in (0,1)$, $\textsc{Reg}(\textsc{ExploreAll}) \geq \frac{1}{q} \cdot \frac{9}{29c} - 1$ on the instance $\mathcal{I}(c,q)$.
\end{lemma}

The next lemma upper bounds the regret of $\textsc{EFA}$.

\begin{lemma}\label{lemma:upper_bound_regret_our_algo}
For any $c\in \mathbb{N}$ and $q < \frac{0.9}{2+2(c+1)^2}$, $\textsc{Reg}(\textsc{EFA}) \leq 2c^2(c+1)$ on the instance $\mathcal{I}(c,q)$.
\end{lemma}
\begin{proof}[Proof of Proposition~\ref{prop:explore_all_is_arbitrarily_suboptimal}.]
The proof follows by combining Lemmas \ref{lemma:lower_bound_regret_ucb} and \ref{lemma:upper_bound_regret_our_algo}
\end{proof}
\begin{proof}[Proof of Lemma~\ref{lemma:lower_bound_regret_ucb}.]
\textsc{ExploreAll} explores $c$ unknown entrants $U_c$ in the first epoch. The expected number of rounds in the epoch is $\tau(U_c;\mathcal{H}_1) = \frac{cq+1}{cq}$ as an unknown entrant is purchased with probability $\frac{cq}{cq+1}$. Note that $U_c$ does not contain known products and thus the only product purchase during the epoch is the one from the entrant at the end of the epoch, which implies $r(U_c;\mathcal{H}_1) = 1$. As a result, given that the expected regret at any round is non-negative, the regret of \textsc{ExploreAll} is at least the regret in the first epoch: \begin{align*}\label{ineq:total_regret_of_exploreall_is_at_least_regret_of_exploreall_in_first_epoch}
\textsc{Reg}(\textsc{ExploreAll}) &\geq \textsc{EpochReg}_1(U_c) =\tau(U_c;\mathcal{H}_1)\textsc{OPT}_1 - r(U_c;\mathcal{H}_1) =  \frac{cq+1}{cq} \textsc{OPT}_1 -1 \nonumber\\
    &\geq  \frac{cq+1}{cq} \cdot \frac{\frac{c}{2} \cdot 0.9}{\frac{c}{2}\cdot 0.9 + 1} - 1  \geq  \frac{cq+1}{cq} \cdot \frac{9}{29} - 1 \geq \frac{1}{cq} \cdot \frac{9}{29} - 1
\end{align*}
where the second inequality uses $\textsc{OPT}_1 \geq \frac{\frac{c}{2} \cdot 0.9}{\frac{c}{2}\cdot 0.9 + 1}$ as there are $\frac{c}{2}$ incumbents with attraction parameter $0.9$, the third inequality uses $\frac{\frac{c}{2} \cdot 0.9}{\frac{c}{2}\cdot 0.9 + 1} \geq \frac{\frac{1}{2} \cdot 0.9}{\frac{1}{2}\cdot 0.9 + 1} = \frac{9}{29}$ as $c \geq 1$, and the fourth inequality uses $cq > 0$.
\end{proof}

In the remainder of the subsection we prove Lemma~\ref{lemma:upper_bound_regret_our_algo}. Let $\mathcal{E}^{0}$ be the event that all of the unknown entrants have an attraction parameter of $0$, and $\mathcal{E}^{1}$ its complement. The following lemma upper bounds the regret per round of \textsc{EFA} conditioned on $\mathcal{E}^{0}$. 

\begin{lemma}\label{lemma:EFA_explores_at_most_c_2_unknown_entrants_conditioned_on_all_zeros}
    Conditioned on $\mathcal{E}^{0}$, the regret of \textsc{EFA} in any round $t$ is upper bounded by $\textsc{OPT} - \textsc{Rev}_t(S_t) \leq cq$
\end{lemma}

\begin{proof}[Proof of Lemma~\ref{lemma:upper_bound_regret_our_algo}.]
Given that all incumbents have parameters smaller than $1$ and there are $c$ unknown entrants with $\bern(q)$ priors, for any round $t$ with at least one unknown entrant the expected ex-post optimum is larger than the maximum revenue from known products, $\textsc{OPT}_t> \textsc{Rev}_t(S^{\star}_{(c)}(\mathcal{I}_t))$. Thus \textsc{EFA} explores at least one unknown entrant until all entrants have been purchased. Let $\tau_i$ be the round at which the $i$-th unknown entrant is purchased for $i \in \{1, \ldots, c\}$ and $\tau_0 = 0$ for ease of notation. Given that $\textsc{EFA}$ incurs zero expected regret after $\tau_c$, its total expected regret is decomposed as
\begin{equation*}
    \textsc{Reg}(\textsc{EFA}) = \expect\Bigg[\sum_{t=0}^{\tau_c} \textsc{OPT} - \textsc{Rev}_t(S_t) \Bigg| \mathcal{E}^{0}\Bigg]\prob[\mathcal{E}^0] + \expect\Bigg[\sum_{t=0}^{\tau_c} \textsc{OPT} - \textsc{Rev}_t(S_t) \Bigg| \mathcal{E}^{1}\Bigg]\prob[\mathcal{E}^{1}].
\end{equation*}
The regret of \textsc{EFA} when $\mathcal{E}^0$ holds is upper bounded as
\begin{align*}\label{ineq:upper_bound_regret_efa_conditioned_on_all_zeros}
   &\expect\left[\sum_{t=1}^{\tau_c} \textsc{OPT} - \textsc{Rev}_t(S_t) \middle| \mathcal{E}^0 \right] \prob[\mathcal{E}^0]\leq \expect\left[\sum_{t=1}^{\tau_c} \textsc{OPT} - \textsc{Rev}_t(S_t) \middle| \mathcal{E}^0 \right] \leq \expect\left[\sum_{t=1}^{\tau_c} cq \middle| \mathcal{E}^0 \right]\\
   &=cq \expect\left[\sum_{i=0}^{c-1} \sum_{t = \tau_i+1}^{\tau_{i+1}}1 \middle| \mathcal{E}^0 \right]  = cq\sum_{i=0}^{c-1} \expect [\tau_{i+1}-\tau_{i} | \mathcal{E}^0] \leq cq\sum_{i=0}^{c-1} \frac{c+1}{q} = cq\frac{c(c+1)}{q} = c^2(c+1).
\end{align*}

The first inequality uses $\prob[\mathcal{E}^0] \leq 1$ and the fact that the expected regret per round is non-negative. The second inequality uses Lemma~\ref{lemma:EFA_explores_at_most_c_2_unknown_entrants_conditioned_on_all_zeros}. Regarding the third inequality, given that the sum of the attraction parameters in any assortment is at most $c$ and the customer's prior attraction parameter for any unknown entrant is $q$, the probability that an unknown entrant purchased is at least $\frac{q}{c+1}$. As a result, the expected number of rounds in any epoch is upper bounded by $\expect[\tau_{i+1} -\tau_{i}|\mathcal{E}^0] \leq \frac{c+1}{q}$ for $i \in \{0, \ldots, c-1\}$.

The regret of \textsc{EFA} when $\mathcal{E}^1$ holds is upper bounded as
\begin{align*}
    &\expect\Bigg[\sum_{t=1}^{\tau_c} \textsc{OPT} - \textsc{Rev}_t(S_t) \Bigg| \mathcal{E}^1 \Bigg] \prob[\mathcal{E}^1] \leq cq\expect\Bigg[\sum_{t=1}^{\tau_c} \textsc{OPT} - \textsc{Rev}_t(S_t) \Bigg| \mathcal{E}^1 \Bigg]  \leq  cq \expect\Bigg[\sum_{t=1}^{\tau_c} 1 \Bigg| \mathcal{E}^1 \Bigg]\\
    &=   cq \expect\left[\sum_{i=0}^{c-1} \sum_{t = \tau_i+1}^{\tau_{i+1}}1 \middle| \mathcal{E}^1 \right] = cq\sum_{i=0}^{c-1} \expect [\tau_{i+1}-\tau_{i} | \mathcal{E}^1] \leq cq \sum_{i=0}^{c-1} \frac{c+1}{q} = cq\frac{c(c+1)}{q} = c^2(c+1)
\end{align*}
The first inequality uses $\prob[\mathcal{E}^1] = 1- (1-q)^c \leq cq$ and the fact that the regret per round is non-negative. The second inequality uses $\textsc{OPT} -\textsc{Rev}_t(S_t) \leq 1$ as $\textsc{OPT} \leq 1$ and $\textsc{Rev}_t(S_t) \geq 0$. The third inequality uses $\expect[\tau_{i+1} -\tau_{i}|\mathcal{E}^1] \leq \frac{c+1}{q}$ for $i \in \{0, \ldots, c-1\}$ due to the same reason as in the previous part. Therefore, $ \textsc{Reg}(\textsc{EFA}) \leq 2 c^2(c+1)$.
\end{proof}

\begin{proof}[Proof of Lemma~\ref{lemma:EFA_explores_at_most_c_2_unknown_entrants_conditioned_on_all_zeros}.]
We first show that \textsc{EFA} explores at most $\frac{c}{2}$ unknown entrants at any round. Let $\mathcal{D}_t^0$ be the event that the attraction parameters of all unknown entrants at round $t$ realize to $0$, and $\mathcal{D}_t^1$ be its complement. The expected ex-post optimum at round $t$ calculated by \textsc{EFA} is be upper bounded by
\begin{equation}\label{eq:expected_ex_post_opt_law_of_total_probability_decoposition}
    \textsc{OPT}_t = \expect[\textsc{OPT}|\mathcal{H}_t] = \expect[\textsc{OPT}|\mathcal{H}_t, \mathcal{D}_t^0] \prob[\mathcal{D}_t^0|\mathcal{H}_t]+\expect[\textsc{OPT}|\mathcal{H}_t, \mathcal{D}_t^1] \prob[\mathcal{D}_t^1|\mathcal{H}_t]\leq \frac{\frac{c}{2}\cdot0.9   + \frac{c}{2} \cdot2q  }{\frac{c}{2}  \cdot  0.9+ \frac{c}{2}\cdot 2q+ 1} + cq.
\end{equation} 
To bound the first term, the inequality uses $\prob[\mathcal{D}_t^0|\mathcal{H}_t] \leq 1$ and $\expect[\textsc{OPT}|\mathcal{H}_t, \mathcal{D}_t^0]  = \frac{\frac{c}{2}\cdot0.9   + \frac{c}{2} \cdot2q  }{\frac{c}{2}  \cdot  0.9+ \frac{c}{2}\cdot 2q+ 1}$ as all entrants conditioned on $\mathcal{H}_t$ and $\mathcal{D}_t^0$ are zero because all known entrants in $\mathcal{H}_t$ are zero (as $\mathcal{E}_0$ holds) and $\mathcal{D}_t^0$ implies all remaining entrants are zero; to bound the second term the inequality uses $\expect[\textsc{OPT}|\mathcal{H}_t, \mathcal{D}_t^1] \leq 1$ (as $\textsc{OPT} \leq 1$) and 
$\prob[\mathcal{D}_t^1|\mathcal{H}_t] = 1-\prob[\mathcal{D}_t^0|\mathcal{H}_t] \leq 1- (1-q)^c \leq cq$ because $\prob[\mathcal{D}_t^0|\mathcal{H}_t] \geq (1-q)^c$ as there are at most $c$ unknown entrants in $\mathcal{H}_t$ are their parameters are independently drawn from $\bern(q)$.

For any number of unknown entrants $\ell \in \{\frac{c}{2} + 1, \ldots, c\}$, the fictitious assortment revenue is greater than the expected ex-post optimum
\begin{align}\label{ineq_equivalence:alpha_t_ell_greater_than_opt_t}
    \alpha_t(\ell) = \frac{c  \cdot 0.9}{ c \cdot 0.9 + 1} &= \frac{\frac{c}{2} \cdot 0.9 + \frac{c}{2}\cdot 2q }{ \frac{c}{2}  \cdot 0.9 +\frac{c}{2}\cdot 2q + 1} + \frac{c  \cdot 0.9}{ c \cdot 0.9 + 1}-\frac{\frac{c}{2} \cdot 0.9 + \frac{c}{2}\cdot 2q }{ \frac{c}{2}  \cdot 0.9 +\frac{c}{2}\cdot 2q + 1} \nonumber\\
    &=  \frac{\frac{c}{2} \cdot 0.9 + \frac{c}{2}\cdot 2q }{ \frac{c}{2}  \cdot 0.9 +\frac{c}{2}\cdot 2q + 1} + \frac{\frac{c}{2} (0.9 - 2q)}{(c \cdot 0.9  + 1)(\frac{c}{2}\cdot 0.9  +\frac{c}{2} \cdot 2q + 1)} \nonumber \\
    &\geq \frac{\frac{c}{2} \cdot 0.9 + \frac{c}{2}\cdot 2q }{ \frac{c}{2}  \cdot 0.9 +\frac{c}{2}\cdot 2q + 1} + \frac{\frac{c}{2} (0.9 - 2q)}{(c+1)^2} > \frac{\frac{c}{2} \cdot 0.9 + \frac{c}{2}\cdot 2q }{ \frac{c}{2}  \cdot 0.9 +\frac{c}{2}\cdot 2q + 1}  + cq \geq  \textsc{OPT}_t.
\end{align}
The first equality uses $\alpha_t(\ell) = \frac{ c \cdot 0.9}{ c\cdot 0.9 + 1}$ as $\ell > \frac{c}{2}$, there are $\frac{c}{2}$ incumbents with parameter $0.9$, and all known entrants in $\mathcal{H}_t$ have parameter $0$ (as $\mathcal{E}^0$ holds). The second equality uses the identity $\frac{A}{A+1} - \frac{B}{B+1} = \frac{A-B}{(A+1)(B+1)}$ for $A=  c\cdot 0.9+ 1$ and $B =\frac{c}{2} \cdot  0.9  + \frac{c}{2}\cdot 2q+ 1$.  The first inequality uses that $c \cdot 0.9  + 1 \leq c+1$ and $\frac{c}{2}\cdot 0.9  +\frac{c}{2} \cdot 2q + 1 \leq c+1$. The second inequality uses $\frac{\frac{c}{2} (0.9 - 2q)}{(c+1)^2} > cq $ as $q < \frac{0.9}{2+2(c+1)^2}$. The last inequality uses \eqref{eq:expected_ex_post_opt_law_of_total_probability_decoposition}. As a result of \eqref{ineq_equivalence:alpha_t_ell_greater_than_opt_t}, \textsc{EFA} explores at most $\frac{c}{2}$ unknown entrants at $t$. Thus, the regret at round $t$ is upper bounded by
\begin{equation*}\label{ineq:regret_per_round_upper_bound_all_zeros}
    \textsc{OPT} - \textsc{Rev}_t(S_t) \leq \frac{\frac{c}{2} \cdot 0.9  + \frac{c}{2} \cdot 2q }{\frac{c}{2} \cdot 0.9  + \frac{c}{2} \cdot 2q + 1} - \frac{\frac{c}{2} \cdot 0.9 +  \frac{c}{2} \cdot q}{\frac{c}{2} \cdot 0.9 +  \frac{c}{2} \cdot q + 1} = \frac{\frac{c}{2} \cdot q }{(\frac{c}{2} \cdot 0.9  + \frac{c}{2} \cdot 2q + 1)(\frac{c}{2} \cdot 0.9 +  \frac{c}{2} \cdot q + 1)} \leq cq.
\end{equation*}
The first inequality uses $\textsc{OPT} = \frac{\frac{c}{2}\cdot0.9   + \frac{c}{2}\cdot 2q }{\frac{c}{2}\cdot0.9   + \frac{c}{2}\cdot 2q + 1}$ (as all unknown entrants have realized to $0$ due to $\mathcal{E}^0$) and $\textsc{Rev}_t(S_t) \geq \frac{ \frac{c}{2}\cdot 0.9 + \frac{c}{2}\cdot q }{\frac{c}{2}\cdot 0.9  + \frac{c}{2}\cdot q  + 1}$ because \textsc{EFA} includes at least $\frac{c}{2}$ incumbents with parameters $0.9$ as it explores at most $\frac{c}{2}$ unknown entrants and each remaining product in the assortment has a parameters at least $q$. The equality uses the identity $\frac{A}{A+1} -\frac{B}{B+1} = \frac{A-B}{(A+1)(B+1)}$ for $(A,B) = (\frac{c}{2} \cdot 0.9   + \frac{c}{2} \cdot 2q ,\frac{c}{2} \cdot 0.9  + \frac{c}{2} \cdot q)$. The last inequality uses that $\frac{c}{2} \cdot 0.9   + \frac{c}{2} \cdot q +1 \geq 1$, $\frac{c}{2}  \cdot 0.9  + \frac{c}{2}  \cdot q  + 1 \geq 1$ and $\frac{c}{2} \cdot  q\leq cq$.
\end{proof}

\subsection{UCB can be arbitrarily suboptimal (Proposition~\ref{prop:ucb_performance_bound_on_instance})}\label{appendix_subsec:ucb_is_arbitrarily_suboptimal_proof}
\begin{proof}[Proof of Proposition~\ref{prop:ucb_performance_bound_on_instance}.]
The second part of the lemma follows by Lemma~\ref{lemma:upper_bound_regret_our_algo}. For the first part, note that that all unknown entrants have the same $\bern(q)$ prior. Thus, for any round $t$ the UCB of each unknown entrant is the same an equal to the $p_t$-th quantile of $\bern(q)$, which we denote by $\psi_t$. As a result, since all $c$ incumbents have parameters smaller than $1$ and greater than $0$, at every round $t$, \textsc{UCB}$(\mathbf{p})$ either does not explore ($\psi_t = 0$) or explores all unknown entrants ($\psi_t = 1$). Given that $\psi_t$ is non-decreasing in $t$ (as $p_t$ is non-decreasing in $t$), there exists some round $H$ such that no entrants are explored at reach round $t \leq H$ and and all entrants are explored at each round $t > H$. We consider cases based on whether $H$ is finite.

\paragraph{Case 1: $H = +\infty$.} This implies \textsc{UCB}$(\mathbf{p})$ never explores and thus $\textsc{Reg}(\textsc{UCB}(\mathbf{p})) = +\infty > \frac{1}{q} \cdot \frac{9}{29c} - 1$. 

\paragraph{Case 2: $H < \infty$.} As \textsc{UCB}$(\mathbf{p})$ offers only incumbents for the first $H$ rounds its regret for this period is non-negative. After round $H$, \textsc{UCB}$(\mathbf{p})$ explores all unknown entrants and thus its regret for this period equals the regret of \textsc{ExploreAll}. As a result, $\textsc{Reg}(\textsc{UCB}(\mathbf{p})) \geq \textsc{Reg}(\textsc{ExploreAll}) \geq \frac{1}{q} \cdot \frac{9}{29c} - 1$; the last inequality uses Lemma~\ref{lemma:lower_bound_regret_ucb}. 
\end{proof}

\subsection{Thompson Sampling can be arbitrarily suboptimal (Proposition \ref{prop:TS_is_arbitrarily_suboptimal})}\label{appendix_subsec:thompson_sampling_is_arbitrarily_suboptimal_proof}
\begin{proof}[Proof of Proposition~\ref{prop:TS_is_arbitrarily_suboptimal}.]
The second part of the proposition follows by applying Lemma~\ref{lemma:upper_bound_regret_our_algo} for $c=2$ and using that $\frac{0.9}{2 + 2(c+1)^2} = 0.045$ and $2c^2(c+1) = 24$.

It remains to show the first part. Each unknown entrant $i \in \{1,2\}$ can have an attraction parameter $1$ (with probability $q$) or $0$ (with probability $1-q$). As a result, the expected ex-post optimum is
\begin{equation*}\label{eq:expected_ex-post_opt_instance}
    \expect[\textsc{OPT}] = (1-q)^2 \frac{0.9 + 2q}{0.9 + 2q + 1} + 2q(1-q) \frac{1+ 0.9}{1 + 0.9 + 1} + q^2 \frac{1+1}{1+1+1}.
\end{equation*}
Let $\textsc{Rev}(S)$ and $q(S)$ be (respectively) the expected revenue and the purchase probability for an unknown product given $S$ at the start. The sample drawn by $\textsc{TS}$ for each unknown product $i \in \{1,2\}$ realizes to $1$ (with probability $q$) or $0$ (with probability $1-q$). As a result, taking an expectation over the assortment $S$ offered by \textsc{TS},
\begin{equation*}\label{eq:expected_purch_prob_TS_instance}
    \expect_{S = \textsc{TS}(\mathcal{H}_1)}[\textsc{Rev}(S)] = (1-q)^2 \frac{0.9 + 2q}{0.9 + 2q + 1} + 2q(1-q) \frac{0.9 + q}{0.9 + q + 1} + q^2 \frac{2q}{2q+1}\quad \text{and}
\end{equation*}
$$\expect_{S = \textsc{TS}(\mathcal{H}_1)}[q(S)] = (1-q)^2 \cdot 0 + 2q(1-q) \cdot \frac{q}{0.9 + q + 1} + q^2 \cdot \frac{2q}{2q+1}\leq 2q(1-q) q + q^2 \cdot 2q = 2q^2
$$
where the inequality uses $\frac{q}{0.9 + q + 1} \leq q$ (as $0.9 + q + 1 \geq 1$) and $\frac{2q}{2q+1} \leq 2q$ (as $2q + 1\geq 1$). 

Under \textsc{TS}, an unknown entrant is purchased with a probability of $\expect_{S = \textsc{TS}(\mathcal{H}_1)}[q(S)]$ independently at every round. As a result, the regret of $\textsc{TS}$ in the first epoch is $\frac{\expect_{S = \textsc{TS}(\mathcal{H}_1)}[\textsc{OPT} - \textsc{Rev}(S)]}{\expect_{S = \textsc{TS}(\mathcal{H}_1)}[q(S)]}$. The expected regret of \textsc{TS} is non-negative in every round, and thus the regret of \textsc{TS} is lower bounded by the regret in the first epoch:
\begin{align*}
    \textsc{Reg}(\textsc{TS}) &\geq \frac{\expect_{S = \textsc{TS}(\mathcal{H}_1)}[\textsc{OPT} - \textsc{Rev}(S)]}{\expect_{S = \textsc{TS}(\mathcal{H}_1)}[q(S)]} \geq \frac{1-q}{q}\Big( \frac{1+ 0.9}{1 + 0.9 + 1}- \frac{0.9 + q}{0.9 + q + 1} \Big)+ \frac{1}{2} \Big(\frac{1+1}{1+1+1}-\frac{2q}{2q+1} \Big)\\
    &= \frac{1-q}{q}\frac{1-q}{(1+0.9 + 1)(0.9 + q + 1)} + \frac{1}{2} \frac{2(1-q)}{(1+1+1)(2q+1)}\geq 
     \frac{1}{q} \cdot \frac{(1-q)^2}{9}.
\end{align*}
The first equality uses the identity $\frac{A}{A+1}-\frac{B}{B+1} = \frac{A-B}{(A+1)(B+1)}$ for $(A,B) \in \{(1+0.9+1, 0.9 + q+1), (1+1+1, 2q+1)\}$. The last inequality uses that $\frac{1-q}{q}\frac{1-q}{(1+0.9 + 1)(0.9 + q + 1)} \geq \frac{1}{q} \cdot \frac{(1-q)^2}{9} $ as $(1+0.9 + 1)(0.9 + q + 1) \leq 9$ and $\frac{1}{2} \frac{2(1-q)}{(1+1+1)(2q+1)} \geq 0$. This concludes the proof of the proposition.
\end{proof}

\subsection{Assortment-size-factor suboptimality of ExploreOne (Proposition~\ref{prop:explore_one_at_a_time_c_factor})}\label{appendix_subsec:explore_one_at_a_time_c_factor_lower_bound_proof}
We focus on the instance $\mathcal{J}(m,c,q,\delta)$ consisting of $m$ entrants with $\mathcal{F} = \bern(q)$ prior and $c$ incumbents each of which with an attraction parameter of $q + \delta$ where $q + \delta \in (0,1)$. The next two lemmas (proofs in Appendix~\ref{appendix_subsec:lower_bound_regret_EFA_proof} and Appendix~\ref{appendix_subsec:proof_lemma_exploreone_regret_lower_bound}) upper bound the regret of \textsc{EFA} and lower bound the regret of \textsc{ExploreOne}.

\begin{lemma}\label{lemma:orfa_regret_upper_bound}
    For any $m \geq 16c^2$ and $q = m^{-\frac{1}{2}}$, the regret of \textsc{EFA} on the instance $\mathcal{J}(m,c,q,\delta)$ is at most
    $$\textsc{Reg}(\textsc{EFA}) \leq  \frac{c}{c+1} \cdot \frac{1-q}{q^2} + \frac{c+1}{q}m \exp\left\{-\frac{9}{32}mq\right\}.$$
\end{lemma}

\begin{lemma}\label{lemma:exploreone_regret_lower_bound}
    For any $m \geq 16c^2$ and $q = m^{-\frac{1}{2}}$, the regret of \textsc{ExploreOne} on $\mathcal{J}(m,c,q,\delta)$ is at least
    $$\textsc{Reg}(\textsc{ExploreOne}) \geq \frac{c}{2}\cdot \frac{1-q-\delta}{q^2} - \Big(\frac{m}{q} + \frac{2c}{q^2} \Big) \exp\left\{-\frac{9}{32} mq \right\}.$$
\end{lemma}
\begin{proof}[Proof of Proposition~\ref{prop:explore_one_at_a_time_c_factor}.]
Consider the instance $\mathcal{J}(m,c,q,\delta)$ with $m \geq 16c^2$, $q = m^{-\frac{1}{2}}$ and $\delta = m^{-\frac{1}{2}}$ where $m$ will be chosen later. By Lemmas \ref{lemma:orfa_regret_upper_bound} and \ref{lemma:exploreone_regret_lower_bound}, the ratio of the regrets of \textsc{ExploreOne} and \textsc{EFA} is at least
\begin{align*}\label{eq:lower_bound_regret_ratio}
    \frac{\textsc{Reg}(\textsc{ExploreOne})}{\textsc{Reg}(\textsc{EFA})} &\geq \frac{\frac{c}{2}\cdot \frac{1-q-\delta}{q^2} - \Big(\frac{m}{q} + \frac{2c}{q^2} \Big) \exp\left\{-\frac{9}{32} mq \right\}}{\frac{c}{c+1} \cdot \frac{1-q}{q^2} + \frac{c+1}{q}m \exp\left\{-\frac{9}{32} mq \right\}} = \frac{\frac{c}{2}(1-q-\delta) - \Big(mq + 2c \Big) \exp\left\{-\frac{9}{32} mq \right\}}{\frac{c}{c+1}  (1-q) + (c+1)mq  \exp\left\{-\frac{9}{32} mq \right\}} \nonumber\\
    &=\underbrace{\frac{\frac{c}{2} (1-2m^{-\frac{1}{2}}) - (m^{\frac{1}{2}} + 2 c) \exp\{-\frac{9}{32}m^{\frac{1}{2}}\}}{\frac{c}{c+1}(1-m^{-\frac{1}{2}})+ (c+1)m^{\frac{1}{2}} \exp\{-\frac{9}{32}m^{\frac{1}{2}}\} }}_{\eqqcolon \textsc{LB}(m,c)}. 
\end{align*}
where the first equality multiplies the numerator and denominator by $q^2$ and the last equality uses $q = m^{-\frac{1}{2}}$ and $\delta =  m^{-\frac{1}{2}}$.
Using that $\lim_{m \to \infty} m^{-\frac{1}{2}} = 0$, $ \lim_{m \to \infty} m^{\frac{1}{2}}\exp\{-\frac{9}{32}m^{\frac{1}{2}}\} = 0$, and $\lim_{m \to \infty} (m^{\frac{1}{2}} + 2c) \exp\{-\frac{9}{32}m^{\frac{1}{2}}\} = 0$ yields $\lim_{m \to \infty}\textsc{LB}(m,c) = \frac{\frac{c}{2}}{\frac{c}{c+1}} = \frac{c+1}{2}$. Thus, there exists $M > 0$ such that for $m > M$, $\textsc{LB}(m,c)> \frac{c}{2}$. Taking $m = \max(M + 1, 16c^2)$ yields the proposition.
\end{proof}

\subsection{Lower bound for the regret of EFA (Lemma~\ref{lemma:orfa_regret_upper_bound})}\label{appendix_subsec:lower_bound_regret_EFA_proof}
To prove Lemma~\ref{lemma:orfa_regret_upper_bound} and Lemma~\ref{lemma:exploreone_regret_lower_bound}, we derive two properties (Lemma~\ref{lemma:probability_bound_event_at_least_c_ones} and Lemma~\ref{lemma:upper_lower_bounds_expectation_num_rounds_between_ones}) which hold any sequence of $m$ i.i.d. $\bern(q)$ random variables. 
Let $V_i$ be independent $\bern(q)$ random variables for $i \in \{1, \ldots, m\}$. Let $\mathcal{E} = \left\{ \sum_{i=1}^m V_i \geq c \right\}$ be the event that there are at least $c$ ones amongst $\left\{V_i: i \in \{1, \ldots, m\} \right\}$ and $\mathcal{E}^{C}$ be its complement. The next lemma upper bounds the probability of $\mathcal{E}^{C}$ and thus shows that $\mathcal{E}$ happens with high probability. 

\begin{lemma}\label{lemma:probability_bound_event_at_least_c_ones}
    For any $m \geq 16c^2$ and $q = m^{-\frac{1}{2}}$, $\prob[\mathcal{E}^{C}] \leq \exp\left\{-\frac{9}{32} mq \right\}$.
\end{lemma}
\begin{proof}[Proof of Lemma~\ref{lemma:probability_bound_event_at_least_c_ones}.]
The probability of $\mathcal{E}^{C}$ is upper bounded by
    \begin{align*}
        \prob[\mathcal{E}^{C}] &= \prob \Big[\sum_{i=1}^{m} V_i \leq c-1\Big] =\prob[\textsc{Binomial}(m,q) \leq c-1] \leq \prob[\textsc{Binomial}(m,q) \leq c] \\
        &\leq \exp \Big\{-\frac{mq}{2}\Big(1 -\frac{c}{mq}\Big)^2\Big\}\leq \exp\Big\{-\frac{mq}{2}\Big(\frac{3}{4}\Big)^2\Big\} = \exp\Big\{-\frac{9}{32}mq \Big\}
    \end{align*}
where the second equality uses that $\sum_{i=1}^{m} V_i \sim \textsc{Binomial}(m,q)$, the second inequality uses the Chernoff bound, and the third inequality uses $1 -\frac{c}{mq} = 1-\frac{c}{m^{\frac{1}{2}}} \geq 1-\frac{c}{4c} = \frac{3}{4}$ as $m\geq 16c^2$. \end{proof}

Let $T_i$ be the index of the $i$-th one in the sequence $(V_1, V_2, \ldots, V_m)$ if the sequence contains at least $i$ ones and be $m$ otherwise for $i \in \{1, \ldots, c\}$. Let $\tau_i = T_{i}-T_{i-1}$ for $i \in \{1, \ldots, c\}$ be number of unknown entrants purchased between the $(i-1)$-th and $i$-th one (assume $T_0 = 0$ for convenience). The following lemma (proof in Appendix~\ref{appendix_subsec:proof_lemma_C11}) provides lower and upper bounds on $\expect[\tau_i|\mathcal{E}]$ and thus shows $\expect[\tau_i|\mathcal{E}]$ concentrates around $\frac{1}{q}$.

\begin{lemma}\label{lemma:upper_lower_bounds_expectation_num_rounds_between_ones}
    For any $m \geq 16c^2$, $q = m^{-\frac{1}{2}}$, and any $i \in \{1, \ldots, c\}$, $$\expect[\tau_i|\mathcal{E}] \geq \frac{1}{q} - \left(\frac{m}{c} + \frac{1}{q} \right)\exp \Big\{-\frac{9}{32}mq \Big\} \quad \text{and} \quad\expect[\tau_i|\mathcal{E}] \leq \frac{1}{q}.$$
\end{lemma}

\begin{proof}[Proof of Lemma~\ref{lemma:orfa_regret_upper_bound}.]
The proof applies Lemma~\ref{lemma:probability_bound_event_at_least_c_ones} and Lemma~\ref{lemma:upper_lower_bounds_expectation_num_rounds_between_ones} where $V_i$ is the attraction parameter of the $i$-th purchased entrant under \textsc{EFA} and $\mathcal{E} = \{ \sum_{i=1}^m V_i \geq c\}$ is the event there are at least $c$ entrants with parameter one. Letting $\textsc{CondReg}(v_1, \ldots, v_m) = \expect \Big[\sum_{t=1}^{\infty} \textsc{OPT} - R_t \Big|(V_1, \ldots, V_m) = (v_1, \ldots, v_m) \Big]$ be the expected regret conditioned on a realization of the attraction parameters $\{V_i\}_{i=1}^m$, by the law of total expectation the regret of \textsc{EFA} is decomposed as
\begin{equation*}\label{eq:law_of_total_probability_regret_efa_with_respect_to_event_at_least_c_ones}
   \textsc{Reg}(\textsc{EFA}) =  \expect_{\substack{v_i \overset{\text{i.i.d.}}{\sim} \bern(q) \\ i \in \{1, \ldots, m\}}}[\textsc{CondReg}(v_1, \ldots, v_m)|\mathcal{E}]\prob[\mathcal{E}] +\expect_{\substack{v_i \overset{\text{i.i.d.}}{\sim} \bern(q) \\ i \in \{1, \ldots, m\}}}[\textsc{CondReg}(v_1, \ldots, v_m)|\mathcal{E}^{C}]\prob[\mathcal{E}^{C}].
\end{equation*}
\paragraph{Bad event.} The per-round regret is at most one and the purchase probability of an entrant in any assortment that contains at least one entrant is at least $\frac{q}{c+1}$. Thus, the expected regret of any epoch is at most $\frac{c+1}{q}$; given that there are at most $m$ epochs and $\prob[\mathcal{E}^{C}] \leq \exp\left\{-\frac{9}{32} mq \right\}$ by Lemma~\ref{lemma:probability_bound_event_at_least_c_ones}, the bad-event term is:
\begin{equation*}\label{ineq:upper_bound_freq_regret_orfa_less_than_c_ones}
    \expect_{\substack{v_i \overset{\text{i.i.d.}}{\sim} \bern(q) \\ i \in \{1, \ldots, m\}}}[\textsc{CondReg}(v_1, \ldots, v_m)|\mathcal{E}^{C}]\prob[\mathcal{E}^{C}] \leq \frac{c+1}{q}m \exp\left\{-\frac{9}{32} mq \right\}.
\end{equation*}

\paragraph{Good event.} For any round $t$ where $i$ of the entrants have realized parameters equal to $1$, given that there are $c$ incumbents with parameters $q + \delta$, the expected ex-post optimum is lower bounded by $\textsc{OPT}_t \geq \frac{i + (c-i) (q + \delta)}{i + (c-i) (q + \delta) + 1}$ and the fictitious assortment revenues at $\ell = c-i$ and $\ell = c-i+1$ unknown entrants are given by $\alpha_t(c-i+1) = \frac{c}{c+1}$ and $\alpha_t(c-i) =  \frac{i + (c-i) (q + \delta)}{i + (c-i) (q + \delta) + 1}$. As a result, when there are $i$ known entrants with parameters equal to $1$, \textsc{EFA} explores $c-i$ unknown entrants along with them (line~\ref{alg_line:set_number_products}, Algorithm~\ref{alg:optimal_explorer}); denote this assortment by $S^i$. Recall that there are $\tau_{i+1}$ epochs with $i$ known entrants with parameters equal to one. Letting $\tau(S^i)$ and $\rev(S^i)$ be the epoch length and revenue under $S^i$, the regret in the good event is at most
\begin{align*}
    & \expect_{\substack{v_i \overset{\text{i.i.d.}}{\sim} \bern(q) \\ i \in \{1, \ldots, m\}}}[\textsc{CondReg}(v_1, \ldots, v_m)|\mathcal{E}]\prob[\mathcal{E}] \leq \expect \left[\textsc{CondReg}(v_1, \ldots, v_m)\middle|\mathcal{E}\right] \nonumber\\
    &=  \expect \left[\sum_{i=0}^{c-1} \tau(S^i)\left(\textsc{OPT}-\rev(S^i) \right) \tau_{i+1} \middle|\mathcal{E}\right] = \expect \left[\sum_{i=0}^{c-1}\frac{(c-i)q + i + 1}{(c-i)q}\Big(\frac{c}{c+1}-\frac{(c-i)q+i }{(c-i)q+i+ 1} \Big) \tau_{i+1} \middle| \mathcal{E} \right] \nonumber\\
    &= \expect \left[\sum_{i=0}^{c-1}\frac{(c-i)q + i + 1}{(c-i)q}\frac{(c-i)(1-q)}{(c+1)((c-i)q+i+ 1)}  \tau_{i+1} \middle| \mathcal{E} \right] = \expect\left[\sum_{i=0}^{c-1}\frac{1-q}{(c+1)q}  \tau_{i+1} \middle| \mathcal{E} \right] \nonumber\\
    &= \sum_{i=0}^{c-1}\frac{1-q}{(c+1)q} \expect \left[ \tau_{i+1} \middle| \mathcal{E} \right] \leq \sum_{i=0}^{c-1}\frac{1-q}{(c+1)q} \frac{1}{q} = \frac{c}{c+1}\cdot \frac{1-q}{q^2}.
\end{align*}
The first inequality uses $\prob[\mathcal{E}]\leq 1$ and that the conditional regret is non-negative. The first equality expresses the conditional regret as a sum over all epoch regrets and uses that there are $\tau_{i+1}$ epochs with $i$ known entrants with parameters equal to $1$. The second equality uses $\tau(S^i) = \frac{(c-i)q + i + 1}{(c-i)q}$, $\rev(S^i) =\frac{(c-i)q+i }{(c-i)q+i+ 1} $, and $\opt = \frac{c}{c+1}$ as there are $c$ entrants with attraction parameters one (as $\mathcal{E}$ holds). The third equality uses the identity $\frac{A}{A+1} -\frac{B}{B+1} = \frac{A-B}{(A+1)(B+1)}$ for $(A,B) \in \{(c, (c-i)q + i)\}_{i=0}^{c-1}$. The second inequality uses $\expect\left[ \tau_{i+1} \middle| \mathcal{E} \right] \leq \frac{1}{q}$ by Lemma~\ref{lemma:upper_lower_bounds_expectation_num_rounds_between_ones}. Therefore, $\textsc{Reg}(\textsc{EFA}) \leq \frac{c}{c+1} \cdot \frac{1-q}{q^2} + \frac{c+1}{q}m \exp\left\{-\frac{9}{32}mq\right\}$.
\end{proof}

\subsection{Upper bound for the regret of ExploreOne (Lemma~\ref{lemma:exploreone_regret_lower_bound})}\label{appendix_subsec:proof_lemma_exploreone_regret_lower_bound}

\begin{proof}[Proof of Lemma~\ref{lemma:exploreone_regret_lower_bound}.]
The proof applies Lemma~\ref{lemma:probability_bound_event_at_least_c_ones} and Lemma~\ref{lemma:upper_lower_bounds_expectation_num_rounds_between_ones} where $V_i$ is the attraction parameter of the $i$-th purchased entrant under \textsc{ExploreOne} and $\mathcal{E} = \{ \sum_{i=1}^m V_i \geq c\}$ is the event there are at least $c$ entrants with parameter one. For a realization $\mathbf{v} = (v_1, \ldots, v_m)$ of the unknown entrants, we define $\textsc{CondReg}(\mathbf{v}) = \expect \Big[\sum_{t=1}^{\infty} \textsc{OPT} - R_t \Big|(V_1, \ldots, V_m) = \mathbf{v} \Big]$ be the expected regret conditioned on a realization of the attraction parameters $\{V_i\}_{i=1}^m$.
Using the law of total expectation and denoting the bad event probability by $\eta \coloneqq \exp\{-\tfrac{9}{32}mq\}$, the regret of \textsc{ExploreOne} is lower-bounded by 
\begin{align*}\label{eq:law_of_total_probability_regret_explore_one_with_respect_to_at_least_c_ones_event}
    \textsc{Reg}(\textsc{ExploreOne}) &= \expect_{\substack{v_i \overset{\text{i.i.d.}}{\sim} \bern(q) \\ i \in \{1, \ldots, m\}}}[\textsc{CondReg}(\mathbf{v})|\mathcal{E}]\prob[\mathcal{E}] +\expect_{\substack{v_i \overset{\text{i.i.d.}}{\sim} \bern(q) \\ i \in \{1, \ldots, m\}}}[\textsc{CondReg}(\mathbf{v})|\mathcal{E}^{C}]\prob[\mathcal{E}^{C}] \nonumber\\
    &\geq \expect_{\substack{v_i \overset{\text{i.i.d.}}{\sim} \bern(q) \\ i \in \{1, \ldots, m\}}}[\textsc{CondReg}(\mathbf{v})|\mathcal{E}]\prob[\mathcal{E}] \geq \expect_{\substack{v_i \overset{\text{i.i.d.}}{\sim} \bern(q) \\ i \in \{1, \ldots, m\}}}[\textsc{CondReg}(\mathbf{v})|\mathcal{E}](1-\eta).
\end{align*}
The first inequality uses that the conditional regret is non-negative; the second inequality holds by 
Lemma~\ref{lemma:probability_bound_event_at_least_c_ones}.

For any round $t$ with $i$ known entrants with parameters equal to $1$, $\textsc{ExploreOne}$ offers the assortment $S^i$ consisting of a single unknown entrant together $i$ known entrants with parameters $1$ and $c-i-1$ incumbents with parameters $q+\delta$. Recall that there are $\tau_{i+1}$ epochs with $i$ known entrants with parameters equal to $1$. Letting $\tau(S^i)$ and $\rev(S^i)$ be the epoch length and revenue under $S^i$, the regret when $\mathcal{E}$ holds is at least:
\begin{align*}
    &\expect \left[\textsc{CondReg}(\mathbf{v})\middle|\mathcal{E}\right] (1-\eta) =  \expect \left[\sum_{i=0}^{c-1} \tau(S^i)\left(\textsc{OPT}-\rev(S^i) \right) \tau_{i+1} \middle|\mathcal{E}\right](1-\eta) \nonumber\\
    &= \expect \left[\sum_{i=0}^{c-1} \frac{q + (c-i-1)(q + \delta) + i + 1}{q}\Big(\frac{c}{c+1} - \frac{q + (c-i-1)(q + \delta) + i }{q + (c-i-1)(q + \delta) + i + 1} \Big)\tau_{i+1} \middle|\mathcal{E}\right](1-\eta) \nonumber\\
    &= \expect \left[\sum_{i=0}^{c-1} \frac{ c-i-q - (c-i-1)(q + \delta)}{q(c+1)}\tau_{i+1} \middle|\mathcal{E}\right](1-\eta)  \geq \expect \left[\sum_{i=0}^{c-1} \frac{ (c-i)(1-q - \delta)}{q(c+1)}\tau_{i+1} \middle|\mathcal{E}\right] (1-\eta) \nonumber\\
    &\geq  \left[\sum_{i=0}^{c-1} \frac{ (c-i)(1-q - \delta)}{q(c+1)} \left(\frac{1}{q}  - \left(\frac{m}{c} + \frac{1}{q} \right)\eta \right) \right](1-\eta)
     \geq \frac{c}{2}\cdot \frac{1-q-\delta}{q^2}- \Big(\frac{m}{q}+\frac{2c}{q^2}\Big)\eta
\end{align*}
The first equality expresses the conditional regret as a sum over all epoch regrets and uses that there are $\tau_{i+1}$ epochs with $i$ known entrants with parameters equal to $1$. The second equality uses that $\tau(S^i) = \frac{q + (c-i-1)(q + \delta) + i + 1}{q}$, $\rev(S^i) = \frac{q + (c-i-1)(q + \delta) + i }{q + (c-i-1)(q + \delta) + i + 1} $, and $\opt = \frac{c}{c+1}$ as there are at least $c$ entrants with parameters one (as $\mathcal{E}$ holds). The third equality uses the identity $\frac{A}{A+1}-\frac{B}{B+1} = \frac{A-B}{(A+1)(B+1)}$ for $(A,B) \in \{(c, q + (c-i-1)(q + \delta) + i)\}_{i=0}^c$ and cancels the term $q + (c-i-1)(q + \delta) + i + 1$ from the numerator of the first term and the denominator of the second term. The first inequality uses $c-i-q - (c-i-1)(q + \delta) = (c-i)(1-q-\delta) + \delta \geq (c-i)(1-q-\delta)$. The second inequality uses $\expect \left[\tau_{i+1} \middle|\mathcal{E}\right] \geq \frac{1}{q} -\left(\frac{m}{c} + \frac{1}{q} \right)\eta $ by Lemma~\ref{lemma:upper_lower_bounds_expectation_num_rounds_between_ones}. The last inequality uses Lemma~\ref{lemma:inequality_algebraic_manipulation} (stated and proven below). 
\end{proof}

\begin{lemma}\label{lemma:inequality_algebraic_manipulation}
For any $m \in \mathbb{N}$, $c \in \mathbb{N}$, $q\in (0,1)$, and $\delta \in (0,1)$ with $q + \delta \in (0,1)$,
   $$ \underbrace{\sum_{i=0}^{c-1} \frac{ (c-i)(1-q - \delta)}{q(c+1)} \left(\frac{1}{q}  - \left(\frac{m}{c} + \frac{1}{q} \right)\eta \right) (1-\eta)}_{=\textsc{L}}\geq \frac{c}{2}\cdot \frac{1-q-\delta}{q^2}- \Big(\frac{m}{q}+\frac{2c}{q^2}\Big)\eta$$
\end{lemma}
\begin{proof}
The left hand side is lower bounded by
    \begin{align*}
        &\textsc{L}\geq \sum_{i=0}^{c-1} \frac{ (c-i)(1-q - \delta)}{q^2(c+1)}  -  \sum_{i=0}^{c-1} \frac{ (c-i)(1-q - \delta)}{q(c+1)} \left(\frac{m}{c} + \frac{2}{q} \right)\eta \geq  \sum_{i=0}^{c-1} \frac{ (c-i)(1-q - \delta)}{q^2(c+1)} - \sum_{i=0}^{c-1} \frac{1}{q}\Big(\frac{m}{c} + \frac{2}{q} \Big)\eta \\
        &= \frac{1-q-\delta}{q^2(c+1)}\sum_{i=0}^{c-1} (c-i)- \frac{c}{q} \left(\frac{m}{c} + \frac{2}{q} \right)\eta = \frac{\frac{c(c+1)}{2}}{c+1} \frac{1-q-\delta}{q^2}- \left(\frac{m}{q} + \frac{2c}{q^2} \right)\eta = \frac{c}{2}\cdot \frac{1-q-\delta}{q^2}- \left(\frac{m}{q} + \frac{2c}{q^2} \right)\eta.
    \end{align*}
    The first inequality uses $\left(\frac{1}{q}  - \left(\frac{m}{c} + \frac{1}{q} \right)\eta \right) (1-\eta) \geq \frac{1}{q}  - \left(\frac{m}{c} + \frac{2}{q} \right)\eta$ and expands the difference $\left(\frac{1}{q}  - \left(\frac{m}{c} + \frac{2}{q} \right)\eta \right)$. The second inequality uses $-\frac{ (c-i)(1-q - \delta)}{q(c+1)} \geq -\frac{1}{q}$ as $\frac{(c-i)(1-q - \delta)}{c+1} \leq 1$ (since $q + \delta \in (0,1)$).
\end{proof}

\subsection{Concentration bound on expected time between two successes (Lemma~\ref{lemma:upper_lower_bounds_expectation_num_rounds_between_ones})}\label{appendix_subsec:proof_lemma_C11}
To prove Lemma~\ref{lemma:upper_lower_bounds_expectation_num_rounds_between_ones}, the following lemma establishes the expectation of any random variable decreases if we learn that it is smaller than an independent random variable. 
\begin{lemma}\label{lemma:if_X_and_Y_are independent_ev_of_X_condtioned_on_X_smaller_than_Y_is_at_most_the_ev_of_x}
    For any independent random variables $X$ and $Y$, 
    $\expect[X| X \leq Y] \leq \expect[X]$
\end{lemma}
\begin{proof}[Proof of Lemma~\ref{lemma:if_X_and_Y_are independent_ev_of_X_condtioned_on_X_smaller_than_Y_is_at_most_the_ev_of_x}.]
It suffices to show that for any $a$, $\prob[X \leq a|X \leq Y] \geq \prob[X \leq a]$. This holds as:
\begin{align*}
    &\prob[X \leq a|X \leq Y] = \frac{\int_{y} \prob[X \leq a,X \leq y|Y = y] p_{Y}(y) dy}{\prob[X \leq Y]}= \frac{\int_{y} \prob[X \leq a,X \leq y] p_{Y}(y) dy}{\prob[X \leq Y]} \\
    &\geq \frac{\int_{y} \prob[X \leq a]\prob[X \leq y] p_{Y}(y) dy}{\prob[X \leq Y]} = \prob[X \leq a]\frac{\int_{y} \prob[X \leq y] p_{Y}(y) dy}{\prob[X \leq Y]}  = \prob[X \leq a]\frac{\prob[X \leq Y]}{\prob[X \leq Y]} = \prob[X \leq a]
\end{align*}
where the second and fourth equalities use that $X$ and $Y$ are independent, and the inequality uses that $\prob[X \leq a, X \leq y] = \min(\prob[X \leq a], \prob[X \leq y]) \geq \prob[X \leq a]\prob[X \leq y]$.
\end{proof}

\begin{proof}[Proof of Lemma~\ref{lemma:upper_lower_bounds_expectation_num_rounds_between_ones}.]
Let $\{\tilde{V}_i\}_{i=1}^{\infty}$ be i.i.d. $\bern(q)$ random variables coupled with $(V_1, \ldots, V_m)$ so that $V_i = \tilde{V}_i$ for $i \in \{1, \ldots, m\}$. Let $\tilde{T}_i \coloneq \min\{j: \sum_{\ell = 1}^{j} \tilde{V}_{\ell} = i\}$ and $\tilde{\tau}_i = \tilde{T}_{i} - \tilde{T}_{i-1}$ be the index of the $i$-th one and the number of indices between the $(i-1)$-th and the $i$-th one in the sequence $\{\tilde{V}_i\}_{i=1}^{\infty}$ for $i \in \mathbb{N}$. The random variables $\{\tilde{\tau}_i\}_{i=1}^{\infty}$ are i.i.d. $\text{Geometric}(q)$ as $\{\tilde{V}_i\}_{i=1}^{\infty}$ are i.i.d. $\bern(q)$. The event that there are at at least $c$ ones in $(V_1, \ldots, V_m)$, $\mathcal{E}$, is exactly the event $\sum_{i=1}^{c} \tilde{\tau}_i \leq m$. Due to the coupling, conditioned on the event $\sum_{i=1}^{c} \tilde{\tau}_i \leq m$, $\tau_i = \tilde{\tau}_i$ for $i \in \{1, \ldots, c\}$. Hence, the conditional expectation in the lemma is 
\begin{equation}\label{eq:condtional_expectation_rewriting_in_terms_of_tilde_taus}
    \expect[\tau_i|\mathcal{E}] = \expect\left[\tau_i \middle|\sum_{i=1}^{c} \tilde{\tau}_i \leq m \right] =  \expect\left[\tilde{\tau}_i \middle|\sum_{i=1}^{c} \tilde{\tau}_i \leq m \right].
\end{equation}
The last conditional expectation is upper bounded by 
\begin{equation}\label{ineq:upper_bound_conditional_expectation} 
    \expect\left[\tilde{\tau}_i \middle| \sum_{j=1}^{c} \tilde{\tau}_j \leq m\right] = \expect\left[\tilde{\tau}_i \middle| \tilde{\tau}_i \leq m-\sum_{j \in \{1, \ldots, c\} \setminus \{i\}} \tilde{\tau}_j\right] \leq \expect\left[\tilde{\tau}_i\right] = \frac{1}{q}.
\end{equation}
where the inequality uses Lemma~\ref{lemma:if_X_and_Y_are independent_ev_of_X_condtioned_on_X_smaller_than_Y_is_at_most_the_ev_of_x} for the independent variables $X = \tau_j$ and $Y = m-\sum_{j \in \{1, \ldots, c\} \setminus \{i\}}\tilde{\tau}_j$, and the last equality uses $\expect\left[\tilde{\tau}_i\right] = \frac{1}{q}$ as $\tilde{\tau}_i \sim \text{Geometric}(q)$. Combining \eqref{eq:condtional_expectation_rewriting_in_terms_of_tilde_taus} and \eqref{ineq:upper_bound_conditional_expectation} shows the second part. 

Letting $\eta \coloneq \exp\left\{-\frac{9}{32} mq \right\}$, the conditional expectation is lower bounded by 
\begin{align}\label{ineq:lower_bound_condtional_expectation}
    \expect\left[\tilde{\tau}_i \middle| \sum_{j=1}^{c} \tilde{\tau}_j \leq m\right] &\geq \expect\left[\tilde{\tau}_i \middle| \sum_{j=1}^{c} \tilde{\tau}_j \leq m\right]\prob\left[\sum_{j=1}^{c} \tilde{\tau}_j \leq m\right] = \expect\left[\tilde{\tau}_i\right] - \expect\left[\tilde{\tau}_i \middle| \sum_{j=1}^{c} \tilde{\tau}_j > m\right] \prob\left[\sum_{j=1}^{c} \tilde{\tau}_j > m\right] \nonumber\\
    &\geq \expect\left[\tilde{\tau}_i\right] -\expect\left[\tilde{\tau}_i \middle| \sum_{j=1}^{c} \tilde{\tau}_j > m\right] \eta  = \frac{1}{q}-\expect\left[\tilde{\tau}_i \middle| \sum_{j=1}^{c} \tilde{\tau}_j > m\right] \eta \geq \frac{1}{q} - \left(\frac{m}{c} + \frac{1}{q}\right) \eta.
\end{align}
The first inequality uses $\prob\left[\sum_{j=1}^{c} \tilde{\tau}_j \leq m\right]  \leq 1$. The first equality uses the law of total expectation. The second inequality uses $\prob\left[\sum_{j=1}^{c} \tilde{\tau}_j > m\right] \leq \eta$ by Lemma~\ref{lemma:probability_bound_event_at_least_c_ones}. The second equality uses $\expect\left[\tilde{\tau}_i\right] = \frac{1}{q}$. To show the last inequality we show $\expect\left[\tilde{\tau}_i \middle| \sum_{j=1}^{c} \tilde{\tau}_j > m\right] \leq \frac{m}{c} + \frac{1}{q}$. To do so, let $\hat{T}_i = \min\{j \geq 1: \sum_{\ell=m+1}^{m+j}\tilde{V}_{\ell} = i\}$ and $\hat{\tau}_i = \hat{T}_{i}-\hat{T}_{i-1}$ be (respectively) the index of the first one after index $m$ and the number of indices between the $(i-1)$-th and the $i$-th one after $m$ for $i \in \mathbb{N}$. Letting $\tilde{V}_{:m} = \sum_{i=1}^m \tilde{V}_i$ be the number of ones in the first $m$ trials, this expectation is upper bounded by
\begin{align*}
       &\expect\left[\tilde{\tau}_i \middle| \sum_{j=1}^{c} \tilde{\tau}_j > m\right] = \frac{1}{c}\expect\left[\sum_{j=1}^{c} \tilde{\tau}_j \middle| \sum_{j=1}^{c} \tilde{\tau}_j > m\right] = \frac{1}{c}\expect\left[\sum_{j=1}^{c} \tilde{\tau}_j \middle| \tilde{V}_{:m} <c\right] = \frac{1}{c}\expect\left[m + \sum_{j=1}^{c-\tilde{V}_{:m}} \hat{\tau}_j \middle| \tilde{V}_{:m} <c\right]\\
       &\leq \frac{1}{c}\expect\left[m + \sum_{j=1}^{c} \hat{\tau}_j \middle| \tilde{V}_{:m} <c\right] = \frac{1}{c}\expect\left[m + \sum_{j=1}^{c} \hat{\tau}_j\right] = \frac{m}{c} + \frac{1}{q}.
\end{align*}
The first equality uses that $\tilde{\tau}_1, \ldots, \tilde{\tau}_c$ are i.i.d.. The second equality uses the event that there are less than $c$ ones in the first $m$ trials, $\tilde{V}_{:m} <c$ is exactly the event that $\sum_{j=1}^c \tilde{\tau}_j > m$. The third equality uses that  $\sum_{j=1}^{c} \tilde{\tau}_j  =m + \sum_{j=1}^{c-\tilde{V}_{:m}} \hat{\tau}_j$ since conditioned on there being less than $c$ ones in the first $m$ trials, the time until the $c$-th one equals to $m$ plus the time until $c-\tilde{V}_{:m}$ ones are reached starting from $m$. The fourth equality uses that $\{\hat{\tau}_j\}_{j=1}^c$ and $V_{:m}$ are independent. The fifth equality uses $\expect[\hat{\tau}_j] = \frac{1}{q}$. Combining \eqref{eq:condtional_expectation_rewriting_in_terms_of_tilde_taus} and \eqref{ineq:lower_bound_condtional_expectation} yields the first part and completes the proof.
\end{proof}

\subsection{Assortment-size-factor optimality of ExploreOne (Remark~\ref{remark:explore_one_at_a_time_c_factor_upper_bound})}\label{appendix_subsec:explore_one_at_a_time_c_factor_upper_bound_proof}
The following proposition establishes that the regret of \textsc{ExploreOne} is at most $c$ times larger than the regret of \textsc{EFA}. This shows that the lower bound of $\frac{c}{2}$ on the ratio $\frac{\textsc{Reg}(\textsc{ExploreOne})}{\textsc{Reg}(\textsc{EFA})}$ established in Proposition~\ref{prop:explore_one_at_a_time_c_factor} is optimal up to a factor of two. 

\begin{proposition}\label{prop:explore_one_at_most_c_orfa}
    For any instance with $w^{\star}_{(c)}(\mathcal{I}_1) > \tildew$, $\textsc{Reg}(\textsc{ExploreOne}) \leq c\cdot\textsc{Reg}(\textsc{EFA})$.
\end{proposition}

To prove the proposition, we need several auxiliary lemmas. The first lemma establishes that the regret of \textsc{ExploreOne} given any terminal history is zero, similar to the second part of Lemma~\ref{lemma:any_policy_has_a_nonnegative_regret_starting_from_a_terminal_history_and_pi_star_as_a_zero_regret_starting_from_a_terminal_history} in Appendix~\ref{appendix_sec:there_exists_optimal_policy_which_is_stationary}. 

\begin{lemma}\label{lemma:analogue_lemma_B1_exploreone}
    For any terminal history $\mathcal{H}$, $\textsc{Reg}(\textsc{ExploreOne};\mathcal{H}) = 0$.
\end{lemma}
\begin{proof}[Proof of Lemma~\ref{lemma:analogue_lemma_B1_exploreone}.]
Given that $\mathcal{H}_1 = \mathcal{H}$ is terminal,  Lemma~\ref{lemma:if_H_1_is_terminal_H_t_is_terminal_with_probability_1} implies that $\mathcal{H}_t$ for any round $t$ is terminal. Thus, \textsc{ExploreOne} offers the myopic revenue-maximizing assortment at any round $t$ starting from $\mathcal{H}_1 = \mathcal{H}$. As a result, starting from $\mathcal{H}$, \textsc{ExploreOne} is the same policy as the policy $\pi^{\star}$ in Lemma~\ref{lemma:epoch_optimal_policy_is_optimal}. Therefore, $\textsc{Reg}(\textsc{ExploreOne};\mathcal{H}) = \textsc{Reg}(\pi^{\star};\mathcal{H})  = 0$, where the last equality uses Lemma~\ref{lemma:any_policy_has_a_nonnegative_regret_starting_from_a_terminal_history_and_pi_star_as_a_zero_regret_starting_from_a_terminal_history}.
\end{proof}

Recall that $\mathcal{E}(\mathcal{H}) = \{S \subseteq \mathcal{N}: |S| \leq c, S \cap (\mathcal{N} \setminus I(\mathcal{H})) \neq \emptyset \}$ if the set of all assortments which contain at least one unknown entrant given a history $\mathcal{H}$ and that $\textsc{EpochReg}(S;\mathcal{H})$ is the expected regret during the epoch if an assortment $S \in \mathcal{E}(\mathcal{H})$ is offered repeatedly. The following lemma shows that the epoch regret of \textsc{ExploreOne} is at most $c$ times larger than the epoch regret of \textsc{EFA}. Recall that $U_{\ell}$ is any set of $\ell$ unknown entrants and let $S^{\star}_{(c-\ell)}(\mathcal{H})$ be a set of the $c-\ell$ most attractive known products given $\mathcal{H}$. 

\begin{lemma}\label{lemma:expoch_regret_exploreone_at_most_c_epoch_regret_orfa_any_history}
For any instance with $w^{\star}_{(c)}(\mathcal{I}_1) > \tildew$ and any non-terminal history $\mathcal{H}$,
    $$\textsc{EpochReg}(S^{\star}_{(c-1)}(\mathcal{H}) \cup U_1;\mathcal{H}) \leq c \min\limits_{S \in \mathcal{E}(\mathcal{H})} \textsc{EpochReg}(S;\mathcal{H}).$$
\end{lemma}

\begin{proof}[Proof of Propositon~\ref{prop:explore_one_at_most_c_orfa}]
We show that $\textsc{Reg}(\textsc{ExploreOne}; \mathcal{H}) \leq c \textsc{Reg}(\textsc{EFA}; \mathcal{H})$ for any history $\mathcal{H}$ by induction on the number of unknown entrants given $\mathcal{H}$.

\paragraph{Base Case.} The history $\mathcal{H}$ contains no unknown entrants; thus $\mathcal{H}$ is terminal. As a result, it holds that $\textsc{Reg}(\textsc{ExploreOne};\mathcal{H})  = 0$ by Lemma~\ref{prop:explore_one_at_most_c_orfa} and $\textsc{Reg}(\textsc{EFA};\mathcal{H})  = 0$ by Lemma~\ref{lemma:any_policy_has_a_nonnegative_regret_starting_from_a_terminal_history_and_pi_star_as_a_zero_regret_starting_from_a_terminal_history}. 

\paragraph{Induction hypothesis.} Suppose that $\textsc{Reg}(\textsc{ExploreOne}; \mathcal{H}') \leq c \textsc{Reg}(\textsc{EFA}; \mathcal{H}')$ for any history $\mathcal{H}'$ with at most $k$ unknown entrants. 

\paragraph{Induction Step.} Let $\mathcal{H}$ be a history with $k+1$ unknown entrants.  If $\mathcal{H}$ is terminal, then $\textsc{Reg}(\textsc{ExploreOne};\mathcal{H})  = 0$ by Lemma~\ref{prop:explore_one_at_most_c_orfa} and $\textsc{Reg}(\textsc{EFA};\mathcal{H})  = 0$ by Lemma~\ref{lemma:any_policy_has_a_nonnegative_regret_starting_from_a_terminal_history_and_pi_star_as_a_zero_regret_starting_from_a_terminal_history}. Suppose $\mathcal{H}$ is not terminal. Let $Z$ be the first round at which an unknown entrant is purchased under \textsc{ExploreOne} if it starts from $\mathcal{H}$. Given that an unknown entrant is purchased with probability at least $\frac{\tildew}{(c+1)w_{\text{max}}}$ where $w_{\text{max}} = \max\{w_i(\mathcal{H}): i \in \mathcal{N} \cup \{0\}\}$, then $\expect[Z] \leq \frac{(c+1)w_{\text{max}}}{\tildew}$. As a result, $\prob[Z < \infty] = 1$. Thus, the regret of \textsc{ExploreOne} is upper bounded by 
\begin{align*}
      &\textsc{Reg}(\textsc{ExploreOne};\mathcal{H}) = \textsc{EpochReg}(S^{\star}_{(c-1)}(\mathcal{H}) \cup U_1;\mathcal{H}) \\
    &\hspace{3.6cm}+ \expect\limits_{\mathcal{H}_{Z}, S_{Z}, Y_{Z}, w_{Y_{Z}}} [\textsc{Reg}(\textsc{ExploreOne};\mathcal{H}_{Z} \cup (S_{Z}, Y_{Z}, w_{Y_{Z}}))|\mathcal{H}_1 = \mathcal{H}]\\
    &\leq \textsc{EpochReg}(S^{\star}_{(c-1)}(\mathcal{H}) \cup U_1;\mathcal{H}) + c\cdot \expect\limits_{\mathcal{H}_{Z}, S_{Z}, Y_{Z}, w_{Y_{Z}}} [\textsc{Reg}(\textsc{EFA};\mathcal{H}_{Z} \cup (S_{Z}, Y_{Z}, w_{Y_{Z}}))|\mathcal{H}_1 = \mathcal{H}]\\
    &= \textsc{EpochReg}(S^{\star}_{(c-1)}(\mathcal{H}) \cup U_1;\mathcal{H}) + c\cdot \expect\limits_{w \sim \mathcal{F}} [\textsc{Reg}(\textsc{EFA}, J(\mathcal{H}) \cup \{w\})]\\
    &\leq c \cdot \left(\min\limits_{S \in \mathcal{E}(\mathcal{H})} \textsc{EpochReg}(S;\mathcal{H}) +  \expect\limits_{w \sim \mathcal{F}} [\textsc{Reg}(\textsc{EFA}, J(\mathcal{H}) \cup \{w\})] \right) = c \cdot \textsc{Reg}(\textsc{EFA};\mathcal{H}).
\end{align*}
The first equality uses $Z < \infty$ almost surely and decomposes the regret of \textsc{ExploreOne} into the first epoch regret and the future regret. The first inequality uses $\textsc{Reg}(\textsc{ExploreOne};\mathcal{H}_{Z} \cup (S_{Z}, Y_{Z}, w_{Y_{Z}})) \leq \textsc{Reg}(\textsc{EFA};\mathcal{H}_{Z} \cup (S_{Z}, Y_{Z}, w_{Y_{Z}}))$ by the induction hypothesis as $\mathcal{H}_{Z} \cup (S_{Z}, Y_{Z}, w_{Y_{Z}})$ has $k$ unknown entrants. The second equality uses $\textsc{Reg}(\textsc{EFA};\mathcal{H}_{Z} \cup (S_{Z}, Y_{Z}, w_{Y_{Z}})) = \textsc{Reg}(\textsc{EFA}, J(\mathcal{H}) \cup \{w_Y\})$ by Lemma~\ref{lemma:the_regret_of_pi_star_depends_on_H_only_through_J_H}, noting that the set of known products' attraction parameters is augmented with $w_{Y}$ after the entrant's purchase, and that $w_{Y} \sim \mathcal{F}$. The second inequality uses $\textsc{EpochReg}(S^{\star}_{(c-1)}(\mathcal{H}) \cup U_1;\mathcal{H}) \leq c \min\limits_{S \in \mathcal{E}(\mathcal{H})} \textsc{EpochReg}(S;\mathcal{H})$ by Lemma~\ref{lemma:expoch_regret_exploreone_at_most_c_epoch_regret_orfa_any_history}. The last equality uses Lemma~\ref{lemma:the_regret_of_pi_star_can_be_decomposed_into_the_first_epoch_regret_and_the_future_regret}. 
\end{proof}

\begin{proof}[Proof of Lemma~\ref{lemma:expoch_regret_exploreone_at_most_c_epoch_regret_orfa_any_history}.]
To prove the lemma, we first derive an expression for the epoch regret of an arbitrary assortment $S =  S^{K}  \cup U_{\ell'}$ consisting of $\ell'$ unknown entrants and a subset of known products $S^{K}$. Any known product $i \in S^{K}$ is purchased $\frac{w_i(\mathcal{H})}{\ell' \tildew}$ rounds in expectation yielding a regret of $\opt-1$ per round. The outside option is purchased $\frac{w_0}{\ell' \tildew}$ rounds in expectation yielding a regret of $\opt$ per round. An unknown entrant is purchased once at the end of the epoch yielding regret $\opt-1$. Hence, the epoch regret of $S$ is 
\begin{equation}\label{eq:epoch_reg_expression_general}
    \textsc{EpochReg}(S;\mathcal{H}) = \sum_{i \in S^K} \frac{w_i(\mathcal{H})}{\ell' \tildew} (\opt-1) + \frac{w_0}{\ell' \tildew} \opt + (\opt-1).
\end{equation}
Letting $\ell^{\star} \in \{1, \ldots, c\}$ be the optimal number of unknown entrants explored by \textsc{EFA}, \textsc{EFA} offers assortment $S^{\star}_{(c-\ell^{\star})}(\mathcal{H}) \cup U_{\ell^{\star}}$ during the epoch. As a result, the $c$-scaled epoch regret of \textsc{EFA} is lower bounded by 
\begin{align*}
    &c \cdot \min\limits_{S \in \mathcal{E}(\mathcal{H})} \textsc{EpochReg}(S;\mathcal{H}) = c\cdot \textsc{EpochReg}(S^{\star}_{(c-\ell^{\star})}(\mathcal{H}) \cup U_{\ell^{\star}};\mathcal{H}) \geq \ell^{\star}\cdot \textsc{EpochReg}(S^{\star}_{(c-\ell^{\star})}(\mathcal{H}) \cup U_{\ell^{\star}};\mathcal{H}) \\
    &= \ell^{\star} \cdot \left( \sum_{k=1}^{c-\ell^{\star}} \frac{w_{(k)}(\mathcal{H})}{\ell^{\star} \tildew} (\opt-1) + \frac{w_0}{\ell^{\star} \tildew} \opt + (\opt-1)\right) \\
    &= \sum_{k=1}^{c-\ell^{\star}} \frac{w_{(k)}(\mathcal{H})}{\tildew} (\opt-1) + \frac{w_0}{\tildew} \opt + (\ell^{\star}-1)(\opt-1) + (\opt-1) \\   
    &\geq \sum_{k=1}^{c-1} \frac{w_{(k)}(\mathcal{H})}{\tildew} (\opt-1) + \frac{w_0}{\tildew} \opt + (\opt-1) = \textsc{EpochReg}(S^{\star}_{(c-1)}(\mathcal{H}) \cup U_1;\mathcal{H}).
\end{align*}
The first equality uses $\min\limits_{S \in \mathcal{E}(\mathcal{H})} \textsc{EpochReg}(S;\mathcal{H})  = \textsc{EpochReg}(S^{\star}_{(c-\ell^{\star})}(\mathcal{H}) \cup U_{\ell^{\star}};\mathcal{H})$ as \textsc{EFA} minimizes epoch regret. The first inequality uses that $c \geq \ell^{\star}$ and the the epoch regret is non-negative. The second equality uses \eqref{eq:epoch_reg_expression_general} for $S = S^{\star}_{(c-\ell^{\star})}(\mathcal{H}) \cup U_{\ell^{\star}}$. The second inequality uses that $(\ell^{\star}-1)(\opt-1) \geq \sum_{k= c-\ell^{\star}+1}^{c} \frac{w_{(k)}(\mathcal{H})}{\tildew} (\opt-1)$ as $\frac{w_{(k)}(\mathcal{H})}{\tildew} \geq 1$ (since $w^{\star}_{(c)}(\mathcal{I}_1) \geq \tildew$) and $\opt - 1 < 0$. The last equality uses \eqref{eq:epoch_reg_expression_general} for $S = S^{\star}_{(c-1)}(\mathcal{H}) \cup U_{1}$. This completes the lemma proof. 
\end{proof}

\section{Heterogeneous rewards (Section~\ref{subsec:heterogeneous_rewards})}

\subsection{Heterogeneous rewards and a single entrant (Theorem~\ref{thm:optimal_policy_different_rewards_single_entrant})}\label{appendix_subsec:optimal_policy_different_rewards_single_entrant}
We use the next lemma which is the analogue of Lemma~\ref{lemma:there_exists_an_optimal_policy_which_offers_the_same_assortment_until_the_entrant_is_purchased} in this heterogeneous-rewards model. 
\begin{lemma}\label{lemma:there_exists_an_optimal_policy_which_offers_the_same_assortment_until_the_entrant_is_purchased_different_rewards}
    For any instance with $\textsc{OPT}_1 > \textsc{Rev}^{\star}_1$, there exists an optimal policy which offers the entrant with a fixed set of incumbents until the entrant is purchased and offers the revenue-maximizing assortment in subsequent rounds. 
\end{lemma}
\begin{proof}[Proof of Theorem~\ref{thm:optimal_policy_different_rewards_single_entrant}]
 For a set $S$ of at most $c-1$ incumbents let $\tau(S)$ be the expected number of rounds until the entrant is purchased if the platform uses $\pi(S)$. Letting $N_i$ be the number of rounds option $i \in S \cup \{0\}$ is purchased, the expectation of $N_i$ is
\begin{equation}\label{eq:expected_number_choices_option_i_single_entrant_diff_rewards}
    \expect[N_i] = \prob[\text{i is purchased}|S] \tau(S) = \frac{\prob[\text{i is purchased}|S] }{\prob[\text{entrant is purchased}|S]} = \frac{\frac{w_i}{\sum_{i \in S} w_i + \tildew + w_0}}{\frac{\tildew}{\sum_{i \in S} w_i + \tildew + w_0}} = \frac{w_i}{\tildew}.
\end{equation}
We refer to the set of rounds up to and including the purchase of the entrant as the \emph{exploration period}. 
For any product $i \in S$, $\pi(S)$ incurs a regret of $\textsc{OPT}_1-r_i$ during the exploration period and this happens $\expect[N_i]$ rounds in expectation. Furthermore, $\pi(S)$ incurs regret of $\textsc{OPT}_1$ every time the outside option is purchased and this happens $\expect[N_0]$ rounds in expectation. Moreover, $\pi(S)$ incurs a regret of $\textsc{OPT}_1 - r$ when the unknown entrant is purchased. As a result, the regret of $\pi(S)$ equals 
\begin{align}\label{eq:regret_of_pi_in_terms_of_expected_regrets_from_each_option}
    \textsc{Reg}(\pi(S)) &= \expect \Bigg[\sum_{i \in S} (\textsc{OPT}_1-r_i) N_i + \textsc{OPT}_1 N_0 + (\textsc{OPT}_1-r)\Bigg] \nonumber\\
    &=  \sum_{i \in S} (\textsc{OPT}_1-r_i) \expect[N_i] + \textsc{OPT}_1 \expect[N_0] + (\textsc{OPT}_1-r) \nonumber\\
    &= \sum_{i \in S} (\textsc{OPT}_1-r_i) \frac{w_i}{\tildew} + \textsc{OPT}_1  \frac{w_0}{\tildew}+ (\textsc{OPT}_1-r) \nonumber\\
    &= \frac{w_0}{\tildew}\sum_{i \in S} \textsc{SIR}_i(1) + \textsc{OPT}_1 \frac{w_0}{\tildew}+ (\textsc{OPT}_1-r)
\end{align}
where the first equality uses that $\pi(S)$ incurs zero regret after the exploration period, the third equality uses \eqref{eq:expected_number_choices_option_i_single_entrant_diff_rewards} and the fourth equality uses $(\textsc{OPT}_1-r_i) \frac{w_i}{\tildew} = \frac{w_0}{\tildew} \textsc{SIR}_i(1)$ by the definition of $\textsc{SIR}_i(1)$. By \eqref{eq:regret_of_pi_in_terms_of_expected_regrets_from_each_option} minimizing $\textsc{Reg}(\pi(S))$ is equivalent to minimizing $\sum_{i \in S} \textsc{SIR}_i(1)$ which is achieved exactly when $S$ consists of the $\min(c-1, n^{<0}_1)$ incumbents with the lowest scaled interim regrets, i.e., $S = S_1$.

Thus, $\pi(S_1)$ yields a lower regret than any other policy $\pi(S)$. Given that there exists an optimal policy which is $\pi(S)$ for some $S$ by Lemma~\ref{lemma:there_exists_an_optimal_policy_which_offers_the_same_assortment_until_the_entrant_is_purchased_different_rewards},  $\pi(S_1)$ is optimal. 
\end{proof}

\begin{proof}[Proof of Lemma~\ref{lemma:there_exists_an_optimal_policy_which_offers_the_same_assortment_until_the_entrant_is_purchased_different_rewards}.]
The lemma is a special case of Lemma~\ref{lemma:epoch_optimal_policy_is_optimal_different_rewards} (stated in Appendix~\ref{appendix_subsec:optimal_policy_multiple_entrants} and proven in Appendix~\ref{appendix_subsec:proof_of_lemma:epoch_optimal_policy_is_optimal_different_rewards}) by taking $\pi^{\star}$ to offer a fixed $S' \in \argmin_{S \in \mathcal{E}_t} \textsc{EpochReg}_t(S)$ until the entrant is purchased. 
\end{proof}

\subsection{Optimal policy under different rewards and multiple entrants (Theorem~\ref{thm:optimal_policy_multiple_entrants})}\label{appendix_subsec:optimal_policy_multiple_entrants}
The proof follows a similar structure to the proof of Theorem~\ref{thm:optimality_of_algo}. Let $\mathcal{E}_t = \{S \subseteq \mathcal{N}: |S| \leq c, S \cap (\mathcal{N} \setminus \mathcal{I}_t) \neq \emptyset \}$ be the set of assortments containing at least one entrant at round $t$. For a set $S \in \mathcal{E}_t$, let $\tau(S; \mathcal{H}_t)$ be the expected number of rounds and $r(S; \mathcal{H}_t)$ be the expected total reward up to and including the first round an entrant is purchased if $S$ is offered repeatedly. Let $\textsc{EpochReg}_t(S) = \textsc{OPT}_t \tau(S; \mathcal{H}_t) - r(S; \mathcal{H}_t)$ be the expected regret until the first purchase of an unknown entrant if $S \in \mathcal{E}_t$ is offered repeatedly. We refer to the set of rounds between two consecutive unknown entrant purchases as an \textit{epoch}. 

The following lemma (which is the analogue of Lemma~\ref{lemma:epoch_optimal_policy_is_optimal} and is proven in Appendix~\ref{appendix_subsec:proof_of_lemma:epoch_optimal_policy_is_optimal_different_rewards}) shows that any policy which minimizes regret within each epoch is optimal.

\begin{lemma}\label{lemma:epoch_optimal_policy_is_optimal_different_rewards}
    A policy $\pi^{\star}$ is optimal if for any round $t$ and history $\mathcal{H}_t$, it satisfies 
    \begin{itemize}
        \item $\pi^{\star}(\mathcal{H}_t) \in \argmin_{S \in \mathcal{E}_t} \textsc{EpochReg}_t(S)$ \text{ if } $\textsc{OPT}_t > \textsc{Rev}^{\star}_{t}$
        \item $\pi^{\star}(\mathcal{H}_t) = S^{\star}_{t, \textsc{known}}$ otherwise
    \end{itemize}
\end{lemma}

The next lemma (which is the analogue of Lemma~\ref{lemma:optimal_completion_given_ell_unknown_prods} and is proven in Appendix~\ref{appendix_subsec:lemma4.2_het_rewards_proof}) establishes that conditioned on including $\ell$ unknown entrants, complementing them the $\min(c-\ell, n_t^{<0})$ known entrants with the lowest scaled interim regrets minimizes the regret in the current epoch. 

\begin{lemma}\label{lemma:optimal_completion_given_ell_unknown_prods_different_rewards}
    For any round $t$ and any $\ell \in \{1, \ldots, k_t\}$, the assortment $S$ that minimizes $\textsc{EpochReg}_t(S)$ subject to including exactly $\ell$ unknown products $U_{\ell}$ complements them with the set $V^{\star}_{(\min(c-\ell, n_t^{<0}))}(\mathcal{I}_t)$ of the $\min(c-\ell, n_t^{<0})$ products with lowest scaled interim regrets, i.e., 
$$V^{\star}_{(\min(c-\ell, n_t^{<0}))}(\mathcal{I}_t) \in \argmin\limits_{S' \subseteq \mathcal{I}_t, |S'| \leq c- \ell} \textsc{EpochReg}_t(U_{\ell} \cup S').$$
\end{lemma}

The following lemma (which is the analogue of Lemma~\ref{lemma:num_unknown_ORFA_minimizes_epoch_reg} and is proven in Appendix~\ref{appendix_subsec:lemma4.3_het_rewards_proof}) establishes that the optimal number of unknown entrants to explore is the maximum between the largest $\ell$ such that $\textsc{OPT}_t \geq \beta_t(\ell)$ and $c-n_t^{<0}$.

\begin{lemma}\label{lemma:num_unknown_ORFA.DR_minimizes_epoch_reg_different_rewards}
    For any round $t$, the number of unknown entrants $\ell_{t}$ selected by Algorithm~\ref{alg:optimal_explorer_different_rewards} minimizes $\textsc{EpochReg}_t(U_{\ell} \cup V^{\star}_{(\min(c-\ell, n_t^{<0}))}(\mathcal{I}_t))$ subject to $\ell \in \{1, \ldots, k_t\}$.
\end{lemma}

\begin{proof}[Proof of Theorem~\ref{thm:optimal_policy_multiple_entrants}]
By Lemma~\ref{lemma:epoch_optimal_policy_is_optimal_different_rewards} a policy is optimal if it minimizes within-epoch regret when it is worth it to explore and maximizes revenue otherwise. HEFA maximizes revenue when it is not worth it to explore (line \ref{alg_different_rewards_line:not_exploration_worthy_case}, Algorithm~\ref{alg:optimal_explorer_different_rewards}) and by Lemmas \ref{lemma:optimal_completion_given_ell_unknown_prods_different_rewards} and \ref{lemma:num_unknown_ORFA.DR_minimizes_epoch_reg_different_rewards} it minimizes within-epoch regret when it is worth it to explore (line \ref{alg_different_rewards_line:worth_it_to_explore}, Algorithm~\ref{alg:optimal_explorer_different_rewards}). Thus, HEFA is optimal. 
\end{proof}

\subsection{Epoch-regret-minimizing policies are optimal (Proof of Lemma~\ref{lemma:epoch_optimal_policy_is_optimal_different_rewards})}\label{appendix_subsec:proof_of_lemma:epoch_optimal_policy_is_optimal_different_rewards}
As in the proof of Lemma~\ref{lemma:epoch_optimal_policy_is_optimal}, Lemma~\ref{lemma:any_policy_has_a_nonnegative_regret_starting_from_a_terminal_history_and_pi_star_as_a_zero_regret_starting_from_a_terminal_history}, Lemma~\ref{lemma:the_regret_of_pi_star_depends_on_H_only_through_J_H}, Lemma~\ref{lemma:the_regret_of_pi_star_can_be_decomposed_into_the_first_epoch_regret_and_the_future_regret}, and Lemma~\ref{lemma:epoch_regret_of_any_policy_is_at_least_the_minimum_epoch_regret_over_fixed_assortments} continue to hold in this setting. The proofs of Lemma~\ref{lemma:the_regret_of_pi_star_depends_on_H_only_through_J_H}, Lemma~\ref{lemma:the_regret_of_pi_star_can_be_decomposed_into_the_first_epoch_regret_and_the_future_regret}, and Lemma~\ref{lemma:epoch_regret_of_any_policy_is_at_least_the_minimum_epoch_regret_over_fixed_assortments} remain unchanged. The proof of Lemma~\ref{lemma:any_policy_has_a_nonnegative_regret_starting_from_a_terminal_history_and_pi_star_as_a_zero_regret_starting_from_a_terminal_history} needs a small modifications and it is provided in Appendix~\ref{appendix_subsec:proof_lemma_any_policy_has_nonegative_regret_given_terminal_history_het_rewards}. Instead of Lemma~\ref{lemma:for_any_history_with_k_unknown_entrants_the_regret_of_pi_star_is_lower_bounded_and_upper_bounded_by_quantities_which_depend_on_k}, we show the following lemma (proof in Appendix~\ref{appendix_subsec:lemmaB5_het_rewards_proof}). 

\begin{lemma}\label{lemma:for_any_history_with_k_unknown_entrants_the_regret_of_pi_star_is_lower_bounded_and_upper_bounded_by_quantities_which_depend_on_k_diff_rewards}
    For any history $\mathcal{H}$, $\textsc{Reg}(\pi^{\star};\mathcal{H}) \geq -m \cdot r$ and  $\textsc{Reg}(\pi^{\star};\mathcal{H}) < \infty$.
\end{lemma}

\begin{proof}[Proof of Lemma~\ref{lemma:epoch_optimal_policy_is_optimal_different_rewards}.]
    The proof combines Lemma~\ref{lemma:any_policy_has_a_nonnegative_regret_starting_from_a_terminal_history_and_pi_star_as_a_zero_regret_starting_from_a_terminal_history}, Lemma~\ref{lemma:the_regret_of_pi_star_depends_on_H_only_through_J_H}, Lemma~\ref{lemma:the_regret_of_pi_star_can_be_decomposed_into_the_first_epoch_regret_and_the_future_regret}, Lemma~\ref{lemma:epoch_regret_of_any_policy_is_at_least_the_minimum_epoch_regret_over_fixed_assortments}, and Lemma~\ref{lemma:for_any_history_with_k_unknown_entrants_the_regret_of_pi_star_is_lower_bounded_and_upper_bounded_by_quantities_which_depend_on_k_diff_rewards} and uses the exact same steps as the proof of Lemma~\ref{lemma:epoch_optimal_policy_is_optimal}, where Lemma~\ref{lemma:for_any_history_with_k_unknown_entrants_the_regret_of_pi_star_is_lower_bounded_and_upper_bounded_by_quantities_which_depend_on_k} is replaced with Lemma~\ref{lemma:for_any_history_with_k_unknown_entrants_the_regret_of_pi_star_is_lower_bounded_and_upper_bounded_by_quantities_which_depend_on_k_diff_rewards}. In the case when the expected length of the first epoch is infinite (Case 1, $\expect[Z] = \infty$), we replace ``Lemma~\ref{lemma:for_any_history_with_k_unknown_entrants_the_regret_of_pi_star_is_lower_bounded_and_upper_bounded_by_quantities_which_depend_on_k} implies that $\expect_{w \sim \mathcal{F}}[\textsc{Reg}(\pi^{\star},J(\mathcal{H}) \cup \{w\})]\prob[Z < \infty] \geq -m$'' with ``Lemma~\ref{lemma:for_any_history_with_k_unknown_entrants_the_regret_of_pi_star_is_lower_bounded_and_upper_bounded_by_quantities_which_depend_on_k_diff_rewards} implies that $\expect_{w \sim \mathcal{F}}[\textsc{Reg}(\pi^{\star},J(\mathcal{H}) \cup \{w\})]\prob[Z < \infty] \geq -m \cdot r$''.
\end{proof}

\subsection{Any policy has non-negative regret given a terminal history (Lemma~\ref{lemma:any_policy_has_a_nonnegative_regret_starting_from_a_terminal_history_and_pi_star_as_a_zero_regret_starting_from_a_terminal_history})}\label{appendix_subsec:proof_lemma_any_policy_has_nonegative_regret_given_terminal_history_het_rewards}

\begin{proof}[Updated proof of Lemma~\ref{lemma:any_policy_has_a_nonnegative_regret_starting_from_a_terminal_history_and_pi_star_as_a_zero_regret_starting_from_a_terminal_history}.]
    
Lemma~\ref{lemma:if_H_1_is_terminal_H_t_is_terminal_with_probability_1} and Lemma~\ref{lemma:if_H_is_terminal_the_expected_opt_equals_the_maximum_instantaneous_revenue} continue to hold in this setting and the proof of Lemma~\ref{lemma:any_policy_has_a_nonnegative_regret_starting_from_a_terminal_history_and_pi_star_as_a_zero_regret_starting_from_a_terminal_history} is completed in the same way as in Section~\ref{appendix_subsec:proof_lemma_any_policy_has_nonegative_regret_given_terminal_history}. The proof of Lemma~\ref{lemma:if_H_1_is_terminal_H_t_is_terminal_with_probability_1} remains unchanged. The proof of Lemma~\ref{lemma:if_H_is_terminal_the_expected_opt_equals_the_maximum_instantaneous_revenue} needs a modification in inequality \eqref{eq:lower_bound_expectation_all_at_least_h}, and otherwise remains unchanged. In \eqref{eq:lower_bound_expectation_all_at_least_h} we need to show that $\expect[\textsc{OPT}|\mathcal{E}^{\geq h}] > \textsc{OPT}(\mathcal{H})$ under the assumption that there exists some $S'$ with at least one unknown entrant such that $\textsc{Rev}(S';\mathcal{H}) > \textsc{OPT}(\mathcal{H})$. To show this, let $S^{K}$ be the set of known products in $S'$ given $\mathcal{H}$. The revenue of $S'$ is at least the revenue of $S^{K}$:
\begin{align}\label{ineq:revenue_of_S_prime_is_at_least_revenue_of_S_K}
    \textsc{Rev}(S';\mathcal{H}) &= \frac{r \cdot\ell \tildew + \sum_{i \in S^K} w_i  r_i }{ \ell \tildew +  \sum_{i \in S^K} w_i + 1} >\textsc{OPT}(\mathcal{H}) = \max\limits_{S \subseteq I(\mathcal{H}), |S| \leq c} \textsc{Rev}(S;\mathcal{H})\nonumber\\
    &\geq \textsc{Rev}(S^K;\mathcal{H}) = \frac{ \sum_{i \in S^K} w_i  r_i }{  \sum_{i \in S^K} w_i + 1}.
\end{align}
The first equality uses the expression for the expected revenue of $S'$ and the first inequality uses the assumption that $\textsc{Rev}(S';\mathcal{H}) > \textsc{OPT}(\mathcal{H})$. The second equality uses that $\mathcal{H}$ is terminal and the second inequality uses that $S^K$ is a set of known products. The third equality uses the expression for the expected revenue of $S^K$. As a result of \eqref{ineq:revenue_of_S_prime_is_at_least_revenue_of_S_K} the reward of each entrant is greater than the revenue of $S^K$: 
\begin{equation}\label{ineq:entrant_reward_is_at_least_revenue_of_S_K}
    \frac{r \cdot\ell \tildew + \sum_{i \in S^K} w_i  r_i }{ \ell \tildew +  \sum_{i \in S^K} w_i + 1} > \frac{ \sum_{i \in S^K} w_i  r_i }{  \sum_{i \in S^K} w_i + 1} \iff r > \frac{ \sum_{i \in S^K} w_i  r_i }{  \sum_{i \in S^K} w_i + 1}
\end{equation}
Then, the conditional expected ex-post optimum $\expect[\textsc{OPT}|\mathcal{E}^{\geq h}]$ is greater than $\textsc{OPT}(\mathcal{H})$:
\begin{align*}
    \expect[\textsc{OPT}|\mathcal{E}^{\geq h}] &\geq \expect \Bigg[ \frac{\sum_{i \in S'} w_i r_i}{\sum_{i \in S'} w_i + 1 }\Bigg|\mathcal{E}^{\geq h} \Bigg] = \expect \Bigg[ \frac{r \cdot \sum_{i \in A_{\ell}} w_i  + \sum_{i \in S^K} w_i r_i}{ \sum_{i \in A_{\ell}} w_i  + \sum_{i \in S^K} w_i + 1 }\Bigg|\mathcal{E}^{\geq h} \Bigg] \\
    &\geq \frac{r \cdot \ell \tildew + \sum_{i \in S^K} w_i r_i}{\ell \tildew + \sum_{i \in S^K} w_i + 1} = \textsc{Rev}(S';\mathcal{H}) >\textsc{OPT}(\mathcal{H}).
\end{align*}
The second inequality uses that $\sum_{i \in A_{\ell}} w_i \geq \ell \tildew$ (due to $\mathcal{E}^{\geq h}$) and that the function $f(x) \coloneqq \frac{r \cdot x + \sum_{i \in S^K} w_i r_i }{x + \sum_{i \in S^K} w_i + 1}$ is increasing since its derivative is positive
$$f'(x) = \frac{r \cdot (\sum_{i \in S^K} w_i + 1) - \sum_{i \in S^K} w_i r_i }{(\sum_{i \in S^K} w_i + 1)^2} = (\sum_{i \in S^K} w_i + 1) \frac{r-\frac{ \sum_{i \in S^K} w_i  r_i }{  \sum_{i \in S^K} w_i + 1}}{(\sum_{i \in S^K} w_i + 1)^2} > 0$$
where the last inequality uses \eqref{ineq:entrant_reward_is_at_least_revenue_of_S_K}.
\end{proof}

\subsection{Optimal policy regret is bounded independently of the history (Lemma~\ref{lemma:for_any_history_with_k_unknown_entrants_the_regret_of_pi_star_is_lower_bounded_and_upper_bounded_by_quantities_which_depend_on_k_diff_rewards})}\label{appendix_subsec:lemmaB5_het_rewards_proof}
To prove Lemma~\ref{lemma:for_any_history_with_k_unknown_entrants_the_regret_of_pi_star_is_lower_bounded_and_upper_bounded_by_quantities_which_depend_on_k_diff_rewards} the following lemma (proven at the end of this subsection) lower bounds the epoch regret of any assortment and any history with the entrant's reward.

\begin{lemma}\label{lemma:epoch_reg_lb_het_rewards}
    For any non-terminal history $\mathcal{H}$, and set $S \in \mathcal{E}(\mathcal{H})$, $\textsc{EpochReg}(S;\mathcal{H}) \geq -r$.
\end{lemma}

The next lemma upper bounds the epoch regret with the epoch regret by the maximum reward multiplied by the maximum epoch length. Let $r_{\text{max}} = \max\{\{r_i: i \in \{m+1, \ldots, n\}\} \cup \{r\}\}$ be the maximum reward. 

\begin{lemma}\label{lemma:epoch_reg_ub_het_rewards}
    For any non-terminal history $\mathcal{H}$, and set $S \in \mathcal{E}(\mathcal{H})$, $\textsc{EpochReg}(S;\mathcal{H}) \leq r_{\text{max}}  \left(1 + \frac{\sum_{i \in I(\mathcal{H})} w_i(\mathcal{H}) + 1}{\tildew} \right)$.
\end{lemma}

\begin{proof}[Proof of Lemma~\ref{lemma:for_any_history_with_k_unknown_entrants_the_regret_of_pi_star_is_lower_bounded_and_upper_bounded_by_quantities_which_depend_on_k_diff_rewards}.]
We show by induction on $k$ that for any history $\mathcal{H}$ with $k$ unknown entrants remaining, 
\begin{equation}\label{ineq:upper_and_lower_bound_claim_for_any_k_diff_rewards}
    \textsc{Reg}(\pi^{\star};\mathcal{H}) \geq -k \cdot r \quad \text{ and } \quad \textsc{Reg}(\pi^{\star};\mathcal{H}) < r_{\text{max}} \left( k + \frac{k \Big(\sum_{i \in I(\mathcal{H})} w_i(\mathcal{H}) + k\expect_{w \sim \mathcal{F}}[w]  + 1\Big)}{\tildew}\right)
\end{equation}

\paragraph{Base of induction.} If $k = 0$, $\mathcal{H}$ is terminal and Lemma~\ref{lemma:any_policy_has_a_nonnegative_regret_starting_from_a_terminal_history_and_pi_star_as_a_zero_regret_starting_from_a_terminal_history} implies that $\textsc{Reg}(\pi^{\star};\mathcal{H}) = 0$. 

\paragraph{Induction step.} Suppose \eqref{ineq:upper_and_lower_bound_claim_for_any_k_diff_rewards} holds for $k$ and we will show it for $k+1$. If $\mathcal{H}$ is terminal, then by Lemma~\ref{lemma:any_policy_has_a_nonnegative_regret_starting_from_a_terminal_history_and_pi_star_as_a_zero_regret_starting_from_a_terminal_history}, $\textsc{Reg}(\pi^{\star};\mathcal{H}) = 0$ which satisfies the conditions. Suppose $\mathcal{H}$ is not terminal. Then by Lemma~\ref{lemma:the_regret_of_pi_star_can_be_decomposed_into_the_first_epoch_regret_and_the_future_regret},
\begin{equation}\label{eq:regret_decomposition_pi_star_different_rewards}
      \textsc{Reg}(\pi^{\star};\mathcal{H}) = \min\limits_{S \subseteq \mathcal{E}(\mathcal{H})} \textsc{EpochReg}(S;\mathcal{H}) + \underbrace{\expect[\textsc{Reg}(\pi^{\star};\mathcal{H}_{Z^{\mathcal{H}}_{\star}} \cup (S_{Z^{\mathcal{H}}_{\star}}, Y_{Z^{\mathcal{H}}_{\star}}, w_{Z^{\mathcal{H}}_{\star}}))|\mathcal{H}_1 = \mathcal{H}]}_{(\star)}.
\end{equation}
Given that $(\star) \geq -k \cdot r$ by the induction hypothesis and $\min\limits_{S \subseteq \mathcal{E}(\mathcal{H})} \textsc{EpochReg}(S;\mathcal{H})  \geq -r$ by Lemma~\ref{lemma:epoch_reg_lb_het_rewards}, the lower bound of \eqref{ineq:upper_and_lower_bound_claim_for_any_k_diff_rewards} follows. Letting $\mu \coloneq \expect_{w \sim \mathcal{F}}[w]$ and 
$W(\mathcal{H}) \coloneq \sum_{i \in I(\mathcal{H})} w_i(\mathcal{H})$ as shorthands, the regret of $\pi^{\star}$ given $\mathcal{H}$ is upper bounded by 
\begin{align*}
     \textsc{Reg}(\pi^{\star};\mathcal{H}) &= \min\limits_{S \subseteq \mathcal{E}(\mathcal{H})} \textsc{EpochReg}(S;\mathcal{H}) + \expect[\textsc{Reg}(\pi^{\star};\mathcal{H}_{Z^{\mathcal{H}}_{\star}} \cup (S_{Z^{\mathcal{H}}_{\star}}, Y_{Z^{\mathcal{H}}_{\star}}, w_{Z^{\mathcal{H}}_{\star}}))|\mathcal{H}_1 = \mathcal{H}] \\
     &\leq  \min\limits_{S \subseteq \mathcal{E}(\mathcal{H})} \textsc{EpochReg}(S;\mathcal{H}) + r_{\text{max}}\expect_{w_{Z^{\mathcal{H}}_{\star}} \sim \mathcal{F}} \Bigg[ k + \frac{k \Big(W(\mathcal{H}) + w_{Z^{\mathcal{H}}_{\star}} + k \mu  +1 \Big)}{\tildew} \Bigg]\\
     &\leq  r_{\text{max}}  \left(1 + \frac{W(\mathcal{H}) + 1}{\tildew} \right)+ r_{\text{max}} \expect_{w_{Z^{\mathcal{H}}_{\star}} \sim \mathcal{F}} \Bigg[ k + \frac{k \Big(W(\mathcal{H}) + w_{Z^{\mathcal{H}}_{\star}} + k \mu +1\Big)}{\tildew} \Bigg]\\
     &\leq r_{\text{max}} \left( k+1 + \frac{(k+1) \Big( W(\mathcal{H}) +(k+1)\mu +  1\Big)}{\tildew} \right).
\end{align*}
The first inequality uses the induction hypothesis; the second inequality uses Lemma~\ref{lemma:epoch_reg_ub_het_rewards}; the third equality uses $\expect_{w_{Z^{\mathcal{H}}_{\star}} \sim \mathcal{F}} [w_{Z^{\mathcal{H}}_{\star}}]  = \mu\geq 0$. This yields the upper bound in \eqref{ineq:upper_and_lower_bound_claim_for_any_k_diff_rewards} and finishes the induction and the proof.
\end{proof}

\begin{proof}[Proof of Lemma~\ref{lemma:epoch_reg_lb_het_rewards}.]
The epoch regret of $S$ is lower bounded by
\begin{align*}\label{ineq:epoch_reg_lower_bound_diff_rewards}
    &\textsc{EpochReg}(S;\mathcal{H}) = \Bigg(\textsc{OPT}(\mathcal{H}) - \frac{\sum_{i \in S^{\textsc{k}}} r_i w_i(\mathcal{H}) + r \sum_{i \in S^{\textsc{u}}} w_i(\mathcal{H})}{\sum_{i \in S^{\textsc{k}}} w_i(\mathcal{H}) + \sum_{i \in S^{\textsc{u}}} w_i(\mathcal{H}) + 1} \Bigg) \frac{\sum_{i \in S^{\textsc{k}}} w_i(\mathcal{H}) + \sum_{i \in S^{\textsc{u}}} w_i(\mathcal{H}) + 1}{\sum_{i \in S^{\textsc{u}}} w_i(\mathcal{H})} \nonumber\\
    &\geq \Bigg(\frac{\sum_{i \in S^{\textsc{k}}} r_i w_i(\mathcal{H}) }{\sum_{i \in S^{\textsc{k}}}  w_i(\mathcal{H}) + 1} - \frac{\sum_{i \in S^{\textsc{k}}} r_i w_i(\mathcal{H}) + r\sum_{i \in S^{\textsc{u}}} w_i(\mathcal{H})}{\sum_{i \in S^{\textsc{k}}} w_i(\mathcal{H}) +  \sum_{i \in S^{\textsc{u}}} w_i(\mathcal{H}) + 1} \Bigg) \frac{\sum_{i \in S^{\textsc{k}}} w_i(\mathcal{H}) + \sum_{i \in S^{\textsc{u}}} w_i(\mathcal{H}) + 1}{\sum_{i \in S^{\textsc{u}}} w_i(\mathcal{H})} \nonumber \\
    &= \frac{\Big(\sum_{i \in S^{\textsc{k}}} r_i w_i(\mathcal{H})  \Big)\Big(\sum_{i \in S^{\textsc{k}}} w_i(\mathcal{H}) + \sum_{i \in S^{\textsc{u}}} w_i(\mathcal{H}) + 1 \Big)}{\Big(\sum_{i \in S^{\textsc{u}}} w_i(\mathcal{H}) \Big) \Big(\sum_{i \in S^{\textsc{k}}}  w_i(\mathcal{H}) + 1 \Big) } - \frac{\sum_{i \in S^{\textsc{k}}} r_i w_i(\mathcal{H})}{\sum_{i \in S^{\textsc{u}}} w_i(\mathcal{H})} -r \nonumber\\
    &= \frac{\sum_{i \in S^{\textsc{k}}} r_i w_i(\mathcal{H})  }{\sum_{i \in S^{\textsc{u}}} w_i(\mathcal{H}) } \Bigg(1 + \frac{\sum_{i \in S^{\textsc{u}}} w_i(\mathcal{H})}{\sum_{i \in S^{\textsc{k}}} w_i(\mathcal{H}) + 1 } \Bigg) - \frac{\sum_{i \in S^{\textsc{k}}} r_i w_i(\mathcal{H})}{\sum_{i \in S^{\textsc{u}}} w_i(\mathcal{H})} -r  \geq -r.
\end{align*}
The first equality expresses the epoch regret as the regret per round of $S$ multiplied by the expected number of rounds in an epoch. The first inequality uses $\textsc{OPT}(\mathcal{H}) \geq \frac{\sum_{i \in S^{\textsc{k}}} r_i w_i(\mathcal{H}) }{\sum_{i \in S^{\textsc{k}}}  w_i(\mathcal{H}) + 1}$. The second equality multiplies each term in the brackets by $\frac{\sum_{i \in S^{\textsc{k}}} w_i(\mathcal{H}) + \sum_{i \in S^{\textsc{u}}} w_i(\mathcal{H}) + 1}{\sum_{i \in S^{\textsc{u}}} w_i(\mathcal{H})}$, cancels the attraction parameter sum $\sum_{i \in S^{\textsc{k}}} w_i(\mathcal{H}) + \sum_{i \in S^{\textsc{u}}} w_i(\mathcal{H}) + 1$ from the denominator of the second term and separates $r$. 
\end{proof}

\begin{proof}[Proof of Lemma~\ref{lemma:epoch_reg_ub_het_rewards}.]
The epoch regret is upper bounded by
\begin{align*}
    \textsc{EpochReg}(S;\mathcal{H}) &= \Bigg(\textsc{OPT}(\mathcal{H}) - \frac{\sum_{i \in S^{\textsc{k}}} r_i w_i(\mathcal{H}) + r \sum_{i \in S^{\textsc{u}}} w_i(\mathcal{H})}{\sum_{i \in S^{\textsc{k}}} w_i(\mathcal{H}) + \sum_{i \in S^{\textsc{u}}} w_i(\mathcal{H}) + 1} \Bigg) \frac{\sum_{i \in S^{\textsc{k}}} w_i(\mathcal{H}) + \sum_{i \in S^{\textsc{u}}} w_i(\mathcal{H}) + 1}{\sum_{i \in S^{\textsc{u}}} w_i(\mathcal{H})} \nonumber\\
    &\leq r_{\text{max}} \left(\frac{\sum_{i \in S^{\textsc{k}}} w_i(\mathcal{H}) + \sum_{i \in S^{\textsc{u}}} w_i(\mathcal{H}) + 1}{\sum_{i \in S^{\textsc{u}}} w_i(\mathcal{H})} \right)
    \leq  r_{\text{max}} \left(1 + \frac{\sum_{i \in I(\mathcal{H})} w_i(\mathcal{H}) + 1}{\tildew} \right).
\end{align*}
The first equality expresses the epoch regret as the regret per round of $S$ multiplied by the expected number of rounds in an epoch. The first inequality upper bounds the regret per round by $\textsc{OPT}(\mathcal{H}) - \frac{\sum_{i \in S^{\textsc{k}}} r_i w_i(\mathcal{H}) + r \sum_{i \in S^{\textsc{u}}} w_i(\mathcal{H})}{\sum_{i \in S^{\textsc{k}}} w_i(\mathcal{H}) + \sum_{i \in S^{\textsc{u}}} w_i(\mathcal{H}) + 1} \leq r_{\text{max}}$ as $\textsc{OPT}(\mathcal{H}) \leq r_{\text{max}}$ and $\frac{\sum_{i \in S^{\textsc{k}}} r_i w_i(\mathcal{H}) + r \sum_{i \in S^{\textsc{u}}} w_i(\mathcal{H})}{\sum_{i \in S^{\textsc{k}}} w_i(\mathcal{H}) + \sum_{i \in S^{\textsc{u}}} w_i(\mathcal{H}) + 1} \geq 0$. The second inequality uses that $\sum_{i \in S^{\textsc{k}}} w_i(\mathcal{H}) \leq \sum_{i \in I(\mathcal{H})} w_i(\mathcal{H})$ (as $S^{\textsc{K}} \subseteq I(\mathcal{H})$) and $\sum_{i \in S^{\textsc{u}}} w_i(\mathcal{H}) \geq \tildew$.
\end{proof}
 
\subsection{Optimal completion to set of unknown entrants (Lemma~\ref{lemma:optimal_completion_given_ell_unknown_prods_different_rewards})}\label{appendix_subsec:lemma4.2_het_rewards_proof}
To prove Lemma \ref{lemma:optimal_completion_given_ell_unknown_prods_different_rewards}, the following lemma provides an expression for the epoch regret of an assortment offering $\ell$ unknown entrants. Recall that $m_t$ is the number of unknown entrants left and that $k_t = \min(c, m_t)$.

\begin{lemma}\label{lemma:expression_for_the_epoch_regret_of_an_assortment_containg_ell_unknown_entrants}
    For any $\ell \in \{1, \ldots, k_t\}$ and any assortment $S'$ of at most $c-\ell$ known products,  $$\textsc{EpochReg}_t(U_\ell \cup S') = \frac{w_0}{\ell \tildew}\sum_{i \in S'} \textsc{SIR}_i(t) + \textsc{OPT}_t \frac{w_0}{\ell \tildew}+ (\textsc{OPT}_t-r)$$
\end{lemma}
\begin{proof}[Proof of Lemma~\ref{lemma:optimal_completion_given_ell_unknown_prods_different_rewards}]
By Lemma~\ref{lemma:expression_for_the_epoch_regret_of_an_assortment_containg_ell_unknown_entrants}, minimizing $\textsc{EpochReg}_t(U_\ell \cup S')$ is equivalent to minimizing the term $\sum_{i \in S'} \textsc{SIR}_i(t)$ which is achieved exactly when $S'$ consists of the $\min(c-\ell, n_t^{<0})$ known products with the lowest scaled interim regrets, i.e., $S = V^{\star}_{(\min(c-\ell, n_t^{<0}))}$. 
\end{proof}

\begin{proof}[Proof of Lemma~\ref{lemma:expression_for_the_epoch_regret_of_an_assortment_containg_ell_unknown_entrants}]
Letting $N_i$ be the number of rounds option $i \in S' \cup \{0\}$ is chosen during the current epoch if $U_{\ell} \cup S'$ is offered, 
\begin{equation}\label{eq:expected_number_choices_product_i_during_epoch_diff_rewards}
    \expect[N_i] = \prob[\text{$i$ is purchased}|U_{\ell} \cup S'] \tau(U_{\ell} \cup S'; \mathcal{I}_t) = \frac{\frac{w_i}{\sum_{j \in S'} w_j + \ell \tildew + w_0}}{\frac{\ell \tildew}{\sum_{j \in S'} w_j + \ell \tildew + w_0}} = \frac{w_i}{\ell \tildew}.
\end{equation}
For any product $i \in S$, $\pi(S)$ incurs a regret of $\textsc{OPT}_t-r_i$ during the epoch and this happens $\expect[N_i]$ rounds in expectation. Furthermore, $\pi(S)$ incurs regret of $\textsc{OPT}_t$ every time the outside option is purchased and this happens $\expect[N_0]$ rounds in expectation. Moreover, $\pi(S)$ incurs a regret of $\textsc{OPT}_t - r$ when the unknown entrant is purchased. Thus the epoch regret equals
\begin{align*}
    \textsc{EpochReg}_t(U_\ell \cup S') &= \expect \Bigg[\sum_{i \in S'} (\textsc{OPT}_t-r_i) N_i + \textsc{OPT}_1 + (\textsc{OPT}_t-r) \Bigg] \nonumber\\
    &=  \sum_{i \in S'} (\textsc{OPT}_t-r_i) \expect[N_i] + \textsc{OPT}_t \expect[N_0] + (\textsc{OPT}_t-r) \nonumber\\
    &= \sum_{i \in S'} (\textsc{OPT}_t-r_i) \frac{w_i}{\ell \tildew} + \textsc{OPT}_t  \frac{w_0}{\ell \tildew}+ (\textsc{OPT}_t-r) \nonumber\\
    &= \frac{w_0}{\ell \tildew}\sum_{i \in S'} \textsc{SIR}_i(t) + \textsc{OPT}_t \frac{w_0}{\ell \tildew}+ (\textsc{OPT}_t-r)
\end{align*}
where the third equality uses \eqref{eq:expected_number_choices_product_i_during_epoch_diff_rewards} and the fourth equality uses $(\textsc{OPT}_t-r_i) \frac{w_i}{\ell \tildew} = \frac{w_0}{\ell \tildew} \textsc{SIR}_i(t)$.
\end{proof}

\subsection{Optimal number of unknown entrants (Lemma~\ref{lemma:num_unknown_ORFA.DR_minimizes_epoch_reg_different_rewards})}\label{appendix_subsec:lemma4.3_het_rewards_proof}

To prove Lemma~\ref{lemma:num_unknown_ORFA.DR_minimizes_epoch_reg_different_rewards}, let $\textsc{EpochReg}^{\star}(\ell) = \textsc{EpochReg}_t( U_{\ell} \cup V^{\star}_{(\min(c-\ell, n_t^{<0}))})$ for ease of notation. The following lemma compares $\textsc{EpochReg}^{\star}(\ell)$ and $\textsc{EpochReg}^{\star}(\ell+1)$ when $\ell \geq c-n_t^{<0}$. This is the analogue of \eqref{ineq:ell_plus_one_smaller_epoch_regret_thanell} and Lemma~\ref{lemma:expression_ratio_reward_loss_time_gain}. 

\begin{lemma}\label{lemma:comaparison_ell_ell_plus_one_larger_than_c_minus_n}
    For any $\ell \geq c-n_t^{<0}$, $\textsc{EpochReg}^{\star}(\ell) \geq \textsc{EpochReg}^{\star}(\ell+1)$ if and only if $\textsc{OPT}_t \geq \beta_t(\ell + 1)$
\end{lemma}

Given that known products with positive scaled interim regret are not included, exploring less than $c-n_t^{<0}$ unknown entrants is never optimal. This is formalized in the next lemma. 

\begin{lemma}\label{lemma:c_minus_n_unknown_entrants_are_better_than_any_smaller_number}
     For any round $t$ with $\textsc{OPT}_t > \textsc{Rev}^{\star}_{t}$ and number of unknown entrants $\ell < c-n_t^{<0}$, $\textsc{EpochReg}^{\star}(\ell) > \textsc{EpochReg}^{\star}(c-n_t^{<0}) $.
\end{lemma}
\begin{proof}[Proof of Lemma~\ref{lemma:num_unknown_ORFA.DR_minimizes_epoch_reg_different_rewards}.]
To prove the lemma it suffices to show that 
\begin{equation*}
    \textsc{EpochReg}^{\star}(\ell) \geq  \textsc{EpochReg}^{\star}(\ell_t^{\star}) \text{ for all $\ell \in \{1, \ldots, k_t\}$}
\end{equation*}
Let $\ell \in \{1, \ldots, k_t\}$. We consider cases based on whether $\ell \geq \ell_t^{\star}$, $\ell_t^{\star} > \ell \geq c-n_t^{<0}$, or $\ell < c-n_t^{<0}$. 

\paragraph{Case 1: $\ell \geq \ell_t^{\star}$.} For any $\ell' \geq \ell_t^{\star}$, it holds that $\ell' \geq \ell_t$ as $\ell_t^{\star} =\max(\ell_t, c-n_t^{<0})$. As a result, $\textsc{OPT}_t < \beta(\ell'+1)$ and Lemma~\ref{lemma:comaparison_ell_ell_plus_one_larger_than_c_minus_n} yields $\textsc{EpochReg}^{\star}(\ell'+1) \geq \textsc{EpochReg}^{\star}(\ell')$. Applying this for $\ell' \in \{\ell-1, \ldots, \ell_t^{\star}\}$ yields $\textsc{EpochReg}^{\star}(\ell) \geq  \textsc{EpochReg}^{\star}(\ell_t^{\star})$.

\paragraph{Case 2: $\ell_t^{\star} > \ell \geq c-n_t^{<0}$.} Given that $\ell_t^{\star} =\max(\ell_t, c-n_t^{<0})$, this case is only possible when $\ell_t^{\star} =\ell_t$. For any $\ell'$ such that $\ell_t > \ell'$ and $\ell' \geq c-n_t^{<0}$, $\textsc{OPT}_t \geq \beta(\ell'+1)$, and Lemma~\ref{lemma:comaparison_ell_ell_plus_one_larger_than_c_minus_n} yields $\textsc{EpochReg}^{\star}(\ell'+1) \leq \textsc{EpochReg}^{\star}(\ell')$. Applying this for $\ell' \in \{\ell, \ldots, \ell_t-1\}$ yields $$\textsc{EpochReg}^{\star}(\ell) \geq  \textsc{EpochReg}^{\star}(\ell_t) =  \textsc{EpochReg}^{\star}(\ell_t^{\star}).$$

\paragraph{Case 3: $\ell < c-n_t^{<0}$.} Then 
$\textsc{EpochReg}^{\star}(\ell) >  \textsc{EpochReg}^{\star}(c-n_t^{<0})\geq  \textsc{EpochReg}^{\star}(\ell_t^{\star})$
where the first inequality uses Lemma~\ref{lemma:c_minus_n_unknown_entrants_are_better_than_any_smaller_number} and the second inequality uses the previous case.
\end{proof}

\begin{proof}[Proof of Lemma~\ref{lemma:comaparison_ell_ell_plus_one_larger_than_c_minus_n}.]
For any $\ell' \geq  c-n_t^{<0}$, the epoch regret is expressed as 
\begin{align}\label{eq:epoch_reg_ell_prime_analytical_expressions}
    \textsc{EpochReg}^{\star}(\ell') &= \textsc{EpochReg}_t( U_{\ell'} \cup V^{\star}_{(\min(c-\ell', n_t^{<0}))}) =\textsc{EpochReg}_t( U_{\ell'} \cup V^{\star}_{(c-\ell')}) \nonumber\\
    &= \frac{\sum_{i =1}^{c-\ell'} w_0 \textsc{SIR}_{(i)}^{\star}(t) + w_0 \textsc{OPT}_t}{\ell' \tildew} + \textsc{OPT}_t -r.
\end{align}
The second equality uses $\min(c-\ell', n_t^{<0}) = c-\ell'$  (as $c-\ell' \leq n_t^{<0}$). The third equality uses  Lemma~\ref{lemma:expression_for_the_epoch_regret_of_an_assortment_containg_ell_unknown_entrants}. Using \eqref{eq:epoch_reg_ell_prime_analytical_expressions} for $\ell' \in \{\ell, \ell + 1\}$, the second term cancels; $\textsc{EpochReg}^{\star}(\ell) \geq \textsc{EpochReg}^{\star}(\ell+1)$ is equivalent to
\begin{align*}
    \frac{\sum_{i =1}^{c-\ell} w_0 \textsc{SIR}_{(i)}^{\star}(t) + w_0 \textsc{OPT}_t}{\ell \tildew} &\geq \frac{\sum_{i =1}^{c-\ell-1} w_0 \textsc{SIR}_{(i)}^{\star}(t) + w_0 \textsc{OPT}_t}{(\ell+1) \tildew}\\
    \iff \frac{\sum_{i =1}^{c-\ell}  \textsc{SIR}_{(i)}^{\star}(t) + \textsc{OPT}_t}{\ell } &\geq \frac{\sum_{i =1}^{c-\ell-1} \textsc{SIR}_{(i)}^{\star}(t) +\textsc{OPT}_t}{(\ell+1) }\\
    \iff (\ell + 1) \sum_{i =1}^{c-\ell} \textsc{SIR}_{(i)}^{\star}(t)  + (\ell + 1) \textsc{OPT}_t &\geq \ell \sum_{i =1}^{c-\ell-1}  \textsc{SIR}_{(i)}^{\star}(t)  + \ell\textsc{OPT}_t \\
    \iff \textsc{OPT}_t &\geq \sum_{i=1}^{c-\ell-1} -\textsc{SIR}_{(i)}^{\star}(t)  - (\ell + 1)\textsc{SIR}_{(c-\ell)}^{\star}(t) = \beta_t(\ell+1)
\end{align*}
\end{proof}

\begin{proof}[Proof of Lemma~\ref{lemma:c_minus_n_unknown_entrants_are_better_than_any_smaller_number}.]
For any $\ell' \leq c-n_t^{<0}$, the epoch regret is expressed as 
\begin{align}\label{eq:epoch_reg_ell_prime_less_than_c_minus_nt_idx_analytical_expression}
    \textsc{EpochReg}^{\star}(\ell') &= \textsc{EpochReg}_t( U_{\ell'} \cup V^{\star}_{(\min(c-\ell', n_t^{<0}))}) = \textsc{EpochReg}_t( U_{\ell'} \cup V^{\star}_{(n_t^{<0})}) \nonumber\\
    &= \frac{\sum_{i=1}^{n_t^{<0}} w_0 \textsc{SIR}_{(i)}^{\star}(t) + w_0 \textsc{OPT}_t}{\ell' \tildew}  + \textsc{OPT}_t-r.
\end{align}
The second equality uses $c-\ell' \geq n_t^{<0}$ (as $\ell' \leq c-n_t^{<0}$); the third equality uses Lemma~\ref{lemma:expression_for_the_epoch_regret_of_an_assortment_containg_ell_unknown_entrants}.
Applying \eqref{eq:epoch_reg_ell_prime_less_than_c_minus_nt_idx_analytical_expression} for $\ell' \in \{c-n_t^{<0}, \ell\}$, the second term cancels and $\textsc{EpochReg}^{\star}(c-n_t^{<0}) < \textsc{EpochReg}^{\star}(\ell)$ is equivalent to
\begin{align}\label{equivalent_condition_epoch_reg_ell_larger_than_epoch_reg_c_minus_n}
\frac{\sum_{i=1}^{n_t^{<0}} w_0 \textsc{SIR}_{(i)}^{\star}(t) + w_0 \textsc{OPT}_t}{(c-n_t^{<0}) \tildew} &< \frac{\sum_{i=1}^{n_t^{<0}} w_0 \textsc{SIR}_{(i)}^{\star}(t) + w_0 \textsc{OPT}_t}{\ell \tildew} \nonumber\\
    \iff \frac{\sum_{i=1}^{n_t^{<0}}  \textsc{SIR}_{(i)}^{\star}(t) + \textsc{OPT}_t}{(c-n_t^{<0}) } &< \frac{\sum_{i=1}^{n_t^{<0}} \textsc{SIR}_{(i)}^{\star}(t) + \textsc{OPT}_t}{\ell } \nonumber\\
    \iff \sum_{i=1}^{n_t^{<0}} \textsc{SIR}_{(i)}^{\star}(t)+ \textsc{OPT}_t > 0 &\iff \textsc{OPT}_t > \sum_{i=1}^{n_t^{<0}} -\textsc{SIR}_{(i)}^{\star}(t) \nonumber\\
    \iff \textsc{OPT}_t > \sum_{i=1}^{n_t^{<0}} (r_{(i)} -\textsc{OPT}_t) \frac{w_{(i)}}{w_0} &\iff \textsc{OPT}_t > \frac{\sum_{i=1}^{n_t^{<0}} r_{(i)} w_{(i)}}{\sum_{i=1}^{n_t^{<0}} w_{(i)} + w_0}
\end{align}
The third equivalence uses that $\frac{x}{\ell} < \frac{x}{c - n_t^{<0}}$ if and only if $x> 0$ (as $\ell < c - n_t^{<0}$) for $x = \sum_{i=1}^{n_t^{<0}} \textsc{SIR}_{(i)}^{\star}(t) + \textsc{OPT}_t$. 

We conclude the proof by showing \eqref{equivalent_condition_epoch_reg_ell_larger_than_epoch_reg_c_minus_n}: 
$\textsc{OPT}_t > \textsc{Rev}^{\star}_{t}$ and $\textsc{Rev}^{\star}_{t} $ is the maximum revenue from known products, thus greater than the revenue from known products with negative scaled interim regrets.
\end{proof}

\subsection{Connection to homogeneous-rewards case (Section~\ref{subsec:heterogeneous_rewards})}\label{appendix_subsec:beta_connection_to_rewards_one}

\begin{proposition}\label{prop:condition_het_rewards_reduces_to_condition_hom_rewards}
    For any instance where $r_i = 1$, $\textsc{OPT}_t \geq \beta_t(\ell)$ if and only if $\textsc{OPT}_t \geq \alpha_t(\ell)$
\end{proposition}
\begin{proof}[Proof of Proposition~\ref{prop:condition_het_rewards_reduces_to_condition_hom_rewards}.]
        \begin{align*}
        \textsc{OPT}_t &\geq \beta_t(\ell)  \iff \textsc{OPT}_t \geq  -\sum_{i=1}^{c-\ell} \textsc{SIR}_{(i)}(t) - \ell \cdot \textsc{SIR}_{(c-\ell+1)}(t)\\
        \iff \textsc{OPT}_t &\geq  \sum_{i=1}^{c-\ell} (1 - \textsc{OPT}_t)\frac{w^{\star}_{(i)}(\mathcal{I}_t)}{w_0} - \ell \cdot (1 - \textsc{OPT}_t)\frac{w^{\star}_{(c-\ell+1)}(\mathcal{I}_t)}{w_0}\\
        \iff \textsc{OPT}_t &\Bigg(\sum_{i=1}^{c-\ell} w^{\star}_{(i)}(\mathcal{I}_t) + \ell w^{\star}_{(c-\ell+1)}(\mathcal{I}_t) + w_0\Bigg) \geq \sum_{i=1}^{c-\ell} w^{\star}_{(i)}(\mathcal{I}_t) + \ell w^{\star}_{(c-\ell+1)}(\mathcal{I}_t) \\
        \iff \textsc{OPT}_t & \ge \frac{\sum_{i=1}^{c-\ell} w^{\star}_{(i)}(\mathcal{I}_t) + \ell w^{\star}_{(c-\ell+1)}(\mathcal{I}_t)}{\sum_{i=1}^{c-\ell} w^{\star}_{(i)}(\mathcal{I}_t) + \ell w^{\star}_{(c-\ell+1)}(\mathcal{I}_t) + w_0} \iff \textsc{OPT}_t \geq \alpha_t(\ell)
    \end{align*}
\end{proof}

\section{Heterogeneous priors (Section~\ref{sec:heterogeneous_priors})}

\subsection{Characterization of the optimal policy with two entrants (Theorem~\ref{thm:characterization_of_the_optimal_policy_with_two_unknown_entrants_with_heterogeneous_priors})}\label{appendix_subsec:characterization_of_the_optimal_policy_with_two_entrants}

We prove the following theorem which is an extension of Theorem~\ref{thm:characterization_of_the_optimal_policy_with_two_unknown_entrants_with_heterogeneous_priors}. This characterizes the exact regret differences between each two of the three exploration policies $\pi^{1}$, $\pi^{2}$, and $\pi^{1,2}$ and the theorem follows directly subsequently. We include this lemma as it helps make the proofs of Propositions~\ref{corollary:stochastic_dominance} and~\ref{prop:different_priors_bernoulli} more concise. 

\begin{lemma}
\label{lemma:characterization_of_the_optimal_policy_with_two_unknown_entrants_with_heterogeneous_priors_regret_differences}
    For any instance in the class $\mathcal{C}$, the regret differences are
$$\textsc{Reg}(\pi^{1,2})-\textsc{Reg}(\pi^{1}) = \expect_{\substack{w_1 \sim \mathcal{F}_1 \\ w_2 \sim \mathcal{F}_2}} \Bigg[\Big(\frac{\tildew_2}{\tildew_1 + \tildew_2} \Big(\frac{\max(w_2, w_3)}{\tildew_1}-\frac{\max(w_1, w_3)}{\tildew_2} \Big) - \frac{w_3}{\tildew_1}  \Big)(\textsc{OPT}-1) - \frac{1}{\tildew_1+ \tildew_2} \textsc{OPT} \Bigg],$$
$$\textsc{Reg}(\pi^{1,2}) - \textsc{Reg}(\pi^{2}) = \expect_{\substack{w_1 \sim \mathcal{F}_1 \\ w_2 \sim \mathcal{F}_2}} \Bigg[\Big( \frac{\tildew_1}{\tildew_1 + \tildew_2} \Big(\frac{\max(w_1, w_3)}{\tildew_2} -\frac{\max(w_2, w_3)}{\tildew_1}\Big)-\frac{w_3}{\tildew_2} \Big)(\textsc{OPT}-1) - \frac{1}{\tildew_1 + \tildew_2} \textsc{OPT} \Bigg],$$
$$
\textsc{Reg}(\pi^{2}) - \textsc{Reg}(\pi^{1}) 
= \expect_{\substack{w_1 \sim \mathcal{F}_1 \\ w_2 \sim \mathcal{F}_2}} 
\left[ 
\left( \frac{w_3}{\tildew_1} + \frac{\max(w_2,w_3)}{\tildew_2} \right)(\textsc{OPT}-1) 
- \left( \frac{w_3}{\tildew_2} + \frac{\max(w_1,w_3)}{\tildew_1} \right)(\textsc{OPT}-1) 
\right].
$$
\end{lemma}

\begin{proof}[Proof of Theorem~\ref{thm:characterization_of_the_optimal_policy_with_two_unknown_entrants_with_heterogeneous_priors}.]
The proof follows directly by Lemma~\ref{lemma:characterization_of_the_optimal_policy_with_two_unknown_entrants_with_heterogeneous_priors_regret_differences}. 
\end{proof}
To prove Lemma~\ref{lemma:characterization_of_the_optimal_policy_with_two_unknown_entrants_with_heterogeneous_priors_regret_differences} the next lemmas (proven in Section~\ref{appendix_subsec_proofs_lemmas_regret_pi1_pi2_pi12}) characterize the regret of $\pi^{1}$, $\pi^{2}$, and $\pi^{1,2}$.

\begin{lemma}\label{lemma:regret_pi1} The regret of $\pi^{1}$, $\textsc{Reg}(\pi^{1}) $, equals
   $$\expect_{\substack{w_1 \sim \mathcal{F}_1 \\ w_2 \sim \mathcal{F}_2}}\left[\left(\frac{w_3}{\tildew_1}+\frac{\max(w_1,w_3)}{\tildew_2} \right)(\textsc{OPT}-1) + \left(\frac{1}{\tildew_1} + \frac{1}{\tildew_2} \right) \textsc{OPT}  + 2(\textsc{OPT}-1) \right].$$
\end{lemma}

\begin{lemma}\label{lemma:regret_pi2} 
    The regret of $\pi^{2}$, $\textsc{Reg}(\pi^{2})$, equals
     \begin{equation*}
      \expect_{\substack{w_1 \sim \mathcal{F}_1 \\ w_2 \sim \mathcal{F}_2}}\left[\left(\frac{w_3}{\tildew_2}+\frac{\max(w_2,w_3)}{\tildew_1} \right)(\textsc{OPT}-1) + \left(\frac{1}{\tildew_1} + \frac{1}{\tildew_2} \right) \textsc{OPT}  + 2(\textsc{OPT}-1) \right].
     \end{equation*}
\end{lemma}

\begin{lemma}\label{lemma:regret_pi12}
The regret of $\pi^{1,2}$, $\textsc{Reg}(\pi^{1,2})$ equals
     \begin{align*}
         \expect_{\substack{w_1 \sim \mathcal{F}_1 \\ w_2 \sim \mathcal{F}_2}}\Bigg[& \frac{1}{\tildew_1 + \tildew_2}\left( \frac{\tildew_1 \max(w_1,w_3)}{\tildew_2} + \frac{\tildew_2\max(w_2,w_3)}{\tildew_1} \right)(\textsc{OPT}-1) + \frac{1+  \nicefrac{\tildew_1}{\tildew_2} +  \nicefrac{\tildew_2}{\tildew_1}}{\tildew_1 + \tildew_2}\textsc{OPT}  + 2(\textsc{OPT}-1) \Bigg].
     \end{align*}
\end{lemma}

\begin{proof}[Proof of Lemma~\ref{lemma:characterization_of_the_optimal_policy_with_two_unknown_entrants_with_heterogeneous_priors_regret_differences}.]
Subtracting the expression of Lemma~\ref{lemma:regret_pi1} from the one of Lemma~\ref{lemma:regret_pi12} and rearranging terms, the difference between $\pi^{1,2}$ and $\pi^1$ equals
\begin{align*}\label{eq:reg_diss_pi12_1_first_part}
    \textsc{Reg}(\pi^{1,2})- \textsc{Reg}(\pi^{1}) 
    &= \expect_{\substack{w_1 \sim \mathcal{F}_1 \\ w_2 \sim \mathcal{F}_2}} \Bigg[\Big(\frac{\tildew_2}{\tildew_1 + \tildew_2} \Big(\frac{\max(w_2, w_3)}{\tildew_1}-\frac{\max(w_1, w_3)}{\tildew_2} \Big) - \frac{w_3}{\tildew_1}  \Big)(\textsc{OPT}-1)  \nonumber\\
     &-\left(\frac{1}{\tildew_1} + \frac{1}{\tildew_2} - \frac{1}{\tildew_1 + \tildew_2} \left(1 +  \frac{\tildew_1}{\tildew_2} +  \frac{\tildew_2}{\tildew_1} \right)\right)\textsc{OPT} \Bigg] \nonumber\\
     = &\expect_{\substack{w_1 \sim \mathcal{F}_1 \\ w_2 \sim \mathcal{F}_2}} \Bigg[\Big(\frac{\tildew_2}{\tildew_1 + \tildew_2} \Big(\frac{\max(w_2, w_3)}{\tildew_1}-\frac{\max(w_1, w_3)}{\tildew_2} \Big) - \frac{w_3}{\tildew_1}  \Big)(\textsc{OPT}-1) - \frac{1}{\tildew_1+ \tildew_2} \textsc{OPT} \Bigg]. 
\end{align*}
 The first equality separately groups all terms with $\textsc{OPT}-1$ and $\textsc{OPT}$. The last equality uses the identity $\frac{1}{\tildew_1} + \frac{1}{\tildew_2} - \frac{1}{\tildew_1 + \tildew_2} \left(1 +  \frac{\tildew_1}{\tildew_2} +  \frac{\tildew_2}{\tildew_1} \right)  = \frac{1}{\tildew_1 + \tildew_2}$ by Lemma~\ref{lemma:equality_reciprocals_h1_h2} (stated and proven right below). Due to an analogous argument replacing $\pi^1$ with $\pi^2$ and using Lemma~\ref{lemma:regret_pi2} instead of Lemma~\ref{lemma:regret_pi1}, 
$$\textsc{Reg}(\pi^{1,2}) - \textsc{Reg}(\pi^{2}) = \expect_{\substack{w_1 \sim \mathcal{F}_1 \\ w_2 \sim \mathcal{F}_2}} \Bigg[\Big( \frac{\tildew_1}{\tildew_1 + \tildew_2} \Big(\frac{\max(w_1, w_3)}{\tildew_2} -\frac{\max(w_2, w_3)}{\tildew_1}\Big)-\frac{w_3}{\tildew_2} \Big)(\textsc{OPT}-1) - \frac{1}{\tildew_1 + \tildew_2} \textsc{OPT} \Bigg].$$
Subtracting the expressions of Lemma~\ref{lemma:regret_pi2} and Lemma~\ref{lemma:regret_pi1}, the regret difference between $\pi^{2}$ and $\pi^{1}$ is
$$\textsc{Reg}(\pi^{2}) - \textsc{Reg}(\pi^{1}) = \expect_{\substack{w_1 \sim \mathcal{F}_1 \\ w_2 \sim \mathcal{F}_2}} \Bigg[ \Bigg( \frac{w_3}{\tildew_1} +\frac{\max(w_2,w_3)}{\tildew_2} \Bigg)(\textsc{OPT}-1) - \Bigg(\frac{w_3}{\tildew_2} + \frac{\max(w_1,w_3)}{\tildew_1}\Bigg) (\textsc{OPT}-1) \Bigg].$$
\end{proof}

\begin{lemma}\label{lemma:equality_reciprocals_h1_h2}
    For any $\tildew_1> 0$ and $\tildew_2 > 0$, 
    $\frac{1}{\tildew_1} + \frac{1}{\tildew_2} - \frac{1}{\tildew_1 + \tildew_2} \left(1 +  \frac{\tildew_1}{\tildew_2} +  \frac{\tildew_2}{\tildew_1} \right) = \frac{1}{\tildew_1 + \tildew_2}$.
\end{lemma}
\begin{proof}
Rearranging terms and factoring out $\frac{1}{\tildew_1}$ and $\frac{1}{\tildew_2}$, the left-hand side is expressed as
  $$\frac{1}{\tildew_1}\Big(1- \frac{\tildew_2}{\tildew_1 + \tildew_2} \Big) + \frac{1}{\tildew_2} \Big(1-\frac{\tildew_1}{\tildew_1 + \tildew_2} \Big) - \frac{1}{\tildew_1 + \tildew_2}= \frac{1}{\tildew_1}\frac{\tildew_1}{\tildew_1 + \tildew_2}  + \frac{1}{\tildew_2} \frac{\tildew_2}{\tildew_1 + \tildew_2}  - \frac{1}{\tildew_1 + \tildew_2}= \frac{1}{\tildew_1 + \tildew_2}.
   $$
\end{proof}

\subsection{Regret characterization of the three policies (Lemmas~\ref{lemma:regret_pi1}, ~\ref{lemma:regret_pi2}, ~\ref{lemma:regret_pi12})}\label{appendix_subsec_proofs_lemmas_regret_pi1_pi2_pi12}
Let $S^{K}$ and $S^{U}$ be the sets of known and unknown products in $S$ respectively. We define $W^{U}(S) \coloneq \sum_{i \in S^{U}} \tildew_i$ to be the sum of the attraction parameters of the unknown entrants in $S$. For an option $i \in S^{K} \cup \{0\}$, we denote by $n_i(S)$ the expected number of rounds that $i$ is chosen during the epoch if $S$ is offered repeatedly. The following lemma provides an expression for $n_i(S)$.

\begin{lemma}\label{lemma:expected_number_of_purchases_of_known_entrant_until_unknown_is_purchased}
    For any assortment $S$, and any option $i \in S^K \cup \{0\}$, $n_i(S) = \frac{w_i}{W^U(S)}$.
\end{lemma}
\begin{proof}[Proof of Lemma~\ref{lemma:expected_number_of_purchases_of_known_entrant_until_unknown_is_purchased}]
Let $\tau(S)$ be the expected number of rounds until an unknown entrant is purchased. Given that an unknown entrant is purchased with probability $\frac{W^U(S)}{\sum_{i \in S} w_i + w_0}$, $\tau(S)  = \frac{\sum_{i \in S} w_i + w_0}{W^U(S)}$. Thus, 
$$n_i(S) = \prob[\text{$i$ is chosen}|S] \tau(S) = \frac{w_i}{\sum_{i \in S} w_i + w_0}\frac{\sum_{i \in S} w_i + w_0}{W^U(S)}  = \frac{w_i}{W^U(S)}.$$
\end{proof}

\begin{proof}[Proof of Lemma~\ref{lemma:regret_pi1}.]
An unknown product is purchased twice (once at the end of each of the two epochs), contributing $2(\opt-1)$ regret in total. By Lemma~\ref{lemma:expected_number_of_purchases_of_known_entrant_until_unknown_is_purchased}, the outside option is purchased $\frac{1}{\tildew_1}$ times in expectation in epoch one and $\frac{1}{\tildew_2}$ times in expectation in epoch two, contributing $\left(\frac{1}{\tildew_1} + \frac{1}{\tildew_2}\right) \opt$ regret in total. By Lemma~\ref{lemma:expected_number_of_purchases_of_known_entrant_until_unknown_is_purchased}, incumbent $3$ is purchased $\frac{w_3}{\tildew_1}$ times in expectation in epoch one and the most attractive known product from $1$ and $3$ is purchased $\frac{\max(w_1,w_3)}{\tildew_2}$ times in expectation in epoch two, contributing $\left(\frac{w_3}{\tildew_1}+\frac{\max(w_1,w_3)}{\tildew_2} \right)(\textsc{OPT}-1)$ regret in total. Adding the regrets and taking expectation yields the lemma. \end{proof}

\begin{proof}[Proof of Lemma~\ref{lemma:regret_pi2}]
   The proof has the same steps as for Lemma~\ref{lemma:regret_pi1} replacing entrant 1 with entrant 2.
\end{proof}

\begin{proof}[Proof of Lemma~\ref{lemma:regret_pi12}.]
An unknown product is purchased twice (once at the end of each of the two epochs), contributing $2(\opt-1)$ regret in total. By Lemma~\ref{lemma:expected_number_of_purchases_of_known_entrant_until_unknown_is_purchased}, the outside option is chosen $\frac{1}{\tildew_1 + \tildew_2}$ times in expectation in epoch one, and $\frac{1}{\tildew_2}$ times in expectation in epoch two if entrant 1 is revealed first (with probability $\frac{\tildew_1}{\tildew_1 + \tildew_2}$) and $\frac{1}{\tildew_1}$ times in expectation in epoch two if entrant 2 is revealed first (with probability $\frac{\tildew_2}{\tildew_1 + \tildew_2}$). As a result, the total regret from the outside option is $\left(\frac{1}{\tildew_1 + \tildew_2} + \frac{1}{\tildew_2} \frac{\tildew_1}{\tildew_1 + \tildew_2} + \frac{1}{\tildew_1} \frac{\tildew_2}{\tildew_1 + \tildew_2} \right) \opt$. By Lemma~\ref{lemma:expected_number_of_purchases_of_known_entrant_until_unknown_is_purchased}, the most attractive from products 1 and 3 is chosen $\frac{\max(w_1,w_3)}{\tildew_2}$ rounds in expectation if entrant $1$ is revealed first (with probability $\frac{\tildew_1}{\tildew_1 + \tildew_2}$) and the most attractive from products 2 and 3 is chosen $\frac{\max(w_2,w_3)}{\tildew_1}$ rounds in expectation if entrant $1$ is revealed first (with probability $\frac{\tildew_2}{\tildew_1 + \tildew_2}$). As a result the total regret from known products is $\left(\frac{\tildew_1}{\tildew_1 + \tildew_2}\frac{\max(w_1,w_3)}{\tildew_2} + \frac{\tildew_2}{\tildew_1 + \tildew_2}\frac{\max(w_2,w_3)}{\tildew_1} \right)(\opt -1)$. Adding the regrets, factoring out $\frac{1}{\tildew_1+\tildew_2}$, and taking an expectation yields the lemma. \end{proof}

\subsection{Stochastic dominance
(Proposition~\ref{corollary:stochastic_dominance})}\label{appendix_subsec_stochastic_dominance}
Let $\max_2(w_1,w_2,w_3,w_4)$ be the sum of the two largest parameters from the set $\{w_1, w_2, w_3, w_4\}$.

\begin{proof}[Proof of Proposition~\ref{corollary:stochastic_dominance}.]
 Given that $\mathcal{F}_1$ stochastically dominates $\mathcal{F}_2$, there exists a random variable $\overline{w}_2$, distributed according to $\mathcal{F}_2$, which is coupled with $w_1$ such that $w_1 \geq \overline{w}_2$ and $\overline{w}_2$ is independent of $w_2$. By Lemma~\ref{lemma:characterization_of_the_optimal_policy_with_two_unknown_entrants_with_heterogeneous_priors_regret_differences} and rearranging terms,
\begin{align*}
    &\textsc{Reg}(\pi^{2}) - \textsc{Reg}(\pi^{1}) = \expect_{\substack{w_1 \sim \mathcal{F}_1 \\ w_2 \sim \mathcal{F}_2}} \Bigg[ \Bigg(  \frac{\max(w_2-w_3,0)}{\tildew_1}-\frac{\max(w_1-w_3,0)}{\tildew_2} \Bigg)(\textsc{OPT}-1) \Bigg]\\
    &> \frac{1}{\tildew_1} \expect_{\substack{w_1 \sim \mathcal{F}_1 \\ w_2 \sim \mathcal{F}_2}} \left[\frac{\max(w_1-w_3,0)-\max(w_2-w_3,0)}{\max_2(w_1,w_2,w_3,w_4) +1} \right] \geq \frac{1}{\tildew_1} \expect_{\substack{\overline{w}_2 \sim \mathcal{F}_2 \\ w_2 \sim \mathcal{F}_2}} \left[\frac{\max(\overline{w}_2-w_3,0)-\max(w_2-w_3,0)}{\max_2(\overline{w}_2,w_2,w_3,w_4) +1} \right] =0
\end{align*}
which establishes the proposition. The first inequality uses Lemma~\ref{lemma:inequality_replace_h_2_with_h_1} (stated and proven below). The second inequality uses that $w_1 \geq \overline{w}_2$ and that the function $f(x) = \frac{(x-w_3)^{+} - (w_2 - w_3)^{+}}{\max_2(x,w_2,w_3,w_4) +1}$ is increasing in $x$ by Lemma~\ref{lemma:rev_diff_func_incr_in_fst_wgt} (stated and proven below). The last equality uses that $\overline{w}_2$ and $w_2$ are independent and identically distributed and that the function inside the expectation is antisymmetric.
\end{proof}

\begin{lemma}\label{lemma:inequality_replace_h_2_with_h_1} For any prior distributions $\mathcal{F}_1$ and $\mathcal{F}_2$ and $\tildew_1$ and $\tildew_2$ with $\tildew_1 > \tildew_2$,
    $$\underbrace{\expect \Bigg[ \Bigg(  \frac{\max(w_2-w_3,0)}{\tildew_1}-\frac{\max(w_1-w_3,0)}{\tildew_2} \Bigg)(\textsc{OPT}-1) \Bigg]}_{=\textsc{LHS}} > \frac{1}{\tildew_1} \expect \left[\frac{\max(w_1-w_3,0)-\max(w_2-w_3,0)}{\max_2(w_1,w_2,w_3,w_4) +1} \right] $$
\end{lemma}
\begin{proof}[Proof of Lemma~\ref{lemma:inequality_replace_h_2_with_h_1}.]
Using $\textsc{OPT}-1 = -\frac{1}{\max_2(w_1,w_2,w_3,w_4) +1}$, the left hand side of the lemma is
\begin{equation}\label{eq:lhs}
    \textsc{LHS} = \expect_{\substack{w_1 \sim \mathcal{F}_1 \\ w_2 \sim \mathcal{F}_2}} \left[\frac{\frac{\max(w_1-w_3,0)}{\tildew_2}}{\max_2(w_1,w_2,w_3,w_4) +1} \right] -\expect_{\substack{w_1 \sim \mathcal{F}_1 \\ w_2 \sim \mathcal{F}_2}} \left[\frac{\frac{\max(w_2-w_3,0)}{\tildew_1}}{\max_2(w_1,w_2,w_3,w_4) +1} \right]
\end{equation}
The first term is lower bounded by
\begin{align}\label{ineq:lb_first_term}
    &\expect_{\substack{w_1 \sim \mathcal{F}_1 \\ w_2 \sim \mathcal{F}_2}} \left[\frac{\frac{\max(w_1-w_3,0)}{\tildew_2}}{\max_2(w_1,w_2,w_3,w_4) +1} \right] = \expect_{\substack{w_1 \sim \mathcal{F}_1 \\ w_2 \sim \mathcal{F}_2}} \left[\frac{\frac{w_1-w_3}{\tildew_2}}{\max_2(w_1,w_2,w_3,w_4) +1} \middle|w_1 > w_3\right]\prob[w_1 > w_3] \nonumber\\
    &> \expect_{\substack{w_1 \sim \mathcal{F}_1 \\ w_2 \sim \mathcal{F}_2}} \left[\frac{\frac{w_1-w_3}{\tildew_1}}{\max_2(w_1,w_2,w_3,w_4) +1} \middle|w_1 > w_3\right]\prob[w_1 > w_3]= \expect_{\substack{w_1 \sim \mathcal{F}_1 \\ w_2 \sim \mathcal{F}_2}} \left[\frac{\frac{\max(w_1-w_3,0)}{\tildew_1}}{\max_2(w_1,w_2,w_3,w_4) +1} \right]
\end{align}
where the inequality uses that $\frac{w_1-w_3}{\tildew_2}> \frac{w_1-w_3}{\tildew_1}$ conditioned on $w_1 > w_3$ since $\tildew_1 > \tildew_2$ and $\prob[w_1 > w_3] > 0$. The lemma combines \eqref{eq:lhs} and \eqref{ineq:lb_first_term} and factors out $\frac{1}{\tildew_1}$.
\end{proof}. 

\begin{lemma}\label{lemma:rev_diff_func_incr_in_fst_wgt}
    For any $w_2$, the function $f(x) \coloneq \frac{(x-w_3)^{+} - (w_2 - w_3)^{+}}{\max_2(x,w_2,w_3,w_4) +1}$ is non-decreasing in $x$.
\end{lemma}
\begin{proof}[Proof of Lemma~\ref{lemma:rev_diff_func_incr_in_fst_wgt}.]
We consider cases based on whether $w_2 \geq w_3$ or $w_3 > w_2$. 
\paragraph{Case 1: $w_2 \geq w_3$.} Using that $w_3 \geq w_4$, 
\[f(x) = \frac{(x-w_3)^{+} - (w_2 - w_3)^{+}}{\max_2(x,w_2,w_3,w_4) +1} = \frac{(x-w_3)^{+} - (w_2 - w_3)}{\max_2(x,w_2,w_3,w_4) +1} =\begin{cases}
\frac{x - w_2}{x+w_3+1}, & x \ge w_2, \\
\frac{w_3 - w_2}{w_3+w_2+1}, & x < w_2.
\end{cases} \]
Thus, $f(x) = \max\left( \frac{x - w_2}{x+w_3+1}, \frac{w_3 - w_2}{w_3+w_3+1}\right)$. $f(x)$ is non-decreasing as $\frac{x - w_2}{x+w_3+1}$ is increasing and the maximum of an increasing function and a constant is non-decreasing. 

\paragraph{Case 2: $w_3 > w_2$.} Using that $w_3 \geq w_4$,
\[
f(x) = \frac{(x-w_3)^{+} - (w_2 - w_3)^{+}}{\max_2(x,w_2,w_3,w_4) +1}
= \frac{(x-w_3)^{+}}{\max_2(x,w_2,w_3,w_4) +1}
=
\begin{cases}
\frac{x-w_3}{x+w_3+1}, & x \ge w_3, \\
0, & x < w_3.
\end{cases}
\]
Thus, $f(x) = \max\left( \frac{x-w_3}{x+w_3+1}, 0 \right)$. $f(x)$ is non-decreasing as $\frac{x-w_3}{x+w_3+1}$ is increasing and the maximum of an increasing function and a constant is non-decreasing.
\end{proof}

\subsection{Optimal policy under two Bernoulli-like priors (Proposition~\ref{prop:different_priors_bernoulli})}\label{appendix_subsec:different_priors}
Let $\mathcal{I}(\mu, p_1, p_2, w_3,w_4)$ be the instance where the attraction parameter of entrant $i \in \{1,2\}$ equals $\frac{\mu}{p_i}$ with probability $p_i$ and $0$ otherwise, the incumbents have attraction parameters $w_3$ and $w_4$ with $w_3 > w_4$, and the cold-start attraction parameter for each entrant is the mean, i.e., $\tilde{w}_1 = \tilde{w}_2 = \mu$. To prove Proposition~\ref{prop:different_priors_bernoulli}, the following lemma (proof in Appendix~\ref{appendix_subsec:proof_lemma_E9}) shows that when the prior upside probabilities $p_1$ and $p_2$ are smaller than $\frac{1}{6}$, exploring both entrants is not optimal. 

\begin{lemma}\label{lemma: exploring_one_is_better}
    For any $\mu, p_1, p_2, w_3, w_4,$ with $p_1, p_2 \in (0, \frac{1}{6})$, $w_3 < \frac{\mu}{p_1}$, $w_3 < \frac{\mu}{p_2}$, $w_3 > 2\mu$, $w_3 > 4$ and $w_4 = \frac{w_3}{2}$, $\textsc{Reg}(\pi^{1}) < \textsc{Reg}(\pi^{1,2})$ under the instance $\mathcal{I}(\mu, p_1, p_2, w_3, w_4)$.
\end{lemma}

The next lemma (proof in Appendix~\ref{appendix_subsec:proof_lemma_E10}) shows that when the most attractive incumbent's attraction parameter lies within a moderate range, either $\pi^1$ or $\pi^2$ can yield lower regret and the behavior depends on a threshold $\theta$.

\begin{lemma}\label{lemma: when_exploring_entrant_2_is_better}
    For any $\mu, p_1,w_3, w_4$ with $w_3< \frac{\mu}{p_1} < 2 w_3$, $w_3 > 2\mu$, there exists some $\theta(\mu , p_1, w_3) \in (0, p_1)$ such that $\textsc{Reg}(\pi^{2}) < \textsc{Reg}(\pi^{1})$ if $p_2 \in (p_1 - \theta( \mu , p_1, w_3), p_1)$ and $\textsc{Reg}(\pi^{2}) > \textsc{Reg}(\pi^{1})$ if $p_2 \in (0, p_1 - \theta(\mu , p_1,w_3))$ under the instance $\mathcal{I}(\mu, p_1, p_2, w_3, w_4)$. 
\end{lemma}

The next lemma (proof in Appendix~\ref{appendix_subsec:proof_lemma_E11}) shows that there exists an optimal policy which offers the same assortment until an unknown entrant is purchased. 

\begin{lemma}\label{lemma:one_one_the_policies_pi_1_pi_2_and_pi_12_is_optimal}
    There exists an optimal policy $\pi^{\star}$ such that $\pi^{\star} \in \{\pi^{1},\pi^{2}, \pi^{1,2}\}$. 
\end{lemma}

\begin{proof}[Proof of Proposition~\ref{prop:different_priors_bernoulli}.]
Any parameters $\mu, p_1,p_2, w_3, w_4$ with $p_2 < p_1$, $p_1 \in (0, \frac{1}{6})$, $w_3< \frac{\mu}{p_1} < 2 w_3$, $w_3 > 2\mu$, $w_3 > 4$, and $w_4 = \frac{w_3}{2}$ satisfy the conditions of Lemma~\ref{lemma: exploring_one_is_better} and Lemma~\ref{lemma: when_exploring_entrant_2_is_better}. By Lemma~\ref{lemma: exploring_one_is_better}, $\pi^{1,2}$ is not optimal. Combining this with Lemma~\ref{lemma: when_exploring_entrant_2_is_better} and Lemma~\ref{lemma:one_one_the_policies_pi_1_pi_2_and_pi_12_is_optimal}, $\pi^{1}$ is optimal if $p_2 < \theta(\mu , p_1,w_3)$ and $\pi^2$ is optimal if $p_2 > \theta(\mu , p_1,w_3)$.
\end{proof}

\subsection{Exploring both entrants is not optimal (Lemma~\ref{lemma: exploring_one_is_better})}\label{appendix_subsec:proof_lemma_E9}
\begin{proof}[Proof of Lemma~\ref{lemma: exploring_one_is_better}]
Letting $\max_2(w_1,w_2,w_3,w_4)$ be the sum of the two largest parameters from the set $\{w_1, w_2, w_3, w_4\}$, the scaled difference of the regrets of $\pi^{1,2}$ and $\pi^{1}$ is expressed and lower bounded as
\begin{align*}
 2\mu \left(\textsc{Reg}(\pi^{1,2}) - \textsc{Reg}(\pi^{1}) \right)&=  \expect_{\substack{w_1 \sim \mathcal{F}_1 \\ w_2 \sim \mathcal{F}_2}}\Bigg[ \frac{- \max(w_2, w_3) + \max(w_1, w_3) +2 w_3 -\max_2(w_1, w_2, w_3, w_4)}{\max_2(w_1, w_2, w_3, w_4) + 1}\Bigg]\\
  &=p_1 p_2 \frac{- \frac{\mu}{p_2} +\frac{\mu}{p_1}+2w_3   - (\frac{\mu}{p_1} +\frac{\mu}{p_2})}{\frac{\mu}{p_1} +\frac{\mu}{p_2} + 1}+ p_1 (1-p_2) \frac{ - w_3 +\frac{\mu}{p_1} +2w_3 - (\frac{\mu}{p_1} +w_3)}{\frac{\mu}{p_1} +w_3 + 1} \nonumber\\
     &+ (1-p_1) p_2 \frac{ - \frac{\mu}{p_2} +w_3 +2w_3 - (\frac{\mu}{p_2} +w_3)}{\frac{\mu}{p_2} +w_3 + 1}+ (1-p_1)(1-p_2) \frac{w_3 - w_4}{w_3 + w_4+1} \nonumber\\
 &= p_1 p_2 \frac{2(w_3 -\frac{\mu}{p_2})}{\frac{\mu}{p_1} +\frac{\mu}{p_2} + 1} + p_2(1-p_1) \frac{2(w_3 - \frac{\mu}{p_2})}{\frac{\mu}{p_2} + w_3 + 1} +(1-p_1)(1-p_2) \frac{w_3}{3w_3 + 2} > 0.
\end{align*}
which shows the lemma (as $\mu > 0$). The first equality expresses the regret difference between $\pi^{1,2}$ and $\pi^1$ using Lemma~\ref{lemma:characterization_of_the_optimal_policy_with_two_unknown_entrants_with_heterogeneous_priors_regret_differences}, uses $\tildew_1 = \tildew_2 = \mu$, $\textsc{OPT} = \frac{\max_2(w_1, w_2, w_3, w_4)}{\max_2(w_1, w_2, w_3, w_4)+1}$, and $\opt-1 = -\frac{1}{\max_2(w_1, w_2, w_3, w_4)+1}$ and cancels the $\mu$ term. The second equality expands the expectation over $w_1$ and $w_2$ and uses $\max(\frac{\mu}{p_1},w_3) = \frac{\mu}{p_1}$ (as $\frac{\mu}{p_1} > w_3$) and $\max(\frac{\mu}{p_2},w_3) = \frac{\mu}{p_2}$ (as $\frac{\mu}{p_2} > w_3$). The third equality cancels the term in multiplying $p_1(1-p_2)$ as $- w_3 +\frac{\mu}{p_1} +2w_3  - (\frac{\mu}{p_1} +w_3) = 0$ and uses $\frac{w_3 - w_4}{w_3 + w_4+1} = \frac{w_3}{3w_3 + 2}$ as $w_4 = \frac{w_3}{2}$. The inequality uses Lemma~\ref{lemma:regret_diff_expression_positive} (stated and proven below). \end{proof}

\begin{lemma}\label{lemma:regret_diff_expression_positive}
    For any $\mu, p_1, p_2, w_3, w_4,$ with $p_1, p_2 \in (0, \frac{1}{6})$, $w_3 < \frac{\mu}{p_1}$, $w_3 < \frac{\mu}{p_2}$, $w_3 > 2\mu$, $w_3 > 4$,
    $$ \underbrace{p_1 p_2 \frac{2(w_3 -\frac{\mu}{p_2})}{\frac{\mu}{p_1} +\frac{\mu}{p_2} + 1} + p_2(1-p_1) \frac{2(w_3 - \frac{\mu}{p_2})}{\frac{\mu}{p_2} + w_3 + 1} +(1-p_1)(1-p_2) \frac{w_3}{3w_3 + 2}}_{\eqqcolon \textsc{LHS}}  > 0$$
\end{lemma}
\begin{proof} The left hand side is lower bounded by
\begin{align*}
       \textsc{LHS}  &\geq p_2 \frac{2(w_3 - \frac{\mu}{p_2})}{\frac{\mu}{p_2} + w_3 + 1} +(1-p_1)(1-p_2) \frac{w_3}{3w_3 + 2} \geq \frac{2 p_2(w_3-\frac{w_3}{2p_2})}{\frac{w_3}{2p_2} +w_3 +1}+(1-p_1)(1-p_2) \frac{w_3}{3w_3 + 2} \\
     &= \frac{(2p_2-1)w_3}{\frac{w_3}{2p_2} +w_3 +1}+(1-p_2)(1-p_1) \frac{w_3}{3w_3 + 2}\geq \frac{(p_2-1)w_3}{4w_3 +1} + (1-p_2) \frac{5 w_3}{18 w_3 + 12} \nonumber\\
     &=  (1-p_2) \left(\frac{w_3}{4w_3 +1} -\frac{5 w_3}{18 w_3 + 12} \right) = (1-p_2) \frac{w_3(2 w_3-7)}{(4w_3 + 1)(18w_3 + 12)} > 0.
\end{align*}
The first inequality uses $\frac{2(w_3 -\frac{\mu}{p_2})}{\frac{\mu}{p_1} +\frac{\mu}{p_2} + 1} \geq \frac{2(w_3 -\frac{\mu}{p_2})}{\frac{\mu}{p_2} +w_3 + 1}$ as $\frac{\mu}{p_1} \geq w_3$ and $w_3 -\frac{\mu}{p_2} \leq 0$ to lower bound the first term of the \textsc{LHS}  and expands the second term to cancel the terms multiplied by $p_1 p_2$. The second inequality uses $p_2 \frac{2(w_3 - \frac{\mu}{p_2})}{\frac{\mu}{p_2} + w_3 + 1} \geq \frac{2 p_2(w_3-\frac{w_3}{2p_2})}{\frac{w_3}{2p_2} +w_3 +1}$ since $\frac{w_3}{2} \geq \mu$ and the function $f(x) \coloneq\frac{w_3-\frac{x}{p_2}}{\frac{x}{p_2} + w_3 + 1}$ is decreasing in $x$. The third inequality uses that $2p_2-1\geq p_2-1$, $\frac{w_3}{2p_2} + w_3 \geq 4w_3$ as $p_2 < \frac{1}{6}$, and $(1-p_1) \frac{w_3}{3w_3 + 2} \geq \frac{5 w_3}{18 w_3 + 12}$ since $1-p_1 > \frac{5}{6}$ (as $p_1 < \frac{1}{6}$). The last inequality uses $2 w_3-7 > 0$ as $w_3 > 4$. 
\end{proof}

\subsection{Threshold dependence of optimal exploration strategy (Lemma~\ref{lemma: when_exploring_entrant_2_is_better})}\label{appendix_subsec:proof_lemma_E10}
\begin{lemma}\label{lemma:sign_diff_reg2_reg1_depends_on_quadratic}
There exists a quadratic $Q_{\mu, p_1,w_3}(x) = x^2 \cdot A_{\mu, p_1,w_3}  + x \cdot B_{\mu, p_1,w_3} + C_{\mu, p_1,w_3}$ such that 
    \begin{enumerate}
    \item $\textsc{Reg}(\pi^2) > \textsc{Reg}(\pi^1)$ if and only if $Q_{\mu, p_1,w_3}(p_2)< 0$,
        \item $A_{\mu, p_1,w_3} > 0$ and the quadratic $Q_{\mu, p_1,w_3}(\cdot)$ has one positive root and one negative root, and 
        \item evaluated at $x = p_1$, $Q_{\mu, p_1,w_3}$ is positive, i.e., $Q_{\mu, p_1,w_3}(p_1) > 0$.
    \end{enumerate}
\end{lemma}
\begin{proof}[Proof of Lemma~\ref{lemma: when_exploring_entrant_2_is_better}.]
    Let $Q_{\mu, p_1,w_3}(x)$ be the quadratic in Lemma~\ref{lemma:sign_diff_reg2_reg1_depends_on_quadratic} and $\theta(\mu, p_1,w_3)$ be its positive root. Given that $A_{\mu, p_1,w_3} > 0$, $Q_{\mu, p_1,w_3}(p_1)$ has one positive and one negative root and $Q_{\mu, p_1,w_3}(p_1) > 0$ (by Lemma~\ref{lemma:sign_diff_reg2_reg1_depends_on_quadratic}), then $\theta(\mu, p_1,w_3) \in (0, p_1)$. 
    
    By Lemma~\ref{lemma:sign_diff_reg2_reg1_depends_on_quadratic}, $\textsc{Reg}(\pi^2) > \textsc{Reg}(\pi^1)$ when $p_2 \in (0, \theta(\mu, p_1,w_3))$ because $Q_{\mu, p_1,w_3}(p_2) < 0$, and $\textsc{Reg}(\pi^2) < \textsc{Reg}(\pi^1)$ when $p_2 \in (\theta(\mu, p_1,w_3), p_1)$ because $Q_{\mu, p_1,w_3}(p_2) > 0$ completing the proof. 
\end{proof}

\begin{proof}[Proof of Lemma~\ref{lemma:sign_diff_reg2_reg1_depends_on_quadratic}.]
Letting $\max_2(w_1, w_2, w_3, w_4)$ be the sum of two largest attraction parameters from $\{w_1, w_2, w_3, w_4\}$, the scaled difference between the regrets of $\pi^{2}$ and $\pi^{1}$ is 
\begin{align}\label{4_diff_pi1_pi2}
    &\mu\left(\textsc{Reg}(\pi^2)-\textsc{Reg}(\pi^1)\right) = \expect_{\substack{w_1 \sim \mathcal{F}_1 \\ w_2 \sim \mathcal{F}_2}} \left[\frac{ \max(w_1, w_3)- \max(w_2, w_3)}{\max_2(w_1, w_2, w_3, w_4) + 1}\right]  \nonumber\\
    &= (p_2-p_1) \frac{p_2 \cdot \mu p_1}{p_2 \cdot (\mu + p_1) + \mu p_1} +  \frac{p_2(1-p_1)(p_2 w_3  -\mu)}{p_2 (w_3+1) + \mu} -  \frac{p_1 (1-p_2)(p_1 w_3-\mu )}{p_1 (w_3 + 1)+\mu} \nonumber\\
    &= (p_2-p_1) \frac{p_2 \cdot \mu p_1}{p_2 \cdot (\mu + p_1) + \mu p_1}+\frac{(p_2-p_1) \Bigg( p_2 \cdot \Big(p_1 \big(w_3(w_3-2\mu) + w_3-\mu \big) + w_3 \mu \Big) + \mu p_1 \left(w_3 - \frac{\mu}{p_1} \right) \Bigg)}{p_2 \cdot (w_3+1)(p_1 (w_3 + 1)+\mu) + \mu (p_1 (w_3 + 1)+\mu)} \nonumber\\
    &= (p_2-p_1) \cdot \frac{p_2^2 \cdot A_{\mu, p_1,w_3} +p_2 \cdot B_{\mu, p_1,w_3} + C_{\mu, p_1,w_3}}{(p_2 \cdot (\mu + p_1) + \mu p_1)(p_2 \cdot (w_3+1)(p_1 (w_3 + 1)+\mu) + \mu (p_1 (w_3 + 1)+\mu))} 
\end{align}
\begin{align*}
    \text{for } A_{\mu, p_1,w_3} &\coloneq \mu p_1 (w_3+1)(p_1(w_3+1)+\mu) + \Big(p_1 \big(w_3(w_3-2\mu) + w_3-\mu \big) + w_3 \mu \Big)(\mu+p_1),\\
    B_{\mu, p_1,w_3} &\coloneq \mu^2 p_1 (p_1(w_3+1)+\mu) + \mu p_1 \left(w_3 - \frac{\mu}{p_1} \right)(\mu+p_1) + \mu p_1 \Big(p_1 \big(w_3(w_3-2\mu) + w_3-\mu \big) + w_3 \mu \Big),\\
C_{\mu, p_1,w_3} &\coloneq \mu^2 p_1^2 \left(w_3 - \frac{\mu}{p_1} \right). 
\end{align*}
The first equality in \eqref{4_diff_pi1_pi2} expresses the regret difference between $\pi^2$ and $\pi^1$ using Lemma~\ref{lemma:characterization_of_the_optimal_policy_with_two_unknown_entrants_with_heterogeneous_priors_regret_differences}, uses $\tildew_1 = \tildew_2= \mu$, $\textsc{OPT} = \frac{\max_2(w_1, w_2, w_3, w_4)}{\max_2(w_1, w_2, w_3, w_4)+1}$, and $\opt-1 = -\frac{1}{\max_2(w_1, w_2, w_3, w_4)+1}$ and cancels the $\mu$ term. The second equality expands the expectation over $w_1$ and $w_2$, using $\max(\frac{\mu}{p_1},w_3) = \frac{\mu}{p_1}$ (as $\frac{\mu}{p_1} > w_3$) and $\max(\frac{\mu}{p_2},w_3) = \frac{\mu}{p_2}$ (as $\frac{\mu}{p_2} > w_3$), and performs an algebraic manipulation on each fraction. The third equality uses Lemma~\ref{lemma:algebraic_manipulation} (stated and proven below). The last equality factors out $p_2-p_1$ and takes a common denominator. 

Let $Q_{\mu, p_1,w_3}(c) \coloneq x^2 \cdot A_{\mu, p_1,w_3} +x \cdot B_{\mu, p_1,w_3} + C_{\mu, p_1,w_3}$. Given that the numerator of the right hand side of \eqref{4_diff_pi1_pi2} is $Q_{\mu, p_1,w_3}(p_2)$, the denominator is positive and $p_2-p_1 < 0$, $\textsc{Reg}(\pi^2) > \textsc{Reg}(\pi^1)$ if and only if $Q_{\mu, p_1,w_3}(p_2) < 0$ showing the first lemma part. The second-degree term of $Q$, $A_{\mu, p_1,w_3}$, is positive 
$$A_{\mu, p_1,w_3} > \Big(p_1 \big(w_3(w_3-2\mu) + w_3-\mu \big) + w_3 \mu \Big)(\mu+p_1) \geq 0$$
where the first inequality uses that $\mu p_1 (w_3+1)(p_1(w_3+1)+\mu) > 0$ as $\mu, p_1, w_3 > 0$ and the second inequality uses that $w_3 > 2 \mu$ and $w_3 > \mu$. The constant term of $Q$, $C_{\mu, p_1,w_3}$, is negative $C_{\mu, p_1,w_3} = \mu^2 p_1^2 \left(w_3 - \frac{\mu}{p_1} \right) < 0$ as $\frac{\mu}{p_1}  > w_3$. As a result, $Q_{\mu, p_1,w_3}(x)$ has one positive and one negative root showing the second lemma part. The value of $Q_{\mu, p_1,w_3}$ at $p_1$ is lower bounded by
\begin{align*}
    Q_{\mu, p_1,w_3}(p_1) &= p_1^2 \cdot A_{\mu, p_1,w_3} +p_1 \cdot B_{\mu, p_1,w_3} + C_{\mu, p_1,w_3}\\
    &\geq p_1^2 \Big(p_1 \big(w_3(w_3-2\mu) + w_3-\mu \big) + w_3 \mu \Big)(\mu+p_1) \\
    &+ p_1 \left( \mu p_1 \left(w_3 - \frac{\mu}{p_1} \right)(\mu+p_1) + \mu p_1 \Big(p_1 \big(w_3(w_3-2\mu) + w_3-\mu \big) + w_3 \mu \Big) \right) + \mu^2 p_1^2 \left(w_3 - \frac{\mu}{p_1} \right)\\
    &=\left(  p_1 \Big(p_1 \big(w_3(w_3-2\mu) + w_3-\mu \big) + w_3 \mu \Big) + \mu p_1 \left(w_3 - \frac{\mu}{p_1} \right) \right) \left(p_1 (\mu + p_1) + \mu p_1 \right)\\
    &= \left(  p_1^2 \big(w_3(w_3-2\mu) + w_3-\mu \big) + \mu p_1 \left(2w_3 - \frac{\mu}{p_1} \right) \right) \left(p_1 (\mu + p_1) + \mu p_1 \right) > 0
\end{align*}
showing the third lemma part and completing the proof. The first inequality uses $\mu p_1 (w_3+1)(p_1(w_3+1)+\mu) > 0$ to lower bound $A_{\mu, p_1,w_3} $ and $\mu^2 p_1 (p_1(w_3+1)+\mu) > 0$ to lower bound $B_{\mu, p_1,w_3}$. The last inequality uses $w_3 > 2\mu$, $w_3 > \mu$, and $2 w_3 > \frac{\mu}{p_1}$.
\end{proof}

\begin{lemma}\label{lemma:algebraic_manipulation}
For any $\mu, p_1, p_1$ and $w_3$,
$$\underbrace{\frac{p_2(1-p_1)(p_2 w_3  -\mu)}{p_2 (w_3+1) + \mu} -  \frac{p_1 (1-p_2)(p_1 w_3-\mu )}{p_1 (w_3 + 1)+\mu}}_{\eqqcolon\textsc{LHS}}= \frac{(p_2-p_1) \left( p_2 \cdot D_1(\mu,p_1,w_3) + D_0(\mu,p_1,w_3)\right)}{p_2 \cdot (w_3+1)(p_1 (w_3 + 1)+\mu) + \mu (p_1 (w_3 + 1)+\mu)}
$$
where $D_1(\mu, p_1,w_3) = p_1 \big(w_3(w_3-2\mu) + w_3-\mu \big) + w_3 \mu$ and $D_0(\mu, p_1,w_3) = \mu p_1 (w_3 - \frac{\mu}{p_1})$.
\end{lemma}
\begin{proof}[Proof of Lemma~\ref{lemma:algebraic_manipulation}]
Taking a common denominator in the left hand side,
\begin{align*}\label{eq:expression_difference_second_and_third_terms}
    \textsc{LHS}&=\frac{p_2 (1-p_1) (p_2 w_3 -\mu)(p_1(w_3+1) + \mu)-p_1 (1-p_2) (p_1 w_3 -\mu)(p_2(w_3+1) + \mu)}{p_2 \cdot (w_3+1)(p_1 (w_3 + 1)+\mu) + \mu (p_1 (w_3 + 1)+\mu)} \\
    &= (p_2-p_1) \frac{p_2 \cdot D_1(\mu, p_1,w_3) + D_0(\mu, p_1,w_3)}{p_2 \cdot (w_3+1)(p_1 (w_3 + 1)+\mu) + \mu (p_1 (w_3 + 1)+\mu)}.
\end{align*}
To obtain the second equality, express the terms in the numerator as
\begin{align*}
    X &\coloneq p_2 (1-p_1) (p_2 w_3 -\mu)(p_1(w_3+1) + \mu)=(p_2^2 w_3 - p_2 \mu - p_1 p_2^2 w_3 + p_2 p_1 \mu)(p_1(w_3+1) + \mu) \\
&=  p_1 p_2^2 w_3 (w_3+1) + \mu p_2^2 w_3
- p_1 p_2 \mu (w_3+1) - p_2 \mu^2  - p_1^2 p_2^2 w_3 (w_3+1) - \mu p_1 p_2^2 w_3
+ p_1^2 p_2 \mu (w_3+1) \\
&+ p_2 p_1 \mu^2, \quad \text{and}\\
    Y &\coloneq p_1 (1-p_2) (p_1 w_3 -\mu)(p_2(w_3+1) + \mu)=(p_1^2 w_3 - p_1 \mu - p_2 p_1^2 w_3 + p_1 p_2 \mu)(p_2(w_3+1) + \mu) \\
&=  p_2 p_1^2 w_3 (w_3+1) + \mu p_1^2 w_3
- p_2 p_1 \mu (w_3+1) - p_1 \mu^2  - p_2^2 p_1^2 w_3 (w_3+1) - \mu p_2 p_1^2 w_3
+ p_2^2 p_1 \mu (w_3+1) \\
&+ p_1 p_2 \mu^2,
\end{align*}
Canceling the common terms $p_2 p_1 \mu (w_3+1)$, $p_2^2 p_1^2 w_3 (w_3+1)$, and $p_1 p_2 \mu^2$, the numerator is
\begin{align*}
    X-Y &= w_3(w_3+1) p_1 p_2 (p_2-p_1) + \mu w_3(p_2-p_1)(p_2 + p_1) - \mu^2(p_2 -p_1) -\mu p_1 p_2 w_3(p_2-p_1) \\
    &+ p_1 p_2 \mu(w_3 + 1)(p_1-p_2)\\
    &= (p_2-p_1) \Bigg( p_2 \Big(p_1 \big(w_3(w_3-2\mu) + w_3-\mu \big) + w_3 \mu \Big) + \mu p_1 \left(w_3 - \frac{\mu}{p_1} \right) \Bigg) \\
    &= (p_2-p_1)(p_2 \cdot D_1(\mu,p_1,w_3) + D_0(\mu,p_1,w_3))
\end{align*}
\end{proof}

\subsection{Optimality of one of the three policies (Lemma~\ref{lemma:one_one_the_policies_pi_1_pi_2_and_pi_12_is_optimal})}\label{appendix_subsec:proof_lemma_E11}

To prove Lemma~\ref{lemma:one_one_the_policies_pi_1_pi_2_and_pi_12_is_optimal} we first characterize the optimal policy through a set of optimal regret functions in Lemma~\ref{lemma:any_policy_which_minimizes_the_Q_function_is_optimal}, which is similar to Lemma~\ref{lemma:epoch_optimal_policy_is_optimal}. To do so, for a history $\mathcal{H} = \{(S_i, Y_i, w_{Y_i})\}_{i=1}^{t-1}$, we define $\mathcal{C}(\mathcal{H}) = \{(i, w_i): i \in I(\mathcal{H})\}$ to be the \textit{condensed history} consisting of the known products and their respective attraction parameters given $\mathcal{H}$. The set of known products and the attraction parameter of each product given a history $\mathcal{H}$ depend on $\mathcal{H}$ only through the condensed history $\mathcal{C}(\mathcal{H})$. As a result, the expected ex-post optimum, $\textsc{OPT}(\mathcal{H})$, the revenue of a set $S$ given $\mathcal{H}$, $\textsc{Rev}(S;\mathcal{H})$, and the epoch regret of a set $S$ given $\mathcal{H}$, $\textsc{EpochReg}(S;\mathcal{H})$, depend on $\mathcal{H}$ only through $\mathcal{C}(\mathcal{H})$ and we let 
$$\textsc{OPT}(\mathcal{C}(\mathcal{H})) \coloneq \textsc{OPT}(\mathcal{H}) \text{,}\quad  \textsc{Rev}(S, \mathcal{C}(\mathcal{H})) \coloneq \textsc{Rev}(S;\mathcal{H}) \text{,}\quad  \textsc{EpochReg}(S,\mathcal{C}(\mathcal{H})) \coloneq \textsc{EpochReg}(S;\mathcal{H}).$$
The sets of known products $I(\mathcal{H})$ and assortments containing an unknown entrant $\mathcal{E}(\mathcal{H})$ depend on $\mathcal{H}$ only through the condensed history; as a result $I(\mathcal{C}(\mathcal{H})) \coloneq I(\mathcal{H})$ and $\mathcal{E}(\mathcal{C}(\mathcal{H})) \coloneq \mathcal{E}(\mathcal{H})$. 

We define optimal regret functions $V(\mathcal{C})$ and $Q(S;\mathcal{C})$ for any condensed history $\mathcal{C}$ and assortment $S \in \mathcal{E}(\mathcal{C})$ containing an unknown entrant, with a recursive definition based on the number of unknown entrants remaining in $\mathcal{C}$. If there are no unknown entrants given $\mathcal{C}$, set $V(\mathcal{C}) \coloneqq 0$. Suppose we have defined $V(\mathcal{C})$ and $Q(S;\mathcal{C})$ for any condensed history $\mathcal{C}$ with at most $k$ remaining unknown entrants and $S \in \mathcal{E}(\mathcal{C})$. A condensed history $\mathcal{C}$ is \textit{terminal} if $\textsc{OPT}(\mathcal{C}) = \max\limits_{S \subseteq I(\mathcal{C}), |S| \leq c} \textsc{Rev}(S,\mathcal{C})$. For any terminal $\mathcal{C}$ with $k+1$ unknown entrants set $Q(S,\mathcal{C}) \coloneq 0$ for all $S \in \mathcal{E}(\mathcal{C})$ and $V(\mathcal{C}) \coloneq 0$. Let $S^{\textsc{u}}(\mathcal{C})$ be the subset of unknown entrants in $S$ given a condensed history $\mathcal{C}$. For any non-terminal condensed history $\mathcal{C}$ with $k + 1$ unknown entrants and $S \in \mathcal{E}(\mathcal{C})$,  $$Q(S;\mathcal{C}) \coloneqq \textsc{EpochReg}(S,\mathcal{C}) + \sum_{i \in S^{\textsc{u}}(\mathcal{C})} \frac{\tildew_i}{ \sum_{j \in S^{\textsc{u}}(\mathcal{C})} \tildew_j} \expect_{w_i \sim \mathcal{F}_i}[ V(\mathcal{C} \cup (i, w_i))] \quad \text{and} \quad V(\mathcal{C}) \coloneqq \min_{S \in \mathcal{E}(\mathcal{C})} Q(S,\mathcal{C}).$$

The following lemma (proven in Appendix~\ref{appendix_subsec:proof_lemma_4.1_het_priors}) establishes that a policy is optimal if it selects the assortment $S \in \mathcal{E}(\mathcal{C})$ minimizing $Q(S;\mathcal{C}(\mathcal{H}))$ if $\mathcal{H}$ is not terminal and the revenue-maximizing assortment otherwise. This lemma is the analogue of Lemma~\ref{lemma:epoch_optimal_policy_is_optimal} in this heterogeneous priors' model and the optimal regret function $Q(S;\mathcal{C}(\mathcal{H}))$ has a similar role as the epoch regret $\textsc{EpochReg}(S;\mathcal{H})$ in Lemma~\ref{lemma:epoch_optimal_policy_is_optimal}.

\begin{lemma}\label{lemma:any_policy_which_minimizes_the_Q_function_is_optimal}
    A policy $\pi^{\star}$ is optimal if for any history $\mathcal{H}$,
    \begin{itemize}
    \item $\pi^{\star}(\mathcal{H}) \in \argmin\limits_{S \in \mathcal{E}(\mathcal{H})}Q(S;\mathcal{C}(\mathcal{H})) $ \text{ if $\mathcal{H}$ is not terminal, and}
    \item $\pi^{\star}(\mathcal{H}) \in \argmax\limits_{S \subseteq I(\mathcal{H}), |S| \leq c} \textsc{Rev}(S; \mathcal{H})$ otherwise.
\end{itemize}
\end{lemma}

\begin{proof}[Proof of Lemma~\ref{lemma:one_one_the_policies_pi_1_pi_2_and_pi_12_is_optimal}.]
Any policy which (i) offers the same assortment until an unknown entrant is purchased, (ii) always explores at least one unknown entrant until all entrants are purchased, and (iii) offers unknown entrants with the most attractive known products is $\pi^1$, $\pi^2$, or $\pi^{1,2}$. Let $\pi^{\star}$ be a policy which satisfies the conditions of Lemma~\ref{lemma:any_policy_which_minimizes_the_Q_function_is_optimal} and satisfies (i), i.e., it offers the same assortment until an unknown entrant is purchased. To prove the lemma it suffices to show that $\pi^{\star}$ explores at least one unknown entrant until all entrants have been purchased and that $\pi^{\star}$ always complements the unknown entrants with the most attractive known products.

\paragraph{(ii) $\pi^{\star}$ explores at least one unknown entrant until all entrants have been purchased. } Given that the upsides of both entrants $\frac{\mu}{p_1}$ and $\frac{\mu}{p_2}$ are greater than the the most attractive incumbent's attraction parameter $w_3$, a history is terminal if and only if both entrants have been purchased. Thus, $\pi^{\star}$ explores at least one entrant until both entrants have been purchased.  
\paragraph{(iii) $\pi^{\star}$ explores the unknown entrants with the most attractive known products. }
For any history $\mathcal{H}$, $\pi^{\star}(\mathcal{H}) \in \argmin\limits_{S \in \mathcal{E}(\mathcal{H})}Q(S;\mathcal{C}(\mathcal{H}))$. Let $S^{K}$ and $S^{U}$ be the known and unknown products in  $\pi^{\star}(\mathcal{H})$ respectively. 
Given that the to-go regret does not depend on $S^{K}$, the minimization problem reduces to minimizing $\textsc{EpochReg}(S;\mathcal{C}(\mathcal{H}))$. When a known product $i \in S^{K}$ is purchased a regret of $\textsc{OPT}-1$ is incurred; this occurs $\frac{w_i(\mathcal{H})}{\sum_{i \in S^{U}}w_i(\mathcal{H})}$ rounds in expectation. When the outside option is chosen a regret of $\textsc{OPT}$ is incurred; this occurs $\frac{w_0}{\sum_{i \in S^{U}}w_i(\mathcal{H})}$ rounds in expectation. When an unknown entrant is purchased a regret of $\textsc{OPT}-1$ is incurred; this happens once. As result, the epoch regret is
$$\textsc{EpochReg}(S,\mathcal{C}(\mathcal{H})) = (\textsc{OPT}-1)\frac{\sum_{i \in S^K} w_i(\mathcal{H})}{\sum_{i \in S^U} w_i(\mathcal{H})} + \textsc{OPT}\frac{w_0}{\sum_{i \in S^U} w_i(\mathcal{H})} + (\textsc{OPT}-1).$$
Thus minimizing $\textsc{EpochReg}(S,\mathcal{C}(\mathcal{H}))$ conditioned on $S^{U}$ is equivalent to maximizing the sum of the known products' attraction parameters $\sum_{i \in S^K} w_i(\mathcal{H})$, i.e., selecting the most attractive $c-|S^U|$ known products. 
\end{proof}

\subsection{Optimality of stationary and deterministic policies (Lemma~\ref{lemma:any_policy_which_minimizes_the_Q_function_is_optimal})}\label{appendix_subsec:proof_lemma_4.1_het_priors}

To prove Lemma~\ref{lemma:any_policy_which_minimizes_the_Q_function_is_optimal} we need several auxiliary lemmas similar to the proof of Lemma~\ref{lemma:epoch_optimal_policy_is_optimal}. Lemma~\ref{lemma:any_policy_has_a_nonnegative_regret_starting_from_a_terminal_history_and_pi_star_as_a_zero_regret_starting_from_a_terminal_history} continues to hold for the policy $\pi^{\star}$ in Lemma~\ref{lemma:any_policy_which_minimizes_the_Q_function_is_optimal} and its proof is the same. The following lemma (proven in Appendix~\ref{appendix_subsec:proof_lemma_B2B3_het_priors}) establishes that the regret of $\pi^{\star}$ can be expressed through the optimal regret function $V$.  This is analogous to Lemma~\ref{lemma:the_regret_of_pi_star_depends_on_H_only_through_J_H} as the regret depends on $\mathcal{H}$ through $\mathcal{C}(\mathcal{H})$ and to Lemma \ref{lemma:the_regret_of_pi_star_can_be_decomposed_into_the_first_epoch_regret_and_the_future_regret} as $V(\mathcal{C}(\mathcal{H}))$ corresponds to a one-epoch-lookahead decomposition.

\begin{lemma}\label{lemmaB2:different_priors}
    For any history $\mathcal{H}$, $\textsc{Reg}(\pi^{\star};\mathcal{H}) = V(\mathcal{C}(\mathcal{H}))$.
\end{lemma}

The following lemma (proven in Appendix~\ref{appendix_subsec:lemmaB4_her_priors_proof}) establishes that the regret of any policy $\pi$ given a history $\mathcal{H}$ is lower bounded by $V(\mathcal{C}(\mathcal{H}))$ and is the analogue of Lemma~\ref{lemma:epoch_regret_of_any_policy_is_at_least_the_minimum_epoch_regret_over_fixed_assortments}. Note that this is stronger than Lemma~\ref{lemma:epoch_regret_of_any_policy_is_at_least_the_minimum_epoch_regret_over_fixed_assortments} since it establishes the inequality in both cases. This is why an analogue of Lemma~\ref{lemma:for_any_history_with_k_unknown_entrants_the_regret_of_pi_star_is_lower_bounded_and_upper_bounded_by_quantities_which_depend_on_k} is not needed. A lemma analogous to Lemma~\ref{lemma:for_any_history_with_k_unknown_entrants_the_regret_of_pi_star_is_lower_bounded_and_upper_bounded_by_quantities_which_depend_on_k} is used in the proof of Lemma~\ref{lemmaB4:different_priors}. 

\begin{lemma}\label{lemmaB4:different_priors}
    For any policy $\pi$ and non-terminal history $\mathcal{H}$, $\textsc{Reg}(\pi;\mathcal{H}) \geq V(\mathcal{C}(\mathcal{H}))$.
\end{lemma}

\begin{proof}[Proof of Lemma~\ref{lemma:any_policy_which_minimizes_the_Q_function_is_optimal}.]
    For any policy $\pi$ it suffices to show $\textsc{Rev}(\pi^{\star};\mathcal{H}) \leq \textsc{Rev}(\pi;\mathcal{H})$ for any history $\mathcal{H}$. If $\mathcal{H}$ is terminal, then Lemma~\ref{lemma:any_policy_has_a_nonnegative_regret_starting_from_a_terminal_history_and_pi_star_as_a_zero_regret_starting_from_a_terminal_history} implies that $\textsc{Rev}(\pi^{\star};\mathcal{H}) = 0$ and $\textsc{Rev}(\pi;\mathcal{H}) \geq 0$. If $\mathcal{H}$ is non-terminal, then 
    $\textsc{Rev}(\pi^{\star};\mathcal{H}) = V(\mathcal{C}(\mathcal{H}))  \leq \textsc{Rev}(\pi;\mathcal{H})$; the equality uses Lemma~\ref{lemmaB2:different_priors} and the inequality uses Lemma~\ref{lemmaB4:different_priors}. 
\end{proof}

\subsection{Optimal policy regret depends only on the condensed history (Lemma~\ref{lemmaB2:different_priors}).}\label{appendix_subsec:proof_lemma_B2B3_het_priors}

\begin{proof}[Proof of Lemma~\ref{lemmaB2:different_priors}.]
We prove that $\textsc{Reg}(\pi^{\star};\mathcal{H}) = V(\mathcal{C}(\mathcal{H}))$ by induction on the number of unknown entrants remaining in $\mathcal{H}$. 

\paragraph{Base Case: $\mathcal{H}$ contains no unknown entrants.} Then 
$\textsc{Reg}(\pi^{\star};\mathcal{H})  = 0$ by Lemma~\ref{lemma:any_policy_has_a_nonnegative_regret_starting_from_a_terminal_history_and_pi_star_as_a_zero_regret_starting_from_a_terminal_history} as $\mathcal{H}$ is terminal and $V(\mathcal{C}(\mathcal{H})) = 0$ by the definition of $V$.

\paragraph{Induction Hypothesis:} $\textsc{Reg}(\pi^{\star};\mathcal{H}') = V(\mathcal{C}(\mathcal{H}'))$ for any history $\mathcal{H}'$ with at most $k$ unknown entrants.

\paragraph{Induction Step:} Let $\mathcal{H}$ be a history with $k+1$ unknown entrants. If $\mathcal{H}$ is terminal, then $\textsc{Reg}(\pi^{\star};\mathcal{H})  = 0$ by Lemma~\ref{lemma:any_policy_has_a_nonnegative_regret_starting_from_a_terminal_history_and_pi_star_as_a_zero_regret_starting_from_a_terminal_history} and $V(\mathcal{C}(\mathcal{H})) = 0$ by the definition of $V$. Suppose $\mathcal{H}$ is not terminal. Let $Z^{\mathcal{H}}_{\star}$ be the first round that an unknown entrant is purchased. Given that that $\pi^{\star}$ offers an unknown entrant at every round until $Z^{\mathcal{H}}_{\star}$ and an unknown entrant is purchased with probability at least $\frac{\min_{i \in \{1, \ldots, m\}} \tildew_i}{(c+1) w_{\textsc{max}}}$ where $w_{\textsc{max}} = \max\{w_i(\mathcal{H}): i \in \mathcal{N} \cup \{0\}\}$, then $\expect[Z^{\mathcal{H}}_{\star}] \leq \frac{(c+1) w_{\textsc{max}}}{\min_{i \in \{1, \ldots, m\}} \tildew_i}$. Thus, $\prob[Z^{\mathcal{H}}_{\star} < \infty] = 1$. As a result, the regret of $\pi^{\star}$ is expressed as
\begin{equation*}\label{eq:regret_pi_star_decompostion_homogeneous_priors}
    \textsc{Reg}(\pi^{\star};\mathcal{H}) = \underbrace{\expect \left[ \sum_{t=1}^{Z^{\mathcal{H}}_{\star}} \textsc{OPT}(\mathcal{H}_t) - \textsc{Rev}(S_t;\mathcal{H}_t) \middle|\mathcal{H}_1 = \mathcal{H}\right]}_{=\textsc{EpochReg}^{\star}} + \underbrace{\expect\left[\textsc{Reg}(\pi^{\star};\mathcal{H}_{Z^{\mathcal{H}}_{\star}} \cup (S_{Z^{\mathcal{H}}_{\star}}, y_{Z^{\mathcal{H}}_{\star}}, w_{Y_{Z^{\mathcal{H}}_{\star}}})) \middle|\mathcal{H}_1 = \mathcal{H} \right]}_{= \textsc{FutureReg}^{\star}} 
\end{equation*}

To express the first term, let $\mathcal{A}(\mathcal{H}) = \argmin\limits_{S \in \mathcal{E}(\mathcal{H})}Q(S;\mathcal{C}(\mathcal{H}))$ be the set of assortments minimizing the optimal regret function, $q(S;\mathcal{H}) = \frac{\sum_{i \in S \setminus I(\mathcal{H})} w_i(\mathcal{H})}{\sum_{i \in S} w_i(\mathcal{H}) + w_0}$ be the probability that an unknown entrant is purchased given $S$, and $N(S)$ be the number of rounds $S$ is offered until $Z^{\mathcal{H}}_{\star}$. Similar to equation \eqref{eq:expressing_regret_until_first_unknown_entrant_purchase} in the proof of Lemma~\ref{lemma:the_regret_of_pi_star_can_be_decomposed_into_the_first_epoch_regret_and_the_future_regret_general_form} (Appendix~\ref{appendix_subsec:proof_lemma_B9}), the regret until $Z^{\mathcal{H}}_{\star}$ can be expressed as
\begin{align*}\label{eq:regret_during_epoch}
   \textsc{EpochReg}^{\star}
    = \sum_{S \in \mathcal{A}(\mathcal{H})} \textsc{EpochReg}(S,\mathcal{C}(\mathcal{H}))  \underbrace{q(S;\mathcal{H})\expect[ N(S) |\mathcal{H}_1 = \mathcal{H}]}_{ = p(S; \mathcal{H})}
\end{align*}
The quantity $p(S; \mathcal{H}) = q(S;\mathcal{H})\expect[ N(S) |\mathcal{H}_1 = \mathcal{H}]$ is the expected number of times an unknown entrant is purchased following an offering from $S$ until $Z^{\mathcal{H}}_{\star}$ as $S$ is offered $\expect[ N(S) |\mathcal{H}_1 = \mathcal{H}]$ rounds in expectation and every time $S$ is offered, an unknown entrant is purchased independently with probability $q(S;\mathcal{H})$. Given that an unknown entrant is purchased only once during the epoch, $p(S;\mathcal{H})$ is the probability that an unknown entrant is purchased following an offering from $S$ during the epoch. Thus, the future regret is expressed as 
\begin{align*}
    \textsc{FutureReg}^{\star} &= \expect[V(\mathcal{C}(\mathcal{H}) \cup (Y_{Z^{\mathcal{H}}_{\star}}, w_{Y_{Z^{\mathcal{H}}_{\star}}}))|\mathcal{H}_1 = \mathcal{H}]\\
    &= \sum_{S \in \mathcal{A}(\mathcal{H})} p(S;\mathcal{H}) \sum_{i \in S^{\textsc{u}}(\mathcal{C}(\mathcal{H}))} \frac{\tildew_i}{\sum_{j \in S^{\textsc{u}}(\mathcal{C}(\mathcal{H})))} \tildew_j}\expect_{w_i \sim \mathcal{F}_i} [V(\mathcal{C}(\mathcal{H}) \cup (i, w_i))].
\end{align*}
The first equality uses the induction hypothesis on $\mathcal{H}' = \mathcal{H}_{Z^{\mathcal{H}}_{\star}} \cup (S_{Z^{\mathcal{H}}_{\star}}, y_{Z^{\mathcal{H}}_{\star}}, w_{Y_{Z^{\mathcal{H}}_{\star}}})$, noting that after the first purchase of the unknown entrant at round $Z^{\mathcal{H}}_{\star}$ the condensed history is augmented by $(Y_{Z^{\mathcal{H}}_{\star}}, w_{Y_{Z^{\mathcal{H}}_{\star}}})$. The second equality decomposes expected future regret over the chosen assortment $S \in \mathcal{A}(\mathcal{H})$ and the purchased entrant $i \in S^{\textsc{u}}(\mathcal{C}(\mathcal{H}))$. Therefore, the regret of $\pi^{\star}$ is expressed as
\begin{align*}
    &\textsc{Reg}(\pi^{\star};\mathcal{H}) = \textsc{EpochReg}^{\star}+ \textsc{FutureReg}^{\star}\\
    &=  \sum_{S \in \mathcal{A}(\mathcal{H})} p(S;\mathcal{H}) \Bigg[\textsc{EpochReg}(S,\mathcal{C}(\mathcal{H})) + \sum_{i \in S^{\textsc{u}}(\mathcal{C}(\mathcal{H}))} \frac{\tildew_i}{\sum_{j \in S^{\textsc{u}}(\mathcal{C}(\mathcal{H})))} \tildew_j}\expect_{w_i \sim \mathcal{F}_i} [V(\mathcal{C}(\mathcal{H}) \cup (i, w_i))] \Bigg]\\
    &=  \sum_{S \in \mathcal{A}(\mathcal{H})} p(S;\mathcal{H}) Q(S; \mathcal{C}(\mathcal{H})) = \left[  \sum_{S \in \mathcal{A}(\mathcal{H})} p(S;\mathcal{H})\right]\min_{S \in \mathcal{E}(\mathcal{C}(\mathcal{H}))} Q(S; \mathcal{C}(\mathcal{H})) = \min_{S \in \mathcal{E}(\mathcal{C}(\mathcal{H}))} Q(S; \mathcal{C}(\mathcal{H}))  = V(\mathcal{C}(\mathcal{H})).
\end{align*}
\end{proof}

\subsection{Regret of any policy is lower bounded by the value function (Lemma~\ref{lemmaB4:different_priors})}\label{appendix_subsec:lemmaB4_her_priors_proof}
To prove Lemma~\ref{lemmaB4:different_priors}, the following lemma lower and upper bounds the optimal regret function $V(\mathcal{C})$ by constants independent of $\mathcal{C}$ and is the analogue of Lemma~\ref{lemma:for_any_history_with_k_unknown_entrants_the_regret_of_pi_star_is_lower_bounded_and_upper_bounded_by_quantities_which_depend_on_k}. Given that entrants have heterogeneous priors, instead of the reciprocal of the customer prior attraction parameter $h(\mathcal{F})$ the upper bound uses the reciprocal of the minimum customer prior attraction parameter $\underline{h} = \min_{i \in \{1, \ldots, m\}} \tildew_i$. 

\begin{lemma}\label{lemmaB3:different_priors}
     For any condensed history $\mathcal{C}$, $V(\mathcal{C}) \geq -m$ and  $V(\mathcal{C})\leq \frac{m}{\underline{h}}$.
\end{lemma}
\begin{proof}[Proof of Lemma~\ref{lemmaB3:different_priors}.]
We show by induction on $k$ that for any history $\mathcal{H}$ with $k$ unknown entrants remaining, 
\begin{equation}\label{ineq:upper_and_lower_bound_claim_for_any_k_het_priors}
    V(\mathcal{C}) \geq -k \quad \text{ and } \quad V(\mathcal{C}) \leq \frac{k}{\overline{h}}
\end{equation}
For the base case $k = 0$, $V(\mathcal{C}) = 0$ by definition. For the induction step, suppose \eqref{ineq:upper_and_lower_bound_claim_for_any_k_het_priors} holds for $k$. If $\mathcal{C}$ is terminal, then $V(\mathcal{C}) = 0$ which satisfies the conditions. Suppose $\mathcal{C}$ is not terminal. Then by definition
\begin{equation*}\label{eq:decomposition_of_regret_of_pi_star_given_H_diff_priors}
Q(S;\mathcal{C}) = \textsc{EpochReg}(S,\mathcal{C}) + \underbrace{\sum_{i \in S^{\textsc{u}}(\mathcal{C})} \frac{\tildew_i}{ \sum_{j \in S^{\textsc{u}}(\mathcal{C})} \tildew_j} \expect_{w_i \sim \mathcal{F}_i}[ V(\mathcal{C} \cup (i, w_i))]}_{=(\star)}.
\end{equation*}
By the induction hypothesis \eqref{ineq:upper_and_lower_bound_claim_for_any_k_het_priors} , $(\star) \geq -k$ and  $(\star)\leq \frac{k}{\overline{h}}$. Given that $V(\mathcal{C}) = \min_{S \in \mathcal{E}(\mathcal{H})} Q(S;\mathcal{C})$, to finish the proof it suffices to show  $\textsc{EpochReg}(S;\mathcal{C})  \geq -1$ and $\textsc{EpochReg}(S;\mathcal{C}) \leq \frac{1}{\overline{h}}$. The former inequality is shown in the same way as \eqref{ineq:epoch_reg_lower_bound} in Lemma~\ref{lemma:for_any_history_with_k_unknown_entrants_the_regret_of_pi_star_is_lower_bounded_and_upper_bounded_by_quantities_which_depend_on_k} (Appendix~\ref{appendix_subsec:regret_of_optimal_policy_is_bounded_below_and_above}). The latter is shown in the same way as \eqref{ineq:epoch_reg_upper_bound} in Lemma~\ref{lemma:for_any_history_with_k_unknown_entrants_the_regret_of_pi_star_is_lower_bounded_and_upper_bounded_by_quantities_which_depend_on_k} (Appendix~\ref{appendix_subsec:regret_of_optimal_policy_is_bounded_below_and_above}) where the last inequality is replaced with $\frac{1}{\sum_{i \in S^{U}} w_i(\mathcal{H})} \leq \frac{1}{\overline{h}}$.
\end{proof}

\begin{proof}[Proof of Lemma~\ref{lemmaB4:different_priors}. ]
We proceed by induction on the number of unknown entrants in~$\mathcal{H}$. 
    \paragraph{Base case.}  If $\mathcal{H}$ has no unknown entrants, $\textsc{Reg}(\pi;\mathcal{H}) \geq 0$ by Lemma~\ref{lemma:any_policy_has_a_nonnegative_regret_starting_from_a_terminal_history_and_pi_star_as_a_zero_regret_starting_from_a_terminal_history} and $V(\mathcal{C}(\mathcal{H})) = 0$ by definition. 
    
    \paragraph{Induction Hypothesis.} $\textsc{Reg}(\pi;\mathcal{H}') \geq V(\mathcal{C}(\mathcal{H}')) $ for any history $\mathcal{H}'$ with at most $k$ unknown entrants. 
    
    \paragraph{Induction Step.} Let $\mathcal{H}$ be a history with $k+1$ unknown entrants. If $\mathcal{H}$ is terminal, then $\textsc{Reg}(\pi;\mathcal{H}) \geq 0$ by Lemma~\ref{lemma:any_policy_has_a_nonnegative_regret_starting_from_a_terminal_history_and_pi_star_as_a_zero_regret_starting_from_a_terminal_history} and $V(\mathcal{C}(\mathcal{H})) = 0$ by definition. If $\mathcal{H}$ is not terminal, let $Z$ be the round at which the first unknown entrant is purchased and $N(S)$ be the number of rounds $S$ is offered until $Z$. Then
\begin{align}\label{ineq:regret_pi_lower_bound_different_priors}
      &\textsc{Reg}(\pi;\mathcal{H}) = \expect\left[\sum_{t=1}^{Z} \textsc{OPT}(\mathcal{H}_t) - \textsc{Rev}(S_t; \mathcal{H}_t) \middle|\mathcal{H}_1 = \mathcal{H}\right] \nonumber\\
      &\hspace{0.65in}+\expect \left[\textsc{Reg}(\pi;\mathcal{H}_{Z} \cup (S_{Z}, Y_{Z}, w_{Y_{Z}})) \middle|\mathcal{H}_1 = \mathcal{H}, Z< \infty \right] \prob[Z < \infty] \nonumber\\
&\geq \underbrace{\sum\limits_{\substack{S \subseteq \mathcal{N}, |S| \leq c}}(\textsc{OPT}(\mathcal{H}) - \textsc{Rev}(S; \mathcal{H})) \expect[N(S)|\mathcal{H}_1 = \mathcal{H}]}_{= \textsc{EpochReg}^{\pi}(\mathcal{H})}+\underbrace{\expect [V(\mathcal{C}(\mathcal{H}) \cup ( Y_{Z}, w_{Y_{Z}}))|\mathcal{H}_1 = \mathcal{H}]}_{=\textsc{FutureReg}^{\pi}(\mathcal{H})} \prob[Z < \infty] 
\end{align}
The first term of the inequality expresses the regret until the first entrant purchase using the same steps as in equation \eqref{eq:regret_until_Z_expression_as_sum_over_assortments} in the proof of Lemma~\ref{lemma:epoch_regret_of_any_policy_is_at_least_the_minimum_epoch_regret_over_fixed_assortments} (Appendix~\ref{appendix_subsec:proof_lemma_epoch_regret_of_any_policy_is_at_least_minimum_epoch_regret_of_static_assortment_policies}). The second term of the inequality applies the induction hypothesis $\mathcal{H}' = \mathcal{H}_{Z} \cup (S_{Z}, Y_{Z}, w_{Y_{Z}})$, noting that after the first purchase of the unknown entrant at round $Z$, the condensed history is augmented by $( Y_{Z}, w_{Y_{Z}})$.

\paragraph{Case 1: $\expect[Z] = +\infty$.} We follow the same steps as in Case 1 in the proof of Lemma~\ref{lemma:epoch_regret_of_any_policy_is_at_least_the_minimum_epoch_regret_over_fixed_assortments} we conclude $\textsc{EpochReg}^{\pi}(\mathcal{H}) = +\infty$. Given that $Z = \sum\limits_{S \subseteq \mathcal{N}, |S| \leq c} N(S)$, there exists some $\tilde{S}$ with $\expect[N(\tilde{S})|\mathcal{H}_1 = \mathcal{H}] = +\infty$.
For any $S \in \mathcal{E}(\mathcal{H})$ it holds that $\expect[N(S)|\mathcal{H}_1 = \mathcal{H}] \leq \frac{(c+1) w_{\textsc{max}}}{\underline{h}}$ as an unknown entrant is purchased with probability at least $\frac{\underline{h}}{(c+1) w_{\textsc{max}}}$ where $w_{\textsc{max}} = \max\{w_i(\mathcal{H}): i \in \mathcal{N} \cup \{0\}\}$.  As a result $\tilde{S}$ contains only known products, i.e., $\tilde{S} \not \in \mathcal{E}(\mathcal{H})$. Given that $\mathcal{H}$ is non-terminal, $\textsc{OPT}(\mathcal{H}) - \textsc{Rev}(\tilde{S}, \mathcal{H}) > 0$, which together with $\expect[N(\tilde{S})|\mathcal{H}_1 = \mathcal{H}] = +\infty$ and \eqref{ineq:regret_pi_lower_bound_different_priors} implies $\textsc{EpochReg}^{\pi}(\mathcal{H}) = +\infty$. Combining this with $\textsc{FutureReg}^{\pi}(\mathcal{H}) \prob[Z < \infty] \geq -m$ by Lemma~\ref{lemmaB3:different_priors} and \eqref{ineq:regret_pi_lower_bound_different_priors} yields $\textsc{Reg}(\pi;\mathcal{H}) = +\infty$. Given that $V(\mathcal{C}(\mathcal{H})) < \infty$ (by Lemma~\ref{lemmaB3:different_priors}), $\textsc{Reg}(\pi;\mathcal{H}) \geq V(\mathcal{C}(\mathcal{H}))$.

\paragraph{Case 2: $\expect[Z] < +\infty$. } In this case, $\prob[Z < \infty] = 1$. Let $q(S;\mathcal{H}) = \frac{\sum_{i \in S \setminus I(\mathcal{H})} w_i(\mathcal{H})}{\sum_{i \in S} w_i(\mathcal{H}) + w_0}$ be the probability that an unknown entrant is purchased given $S$ and $p(S;\mathcal{H}) = q(S;\mathcal{H}) \expect[N(S)|\mathcal{H}_1 = \mathcal{H}]$ be the probability the epoch ends following an offering of $S$. By the same steps as in \eqref{ineq:lower_bound_on_the_regret_of_the_first_epoch_of_any_policy} in the proof of Lemma~\ref{lemma:epoch_regret_of_any_policy_is_at_least_the_minimum_epoch_regret_over_fixed_assortments} (Appendix~\ref{appendix_subsec:proof_lemma_epoch_regret_of_any_policy_is_at_least_minimum_epoch_regret_of_static_assortment_policies})
$$\textsc{EpochReg}^{\pi}(\mathcal{H})\geq  \sum\limits_{S \in \mathcal{E}(\mathcal{H})}\textsc{EpochReg}(S,\mathcal{C}(\mathcal{H}))p(S;\mathcal{H}).$$ 
Following the same steps as in the proof of the previous Lemma~\ref{lemmaB2:different_priors} (Appendix~\ref{appendix_subsec:proof_lemma_B2B3_het_priors}), 
$$\textsc{FutureReg}^{\pi}(\mathcal{H})= \sum_{S \in \mathcal{E}(\mathcal{H})} p(S;\mathcal{H})\sum_{i \in S^{\textsc{u}}(\mathcal{C}(\mathcal{H}))} \frac{\tildew_i}{\sum_{j \in S^{\textsc{u}}(\mathcal{C}(\mathcal{H})))} \tildew_j}\expect_{w_i \sim \mathcal{F}_i} [V(\mathcal{C}(\mathcal{H}) \cup (i, w_i))].$$
As a result, the regret of $\pi$ is lower bounded by
\begin{align*}
    &\textsc{Reg}(\pi;\mathcal{H}) \geq \textsc{EpochReg}^{\pi}(\mathcal{H}) + \textsc{FutureReg}^{\pi}(\mathcal{H})\\
    &\geq 
    \sum\limits_{S \in \mathcal{E}(\mathcal{H})}p(S;\mathcal{H}) \left[\textsc{EpochReg}(S,\mathcal{C}(\mathcal{H})) + 
    \sum_{i \in S^{\textsc{u}}(\mathcal{C}(\mathcal{H}))} \frac{\tildew_i}{\sum_{j \in S^{\textsc{u}}(\mathcal{C}(\mathcal{H})))} \tildew_j}\expect_{w_i \sim \mathcal{F}_i} [V(\mathcal{C}(\mathcal{H}) \cup (i, w_i))] \right]\\
    &= \sum\limits_{S \in \mathcal{E}(\mathcal{H})} p(S;\mathcal{H}) Q(S; \mathcal{C}(\mathcal{H})) \geq \Big(\sum\limits_{S \in \mathcal{E}(\mathcal{H})} p(S;\mathcal{H}) \Big) \Big(\min\limits_{S \in \mathcal{E}(\mathcal{H})}  Q(S; \mathcal{C}(\mathcal{H})) \Big) = \min\limits_{S \in \mathcal{E}(\mathcal{H})}  Q(S; \mathcal{C}(\mathcal{H})) = V(\mathcal{C}(\mathcal{H}))
\end{align*}
where the third equality uses $\sum\limits_{S \in \mathcal{E}(\mathcal{H})} p(S;\mathcal{H})  = 1$. This completes the induction step and the proof.
\end{proof}

\section{Single entrant and noisy observations (Section~\ref{sec:single_entrant_noisy_obs})}
\subsection{Optimal policy with a single entrant and noisy observations (Theorem~\ref{theorem:single_entrant_non-full_realizability})}\label{appendix_subsec:proof_characterization_opt_policy_single_entrant_noisy_obs}

To prove Theorem~\ref{theorem:single_entrant_non-full_realizability}, in Lemma~\ref{lemma:analogue_lemma4.1_noisy_observations} we show that any policy which minimizes the regret until the next entrant review is optimal. To do so, we partition the time horizon into \textit{epochs}; within each epoch the set of reviews for the entrant remains the same. Let $\mathcal{E} = \{ S \subseteq \mathcal{N}: |S| \leq c, 1 \in S\}$ be the set of assortments containing the entrant. A history $\mathcal{H}$ is terminal if it contains at least $k$ reviews for the entrant. Let $\tau(S;\mathcal{H})$ and $r(S;\mathcal{H})$ be (respectively) the expected number of rounds and total expected reward in an epoch if assortment $S \in \mathcal{E}$ is offered repeatedly given a non-terminal history $\mathcal{H}$. We refer to the expected regret in an epoch if $S$ is offered repeatedly given a non-terminal history $\mathcal{H}$ as \emph{epoch regret} and denote it by 
$$\textsc{EpochReg}(S;\mathcal{H}) =\textsc{OPT}(\mathcal{H}) \tau(S;\mathcal{H}) - r(S;\mathcal{H}) \quad \text{ where } \quad \textsc{OPT}(\mathcal{H}) = \expect[\textsc{OPT}|\mathcal{H}].$$
Recall that $\mathcal{F}(\mathcal{H})$ is the posterior of the entrant given a non-terminal history $\mathcal{H}$. Let $w_1(\mathcal{H}) = h(\mathcal{F}(\mathcal{H}))$ if $\mathcal{H}$ is non-terminal and $w_1(\mathcal{H}) = w_1$ otherwise be the attraction parameter of the entrant. Recall that $w_i(\mathcal{H}) = w_i$ is the attraction parameter of incumbent $i$ for $i \in \{2, \ldots, n\}$ and $\textsc{Rev}(S;\mathcal{H}) = \frac{\sum_{i \in S} w_i(\mathcal{H})}{\sum_{i \in S} w_i(\mathcal{H}) + w_0}$ is the expected revenue of $S$ given a history $\mathcal{H}$.

The following lemma establishes that a policy is optimal if it selects an assortment $S$ minimizing the epoch regret if the history is not terminal, and selects the revenue-maximizing assortment otherwise. This lemma is the analogue of Lemma~\ref{lemma:epoch_optimal_policy_is_optimal} and is proven in Appendix~\ref{appendix_subsec:proof_F1}. 

\begin{lemma}\label{lemma:analogue_lemma4.1_noisy_observations}
A policy $\pi^{\star}$ is optimal if
 \begin{itemize}
     \item $\pi^{\star}(\mathcal{H}) \in \argmin\limits_{S \in \mathcal{E}} \textsc{EpochReg}(S;\mathcal{H})$ if $\mathcal{H}$ is not terminal, and
     \item $\pi^{\star}(\mathcal{H}) \in \argmax\limits_{S \subseteq \mathcal{N}, |S| \leq c} \textsc{Rev}(S; \mathcal{H})$ otherwise.
 \end{itemize}
\end{lemma}

\begin{proof}[Proof of Theorem~\ref{theorem:single_entrant_non-full_realizability}.]
Let $\mathcal{H}$ be a non-terminal history. For a set of at most $c-1$ incumbents $S$, let $o(S;\mathcal{H})$ be the expected number of rounds in which the outside option is chosen in the epoch if the entrant is repeatedly offered with $S$. When the outside option is chosen the reward is $0$ and the regret incurred is $\textsc{OPT}(\mathcal{H})$. When a product is chosen the reward is $1$ and the regret incurred is $\textsc{OPT}(\mathcal{H})-1$; this happens $\tau(S;\mathcal{H})-o(S;\mathcal{H})$ rounds in expectation. Thus, the expected regret until the next review is
\begin{equation}\label{eq:epoch_regret_non-full_realizability_single_entrant}
    \textsc{EpochReg}(S \cup \{1\};\mathcal{H}) = (\textsc{OPT}(\mathcal{H})-1)(\tau(S;\mathcal{H})-o(S;\mathcal{H})) +\textsc{OPT}(\mathcal{H}) o(S;\mathcal{H}).
\end{equation}
Next, we show that $o(S;\mathcal{H})$ is independent of $S$.  Let $\mathbf{u} = (u_1, \ldots, u_{\ell})$ be the reviews for the entrant in $\mathcal{H}$ and $\mathcal{F}_{\mathbf{u}}$ be the posterior distribution for the entrant given $\mathbf{u}$. Given that the entrant is purchased with probability $\frac{h(\mathcal{F}_{\mathbf{u}})}{\sum_{i \in S} w_i + h(\mathcal{F}_{\mathbf{u}}) + w_0}$ every time it is offered with $S$ in the epoch,
\begin{equation}\label{eq:expression_tau_single_entrant_noisy_obs}
    \tau(S;\mathcal{H}) = \frac{1}{\prob[\text{entrant is chosen}|S]} = \frac{\sum_{i \in S} w_i + h(\mathcal{F}_{\mathbf{u}}) + w_0}{h(\mathcal{F}_{\mathbf{u}})}.
\end{equation}
Using this we can express $o(S;\mathcal{H})$ as
\begin{equation*}
    o(S;\mathcal{H}) = \prob[\text{outside option chosen}|S] \tau(S;\mathcal{H}) 
    = \frac{\prob[\text{outside option chosen}|S]}{\prob[\text{entrant chosen}|S]}= \frac{\frac{w_0}{\sum_{i \in S} w_i  + h(\mathcal{F}_{\mathbf{u}}) + w_0} }{\frac{h(\mathcal{F}_{\mathbf{u}})}{\sum_{i \in S} w_i  + h(\mathcal{F}_{\mathbf{u}}) + w_0} }=
    \frac{w_0}{h(\mathcal{F}_{\mathbf{u}})}.
\end{equation*}
Combining this with \eqref{eq:epoch_regret_non-full_realizability_single_entrant}, minimizing $ \textsc{EpochReg}(S \cup \{1\};\mathcal{H})$ is equivalent to maximizing $\tau(S;\mathcal{H})$ (as $\textsc{OPT}(\mathcal{H}) < 1$). By \eqref{eq:expression_tau_single_entrant_noisy_obs}, $\tau(S;\mathcal{H})$ is maximized when $\sum_{i \in S} w_i$ is maximized, which occurs when $S =  S^{\star}_{(c-1)}(\mathcal{I}_1)$. 

Therefore, offering the assortment $S^{\star}_{(c-1)}(\mathcal{I}_1) \cup \{1\}$ minimizes $\textsc{EpochReg}(S;\mathcal{H})$ when $\mathcal{H}$ is not terminal. Combining this with Lemma~\ref{lemma:analogue_lemma4.1_noisy_observations}, it is optimal to offer the entrant with the most attractive $c-1$ products until the $k$-th review and offer the $c$ most attractive products in subsequent rounds, which yields the lemma.
\end{proof}

\subsection{Optimality of stationary and deterministic policies (Lemma~\ref{lemma:analogue_lemma4.1_noisy_observations})}\label{appendix_subsec:proof_F1}
Similar to the proof of Lemma~\ref{lemma:epoch_optimal_policy_is_optimal}, we use several auxiliary lemmas. The proof follows the same structure as that of Lemma~\ref{lemma:epoch_optimal_policy_is_optimal} in Appendix~\ref{appendix_sec:there_exists_optimal_policy_which_is_stationary}, where the role of a new entrant purchase is replaced with by the role of new review for the entrant.

Lemma~\ref{lemma:any_policy_has_a_nonnegative_regret_starting_from_a_terminal_history_and_pi_star_as_a_zero_regret_starting_from_a_terminal_history} continues to hold for the policy $\pi^{\star}$ in Lemma~\ref{lemma:analogue_lemma4.1_noisy_observations} and its proof (provided in Appendix~\ref{appendix_subsec:proof_lemma_B1_single_entrant_noisy_obs}) is simpler since a history can be terminal only when all products are known. For a history $\mathcal{H}$, let $p(\mathcal{H})$ be the number of times the entrant has been purchased and let $\mathcal{C}(\mathcal{H}) = (u_1, \ldots, u_{p(\mathcal{H})})$ be the condensed history consisting of the reviews obtained for the unknown entrant. The following lemma establishes that the regret of $\pi^{\star}$ given a history $\mathcal{H}$ depends on $\mathcal{H}$ only through $\mathcal{C}(\mathcal{H})$ and provides a one-epoch-lookahead decomposition for the regret of $\pi^{\star}$. This lemma is the analogue of Lemma~\ref{lemma:the_regret_of_pi_star_depends_on_H_only_through_J_H} (and it is proven in Appendix~\ref{appendix_subsec:proof_lemma_B2_single_entrant_noisy}).

\begin{lemma}\label{lemmaB2:noisy_observations}
    For any history $\mathcal{H}$, $\textsc{Reg}(\pi^{\star};\mathcal{H})$ depends on $\mathcal{H}$ only through $\mathcal{C}(\mathcal{H})$. 
\end{lemma}    
Hence, the regret of the optimal policy depends only on the reviews for the unknown entrant and we can express the regret of $\pi^{\star}$ given a history $\mathcal{H}$ as $\textsc{Reg}(\pi^{\star},\mathcal{C}(\mathcal{H})) = \textsc{Reg}(\pi^{\star};\mathcal{H})$. The next lemma (proof in Appendix~\ref{appendix_subsec:lemma_B3_single_entrant_noisy}) provides a one-epoch-lookahead decomposition of the regret of~$\pi^{\star}$, similar to Lemma~\ref{lemma:the_regret_of_pi_star_can_be_decomposed_into_the_first_epoch_regret_and_the_future_regret}.
\begin{lemma}\label{lemma:the_regret_of_pi_star_can_be_decomposed_into_the_first_epoch_regret_and_the_future_regret_single_entrant_noisy}
     For any non-terminal history $\mathcal{H}$, 
    $$\textsc{Reg}(\pi^{\star},\mathcal{C}(\mathcal{H})) = \min\limits_{ S \in \mathcal{E}}\textsc{EpochReg}(S;\mathcal{H}) + \expect_{u \sim \bern(w_1)}[ \textsc{Reg}(\pi^{\star},\mathcal{C}(\mathcal{H}) \cup \{u\})].$$
\end{lemma}
Let $Z^{\pi, \mathcal{H}}$ be the first round at which the next review for the entrant is obtained if a policy $\pi$ starts from a history $\mathcal{H}$. The next lemma establishes that the expected regret of $\pi$ until $Z^{\pi, \mathcal{H}}$ is at least the minimum epoch regret achieved by repeatedly offering a fixed assortment. This lemma is the analogue of Lemma~\ref{lemma:epoch_regret_of_any_policy_is_at_least_the_minimum_epoch_regret_over_fixed_assortments} and is proven in Appendix~\ref{appendix_subsec:proof_lemma_B4_single_entrant_noisy}. Let $\textsc{EpochReg}^{\pi}(\mathcal{H}) = \expect\Bigg[\sum_{t=1}^{Z^{\pi, \mathcal{H}}} \textsc{OPT}(\mathcal{H}_t) - \textsc{Rev}(S_t, \mathcal{H}_t)|\mathcal{H}_1 = \mathcal{H}\Bigg]$ be the expected regret of $\pi$ until the first purchase of an unknown entrant.

\begin{lemma}\label{lemma:epoch_regret_of_any_policy_is_at_least_the_minimum_epoch_regret_over_fixed_assortments_single_entrant_noisy}
    For any policy $\pi$ and non-terminal history $\mathcal{H}$, 
    $\textsc{EpochReg}^{\pi}(\mathcal{H}) \geq \min\limits_{S \in \mathcal{E}}  \textsc{EpochReg}(S;\mathcal{H})$
    if $\expect[Z^{\pi, \mathcal{H}}] < \infty$, and $\textsc{EpochReg}^{\pi}(\mathcal{H}) = +\infty$ otherwise. 
\end{lemma}

The next lemma (proof in Appendix~\ref{appendix_subsec:regret_of_optimal_policy_is_bounded_below_and_above_single_entrant_noisy}) establishes that the regret of $\pi^{\star}$ for any history $\mathcal{H}$ is bounded from below and from above by constants independent of $\mathcal{H}$. This lemma is the analogue of Lemma~\ref{lemma:for_any_history_with_k_unknown_entrants_the_regret_of_pi_star_is_lower_bounded_and_upper_bounded_by_quantities_which_depend_on_k}. Let $\underline{h} = \min\limits_{(x,y): x+ y \leq k-1} h(\mathrm{Beta}(a + x, b + y))$ be the minimum possible value of the attraction parameter of the entrant until the $k$-th review. Given that the entrant can have a prior $h(\mathrm{Beta}(a + x, b + y))$ for any $(x,y)$ with $x+ y \leq k-1$, instead of the reciprocal of the customer prior attraction parameter $h(\mathcal{F})$ the upper bound uses the reciprocal of the minimum customer prior attraction parameter $\underline{h} = \min\limits_{(x,y): x+ y \leq k-1} h(\mathrm{Beta}(a + x, b + y))$.

\begin{lemma}\label{lemmaB3:noisy_observations}
     For any history $\mathcal{H}$, $\textsc{Reg}(\pi^{\star};\mathcal{H}) \geq -m$ and $\textsc{Reg}(\pi^{\star};\mathcal{H}) \leq \frac{m}{\underline{h}}$.
\end{lemma}

\begin{proof}[Proof of Lemma~\ref{lemma:analogue_lemma4.1_noisy_observations}.]
The proof combines Lemma~\ref{lemma:any_policy_has_a_nonnegative_regret_starting_from_a_terminal_history_and_pi_star_as_a_zero_regret_starting_from_a_terminal_history}, Lemma~\ref{lemmaB2:noisy_observations},  Lemma~\ref{lemma:the_regret_of_pi_star_can_be_decomposed_into_the_first_epoch_regret_and_the_future_regret_single_entrant_noisy},  Lemma~\ref{lemma:epoch_regret_of_any_policy_is_at_least_the_minimum_epoch_regret_over_fixed_assortments_single_entrant_noisy}, and  Lemma~\ref {lemmaB3:noisy_observations} and follows the same steps as the proof of Lemma~\ref{lemma:epoch_optimal_policy_is_optimal} with the following changes: 
\begin{itemize}
\item Lemma~\ref{lemma:the_regret_of_pi_star_depends_on_H_only_through_J_H} is replaced with Lemma~\ref{lemmaB2:noisy_observations}, Lemma~\ref{lemma:the_regret_of_pi_star_can_be_decomposed_into_the_first_epoch_regret_and_the_future_regret} is replaced with Lemma~\ref{lemma:the_regret_of_pi_star_can_be_decomposed_into_the_first_epoch_regret_and_the_future_regret_single_entrant_noisy}, Lemma~\ref{lemma:epoch_regret_of_any_policy_is_at_least_the_minimum_epoch_regret_over_fixed_assortments} is replaced with Lemma~\ref{lemma:epoch_regret_of_any_policy_is_at_least_the_minimum_epoch_regret_over_fixed_assortments_single_entrant_noisy}, and Lemma~\ref{lemma:for_any_history_with_k_unknown_entrants_the_regret_of_pi_star_is_lower_bounded_and_upper_bounded_by_quantities_which_depend_on_k} is replaced with Lemma~\ref {lemmaB3:noisy_observations}.
\item Instead of showing $\textsc{Reg}(\pi;\mathcal{H}) \geq \textsc{Reg}(\pi^{\star};\mathcal{H})$ by induction on the number of unknown entrants in $\mathcal{H}$, we show $\textsc{Reg}(\pi;\mathcal{H}) \geq \textsc{Reg}(\pi^{\star};\mathcal{H})$ by induction on the number of entrant reviews renaming in $\mathcal{H}$ until the entrant is known. 
    \item Instead of $Z$ being the first round an unknown entrant is purchased, $Z$ is the first round a new review for the entrant is obtained. 
    \item In the future regret term $\expect[\textsc{Reg}(\pi;\mathcal{H}_{Z} \cup (S_{Z}, Y_{Z}, w_{Y_{Z}}))|\mathcal{H}_1 = \mathcal{H}, Z< \infty]$, the unknown entrant purchase is replaced with the next entrant review, i.e.,
    $\expect_{u \sim \bern(w_1)}[\textsc{Reg}(\pi;\mathcal{H}_{Z} \cup (S_{Z}, 1, u))|\mathcal{H}_1 = \mathcal{H}, Z< \infty]$.
    Similarly, for $\pi^{\star}$, $\expect_{w \sim \mathcal{F}}[\textsc{Reg}(\pi^{\star},J(\mathcal{H}) \cup \{w\})]$ is replaced with $\expect\limits_{u\sim \bern(w_1)} [\textsc{Reg}(\pi^{\star}, \mathcal{C}(\mathcal{H}) \cup \{u\})]$.
\end{itemize}
\end{proof}

\subsection{Any policy yields non-negative regret given a terminal history (Lemma~\ref{lemma:any_policy_has_a_nonnegative_regret_starting_from_a_terminal_history_and_pi_star_as_a_zero_regret_starting_from_a_terminal_history})}\label{appendix_subsec:proof_lemma_B1_single_entrant_noisy_obs}
\begin{proof}[Updated proof of Lemma~\ref{lemma:any_policy_has_a_nonnegative_regret_starting_from_a_terminal_history_and_pi_star_as_a_zero_regret_starting_from_a_terminal_history}.]
We show that Lemmas \ref{lemma:if_H_1_is_terminal_H_t_is_terminal_with_probability_1} and \ref{lemma:if_H_is_terminal_the_expected_opt_equals_the_maximum_instantaneous_revenue} continue to hold. Given Lemmas \ref{lemma:if_H_1_is_terminal_H_t_is_terminal_with_probability_1} and \ref{lemma:if_H_is_terminal_the_expected_opt_equals_the_maximum_instantaneous_revenue}, the updated proof of Lemma~\ref{lemma:any_policy_has_a_nonnegative_regret_starting_from_a_terminal_history_and_pi_star_as_a_zero_regret_starting_from_a_terminal_history} follows the same steps as the proof described in Appendix~\ref{appendix_subsec:proof_lemma_any_policy_has_nonegative_regret_given_terminal_history}. 

If $\mathcal{H}_1$ is terminal, the entrant has at least $k$ purchases; thus for any round $t$ the entrant has at least $k$ purchases, i.e., $\mathcal{H}_t$ is terminal. Hence, Lemma \ref{lemma:if_H_1_is_terminal_H_t_is_terminal_with_probability_1} holds. If $\mathcal{H}$ is terminal, the entrant's attraction parameter is known and as a result $\textsc{OPT}(\mathcal{H}) = \textsc{OPT} = \max\limits_{S \subseteq \mathcal{N}, |S| \leq c}\textsc{Rev}(S;\mathcal{H})$.
Hence, Lemma \ref{lemma:if_H_is_terminal_the_expected_opt_equals_the_maximum_instantaneous_revenue} holds.
\end{proof}

\subsection{Optimal policy regret depends only on entrant reviews (Lemma~\ref{lemmaB2:noisy_observations})}\label{appendix_subsec:proof_lemma_B2_single_entrant_noisy}
The proof follows a similar structure to the proof of Lemma~\ref{lemma:the_regret_of_pi_star_depends_on_H_only_through_J_H} (in Appendix~\ref{appendix_subsec:proof_lemma_optimal_regret_depends_only_on_known_products_attraction_parameters}). The following lemma establishes that the epoch regret of any assortment $S$ and history $\mathcal{H}$ depends on $\mathcal{H}$ only through the entrant reviews $\mathcal{C}(\mathcal{H})$. This lemma is a simpler analogue of Lemma~\ref{lemma:for_any_history_H_the_expected_ex_post_optimum_the_maximum_revenue_from_known_products_and_the_minimum_epoch_regret_depend_on_H_only_through_J(H)}. 
\begin{lemma}\label{lemma:B8_single_entrant_noisy}
    For any non-terminal history $\mathcal{H}$ and assortment $S \in \mathcal{E}$, $\textsc{EpochReg}(S;\mathcal{H})$ depends on $\mathcal{H}$ only through $\mathcal{C}(\mathcal{H})$.
\end{lemma}

\begin{proof}[Proof of Lemma~\ref{lemma:B8_single_entrant_noisy}.]
The entrant posterior given $\mathcal{H}$ depends on $\mathcal{H}$ only through the reviews $\mathcal{C}(\mathcal{H})$; as a result we denote it by $\mathcal{F}(\mathcal{C}(\mathcal{H}))$. The customer's attraction parameter for the entrant is $w_1(\mathcal{C}(\mathcal{H})) = h(\mathcal{F}(\mathcal{C}(\mathcal{H})))$. If the platform offers a set $S = S^{K} \cup \{1\}$ containing the entrant repeatedly, every known product $i \in S^{K}$ is purchased $\frac{w_i}{w_1(\mathcal{C}(\mathcal{H}))}$ rounds in expectation yielding a regret of $\textsc{OPT}-1$ per round, the outside option is chosen $\frac{w_0}{w_1(\mathcal{C}(\mathcal{H}))}$ rounds in expectation yielding a regret of $\textsc{OPT}$ per round, and the entrant is chosen once yielding a regret of $\textsc{OPT}-1$. As a result the epoch regret is 
$$\textsc{EpochReg}(S;\mathcal{H}) = (\textsc{OPT}-1)\frac{\sum_{i \in S^k}w_i}{w_1(\mathcal{C}(\mathcal{H}))} + \textsc{OPT}\frac{w_0}{w_1(\mathcal{C}(\mathcal{H}))}  + (\textsc{OPT}-1).$$
Given that the right hand side depends on $\mathcal{H}$ only through $\mathcal{C}(\mathcal{H})$ the lemma follows.
\end{proof}

The next lemma provides a one-epoch-lookahead decomposition for the regret of $\pi^{\star}$. For a history $\mathcal{H}$, let $Z^{\mathcal{H}}_{\star}$ be the first round the next entrant review is obtained if $\pi^{\star}$ starts from a history $\mathcal{H}$. This is the analogue of Lemma~\ref{lemma:the_regret_of_pi_star_can_be_decomposed_into_the_first_epoch_regret_and_the_future_regret_general_form}. Let $\mathcal{H}_{Z^{\mathcal{H}}_{\star}}$ and $S_{Z^{\mathcal{H}}_{\star}}$ be the history and the assortment offered at round $Z^{\mathcal{H}}_{\star}$ respectively. 

\begin{lemma}\label{lemma:B9_single_entrant_noisy}
    For any non-terminal history $\mathcal{H}$, the regret of $\pi^{\star}$ can be decomposed as
     $$\textsc{Reg}(\pi^{\star};\mathcal{H}) = \min\limits_{S \in \mathcal{E}} \textsc{EpochReg}(S;\mathcal{H}) \nonumber + \expect_{u \sim \bern(w_1)}[\textsc{Reg}(\pi^{\star};\mathcal{H}_{Z^{\mathcal{H}}_{\star}} \cup (S_{Z^{\mathcal{H}}_{\star}}, 1, u))|\mathcal{H}_1 = \mathcal{H}].$$
\end{lemma}

\begin{proof}[Proof of Lemma~\ref{lemma:B9_single_entrant_noisy}.]
The proof follows the same steps as the proof of Lemma~\ref{lemma:the_regret_of_pi_star_can_be_decomposed_into_the_first_epoch_regret_and_the_future_regret_general_form} (in Appendix~\ref{appendix_subsec:proof_lemma_B9}) with two modifications. First, $\expect[Z^{\mathcal{H}}_{\star}] < \frac{(c+1) w_{\text{max}}}{h(\mathcal{F})}$ is replaced with $\expect[Z^{\mathcal{H}}_{\star}] < \frac{(c+1) w_{\text{max}}}{\underline{h}}$ as the entrant is purchased with probability at least $\frac{\underline{h}}{(c+1) w_{\text{max}}}$ at every round where $w_{\text{max}} = \max\{w_i(\mathcal{H}): i \in \mathcal{N} \cup \{0\}\}$. Second, in the future regret term $\expect[\textsc{Reg}(\pi^{\star};\mathcal{H}_{Z^{\mathcal{H}}_{\star}} \cup (S_{Z^{\mathcal{H}}_{\star}}, Y_{Z^{\mathcal{H}}_{\star}}, w_{Y_{Z^{\mathcal{H}}_{\star}}}))|\mathcal{H}_1 = \mathcal{H}]$, the unknown entrant purchase         is replaced with the next entrant review, i.e., 
        $\expect_{u \sim \bern(w_1)}[\textsc{Reg}(\pi^{\star};\mathcal{H}_{Z^{\mathcal{H}}_{\star}} \cup (S_{Z^{\mathcal{H}}_{\star}}, 1, u))|\mathcal{H}_1 = \mathcal{H}]$.
\end{proof}

\begin{proof}[Proof of Lemma~\ref{lemmaB2:noisy_observations}.]
To prove the lemma it suffices to show that 
\begin{equation}\label{preoperty:regret_of_pi_star_is_the_for_any_two_histories_with_the_same_set_of_known_products_attraction_parameters_single_entrant_noisy}
    \textsc{Reg}(\pi^{\star};\mathcal{H}^1) = \textsc{Reg}(\pi^{\star};\mathcal{H}^2) \text{ for any $\mathcal{H}^1$ and $\mathcal{H}^2$ with $\mathcal{C}(\mathcal{H}^1) =\mathcal{C}(\mathcal{H}^2)$}.
\end{equation}
We prove \eqref{preoperty:regret_of_pi_star_is_the_for_any_two_histories_with_the_same_set_of_known_products_attraction_parameters_single_entrant_noisy} by induction on the number of reviews remaining until the entrant is known given $\mathcal{H}^1$ and $\mathcal{H}^2$. 

\paragraph{Base case.} The entrant has at least $k$ reviews under $\mathcal{H}^1$ and $\mathcal{H}^2$; thus both $\mathcal{H}^1$ and $\mathcal{H}^2$ are terminal. Therefore,  $\textsc{Reg}(\pi^{\star};\mathcal{H}^1) = \textsc{Reg}(\pi^{\star};\mathcal{H}^2) = 0$ by Lemma~\ref{lemma:any_policy_has_a_nonnegative_regret_starting_from_a_terminal_history_and_pi_star_as_a_zero_regret_starting_from_a_terminal_history}.

\paragraph{Induction step.}
Suppose \eqref{preoperty:regret_of_pi_star_is_the_for_any_two_histories_with_the_same_set_of_known_products_attraction_parameters_single_entrant_noisy} holds for any two histories $\tilde{\mathcal{H}}^1$ and $\tilde{\mathcal{H}}^2$ with $\mathcal{C}(\tilde{\mathcal{H}}^1) = \mathcal{C}(\tilde{\mathcal{H}}^2)$ and at most $k$ unknown entrants left given $\tilde{\mathcal{H}}^1$ and $\tilde{\mathcal{H}}^2$. Let $\mathcal{H}^1$ and $\mathcal{H}^2$ be such that $\mathcal{C}(\mathcal{H}^1) = \mathcal{C}(\mathcal{H}^2)$ and there are $k+1$ unknown entrants left given $\mathcal{H}^1$ and $\mathcal{H}^2$. We consider cases based on whether $\mathcal{H}^1$ is terminal. 
\begin{itemize}
    \item \textbf{$\mathcal{H}^1$ is terminal.} Then the entrant is known under $\mathcal{H}^1$. Thus, the entrant is also known under $\mathcal{H}^2$. As a result, both both $\mathcal{H}^1$ and $\mathcal{H}^2$ are terminal and Lemma~\ref{lemma:any_policy_has_a_nonnegative_regret_starting_from_a_terminal_history_and_pi_star_as_a_zero_regret_starting_from_a_terminal_history} implies $\textsc{Reg}(\pi^{\star};\mathcal{H}^1) = \textsc{Reg}(\pi^{\star};\mathcal{H}^2) = 0$
    \item  \textbf{$\mathcal{H}^1$ is not terminal} This case is completed following the same steps as in the case where $\mathcal{H}^1$ is not terminal in the proof of Lemma~\ref{lemma:the_regret_of_pi_star_depends_on_H_only_through_J_H} (Appendix~\ref{appendix_subsec:proof_lemma_optimal_regret_depends_only_on_known_products_attraction_parameters}) with several small modifications. First, Lemma~\ref{lemma:for_any_history_H_the_expected_ex_post_optimum_the_maximum_revenue_from_known_products_and_the_minimum_epoch_regret_depend_on_H_only_through_J(H)} is replaced with Lemma~\ref{lemma:B8_single_entrant_noisy} and Lemma~\ref{lemma:the_regret_of_pi_star_can_be_decomposed_into_the_first_epoch_regret_and_the_future_regret_general_form} is replaced with Lemma~\ref{lemma:B9_single_entrant_noisy}. Second, the set of known products' attraction parameter $J(\mathcal{H}^1)$ and $J(\mathcal{H}^1)$ are replaced with the entrant reviews $\mathcal{C}(\mathcal{H}^1)$ and $\mathcal{C}(\mathcal{H}^2)$. Third, instead of $Z^1$ being the first round at which an unknown entrant is purchased starting from $\mathcal{H}^1$, $Z^1$ is the first round at which the next entrant review is obtained starting from $\mathcal{H}^1$. Fourth, similar to the proof of Lemma~\ref{lemma:B9_single_entrant_noisy} the term $\expect\limits_{\mathcal{H}_{Z^1}, S_{Z^1}, Y_{Z^1}, w_{Y_{Z^1}}} [\textsc{Reg}(\pi^{\star};\mathcal{H}_{Z^1} \cup (S_{Z^1}, Y_{Z^1}, w_{Y_{Z^1}}))|\mathcal{H}_1 = \mathcal{H}^1] $
    is replaced by 
    $\expect_{u \sim \bern(w_1)}[\textsc{Reg}(\pi^{\star};\mathcal{H}_{Z^{1}} \cup (S_{Z^{1}}, 1, u))|\mathcal{H}_1 = \mathcal{H}^1]$. Finally, the future regrets after the first unknown entrant purchase given $\mathcal{H}^1$ and $\mathcal{H}^2$, $\expect\limits_{w \sim \mathcal{F}} [f(J(\mathcal{H}^1) \cup \{w\})]$ and $\expect\limits_{w \sim \mathcal{F}} [f(J(\mathcal{H}^1) \cup \{w\})]$, are replaced with the future regrets after the next entrant review given $\mathcal{H}^1$ and $\mathcal{H}^2$, $\expect\limits_{u \sim \bern(w_1)} [f(\mathcal{C}(\mathcal{H}^1) \cup \{u\})]$ and $\expect\limits_{u \sim \bern(w_1)} [f(\mathcal{C}(\mathcal{H}^1) \cup \{u\})]$.
\end{itemize}
\end{proof}

\subsection{Extended one-epoch-lookahead decomposition (Lemma~\ref{lemma:the_regret_of_pi_star_can_be_decomposed_into_the_first_epoch_regret_and_the_future_regret_single_entrant_noisy})}\label{appendix_subsec:lemma_B3_single_entrant_noisy}
\begin{proof}[Proof of Lemma~\ref{lemma:the_regret_of_pi_star_can_be_decomposed_into_the_first_epoch_regret_and_the_future_regret_single_entrant_noisy}.]
The proof uses the same steps as the proof of Lemma~\ref{lemma:the_regret_of_pi_star_can_be_decomposed_into_the_first_epoch_regret_and_the_future_regret} (in Appendix~\ref{appendix_subsec:proof_lemma_one_epoch_lookahead_decomposition_optimal_policy}) with several modifications. First, Lemma~\ref{lemma:the_regret_of_pi_star_depends_on_H_only_through_J_H} is replaced with Lemma~\ref{lemmaB2:noisy_observations} and Lemma~\ref{lemma:the_regret_of_pi_star_can_be_decomposed_into_the_first_epoch_regret_and_the_future_regret_general_form} is replaced with Lemma~\ref{lemma:B9_single_entrant_noisy}. Second, $\textsc{Reg}(\pi^{\star},J(\mathcal{H}))$ is replaced with $\textsc{Reg}(\pi^{\star},\mathcal{C}(\mathcal{H}))$. Third, as in the proof of Lemma~\ref{lemma:B9_single_entrant_noisy}, in the future regret term $\expect[\textsc{Reg}(\pi^{\star};\mathcal{H}_{Z^{\mathcal{H}}_{\star}} \cup (S_{Z^{\mathcal{H}}_{\star}}, Y_{Z^{\mathcal{H}}_{\star}}, w_{Y_{Z^{\mathcal{H}}_{\star}}}))|\mathcal{H}_1 = \mathcal{H}]$, the unknown entrant purchase         is replaced with the next entrant review, i.e., 
        $\expect_{u \sim \bern(w_1)}[\textsc{Reg}(\pi^{\star};\mathcal{H}_{Z^{\mathcal{H}}_{\star}} \cup (S_{Z^{\mathcal{H}}_{\star}}, 1, u))|\mathcal{H}_1 = \mathcal{H}]$. Similarly, the term $\expect_{w \sim \mathcal{F}}[\textsc{Reg}(\pi^{\star}, J(\mathcal{H}) \cup \{w\})|\mathcal{H}_1 = \mathcal{H}]$ is replaced with $\expect_{u \sim \bern(w_1)}[\textsc{Reg}(\pi^{\star}, \mathcal{C}(\mathcal{H}) \cup \{u\})|\mathcal{H}_1 = \mathcal{H}]$.
\end{proof}

\subsection{Static assortment policies minimize epoch regret (Lemma~\ref{lemma:epoch_regret_of_any_policy_is_at_least_the_minimum_epoch_regret_over_fixed_assortments_single_entrant_noisy})}\label{appendix_subsec:proof_lemma_B4_single_entrant_noisy}

\begin{proof}[Proof of Lemma~\ref{lemma:epoch_regret_of_any_policy_is_at_least_the_minimum_epoch_regret_over_fixed_assortments_single_entrant_noisy}.]
The proof follows the same steps as the one of Lemma~\ref{lemma:epoch_regret_of_any_policy_is_at_least_the_minimum_epoch_regret_over_fixed_assortments} (in Appendix~\ref{appendix_subsec:proof_lemma_epoch_regret_of_any_policy_is_at_least_minimum_epoch_regret_of_static_assortment_policies}) with two modifications. First, the upper bound $\expect[N(S)|\mathcal{H}_1 = \mathcal{H}\Big] \leq \frac{\sum_{i\in S} w_i(\mathcal{H}) + w_0}{h(\mathcal{H})}$ for $S \in \mathcal{E}(\mathcal{H})$ is replaced with $\expect[N(S)|\mathcal{H}_1 = \mathcal{H}\Big] \leq \frac{\sum_{i\in S} w_i(\mathcal{H}) + w_0}{\underline{h}}$ for $S \in \mathcal{E}$ as the entrant is purchased with probability at least $\frac{\underline{h}}{\sum_{i\in S} w_i(\mathcal{H}) + w_0}$ each time $S$ is offered. Second, $\mathcal{E}(\mathcal{H})$ is replaced with $\mathcal{E}$.
\end{proof}

\subsection{Optimal policy regret is bounded independently of the history (Lemma~\ref{lemmaB3:noisy_observations})}\label{appendix_subsec:regret_of_optimal_policy_is_bounded_below_and_above_single_entrant_noisy}
\begin{proof}[Proof of Lemma~\ref{lemmaB3:noisy_observations}.]
The proof follows the same steps as the proof of Lemma~\ref{lemma:for_any_history_with_k_unknown_entrants_the_regret_of_pi_star_is_lower_bounded_and_upper_bounded_by_quantities_which_depend_on_k} (in Appendix~\ref{appendix_subsec:regret_of_optimal_policy_is_bounded_below_and_above}) with several modifications. First, Lemma~\ref{lemma:the_regret_of_pi_star_can_be_decomposed_into_the_first_epoch_regret_and_the_future_regret_general_form} is replaced with Lemma~\ref{lemma:B9_single_entrant_noisy}. Second, instead of showing $\textsc{Reg}(\pi^{\star};\mathcal{H}) \geq -k$ and $\textsc{Reg}(\pi^{\star};\mathcal{H}) \leq \frac{k}{h(\mathcal{F})}$ for any history $\mathcal{H}$ with $k$ unknown entrants, we show $\textsc{Reg}(\pi^{\star};\mathcal{H}) \geq -k$ and $\textsc{Reg}(\pi^{\star};\mathcal{H}) \leq \frac{k}{\underline{h}}$ for any history $\mathcal{H}$ with $k$ reviews remaining until the entrant is known. Third, as in the proof of Lemma~\ref{lemma:B9_single_entrant_noisy}, in the future regret term $\expect[\textsc{Reg}(\pi^{\star};\mathcal{H}_{Z^{\mathcal{H}}_{\star}} \cup (S_{Z^{\mathcal{H}}_{\star}}, Y_{Z^{\mathcal{H}}_{\star}}, w_{Y_{Z^{\mathcal{H}}_{\star}}}))|\mathcal{H}_1 = \mathcal{H}]$, the unknown entrant purchase is replaced with the next entrant review, i.e., $\expect_{u \sim \bern(w_1)}[\textsc{Reg}(\pi^{\star};\mathcal{H}_{Z^{\mathcal{H}}_{\star}} \cup (S_{Z^{\mathcal{H}}_{\star}}, 1, u))|\mathcal{H}_1 = \mathcal{H}]$. Similarly, the term $\expect_{w \sim \mathcal{F}}[\textsc{Reg}(\pi^{\star}, J(\mathcal{H}) \cup \{w\})|\mathcal{H}_1 = \mathcal{H}]$ is replaced with $\expect_{u \sim \bern(w_1)}[\textsc{Reg}(\pi^{\star}, \mathcal{C}(\mathcal{H}) \cup \{u\})|\mathcal{H}_1 = \mathcal{H}]$. Fourth, in inequality \eqref{ineq:epoch_reg_upper_bound}, the last inequality $\frac{1}{\sum_{i \in S^{\textsc{u}}} w_i(\mathcal{H})} \leq \frac{1}{h(\mathcal{H})}$ is replaced with $\frac{1}{\sum_{i \in S^{\textsc{u}}} w_i(\mathcal{H})} \leq \frac{1}{\underline{h}}$. \end{proof}

\end{document}